\documentclass[12pt,a4paper]{article}
\usepackage{color, amsfonts, amsmath,  amsthm, amssymb, nccmath, bm, float, graphicx, physics, multirow, BOONDOX-cal}
\usepackage{mathtools}  

\usepackage{enumerate}
\usepackage{enumitem}

\usepackage{natbib}
\usepackage{url} 
\usepackage[margin=1in]{geometry}

\usepackage{setspace}
\DeclareMathOperator*{\argmin}{\arg\!\min}
\DeclareMathOperator*{\argmax}{\arg\!\max}

\DeclareMathOperator*{\diag}{\normalfont\textrm{diag}}

\allowdisplaybreaks[4]

\newtheorem{corollary}{Corollary}

\newtheorem{lemma}{Lemma}

\newtheorem{theorem}{Theorem}
\newtheorem{assumption}{Assumption}

\numberwithin{corollary}{section}
\numberwithin{definition}{section}
\numberwithin{equation}{section}
\numberwithin{lemma}{section}
\numberwithin{proposition}{section}
\numberwithin{remark}{section}
\numberwithin{theorem}{section}

\begin{document}

 \begin{titlepage}
   \begin{center}
{\large \textbf{Time-Varying Vector Error-Correction Models:\\ Estimation and Inference}}

    \bigskip
    
$^{\ast}${\sc Jiti Gao} and $^{\ast}${\sc Bin Peng} and $^{\ddag}${\sc Yayi Yan}
\bigskip

$^{\ast}$Department of Econometrics and Business Statistics,  Monash University\\
\medskip   
    
$^{\dag}$School of Statistics and Management, Shanghai University of Finance and Economics
    
 \bigskip
    
\today
\end{center}
  \medskip

\begin{abstract}
This paper considers a time-varying vector error-correction model that allows for different time series behaviours (e.g., unit-root and locally stationary processes) to interact with each other to co-exist. From practical perspectives, this framework can be used to estimate shifts in the predictability of non-stationary variables, test whether economic theories hold periodically, etc. We first develop a time-varying Granger Representation Theorem, which facilitates the establishment of asymptotic properties for the model, and then propose estimation and inferential methods and theory for both short-run and long-run coefficients. We also propose an information criterion to estimate the lag length, a singular-value ratio test to determine the cointegration rank, and a hypothesis test to examine the parameter stability. To validate the  theoretical findings, we conduct extensive simulations. Finally, we demonstrate the empirical relevance by applying the framework to investigate the rational expectations hypothesis of the U.S. term structure.

\medskip

\noindent \textbf{Keywords:} Cointegration, Gaussian Approximations, Granger Representation Theorem, Iterated Time-Varying Functions, Term Structure of Interest Rates

\medskip
\medskip

\noindent%



\end{abstract}
 \end{titlepage}

\section{Introduction}\label{Sec1}

Vector error-correction models (VECM) have been widely used in practice to study the short-run dynamics and long-run equilibrium of multiple non-stationary time series (e.g., \citealp{barigozzi2022inference}). These models have proven to be useful in forecasting non-stationary time series and testing the validity of various hypotheses and theories, such as the rational expectations hypothesis of the term structure (\citealp{bauer2020interest}), money demand theory (\citealp{benati2021international}), real business-cycle theory (\citealp{king1991stochastic}), and more. However, it is important to note that all of these studies are based on the assumption of constant parameters. In practice, parameters may evolve over time, and failure to take these changes into account leads to incorrect policy implications and predictions (\citealp[p. 15]{fan2003nonlinear}).

To model the changes over time, various parametric VECM models  (e.g., \citealp{hansen2002testing,hansen2003structural,bergamelli2019combining}) have been proposed to allow for abrupt structural breaks. However, for this line of research, the corresponding estimation and inferential theory has not been well established, which can undermine the reliability of these models, perhaps because the relevant empirical process theory does not typically apply to this case when the data generating process exhibits some uncertain information (e.g., unknown structural breaks and dates) and involves different time series behaviours (e.g., unit-root and piecewise stationary processes). For instance, \cite{hansen2002testing} propose a method to test for the presence of a threshold effect in the VECM model with a single cointegration vector, but do not establish the corresponding estimation theory. Similarly, \cite{hansen2003structural} assumes the change points and the number of cointegration relations are given, while \cite{bergamelli2019combining} consider testing the structural breaks of a VECM model with unknown break dates but do not provide estimation and inferential theories for model parameters. What is more, model misspecification and parameter instability may undermine the performance of parametric time-varying VECM models.  \cite{hansen2001new} notes that it is unlikely that a structural break would be immediate and that it might be more reasonable to allow for a structural change to take a period of time to take effect. Therefore, models that allow for smooth changes over time may provide a more accurate representation of the dynamics in reality.

Having said that, we consider the time-varying vector error-correction model (VECM):
\begin{equation}\label{Eq2.1}
	\Delta \mathbf{y}_t = \bm{\Pi}(\tau_t)\mathbf{y}_{t-1} + \sum_{j=1}^{p_0-1}\bm{\Gamma}_j(\tau_t)\Delta \mathbf{y}_{t-j} + \mathbf{u}_t,\quad \mathbf{u}_t = \bm{\omega}(\tau_t)\bm{\varepsilon}_t
\end{equation}
for $1\leq t\leq T$, where $\bm{\Pi}(\tau) = \bm{\alpha}(\tau)\bm{\beta}^\top$, $\bm{\alpha}(\tau)$ is the $d\times r_0$ adjustment coefficients, $\bm{\beta}$ is the $d\times r_0$ cointegration matrix, and $r_0$ is the cointegration rank to be determined by the dataset. Here, $\bm{\omega}(\cdot)$ governs the time-varying dynamics of the covariance matrix of $\{\mathbf{u}_t\}$, so it characterizes the permanent changes in unconditional volatility. Here, $\{ \mathbf{y}_t \}$ and $\{\Delta\mathbf{y}_t \}$ naturally connect two types of nonstationarity using a time-varying setup. Note that we assume the cointegration matrix $\bm{\beta}$ to be time-invariant, since there are ample empirical evidence showing that the short-run dynamics should be time-varying, while the long-run relationship between economic variables are quite stale (e.g., the expectations hypothesis of the term structure in \citealp{hansen2003structural}; long-run money demand theory in  \citealp{benati2021international}). In some specific cases, such as testing the present-value theory for stock returns  (e.g., \citealp{campbell1987cointegration}), the cointegration matrix is even known a priori. Mathematically, the decomposition of $\bm{\Pi}(\tau) = \bm{\alpha}(\tau)\bm{\beta}^\top$ is simply an identification restriction. Without any structure, it is impossible to identify the elements in the decomposition of  $\bm{\Pi}(\tau)$ due to the multiplication form, so we go along with the existing literature by retaining that $\bm{\beta}$ is a time--invariant matrix.

Having presented our model in \eqref{Eq2.1}, we comment on a challenge from the methodological perspective. In the extant literature of time-varying dynamic models, one often adopts the so-called ``stationary approximation'' technique in order to find a weakly dependent stationary approximation for time-varying dynamical processes before being able to establish asymptotics  (e.g., \citealp{dahlhaus1996kullback,chandler2012mode,zhang2012inference} on time-varying AR models; \citealp{dahlhaus2009empirical} on time-varying ARMA models; \citealp{dahlhaus2006statistical,truquet2017parameter} on time-varying ARCH models; \citealp{karmakar2021simultaneous} on time-varying AR-ARCH models; \citealp{gao2022time} on time-varying VARMA-GARCH models). However, for \eqref{Eq2.1}, $\Delta\mathbf{y}_t$ is expressed in terms of iterated time-varying functions, so the properties of $\Delta\mathbf{y}_t$ involve the integrated processes $\{\mathbf{y}_s\}_{s<t}$ and also depend on infinity past points $\{s/T\}_{s<t}$. As a consequence, how to approximate  $\{\Delta\mathbf{y}_t\}$ and  $\{ \mathbf{y}_t\}$ becomes challenging and the extant empirical process theories for stationary or locally stationary processes do not apply in this case, to the best of our knowledge.

In view of the aforementioned issues, in this study our contributions to the literature are as follows:

\begin{enumerate}
\item We first develop a time-varying version of Granger Representation Theorem for model \eqref{Eq2.1}, which indicates that $\Delta\mathbf{y}_t$ can be approximated by a time-varying vector moving average infinity (VMA($\infty$)) process, and thus $\mathbf{y}_t$ can be approximated by a partial sum of time-varying VMA($\infty$) processes;

\item We then provide a comprehensive study on time-varying VMA($\infty$) processes, including the Nagaev-type inequality, Gaussian approximation and the limit theorem for quadratic forms. With these fundamental results in hand, we are able to establish an estimation theory and the corresponding properties for both short-run and long-run coefficients.

\item In addition, we propose an information criterion to estimate the lag length, a singular-value ratio test to determine the cointegration rank, and a hypothesis test to examine the parameter stability.
\end{enumerate}

In our empirical analysis, we utilize the newly developed framework to investigate the rational expectations hypothesis of the interest rate term structure. Our results indicate that (1). the predictability of the term structure exhibits significant time-varying behaviour, and (2). the expectations hypothesis of the term structure holds periodically, particularly during periods of unusually high inflation. These findings lend support to the work of \cite{andreasen2021yield}, who propose a macro-finance model of the term structure to explain variations in bond return predictability and identify the role of Federal Reserve monetary policy in stabilizing inflation as a key factor.

The rest of the paper is organized as follows. Section \ref{Sec2} presents the estimation methodology and theory. In Section \ref{Sec3}, we conduct extensive simulation studies to examine the theoretical findings. Section \ref{Sec4} investigates the rational expectations hypothesis of the U.S. term structure. Section \ref{Sec5} concludes. All proofs are collected in the online supplementary Appendix A of the paper.

Before proceeding further, it is convenient to introduce some notations: $|\cdot|$ denotes the absolute value of a scalar or the spectral norm of a matrix; for a random variable $\mathbf{v}$, let $\|\mathbf{v}\|_q= (E|\mathbf{v}|^q )^{1/q}$ for $q\ge 1$; $\otimes$ denotes the Kronecker product; $\mathbf{I}_a$ stands for an $a\times a$ identity matrix; $\mathbf{0}_{a\times b}$ stands for an $a\times b$ matrix of zeros, and we write $\mathbf{0}_a$ for short when $a=b$; for a function $g(w)$, let $g^{(j)}(w)$ be the $j^{th}$ derivative of $g(w)$, where $j\ge 0$ and $g^{(0)}(w) \equiv g(w)$; let $\tilde{c}_k =\int_{-1}^{1} u^k K(u) \mathrm{d}u$ and $\tilde{v}_k= \int_{-1}^{1} u^k K^2(u) \mathrm{d}u$ for integer $k\ge 0$; $\mathrm{vec}(\cdot)$ stacks the elements of an $m\times n$ matrix as an $mn \times 1$ vector; $\mathbf{A}^{+}$ denotes the Moore-Penrose (MP) inverse of matrix $\mathbf{A}$; for a matrix $\mathbf{A}$ with full column rank, we let $\mathbf{P}_{\mathbf{A}} = \mathbf{A}(\mathbf{A}^\top \mathbf{A})^{-1}\mathbf{A}^\top$ and  $\overline{\mathbf{A}} = \mathbf{A}(\mathbf{A}^\top\mathbf{A})^{-1}$; for $m\geq n$, we denote by $\mathbf{M}_{\perp}$ an orthogonal matrix complement of the $m\times n$ matrix $\mathbf{M}$ with $\mathrm{Rank}(\mathbf{M}) = n$;  $\to_P$, $\to_D$ and $\Rightarrow$ denote convergence in probability, convergence in distribution, weak convergence with respect to the uniform metric; let $\mathbf{W}_d(\cdot,\bm{\Sigma})$ be a $d$-dimensional Brownian motion with covariance matrix $\bm{\Sigma}$.

\section{The Methodology and Asymptotics}\label{Sec2}

In what follows, Section \ref{Sec2.1} studies the properties of $\{\mathbf{y}_t \}$ and  $\{\Delta\mathbf{y}_t \}$, and provide (integrated) time-varying VMA($\infty$) approximations for both processes. In Section \ref{Sec2.2}, we assume $r_0$ and $p_0$ are known for simplicity, and consider the estimation of $\bm{\alpha}(\cdot)$ and $\bm{\beta}$. Section \ref{Sec2.3} proposes an information criterion to estimate the lag length (i.e., $p_0$), and a singular-value ratio test to determine the cointegration rank (i.e., $r_0$). Finally, Section \ref{Sec2.5} gives a parameter stability test, which provides statistical evidence to support the necessity of time-varying VECM models practically.

\subsection{Approximations of $\{\mathbf{y}_t \}$ and $\{\Delta\mathbf{y}_t \}$}\label{Sec2.1}

To investigate \eqref{Eq2.1}, the first obstacle lies in the complexity of the dependence structure of $\Delta\mathbf{y}_t$, which is expressed in terms of iterated time-varying functions with infinity memory. This implies that the population properites of $\Delta\mathbf{y}_t$ involve the integrated processes $\{\mathbf{y}_s\}_{s<t}$ and also depend on infinity past points $\{s/T\}_{s<t}$. As a result, the extant empirical process theory does not apply. To address this issue, we initiate our analysis by seeking local (not necessarily stationary) approximations for each $\Delta\mathbf{y}_t$ and $\mathbf{y}_t$, and then establish some necessary asymptotic properties for these processes.

We now introduce some necessary assumptions.

\begin{assumption}\label{Ass1}
\item
	\begin{enumerate}[wide, labelwidth=!, labelindent=0pt]
		\item[1.] Define $\mathbf{C}_{\tau}(L) = (1-L)\mathbf{I}_d - \bm{\alpha}(\tau)\bm{\beta}^\top L - \sum_{j=1}^{p_0-1}\bm{\Gamma}_j(\tau)(1-L)L^j$. 
Suppose that
\begin{enumerate}
\item[a.] $\mathrm{det}(\mathbf{C}_{\tau}(L)) = 0$ if and only if $|L| > 1$ or $L = 1$ uniformly over $\tau \in [0,1]$; the number of unit roots, $L = 1$, is exactly $d-r_0$;

\item[b.] $\mathrm{Rank}(\bm{\alpha}(\tau)) =r_0$ uniformly over $\tau \in [0,1]$, and $\mathrm{Rank}(\bm{\beta}) = r_0$;

\item[c.] $\bm{\alpha}_{\perp}^\top(\tau)\left[\mathbf{I}_d - \sum_{i=1}^{p_0-1}\bm{\Gamma}_i(\tau)\right]\bm{\beta}_{\perp}$ is nonsingular for each given $\tau \in [0,1]$.

\end{enumerate}

\item[2.]  Suppose that

\begin{enumerate}
\item[a.] Each element of $\bm{\omega}(\tau)$, $\bm{\alpha}(\tau)$ and $\bm{\Gamma}_j(\tau)$ with $j=1,\ldots, p_0-1$ is third order continuously differentiable on $[0,1]$;

\item[b.] $\mathbf{y}_0 = O_P(1)$, and,  for $\tau<0$, $\bm{\omega}(\tau)=\bm{\omega}(0)$, $\bm{\alpha}(\tau)=\bm{\alpha}(0)$, and $ \bm{\Gamma}_{j}(\tau)  = \bm{\Gamma}_{j}(0)$ with $j=1,\ldots, p_0-1$.
\end{enumerate}		
		
\item[3.] $\{\bm{\varepsilon}_t\}$ is a sequence of martingale differences such that $E(\bm{\varepsilon}_t\bm{\varepsilon}_t^\top\mid \mathcal{F}_{t-1}) = \mathbf{I}_d$ almost surely (a.s.), where $\mathcal{F}_{t} = \sigma(\bm{\varepsilon}_t,\bm{\varepsilon}_{t-1},\ldots)$, and $\max_t\|\bm{\varepsilon}_t\|_\delta < \infty$ for some $\delta > 4$; $\bm{\Omega}(\tau)=\bm{\omega}(\tau)\bm{\omega}(\tau)^\top >0$ is uniformly over $\tau \in [0,1]$.
	\end{enumerate}
\end{assumption}

Assumption \ref{Ass1}.1 explicitly excludes explosive processes, and ensures that $\mathbf{y}_t$ is an integrated process of order one with $d-r_0$ common unit root components, as well as $r_0$ cointegration relationships. In Assumption \ref{Ass1}.1.c,  $\bm{\alpha}_{\perp}^\top(\tau) [\mathbf{I}_d - \sum_{i=1}^{p_0-1}\bm{\Gamma}_i(\tau) ]\bm{\beta}_{\perp}$ being nonsingular ensures the existence of a time-varying Granger Representation Theorem for $\mathbf{y}_t$, which facilitates the asymptotic development later on.

Assumption \ref{Ass1}.2.a allows the underlying data generating process to evolve over time in a smooth manner.   Assumption \ref{Ass1}.2.b regulates the behaviour of $\mathbf{y}_t$ for $t\leq 0$, which is standard in the literature of locally stationary models (e.g., \citealp{vogt2012nonparametric}) and unit-root processes (e.g., \citealp{li2019kernel}). 

Assumption \ref{Ass1}.3 imposes conditions on the error innovations, which are standard in the time series literature (e.g., \citealp{lutkepohl2005new}).

\medskip

With these conditions in hand, we are able to provide local approximations for each $\Delta\mathbf{y}_t$ and $\mathbf{y}_t$ with $t \geq 1$. For notational simplicity, we denote  

\begin{eqnarray*}
 \mathbf{B}_\tau(L) =\mathbf{J}^{-1}\mathbf{B}_\tau^*(L)\mathbf{J} \quad \text{and}\quad \mathbf{B}_\tau^*(L) = \mathbf{J}[\bm{\Gamma}_\tau(L)\overline{\bm{\beta}}(1-L)-\bm{\alpha}(\tau)L,\ \bm{\Gamma}_{\tau}(L)\bm{\beta}_{\perp}]
\end{eqnarray*}		
in which	$\mathbf{J} = [\bm{\beta},\overline{\bm{\beta}}_{\perp}]^\top$ and $\bm{\Gamma}_{\tau}(L)=\mathbf{I}_d - \sum_{i=1}^{p_0-1}\bm{\Gamma}_i(\tau)L^i$.

\medskip

\begin{lemma}\label{L1}		
Let Assumption \ref{Ass1} hold.

\begin{enumerate}[wide, labelwidth=!, labelindent=0pt]
\item[1.] We obtain that $\mathrm{det}\left(\mathbf{B}_\tau(L)\right)\neq 0$ for all $|L|\leq 1$, and $\mathbf{B}_\tau(L) $ admits an expression of the form: $\mathbf{B}_\tau(L)  =\mathbf{I}_d - \sum_{i=1}^{p}\mathbf{B}_i(\tau)L^i $, thus we can denote   
\begin{eqnarray*}
\mathbf{B}_\tau^{-1}(L):= \bm{\Psi}_\tau(L)=\sum_{j=0}^{\infty}\bm{\Psi}_j(\tau) L^j 
\end{eqnarray*}
that satisfies $\mathbf{P}_{\bm{\beta}_{\perp}}\bm{\Psi}_\tau(1) = \bm{\beta}_{\perp} [\bm{\alpha}_{\perp}^\top(\tau)\bm{\Gamma}_{\tau}(1) \bm{\beta}_{\perp}]^{-1} \bm{\alpha}_{\perp}^\top(\tau).$

\item[2.] Equation \eqref{Eq2.1} admits the following representation:
			\begin{eqnarray*}
			\mathbf{y}_t =\mathbf{P}_{ \bm{\beta}_{\perp}}\sum_{j=1}^{t}\mathbf{z}_j +\mathbf{P}_{\bm{\beta}} \mathbf{z}_t + \mathbf{P}_{ \bm{\beta}_{\perp}} \mathbf{y}_0 \quad \text{in which}\quad	\mathbf{z}_t = \sum_{i=1}^{p}\mathbf{B}_i(\tau_t)\mathbf{z}_{t-i} + \mathbf{u}_t.
			\end{eqnarray*} 
	
\item[3.]   For any $\tau \in[0,1]$, model (\ref{Eq2.1}) can be approximated by
	\begin{equation*}
		\Delta \widetilde{\mathbf{y}}_t(\tau) =  \bm{\Pi}(\tau) \widetilde{\mathbf{y}}_{t-1}(\tau) + \sum_{j=1}^{p_0-1}\bm{\Gamma}_j(\tau)\Delta \widetilde{\mathbf{y}}_{t-j}(\tau) + \widetilde{\mathbf{u}}_t(\tau)\quad\text{with}\quad \widetilde{\mathbf{u}}_t(\tau) = \bm{\omega}(\tau)\bm{\varepsilon}_t,
	\end{equation*}
	which  admits the following representation:
			\begin{eqnarray*}
				\widetilde{\mathbf{y}}_t(\tau) &=&\mathbf{P}_{ \bm{\beta}_{\perp}}\bm{\Psi}_\tau(1)\sum_{i=1}^{t}\widetilde{\mathbf{u}}_i(\tau) +\mathbf{P}_{ \bm{\beta} } [\bm{\Psi}_\tau(1)\widetilde{\mathbf{u}}_t(\tau) +  \bm{\Psi}_\tau^*(L)\widetilde{\mathbf{u}}_{t-1}(\tau)] - \bm{\Psi}_\tau^*(L) \widetilde{\mathbf{u}}_t(\tau) +\widetilde{\mathbf{y}}_0^*(\tau),\\
				\widetilde{\mathbf{y}}_0^*(\tau) &=&\mathbf{P}_{ \bm{\beta}_{\perp}} \bm{\Psi}_\tau^*(L)\widetilde{\mathbf{u}}_0(\tau) +\mathbf{P}_{ \bm{\beta}_{\perp}} \widetilde{\mathbf{y}}_0(\tau),
			\end{eqnarray*}
			
				\noindent such that
	\begin{enumerate}
		\item[a.]  $\max_{1\leq t \leq T} \|\Delta\mathbf{y}_t-\Delta\widetilde{\mathbf{y}}_t(\tau_t)\|_\delta  = O(1/T)$,
		
		\item[b.] $\sup_{\tau,\tau'\in[0,1]}\|\Delta\widetilde{\mathbf{y}}_t(\tau)-\Delta\widetilde{\mathbf{y}}_t(\tau')\|_\delta = O(|\tau - \tau'|)$,
		
		\item[c.] $\max_{1\leq t\leq T}\|\mathbf{y}_t - \mathbf{y}_0 - \sum_{j=1}^{t} \Delta\widetilde{\mathbf{y}}_j(\tau_j)\|_\delta = O(1/\sqrt{T})$,
		
		\end{enumerate}		
			where   $\widetilde{\mathbf{y}}_0(\tau)$ is an initial value of $\widetilde{\mathbf{y}}_t(\tau)$,  $\bm{\Psi}_\tau^*(L) = \sum_{j=0}^\infty \bm{\Psi}_j^*(\tau)$, and $\bm{\Psi}_j^*(\tau) = \sum_{i=j+1}^{\infty}\bm{\Psi}_i(\tau)$.
			
			\item[4.]  Finally, we  obtain that
\begin{eqnarray*}
 T^{-1/2}\mathbf{y}_{\lfloor Tu \rfloor} \Rightarrow \mathbf{W}_d(u,\bm{\Sigma}_{\mathbf{y}}(u)) 
\end{eqnarray*}			
uniformly over $u\in [0,1]$,  where $\bm{\Sigma}_{\mathbf{y}}(u) =\mathbf{P}_{ \bm{\beta}_{\perp}}\int_{0}^{u}  \bm{\Psi}_{s}(1) \bm{\Omega}(s) \bm{\Psi}_{s}^\top(1) \mathrm{d}s \mathbf{P}_{ \bm{\beta}_{\perp}} .$
\end{enumerate}
\end{lemma}

Lemma \ref{L1}  presents a local approximation of $\{\Delta \mathbf{y}_t\}$, which is the foundation of our asymptotic developments. Using Lemma \ref{L1} we are able to establish basic properties of $\{\Delta \mathbf{y}_t\}$ and $\{  \mathbf{y}_t\}$. To put it in a nutshell,  studying $\{\Delta \mathbf{y}_t\}$ and $\{  \mathbf{y}_t\}$  directly is technically challenging, we therefore consider their local approximations, which can help us understand $\{\Delta \mathbf{y}_t\}$ and $\{  \mathbf{y}_t\}$ in every small neighbourhood.  In addition, Lemma \ref{L1} indicates that $\widetilde{\mathbf{y}}_t(\tau)$ admits a time-varying version of Granger Representation Theorem (\citealp{johansen1995likelihood}, Theorem 4.2), so that $\Delta \widetilde{\mathbf{y}}_t(\tau)$ has a time-varying VMA($\infty$) representation. In the online supplementary Appendix A, we provide a comprehensive study on time-varying VMA($\infty$) processes (such as the Nagaev-type inequality, Gaussian approximation and the limit theorem for quadratic forms), which allows us to further derive estimation and inferential properties for \eqref{Eq2.1} in the following subsections.

\subsection{Model Estimation}\label{Sec2.2}

We assume $p_0$ and $r_0$ are known in this subsection, and shall come back to work on their estimation in Section \ref{Sec2.3}.

The local linear estimator of $[\bm{\Pi}(\tau),\bm{\Gamma}(\tau)]$ with $\bm{\Gamma}(\tau) = [\bm{\Gamma}_1(\tau),\ldots,\bm{\Gamma}_{p_0-1}(\tau)]$ is given by
\begin{eqnarray*}
	[\widehat{\bm{\Pi}}(\tau), \widehat{\bm{\Gamma}}(\tau)] &=&\left[ 
		\mathbf{V}_{T,0}(\tau) , \mathbf{V}_{T,1}(\tau)\right] \cdot  \left[\begin{matrix}
		\mathbf{S}_{T,0}(\tau) & \mathbf{S}_{T,1}(\tau)\\ \mathbf{S}_{T,1}(\tau) & \mathbf{S}_{T,2}(\tau)
	\end{matrix} \right]^{+} \cdot \left[\begin{matrix}
		\mathbf{I}_{dp_0} \\
		\mathbf{0}_{dp_0}
	\end{matrix} \right],
\end{eqnarray*}
where $\mathbf{h}_t = \left[\mathbf{y}_{t}^\top,\Delta \mathbf{x}_t^\top\right]^\top$,
\begin{eqnarray*}
\mathbf{V}_{T,l}(\tau) &=& \sum_{t=1}^{T}\Delta\mathbf{y}_{t}\mathbf{h}_{t-1}^\top \left(\frac{\tau_t-\tau}{h}\right)^lK\left(\frac{\tau_t-\tau}{h}\right) \text{ for } l =0,1,\\
\mathbf{S}_{T,l}(\tau) &=& \sum_{t=1}^{T}\mathbf{h}_{t-1}\mathbf{h}_{t-1}^\top \left(\frac{\tau_t-\tau}{h}\right)^lK\left(\frac{\tau_t-\tau}{h}\right) \text{ for } l =0,1, 2.
\end{eqnarray*}
Here, we use MP inverse since $\mathbf{S}_{T,l}(\tau)$ is asymptotically singular (cf., Lemma \ref{L4} of the supplementary Appendix A), which is referred to as ``kernel-induced degeneracy'' in the unit-root literature (\citealp{phillips2017estimating}).


Accordingly, the local linear estimator of $\bm{\Omega}(\tau)$ is defined as 
$$
\widehat{\bm{\Omega}}(\tau) = \frac{1}{T}\sum_{t=1}^{T}\widehat{\mathbf{u}}_t\widehat{\mathbf{u}}_t^\top w_t(\tau),
$$
where $\widehat{\mathbf{u}}_t = \Delta \mathbf{y}_t - \widehat{\bm{\Pi}}(\tau_t)\mathbf{y}_{t-1} - \widehat{\bm{\Gamma}}(\tau_t)\Delta \mathbf{x}_{t-1}$, $\Delta \mathbf{x}_{t} = [\Delta \mathbf{y}_{t}^\top,\ldots,\Delta \mathbf{y}_{t-p_0+2}^\top]^\top$,

\begin{eqnarray*}
w_t(\tau) &=& h^{-1}K\left(\frac{\tau_t-\tau}{h}\right)\frac{P_{h,2}(\tau)-\frac{\tau_t-\tau}{h}P_{h,1}(\tau)}{P_{h,0}(\tau)P_{h,2}(\tau)-P_{h,1}^2(\tau)},\\
P_{h,l}(\tau) &=& \frac{1}{Th}\sum_{t=1}^{T}\left(\frac{\tau_t-\tau}{h}\right)^lK\left(\frac{\tau_t-\tau}{h}\right) \text{ for }l=0,1,2.
\end{eqnarray*}

\medskip

To proceed, we require the following conditions to hold.
\begin{assumption}\label{Ass2}
\item 
\begin{enumerate}[wide, labelwidth=!, labelindent=0pt]
\item[1.]  Suppose that $K(\cdot)$ is a symmetric and positive kernel function defined on $[-1,1]$, and $\int_{-1}^{1}K(u)\mathrm{d}u = 1$. Moreover, $K(\cdot)$ is Lipschitz continuous on $[-1,1]$. As $(T,h) \to (\infty, 0)$, $Th\to \infty$.

\item[2.]  $\{\bm{\varepsilon}_t\}$ is a sequence of independent random variables. Let $T^{1-\frac{4}{\delta}}h/(\log T)^{1-\frac{4}{\delta}} \to \infty$ as $T\rightarrow \infty$.
\end{enumerate}

\end{assumption}

Assumption \ref{Ass2}.1 imposes a set of regular conditions on the kernel function and the bandwidth.  On top of Assumption \ref{Ass1}.3, Assumption \ref{Ass2}.2 imposes more structure on $\{\bm{\varepsilon}_t\}$, and is used to derive Gaussian approximation for the sum of time-varying VECM process. This condition can be relaxed if we impose more dependence structure (e.g., nonlinear system theory as in \citealp{wu2005nonlinear}).  Provided $\delta>5$, the usual optimal bandwidth $h_{opt}=O (T^{-1/5} )$ satisfies the condition $T^{1-\frac{4}{\delta}}h/(\log T)^{1-\frac{4}{\delta}} \to \infty$. In addition,  Gaussian approximation together with $T^{1-\frac{4}{\delta}}h/(\log T)^{1-\frac{4}{\delta}} \to \infty$ is used to derive the uniform convergence of our nonparametric estimators, which is further used for asymptotic covariance estimation, semiparametric estimation and model specification testing in the below.

\medskip

With these conditions in hand, we further let

\begin{eqnarray*}
 \widetilde{\mathbf{w}}_t(\tau) = [\widetilde{\mathbf{y}}_t^\top(\tau)\bm{\beta},\Delta\widetilde{\mathbf{x}}_t^\top(\tau)]^\top \quad \text{and}\quad \widetilde{\mathbf{x}}_t(\tau) = [\Delta \widetilde{\mathbf{y}}_{t}^\top(\tau),\ldots,\Delta \widetilde{\mathbf{y}}_{t-p_0+2}^\top(\tau)]^\top,
\end{eqnarray*}
and summarize the asymptotic properties of the local linear estimators in the following theorem.

\begin{theorem}\label{Thm1}
	Let Assumptions \ref{Ass1}--\ref{Ass2} hold.
	\begin{enumerate}[wide, labelwidth=!, labelindent=0pt]
		\item [1.] If $Th^7\to 0$, then for any $\tau \in (0,1)$,
		$$
		\sqrt{Th}\mathrm{vec}\left([\widehat{\bm{\Pi}}(\tau), \widehat{\bm{\Gamma}}(\tau)] - [\bm{\Pi}(\tau), \bm{\Gamma}(\tau)] -\frac{1}{2}h^2\widetilde{c}_2[\bm{\Pi}^{(2)}(\tau), \bm{\Gamma}^{(2)}(\tau)]\right)\to_D N\left(\mathbf{0}, \widetilde{v}_0\bm{\Sigma}_{\mathrm{co}}(\tau) \right),
		$$
		where  $\bm{\Sigma}_{\mathrm{co}}(\tau) =\left(  \diag( \bm{\beta}, \mathbf{I}_{d(p_0-1)})   E(\widetilde{\mathbf{w}}_t(\tau)\widetilde{\mathbf{w}}_t^{\top,-1}(\tau) )\diag( \bm{\beta}, \mathbf{I}_{d(p_0-1)})^\top \right)\otimes \bm{\Omega}(\tau)$;
 
			\item [2.] $\sup_{\tau \in [0,1]} |[\widehat{\bm{\Pi}}(\tau), \widehat{\bm{\Gamma}}(\tau)] - [\bm{\Pi}(\tau), \bm{\Gamma}(\tau)] | = O_P\left(h^2 + \sqrt{\log T/(Th)} \right)$;
			
			\item [3.] $\sup_{\tau \in [0,1]} |\widehat{\bm{\Sigma}}_{\mathrm{co}}(\tau) - \bm{\Sigma}_{\mathrm{co}}(\tau) | = O_P\left(h + \sqrt{\log T/(Th)} \right)$,  \\ where $\widehat{\bm{\Sigma}}_{\mathrm{co}}(\tau) = \sum_{t=1}^{T}K\left(\frac{\tau_t-\tau}{h} \right)\mathbf{S}_{T,0}^+(\tau) \otimes \widehat{\bm{\Omega}}(\tau).$
		\end{enumerate}
	\end{theorem}

Theorem \ref{Thm1} establishes the asymptotic distribution yielded by $[\widehat{\bm{\Pi}}(\tau), \widehat{\bm{\Gamma}}(\tau)]$, and also gives the rates of uniform  convergence which will be extensively used in the following development.

\medskip

In what follows, we consider the estimation of cointegration matrix by utilizing the reduced rank structure of $\bm{\Gamma}(\cdot)$ and using the profile likelihood method  (e.g., \citealp{fan2005profile}). To ensure a unique cointegration matrix, we assume
	\begin{eqnarray}\label{eqbp2}
	\bm{\beta} = \left[\begin{matrix}
		\mathbf{I}_{r_0} \\
		\bm{\beta}^*
	\end{matrix} \right]
	\end{eqnarray}
	where $\bm{\beta}^*$ is a $(d-r_0)\times r_0$ matrix. Using \eqref{eqbp2}, $\widehat{\bm{\alpha}}(\tau)$ is the first $r_0$ columns of $\widehat{\bm{\Pi}}(\tau)$, i.e., $\widehat{\bm{\Pi}}(\tau) = [\widehat{\bm{\alpha}}(\tau) , \widehat{\bm{\Pi}}_2(\tau)],$ where the definition of $\widehat{\bm{\Pi}}_2(\tau)$ is obvious.

Specifically, given $\bm{\Pi}(\tau)$, we can estimate the short-run time-varying parameters $\bm{\Gamma}(\tau)$ by
\begin{eqnarray*}
\widehat{\bm{\Gamma}}(\tau,\bm{\Pi}) &=& \sum_{s=1}^{T}(\Delta\mathbf{y}_{s}-\bm{\Pi}(\tau_s)\mathbf{y}_{s-1})\Delta\mathbf{x}_{s-1}^{*,\top} K\left(\frac{\tau_s-\tau}{h}\right)\\
&&\times\left(\sum_{s=1}^{T}\Delta\mathbf{x}_{s-1}^*\Delta\mathbf{x}_{s-1}^{*,\top} K\left(\frac{\tau_s-\tau}{h}\right)\right)^{-1}\left[\begin{matrix}
	\mathbf{I}_{d(p_0-1)} \\
	\mathbf{0}_{d(p_0-1)}
\end{matrix} \right],
\end{eqnarray*}
where $\Delta\mathbf{x}_{t}^{*} = \Delta\mathbf{x}_{t}\otimes\left[1,\frac{\tau_{t+1}-\tau}{h}\right]^\top$. In connection with \eqref{eqbp2},  we can write
$$
\widetilde{\mathbf{r}}_t(\bm{\alpha}) = \widetilde{\mathbf{R}}_t^\top(\bm{\alpha}) \mathrm{vec}\left(\bm{\beta}^{*,\top}\right) + \mathbf{u}_t^*,
$$
where $\mathbf{u}_t^* = \mathbf{u}_t + \left[\bm{\Gamma}(\tau_t)-\widehat{\bm{\Gamma}}(\tau_t,\bm{\Pi})\right]\Delta \mathbf{x}_{t-1}$, $\mathbf{r}_t(\bm{\alpha}) = \Delta \mathbf{y}_t - \bm{\alpha}(\tau_t)\mathbf{y}_{t-1}^{(1)}$, $\mathbf{y}_{t}^{(1)}$ contains the first $r_0$ elements of $\mathbf{y}_{t}$, $\mathbf{r}^\top(\bm{\alpha}) = [\mathbf{r}_1(\bm{\alpha}),\ldots,\mathbf{r}_T(\bm{\alpha})]$,
$$
\widetilde{\mathbf{r}}_t(\bm{\alpha}) = \mathbf{r}_t(\bm{\alpha}) - \mathbf{r}^\top(\bm{\alpha}) \mathbf{K}(\tau_t) \Delta \mathbf{x}^*\left(\Delta \mathbf{x}^{*,\top}\mathbf{k}(\tau_t) \Delta \mathbf{x}^*\right)^{-1}[\mathbf{I}_{d(p_0-1)},\mathbf{0}_{d(p_0-1)}]^\top\Delta\mathbf{x}_{t-1},
$$
$\mathbf{k}(\tau)= \mathrm{diag}\left[K\left(\frac{\tau_1-\tau}{h}\right),\ldots,K\left(\frac{\tau_T-\tau}{h}\right)\right]$, $\Delta \mathbf{x}^{*,\top} = \left[\Delta \mathbf{x}_0^*,\ldots,\Delta \mathbf{x}_{T-1}^* \right]$, $\mathbf{R}_t(\bm{\alpha}) = \mathbf{y}_{t-1}^{(2)} \otimes \bm{\alpha}^\top(\tau_t)$, $\mathbf{y}_{t}^{(2)}$ contains the last $d-r_0$ elements of $\mathbf{y}_{t}$,
$$
\widetilde{\mathbf{R}}_t(\bm{\alpha}) = \mathbf{R}_t(\bm{\alpha}) -\mathbf{R}^\top(\bm{\alpha})\mathbf{K}(\tau_t)\Delta \mathbf{X}^{*}\left(\Delta \mathbf{X}^{*,\top}\mathbf{K}(\tau_t) \Delta \mathbf{X}^*\right)^{-1}\left[\mathbf{I}_{d^2(p_0-1)},\mathbf{0}_{d^2(p_0-1)} \right]^\top \Delta\mathbf{X}_{t-1},
$$
$\mathbf{R}^\top(\bm{\alpha}) = [\mathbf{R}_1(\bm{\alpha}),\ldots,\mathbf{R}_T(\bm{\alpha})]$, $\Delta\mathbf{X}^{*,\top}=\left[\Delta \mathbf{X}_0^*,\ldots,\Delta \mathbf{X}_{T-1}^*\right]$, $\Delta\mathbf{X}_t = \Delta\mathbf{x}_t \otimes \mathbf{I}_d$, $\Delta\mathbf{X}_t^* = \Delta\mathbf{x}_t^* \otimes \mathbf{I}_d$ and $\mathbf{K}(\tau) = \mathbf{k}(\tau)\otimes \mathbf{I}_d$. Replacing $\bm{\alpha}(\tau)$ with $\widehat{\bm{\alpha}}(\tau)$, the weighted least squares (WLS) estimator of $\bm{\beta}^*$ is given by
\begin{equation} 
	\mathrm{vec}\left[\widehat{\bm{\beta}}^{*,\top}\right] = \left(\sum_{t=1}^{T}\widetilde{\mathbf{R}}_t(\widehat{\bm{\alpha}})\widehat{\bm{\Omega}}^{-1}(\tau_t)\widetilde{\mathbf{R}}_t^\top(\widehat{\bm{\alpha}})\right)^{-1}\sum_{t=1}^{T}\widetilde{\mathbf{R}}_t(\widehat{\bm{\alpha}})\widehat{\bm{\Omega}}^{-1}(\tau_t)\widetilde{\mathbf{r}}_t(\widehat{\bm{\alpha}}).
\end{equation}

	\medskip

	The next theorem summaries the asymptotic distribution associated with $\widehat{\bm{\beta}}^{*}$.
	\begin{theorem}\label{Thm3}
		Let Assumptions \ref{Ass1}--\ref{Ass2} hold. Suppose further that $\frac{Th^2}{(\log T)^2} \to \infty$ and $Th^6 \to 0$, then
{\small
\begin{enumerate}[wide, labelwidth=!, labelindent=0pt]
\item[1.] $	T\mathrm{vec}\left[\widehat{\bm{\beta}}^{*,\top} - \bm{\beta}^{*,\top}\right] \to_D \left(\int_{0}^{1}\mathbf{W}_{d-r_0}(u) \mathbf{W}_{d-r_0}^\top(u) \otimes \bm{\alpha}^\top(u)\bm{\Omega}^{-1}(u)\bm{\alpha}(u)\mathrm{d}u\right)^{-1}    \int_{0}^{1}\mathbf{W}_{d-r_0}(u)\otimes \mathrm{d}\mathbf{W}_{r_0}(u)$,

\item[2.] $\left(\sum_{t=1}^{T}\mathbf{y}_{t-1}^{(2)}\mathbf{y}_{t-1}^{(2),\top}  \otimes \widehat{\bm{\alpha}}^\top(\tau_t)\widehat{\bm{\Omega}}^{-1}(\tau_t) \widehat{\bm{\alpha}}(\tau_t)\right)^{1/2}\mathrm{vec}\left[\widehat{\bm{\beta}}^{*,\top} - \bm{\beta}^{*,\top}\right] \to_D N\left(\mathbf{0}, \mathbf{I}_{(d-r_0)r_0}\right)$,
\end{enumerate}
		where $\mathbf{W}_{d-r_0}(u) = \left[\mathbf{0}_{(d-r_0)\times r_0},\mathbf{I}_{d-r_0}\right]\mathbf{W}_d(u,\bm{\Sigma}_{\mathbf{y}}(u))$, $\mathbf{W}_{r_0}(u) = \mathbf{W}_{r_0}(u,\bm{\alpha}^\top(u)\bm{\Omega}^{-1}(u)\bm{\alpha}(u))$, and $\mathbf{W}_{d-r_0}(\cdot)$ is independent of $\mathbf{W}_{r_0}(\cdot)$. 
		}
	\end{theorem}
	
The first result shows that the cointegration matrix can be estimated at a super consistent rate $T$, while the second result concerning asymptotic normality of a properly standardized from of $\widehat{\bm{\beta}}^* - \bm{\beta}^*$ indicates how to construct confidence interval practically.

\subsection{On Lag Length and Cointegration Rank}\label{Sec2.3}
	
We now consider the estimation of $p_0$ and $r_0$. Specifically, we first propose an information criterion that can  estimate the lag length ($p_0$), and then consider the estimation of cointegration rank ($r_0$). 
	
	
To estimate $p_0$, we minimize an information criterion as follows:
	
	\begin{eqnarray}\label{Eq2.4}
		\widehat{p} =  \argmin_{1\le p\le P }\  \text{IC}(p),
	\end{eqnarray}
	where $
\text{IC}(p) = \log \left\{\text{RSS}(p)\right\}+p\cdot\chi_T$, $\text{RSS}(p)=\frac{1}{T}\sum_{t=1}^{T}\widehat{\mathbf{u}}_{p,t}^\top \widehat{\mathbf{u}}_{p,t}$, $\chi_T$ is the penalty term, $\widehat{\mathbf{u}}_{p,t}$ is the value of $\widehat{\mathbf{u}}_{t}$ by letting the number of lagged differences be $p-1$, and $P$ is a sufficiently large fixed positive integer. The following proposition summaries the asymptotic property of \eqref{Eq2.4}.
	
	\begin{theorem}\label{Prop2.1}
		Let Assumptions \ref{Ass1}--\ref{Ass2} hold. Suppose that $\chi_T\to 0$ and $c_T^{-2}\chi_T\to \infty$, where $c_T=h^2+\left(\frac{\log T}{Th}\right)^{1/2}$. Then $\Pr\left(\widehat{p}=p_0\right)\to 1$ as $T\rightarrow \infty$.
	\end{theorem}
	
It is noteworthy that  Theorem \ref{Prop2.1} does not require any knowledge of cointegration rank. In view of the conditions on $\chi_T$, a natural choice is 
	$$
	\chi_T = \frac{\log(\log(Th))}{3} \left(h^4 +h^2\left(\frac{\log T}{Th}\right)^{1/2} + \frac{\log T}{Th}\right).
	$$

	\medskip
	
	We next consider the estimation of cointegration rank $r_0$. The basic principle of our method is to separate the $r_0$ relevant singular values of $\bm{\Pi}(\tau)$ from the zero ones, while the number of nonzero ones corresponds to the cointegration rank. Before proceeding further, it should be pointed out that the choice of lag length $p_0$ is irrelevant with the determination of cointegration rank since we can set the lag length $P$ as sufficiently large but fixed (i.e., $p_0\leq P$), in which case $\bm{\Gamma}_j(\tau) = \mathbf{0}_d$ for $j=p_0,...,P$.
	
	The method is based on the QR decomposition with column-pivoting of $\int_{0}^{1}\bm{\Pi}^\top(\tau)\mathrm{d}\tau$, i.e.,  $\int_{0}^{1}\bm{\Pi}^\top(\tau)\mathrm{d}\tau = \bm{\beta}\int_{0}^{1}\bm{\alpha}^\top(\tau)\mathrm{d}\tau = \mathbf{S}\mathbf{R}$, where $\mathbf{S}^\top\mathbf{S}=\mathbf{I}_d$, and $\mathbf{R}$ is an upper triangular matrix with the diagonal elements being nonincreasing. The estimator of $\int_{0}^{1}\bm{\Pi}(\tau)\mathrm{d}\tau$ is naturally given by $\int_{0}^{1}\widehat{\bm{\Pi}}(\tau)\mathrm{d}\tau$ and its QR decomposition with column-pivoting is defined as

\begin{eqnarray*}
\int_{0}^{1}\widehat{\bm{\Pi}}(\tau)\mathrm{d}\tau = \widehat{\mathbf{R}}^\top \widehat{\mathbf{S}}^\top = \left[\begin{matrix}
		\widehat{\mathbf{R}}_{11}^\top & \mathbf{0}_{r_0\times (d-r_0)}\\
		\widehat{\mathbf{R}}_{12}^\top & \widehat{\mathbf{R}}_{22}^\top
	\end{matrix}  \right] \cdot \left[\begin{matrix}
		\widehat{\mathbf{S}}_{1}^\top \\
		\widehat{\mathbf{S}}_{2}^\top
	\end{matrix}  \right] ,
\end{eqnarray*}
where $\widehat{\mathbf{S}}$ is $d\times d$ orthonormal, and the partition of the second step should be obvious so we omit the descriptions for each block. By Theorem \ref{Thm1},  $\int_{0}^{1}\widehat{\bm{\Pi}}(\tau)\mathrm{d}\tau\to_P \int_{0}^{1}\bm{\Pi}(\tau)\mathrm{d}\tau .$	Therefore, $\widehat{\mathbf{R}}_{22}$ is expected to be small, which motivates the use of following procedure. Let $\widehat{\mu}_k = \sqrt{\sum_{j=k}^{d} \widehat{\mathbf{R}}^2(k,j)}$, where $\widehat{\mathbf{R}}(k,j)$ denotes the element of $k^{th}$ row and $j^{th}$ column of $\widehat{\mathbf{R}}$. 

	
	We consider the following singular value ratio test, taking a suggestion from the literature (e.g., \citealp{lam2012factor,zhang2019identifying}):
	\begin{eqnarray}\label{Eq2.5}
		\widehat{r} =  \argmax_{0 \le r \le d-1}\left( \frac{\widehat{\mu}_r}{\widehat{\mu}_{r+1}}I\left(\widehat{\mu}_r\geq w_T \right)+I\left(\widehat{\mu}_r < w_T \right)\right) 
	\end{eqnarray}
	where $\widehat{\mu}_0 = \widehat{\mu}_1 + w_T$ is the ``mock'' singular value since $\frac{\widehat{\mu}_r}{\widehat{\mu}_{r+1}}$ is not defined for $r=0$, $w_T =\frac{\log T}{Th} \log(\log (Th)) $ and $ \widehat{\mu}_1\ge\cdots \ge\widehat{\mu}_d$. In addition, the indicator function is used to ensure that the estimator $\widehat{r}$ is consistent. Note that  similar to the case of eigenvalue ratio test for the factor model, both the numerator and denominator converge to zeros at the same rate when $r>r_0$. Therefore, if $\widehat{\mu}_r$ is ``small'', we take it as a sign of $r > r_0$ and set $\widehat{\mu}_r/\widehat{\mu}_{r+1}$ to one. 
	
The following theorem summaries the asymptotic property of \eqref{Eq2.5}.
	
	\begin{theorem}\label{Thm2}
		Let Assumptions \ref{Ass1}--\ref{Ass2} hold. Then $\Pr \left(\widehat{r} = r_0\right) \to 1$ as $T\rightarrow \infty$.
	\end{theorem}
If $r_0 = 0$, $\mathbf{y}_t$ is a pure unit root process with time-varying vector autoregressive errors. Hence, our procedure is also able to test the existence of cointegration relationship.

\subsection{Testing for Parameter Stability}\label{Sec2.5}
	
Practically, it is necessary to test whether the coefficients of \eqref{Eq2.1} are time-varying before applying the aforementioned framework. Formally, we consider a hypothesis test of the form:
	\begin{equation}\label{Eq2.7}
		\mathbb{H}_0: \mathbf{C}\mathbf{b}(\cdot) = \mathbf{c} \text{ for some unknown $\mathbf{c}\in \mathbb{R}^s$},
	\end{equation}
	where $\mathbf{b}(\tau) =  \mathrm{vec} (\bm{\alpha}(\tau),\bm{\Gamma}(\tau) )$, $\mathbf{C}$ is a selection matrix of full row rank, and $s$ is the number of restrictions. The choice of $\mathbf{C}$ and $\mathbf{c}$ should be theory/data driven. For example, one can let $\mathbf{C} =\left[\mathbf{I}_{r_0},\mathbf{0}_{r_0\times (d-1)r_0+d^2(p_0-1)}\right]$ and $\mathbf{c} = \mathbf{0}$ to test whether there exists an error-correction term for $\Delta y_{1,t}$ over the long-run.
	
	The test statistic is constructed based on the weighted integrated squared errors:
	\begin{equation}\label{Eq2.8}
		\widehat{Q}_{\mathbf{C},\mathbf{H}}=\frac{1}{T}\sum_{t=1}^{T} \left\{\mathbf{C}\widehat{\mathbf{b}}(\tau_t)-\widehat{\mathbf{c}}\right\}^\top \mathbf{H}(\tau_t)\left\{\mathbf{C}\widehat{\mathbf{b}}(\tau_t)-\widehat{\mathbf{c}}\right\},
	\end{equation}
	where $\widehat{\mathbf{b}}(\tau) =  \mathrm{vec} (\widehat{\bm{\alpha}}(\tau),\widehat{\bm{\Gamma}}(\tau) )$ should be obvious, and $\widehat{\mathbf{c}} = \frac{1}{T}\sum_{t=1}^{T}\mathbf{C}\widehat{\mathbf{b}}(\tau_t)$ is the semiparametric estimator of $\mathbf{c}$. In \eqref{Eq2.8}, $\mathbf{H}(\cdot)$ is an $s\times s$ positive definite weighting matrix, and is typically set as the precision matrix associated with $\widehat{\mathbf{b}}(\cdot)$. We present the asymptotic distribution of the semiparametric estimator $\widehat{\mathbf{c}}$ and the proposed test in the following theorem.
	\begin{theorem}\label{Thm4}
		Let Assumptions \ref{Ass1}--\ref{Ass2} hold. Suppose further that $\frac{Th^2}{(\log T)^2} \to \infty$ and $Th^{5.5} \to 0$. Then
		
		\begin{enumerate}[wide, labelwidth=!, labelindent=0pt]
			\item [1.] $\sqrt{T}\left(\widehat{\mathbf{c}} - \mathbf{c} - \frac{1}{2}h^2\widetilde{c}_2\int_{0}^{1}\mathbf{C}\mathbf{b}^{(2)}(\tau)\mathrm{d}\tau \right) \to_D N\left( \mathbf{0}, \int_{0}^{1}\mathbf{C}(\bm{\Sigma}_{\mathbf{w}}^{-1}(\tau)\otimes\bm{\Omega}(\tau))\mathbf{C}^\top\mathrm{d}\tau\right)$;
			
			\item [2.] Under $\mathbb{H}_0$, 	$T\sqrt{h}\left(\widehat{Q}_{\mathbf{C}, \widehat{\mathbf{H}}} -\frac{1}{Th}s\widetilde{v}_0\right)\to_D N(0,4s C_B),$	where $\widehat{\mathbf{H}}(\tau) =  (\mathbf{C}\widehat{\mathbf{V}}_{\mathbf{b}}(\tau)\mathbf{C}^\top )^{-1}$, $\widehat{\mathbf{V}}_{\mathbf{b}}(\tau) = \widehat{\bm{\Sigma}}_{\mathbf{w}}^{-1}(\tau)\otimes\widehat{\bm{\Omega}}(\tau)$, and $C_B = \int_{0}^{2} \left(\int_{-1}^{1-v}K(u)K(u+v) \mathrm{d}u\right)^2\mathrm{d}v.$
		\end{enumerate}
	
	\end{theorem}
	In Theorem \ref{Thm4}.1, the bias term $\frac{1}{2}h^2\widetilde{c}_2\int_{0}^{1}\mathbf{C}\mathbf{b}^{(2)}(\tau)\mathrm{d}\tau$ vanishes under $\mathbb{H}_0$, and thus the parametric component in the corresponding semiparametric model can have a $\sqrt{T}$-consistent estimate $\widehat{\mathbf{c}}$. Theorem \ref{Thm4}.2 states that the test statistic converges to a normal distribution and is asymptotically pivotal. The bias term $s\widetilde{v}_0$ can easily be calculated for any given kernel function, and it arises due to the quadratic form of the test statistic.
	
	To close our theoretical investigation, we note further that the online supplementary Appendix \ref{App0}  provides a local alternative of the parameter stability test, and also gives  a  simulation-assisted testing procedure to improve the finite sample performance of the test.

	\section{Simulation}\label{Sec3}
	
	In this section, we first provide some details of the numerical implementation in Section \ref{Sec3.1}, and then respectively examine the estimation and hypothesis testing in Sections \ref{Sec3.2} and \ref{Sec3.3}.  
	
	\subsection{Numerical Implementation}\label{Sec3.1}
	Throughout the numerical studies, Epanechnikov kernel (i.e., $K(u)=0.75(1-u^2)I(|u|\leq 1)$) is adopted. The optimal lag length and the cointegration rank are estimated based on \eqref{Eq2.4} and \eqref{Eq2.5} respectively.  For each given $p$ of \eqref{Eq2.4}, the bandwidth $\widehat{h}_{cv}$ is always chosen by minimizing the following leave-one-out cross-validation criterion function:
	\begin{equation}\label{Eq3.1}
		\widehat{h}_{cv}=\arg\min_{h}\sum_{t=1}^{T}\left\|\Delta\mathbf{y}_t-\widehat{\bm{\Pi}}_{-t}(\tau_t)\mathbf{y}_{t-1}-\sum_{j=1}^{p-1}\widehat{\bm{\Gamma}}_{j,-t}(\tau_t) \Delta\mathbf{y}_{t-j}\right\|^2,
	\end{equation}
	where $\widehat{\bm{\Pi}}_{-t}(\cdot)$ and $\widehat{\bm{\Gamma}}_{j,-t}(\cdot)$ are obtained based on the local linear estimator of Section \ref{Sec2.1} but leaving the $t^{th}$ observation out. Once $\widehat{p}$, $\widehat{r}$ and $\widehat{h}_{cv}$ are obtained, the estimation procedure is relatively straightforward. As shown in \cite{richter2019cross}, the leave-one-out cross validation method works well as long as the error terms are uncorrelated, which implies that this desirable property should hold in our case.
	
	\subsection{Examining the Estimation Results}\label{Sec3.2}
	
The data generating process (DGP) is as follows:
	\begin{equation}\label{Eq3.2}
		\Delta\mathbf{y}_t=\bm{\alpha}(\tau_t)\bm{\beta}^\top\mathbf{y}_{t-1}+\bm{\Gamma}_1(\tau_t)\Delta\mathbf{y}_{t-1}+\mathbf{u}_t \ \ \mbox{with} \ \ \mathbf{u}_t=\bm{\omega}(\tau_t)\bm{\varepsilon}_t \ \ \text{for}\ \ t=1,\ldots ,T,
	\end{equation}
	where $T\in \{200, 400, 800\}$, $\bm{\varepsilon}_t$'s are i.i.d. draws from $N(\mathbf{0}_{2\times 1}, \mathbf{I}_2)$, and
	\begin{eqnarray*}
		\bm{\Gamma}_1(\tau)&=&\left[\begin{matrix}
			0.5\exp{\tau-0.5} & -0.2\exp{\tau-1} \\
			-0.2\cos{\pi\tau} & 0.6\exp{-\tau-0.5}
		\end{matrix}\right], \\
		\bm{\omega}(\tau)&=&\left[\begin{matrix}
			0.8\exp{-0.5\tau}+0.5 & 0 \\
			0.1\exp{0.5-\tau}     & 0.5(\tau-0.5)^2 + 1
		\end{matrix}\right].
	\end{eqnarray*}
	In this case, we have $p_0 = 2$. To test the null hypothesis of no cointegration relations, we consider two sets of $\bm{\alpha}(\tau)$ and $\bm{\beta}$:
	\begin{enumerate}[wide, labelwidth=!, labelindent=0pt]
	\item[] DGP 1 --- $\bm{\alpha}(\tau) = [0.2\sin(\tau)-0.5,0.2\cos(\tau)+0.4]^\top$ and $\bm{\beta} = [1,-0.8]^\top$, so there exists one cointegration relationship, i.e. $r_0=1$. 
	
	\item[] DGP 2 --- $\bm{\alpha}(\tau) = \bm{\beta} =0$, so $r_0=0$ and $\mathbf{y}_t$ is a pure unit-root process with time-varying vector autoregressive errors.
	\end{enumerate}
For each generated dataset, we carry on the methodologies documented in Section \ref{Sec2}, and conduct 1000 replications.
	
	First, we evaluate the performance of the lag length selection procedure (i.e., \eqref{Eq2.4}), and report the percentages of $\widehat{p} < 2$, $\widehat{p} = 2$, and $\widehat{p} > 2$ respectively based on 1000 replications. Table 1 shows that the information criterion \eqref{Eq2.4} performs reasonably well, as the percentages associated with $\widehat{p}=2$ are sufficiently close to 1 except for the case with $T=200$. In addition, this information criterion works well in both time-varying VECM (DGP 1) and time-varying VAR models (DGP 2). 
		
		\begin{center}
		INSERT TABLE 1 ABOUT HERE
		\end{center}
	
	Next, we evaluate the performance of the conintegration rank estimator \eqref{Eq2.5} and report the percentages of $\widehat{r} = 0$ and $\widehat{r} = 1$ respectively based on 1000 replications. Table 2 shows that the  singular value ratio method of \eqref{Eq2.5} performs reasonably well. However, when $r_0=0$, the estimator \eqref{Eq2.5} tends to identify a false cointegration relationship for small sample size (i.e., $T=200$).
	
\begin{center}
INSERT TABLE 2 ABOUT HERE
\end{center}
	
	Finally, we evaluate the estimates of $\bm{\alpha}(\tau)$, $\bm{\beta}$ and $\bm{\Gamma}_1(\tau)$ for DGP 1, and calculate the root mean square error (RMSE) as follows
	\begin{equation*}
		\left\{\frac{1}{1000T} \sum_{n=1}^{1000} \sum_{t=1}^{T}\|\widehat{\bm{\theta}}^{(n)}(\tau_t)-\bm{\theta}(\tau_t)\|^2\right\}^{1/2}
	\end{equation*}
	for $\bm{\theta}(\cdot)\in\left\{\bm{\alpha}(\cdot),\bm{\Gamma}(\cdot)\right\}$, where $\widehat{\bm{\theta}}^{(n)}(\tau)$ is the estimate of $\bm{\theta}(\tau)$ for the $n$-th replication. Of interest, we also examine the finite sample coverage probabilities of the confidence intervals based on our asymptotic theories. In the following, we compute the average of coverage probabilities for grid points in $\{\tau_t,t=1,\ldots,T\}$, and then further take an average across the elements of $\bm{\theta}(\cdot)$. The RMSEs and empirical coverage probabilities are reported in Table 3, which reveals several notable points. First, the RMSE decreases as the sample size goes up. Second, the RMSE of $\bm{\beta} $ is much smaller than those of $\bm{\alpha}(\tau)$ and $\bm{\Gamma}_1(\tau)$, which should be expected. Third, the finite sample coverage probabilities are smaller than their nominal level (95\%) for small $T$, but are fairly close to 95\% as $T$ increases.
	
	\begin{center}
	INSERT TABLE 3 ABOUT HERE
	\end{center}

	\subsection{Examining the Parameter Stability Test}\label{Sec3.3}
	To evaluate the size and local power of the proposed test statistic, we consider the following DGP:
	\begin{equation}\label{Eq3.3}
		\Delta\mathbf{y}_t=\bm{\alpha}(\tau_t)\bm{\beta}^\top\mathbf{y}_{t-1}+\bm{\Gamma}_1(\tau_t)\Delta\mathbf{y}_{t-1}+\mathbf{u}_t,
	\end{equation}
	where $\bm{\beta}$, $\bm{\Gamma}_1(\cdot)$ and $\mathbf{u}_t$ are generated in the same way as DGP 1 in Section \ref{Sec3.2}, and
	$$
	\bm{\alpha}(\tau)=\left[\begin{matrix}
		-0.4  \\
		0.4
	\end{matrix} \right]+b\times d_T\times \left[\begin{matrix}
		\sin(\tau)  \\
		\cos(\pi\tau)
	\end{matrix} \right],
	$$
	in which $d_T = T^{-1/2}h^{-1/4}$ and $b$ is set to be $0$, $1$ or $2$ in order to investigate the size and local power of the proposed test. We use the proposed testing procedure to test whether the coefficient $\bm{\alpha}(\cdot)$ is time-varying. Again, we let $T\in \{200,400,800\}$ and conduct $1000$ replications for each choice of $T$. We use the simulation-assisted testing procedure of Appendix \ref{App0} to get the empirical critical value $\widehat{q}_{1-\alpha}$ after $1000$ bootstrap replications. We consider a sequence of bandwidths to check the robustness of the proposed test with respect to a sequence of bandwidths: 
	\begin{equation}\label{h_test}
		h = \alpha_1 T^{-1/5}, \quad \alpha_1= 0.6,\ldots,2.
	\end{equation}
	
	Table 4 reports the rejection rates at the $5\%$ and $10\%$ nominal levels. A few facts emerge. First, our test has reasonable sizes using the empirical critical values obtained by the  bootstrap procedure if the sample size is not so small. Second, the size behaviour of our test is not sensitive to the choices of bandwidths. As discussed in \cite{gao2008bandwidth}, the estimation-based optimal bandwidths may also be optimal for testing purposes, so for simplicity we use the cross-validation based bandwidth or the rule-of-thumb bandwidth in practice. Third, the local power of our test increases rapidly as $b$ increases.

\begin{center}
INSERT TABLE 4 ABOUT HERE
\end{center}

	\section{A Real Data Example}\label{Sec4}
	In this section, we assess the time-varying predictability of the term structure (i.e., the yield curve) of interest rates, and investigate whether the expectations hypothesis of the term structure holds periodically by using the proposed time-varying VECM model, which allows for shifts in the predictability of the term structure.
	
We now briefly review the literature. The term structure is crucial to both monetary policy analysis and private individuals.  According to the rational expectations hypothesis (\citealp{campbell1987cointegration}), the term structure (or term spread) should provide information on the future changes in both short-term and long-term interest rates. For example, if a long bond yield exceeds a short yield, the long rate subsequently tends to rise, which generates expected capital losses on the long bond and thus offsets the current yield advantage. This also implies that bond returns are predictable from the yield spread, and the expected bond returns and the yield spread should be negatively correlated.
	
The literature on the term structure and bond return predictability is enormous and it continues to expand (e.g.,  \citealp{campbell1987cointegration,bauer2020interest,andreasen2021yield,vayanos2021preferred,he2022treasury}). However, the existing results on the expectations hypothesis of the term structure present many discrepancies, which may be due to the fact that the relationship evolves with time. For example, \cite{borup2021predicting} document that bond return predictability depends on the economic states, while linear forecasting models yield little evidence of unconditional predictability. Along this line of research, one important question is that whether the U.S. monetary and financial system has changed over time so that estimates from historical data are unreliable for modern policy analysis and investment activities. Although the literature has begun to explore whether the predictability of the term structure depends on the economy states (e.g., \citealp{andreasen2021yield,borup2021predicting}), few studies aim to quantify the varying predictability of the term structure over time. In what follows, we address this issue using the newly proposed framework. The estimation procedure is conducted in exactly the same way as in Section \ref{Sec3}, so we no longer repeat the details.

	\subsection{Empirical Analysis}\label{Sec4.2}
	To study the time-varying predictability of the term structure, we consider the following time-varying bivariate VECM($p_0-1$) model:
	$$
	\Delta \mathbf{y}_t = \bm{\alpha}(\tau_t)\bm{\beta}^\top\mathbf{y}_{t-1} + \sum_{j=1}^{p_0-1}\bm{\Gamma}_j(\tau_t)\Delta \mathbf{y}_{t-j} + \mathbf{u}_t,\quad \mathbf{u}_t = \bm{\omega}(\tau_t)\bm{\varepsilon}_t,
	$$
	where $\mathbf{y}_t^\top = [l_t, s_t]$, $s_T$ is the interest rate on a one-period bond and $l_t$ is the interest rate on a multi-period bond. According to the expectations hypothesis of the term structure, $l_t$ and $s_t$ should be integrated with $\bm{\beta}^\top = [1,-1]$ (and thus $\beta^* = -1$ under the identification condition), while the error-correction term is the term spread. {In addition, the adjustment coefficient $\bm{\alpha}(\tau)$ measures the predictability of the term structure and the sign of the elements in $\bm{\alpha}(\tau)$ should be positively significant according to the theory. }
	
	We use a selection of bond rates with maturities ranging from 1 to 5 years. The interest rates are estimated from the prices of U.S. Treasury securities and correspond to zero-coupon bonds. The data are monthly observations from 1961:M6 to 2022:M12, which are collected from Nasdaq Data Link at \url{https://data.nasdaq.com}. Figure \ref{Fg1} plots these five variables.

\begin{center}
INSERT FIGURE 1 ABOUT HERE.
\end{center}

\subsection{Constant Parameter VECM}

We begin by presenting the results using a constant parameter VECM model in order to get some intuition, which is also served as a benchmark for evaluating our time-varying VECM model. Recall that one of the economic implications of the expectations hypothesis of the term structure is that the long and short rates should be cointegrated with $\bm{\beta}=[1,-1]^\top$. To empirically test this implication, we first detect the presence of cointegration, using the Johansen's likelihood ratio test. The testing results are reported in Table 5. From Table 5, for all possible bivariate pairs and lag lengths considered, the tests strongly reject the null of no integration, but do not reject the null of a single cointegration relationship at the 5\% significance level.

	\begin{center}
	INSERT TABLE 5 ABOUT HERE
	\end{center}

 We then report the parameter estimates for constant VECM models, while the optimal lag $p$ is set to be $2$ according to the literature (e.g., \citealp{hansen2003structural}). Table 6 reports the parameter estimates as well as their 95\% confidence intervals, in which $\alpha_i$ denotes the $i^{th}$ elements of adjustment coefficient $\bm{\alpha}$. From Table 6, we can see that for all considered bivariate pairs, the estimated cointegration parameter $\widehat{\beta}^*$ is quite close to unity, which is consistent with the theory. {However, for all bivariate pairs considered, the estimates of  $\widehat{\alpha}_1$ and $\widehat{\alpha}_2$ are not positively significant (or even negative), which contradicts to the economic implications of the term structure theory. Note that according to the expectations hypothesis, the sign of adjustment coefficients $\bm{\alpha}$ should be positive. These results are in line with \cite{borup2021predicting}, who find that linear forecasting models yield little evidence of unconditional bond return predictability.}

\begin{center}
	INSERT TABLE 6 ABOUT HERE
	\end{center}

\subsection{Time-Varying VECM}

In order to solve the puzzle raised by constant parameter VECM models, we then investigate the possibility that the expectations hypothesis of the term structure holds periodically. This is mainly motivated by the following facts. First, several studies have shown that the bond return predictability depends on the economic states (e.g., \citealp{andreasen2021yield,borup2021predicting}). Second the rational expectations hypothesis also implies that the future bond returns should be negatively correlated with the current term spread (i.e., $l_t-s_t$). The above two facts mean that the predictability of term structure shifts over time and the above puzzling contradictions may be explained by using time-varying VECM models.

Table 7 report the estimation and testing results of time-varying VECM models, as well as some robustness checks. For the time-varying VECM($p-1$) model, the optimal lag is $\widehat{p} = 2$ or $\widehat{p} = 3$ by our approach, which is consistent with the literature (e.g., \citealp{hansen2002testing}). We further check whether the model coefficients are indeed time-varying. We employ the proposed test statistic to examine the constancy of model coefficients. The associated $p$-value for these considered bivariate pairs are around 0.000--0.001, which suggest that we should choose the time-varying VECM model over a constant one. Certainly, one may examine each element of these coefficient matrices. However, it will lead to a quite lengthy presentation. In order not to deviate from our main goal, we no longer conduct more testing along this line. These testing results also suggest that the predictability of term structure are time-varying for a range of short and long rates. We then conduct robustness check to see whether the error innovations $\{\bm{\varepsilon}_t\}$ exhibit serial correlation. We use the multivariate version of Breusch-Godfrey LM test (\citealp{godfrey1978testing}) to test the serial-correlations in the innovations $\{\bm{\varepsilon}_t\}$, in which the null hypothesis is $H_0: E(\bm{\varepsilon}_t\bm{\varepsilon}_{t+1}^\top) = \mathbf{0}$. Based on the estimates $\widehat{\bm{\varepsilon}}_t = \widehat{\bm{\Omega}}^{-1/2}(\tau_t)\widehat{\mathbf{u}}_t$, the corresponding $p$-value ranges from 0.559 to 0.945, suggesting that the time-varying VECM model fits the data quite well for all considered bivariate pairs. 

\begin{center}
	INSERT TABLE 7 ABOUT HERE
	\end{center}

We then apply the singular ratio test to detect the presence of cointegration and then check whether $\beta^* = -1$. Table 7 show that the estimates of cointegration rank are $1$ (i.e., $\widehat{r} = 1$) for all considered bivariate pairs. Based on these singular ratio tests, we find strong evidence of the existence of cointegration, which is consistent with Figure \ref{Fg1} and the results of constant parameter VECM models. In addition, this result, i.e., $\widehat{r} = 1$, is robust to different choices of $p$ ranging from $2$--$6$. We also report the point estimates of $\beta^*$ and their 95\% confidence intervals in Table 7. According to Table 7, we find a long-run relationship between the long rate and the short rate, and we cannot reject the null that $H_0: \beta^*=-1$ at a 5\% significance level for most considered bivariate pairs. Interestingly, we find that the estimates of $\beta^*$ between these two models are almost identical, while the time-varying VECM models yields quite narrow confidence bands (i.e., smaller standard errors).

Finally, we investigate whether the term structure (or bond returns) is predictable over the long-run by using the term spread as a predictor, i.e., $\bm{\alpha}(\tau) = \mathbf{0}$, and we also examine the sign of $\bm{\alpha}(\tau)$, which should be positive according to the expectations hypothesis. {In order to confirm whether the error--correction component is significant, We test the null hypothesis $\mathbb{H}_0:\ \bm{\alpha}(\tau) = \mathbf{0}$. The testing results are reported in the last column of Table 7, which indicates that we should reject the null at all conventional levels. Therefore, the term structure is predictable in the long-run, at least in some local periods, while the constant parameter VECM model yields little evidence of the predictability of the term structure.} We also investigate the time-varying pattern of the term structure predictability. Figure \ref{Fg2} plots the estimates of $\alpha_1(\cdot)$ and $\alpha_2(\cdot)$ and their 95\% point-wise confidence intervals, as well as the U.S. core inflation. Here, the core inflation data are monthly observations from 1967:M1 to 2022:M12, collected from the Federal Reserve Bank of St. Louis economic database. 

{In summary, our finding of great interest is that the estimated error-correction effects for the short and long rates are only positively significant in the 1980s and in recent years, and vary significantly over time for these bivariate pairs considered.} What is more is that the time-varying patterns of $\alpha_1(\cdot)$ and $\alpha_2(\cdot)$ are almost identical, and are similar to the pattern of time-variations in U.S. core inflation rate. Our findings suggest that the expectations hypothesis holds periodically, especially in the period of unusual high inflation. Importantly, our results also provide us with evidence for the economic implications of the theoretical macro-finance term structure model proposed by \cite{andreasen2021yield}, who find that the monetary policy decisions by the Federal Reserve with respect to stabilizing inflation is a key driver of this switch in bond return predictability.

\begin{center}
INSERT FIGURE 2 ABOUT HERE
\end{center}

\section{Conclusions}\label{Sec5}
In this paper, we propose a time-varying vector error-correction model that allows for different time series behaviours (e.g., unit-root and locally stationary processes) interacting with each other to co-exist. From practical perspectives, this framework can be used to estimate shifts in the predictability of non-stationary variables, and test whether economic theories hold periodically. We first develop a time-varying Granger Representation Theorem, which facilitates establishing asymptotic properties, and then propose estimation and inferential theories for both short-run and long-run coefficients. We also propose an information criterion to estimate the lag length, a singular-value ratio test to determine the cointegration rank, and a hypothesis test to examine the parameter stability. To validate the  theoretical findings, we conduct extensive simulations. Finally, we demonstrate the empirical relevance  by applying the framework to investigate the rational expectations hypothesis of the U.S. term structure. We conclude that the predictability of the term structure vary significantly over time and the expectations hypothesis of the term structure holds periodically, especially in the period of unusual high inflation.

\section{Acknowledgements}

Gao and Peng acknowledge financial support from the Australian Research Council Discovery Grants Program under Grant  Numbers: DP200102769 \& DP210100476, respectively. Yan acknowledges the financial support by Fundamental Research Funds for the Central Universities (Grant Numbers: 2022110877 \& 2023110099).

 {\footnotesize

\bibliographystyle{rss}
\bibliography{vecm}

@article{chandler2012mode,
	title={Mode Identification of volatility in time-varying autoregression},
	author={Chandler, Gabriel and Polonik, Wolfgang},
	journal={Journal of the American Statistical Association},
	volume={107},
	number={499},
	pages={1217--1229},
	year={2012},
	publisher={Taylor \& Francis}
}

@article{zhang2019identifying,
	title={Identifying cointegration by eigenanalysis},
	author={Zhang, Rongmao and Robinson, Peter and Yao, Qiwei},
	journal={Journal of the American Statistical Association},
	volume={114},
	number={526},
	pages={916--927},
	year={2019},
	publisher={Taylor \& Francis}
}

@article{barigozzi2022inference,
	title={Inference in Heavy-Tailed Nonstationary Multivariate Time Series},
	author={Barigozzi, Matteo and Cavaliere, Giuseppe and Trapani, Lorenzo},
	journal={Journal of the American Statistical Association},
        volume={117},
	number={540},
	pages={1--17},
	year={2022},
	publisher={Taylor \& Francis}
}

@book{fan2003nonlinear,
	title={Nonlinear Time Series: Nonparametric and Parametric Methods},
	author={Fan, Jianqing and Yao, Qiwei},
	volume={20},
	year={2003},
	publisher={Springer}
}

@article{godfrey1978testing,
	title={Testing against general autoregressive and moving average error models when the regressors include lagged dependent variables},
	author={Godfrey, Leslie G},
	journal={Econometrica},
	volume={46},
	number={6},
	pages={1293--1301},
	year={1978},
	publisher={JSTOR}
}

@article{benati2021international,
	title={International evidence on long-run money demand},
	author={Benati, Luca and Lucas Jr, Robert E and Nicolini, Juan Pablo and Weber, Warren},
	journal={Journal of Monetary Economics},
	volume={117},
	pages={43--63},
	year={2021},
	publisher={Elsevier}
}

@article{gao2022time,
	title={Time-Varying Multivariate Causal Processes},
	author={Gao, Jiti and Peng, Bin and Wu, Wei Biao and Yan, Yayi},
	journal={arXiv preprint arXiv:2206.00409},
	year={2022}
}

@article{hansen2003structural,
	title={Structural changes in the cointegrated vector autoregressive model},
	author={Hansen, Peter Reinhard},
	journal={Journal of Econometrics},
	volume={114},
	number={2},
	pages={261--295},
	year={2003},
	publisher={Elsevier}
}

@article{bergamelli2019combining,
	title={Combining p-values to test for multiple structural breaks in cointegrated regressions},
	author={Bergamelli, Michele and Bianchi, Annamaria and Khalaf, Lynda and Urga, Giovanni},
	journal={Journal of Econometrics},
	volume={211},
	number={2},
	pages={461--482},
	year={2019},
	publisher={Elsevier}
}

@article{hansen2002testing,
	title={Testing for two-regime threshold cointegration in vector error-correction models},
	author={Hansen, Bruce E and Seo, Byeongseon},
	journal={Journal of Econometrics},
	volume={110},
	number={2},
	pages={293--318},
	year={2002},
	publisher={Elsevier}
}

@article{king1991stochastic,
	title={Stochastic trends and economic fluctuations},
	author={King, Robert G and Plosser, Charles I and Stock, James H and Watson, Mark W},
journal={The American Economic Review},
volume={81},
	number={4},
pages={819--840},
year={1991},
}

@article{borup2021predicting,
	title={Predicting bond return predictability},
	author={Borup, Daniel and Eriksen, Jonas N and Kj{\ae}r, Mads Markvart and Thyrsgaard, Martin},
	journal={Available at SSRN 3513340},
	year={2021}
}

@article{campbell1987cointegration,
	title={Cointegration and tests of present value models},
	author={Campbell, John Y and Shiller, Robert J},
	journal={Journal of Political Economy},
	volume={95},
	number={5},
	pages={1062--1088},
	year={1987},
	publisher={The University of Chicago Press}
}

@article{he2022treasury,
	title={Treasury inconvenience yields during the {COVID}-19 crisis},
	author={He, Zhiguo and Nagel, Stefan and Song, Zhaogang},
	journal={Journal of Financial Economics},
	volume={143},
	number={1},
	pages={57--79},
	year={2022},
	publisher={Elsevier}
}

@article{andreasen2021yield,
	title={The yield spread and bond return predictability in expansions and recessions},
	author={Andreasen, Martin M and Engsted, Tom and M{\o}ller, Stig V and Sander, Magnus},
	journal={The Review of Financial Studies},
	volume={34},
	number={6},
	pages={2773--2812},
	year={2021},
	publisher={Oxford University Press}
}

@article{bauer2020interest,
	title={Interest rates under falling stars},
	author={Bauer, Michael D and Rudebusch, Glenn D},
	journal={The American Economic Review},
	volume={110},
	number={5},
	pages={1316--54},
	year={2020}
}

@article{vayanos2021preferred,
	title={A preferred-habitat model of the term structure of interest rates},
	author={Vayanos, Dimitri and Vila, Jean-Luc},
	journal={Econometrica},
	volume={89},
	number={1},
	pages={77--112},
	year={2021},
	publisher={Wiley Online Library}
}

@article{lam2012factor,
	title={Factor modeling for high-dimensional time series: inference for the number of factors},
	author={Lam, Clifford and Yao, Qiwei},
	journal={The Annals of Statistics},
	volume={40},
	number={2},
	pages={694--726},
	year={2012},
	publisher={JSTOR}
}

@article{borovkov1973notes,
	title={Notes on inequalities for sums of independent variables},
	author={Borovkov, AA},
	journal={Theory of Probability and its Applications},
	volume={17},
	number={3},
	pages={556},
	year={1973},
	publisher={Society for Industrial and Applied Mathematics}
}

@article{eberlein1986strong,
	title={On strong invariance principles under dependence assumptions},
	author={Eberlein, Ernst},
	volume={14},
	number={1},
	journal={The Annals of Probability},
	pages={260--270},
	year={1986},
	publisher={JSTOR}
}

@article{karmakar2021simultaneous,
	title={Simultaneous inference for time-varying models},
	author={Karmakar, Sayar and Richter, Stefan and Wu, Wei Biao},
	journal={Journal of Econometrics},
	volume = {227},
number = {2},
pages = {408-428},
	year={2022},
	publisher={Elsevier}
}

@article{wu2005nonlinear,
	title={Nonlinear system theory: Another look at dependence},
	author={Wu, Wei Biao},
	journal={Proceedings of the National Academy of Sciences},
	volume={102},
	number={40},
	pages={14150--14154},
	year={2005},
	publisher={National Academy of Sciences}
}

@book{johansen1995likelihood,
	title={{Likelihood-based Inference in Cointegrated Vector Autoregressive Models}},
	author={Johansen, S{\o}ren},
	year={1995},
	publisher={Oxford University Press on Demand}
}

@article{gao2008bandwidth,
  title={Bandwidth selection in nonparametric kernel testing},
  author={Gao, Jiti and Gijbels, Irene},
  journal={Journal of the American Statistical Association},
  volume={103},
  number={484},
  pages={1584--1594},
  year={2008},
  publisher={Taylor \& Francis}
}

@ARTICLE {hansen2001new,
    author  = "Hansen, Bruce E",
    title   = "The new econometrics of structural change: {D}ating breaks in {US} labour productivity",
    journal = "Journal of Economic Perspectives",
    year    = "2001",
    volume  = "15",
    number  = "4",
    pages   = "117-128"
}

@ARTICLE {dahlhaus1996kullback,
    author    = "Dahlhaus, Rainer",
    title     = "On the Kullback-Leibler information divergence of locally stationary processes",
    journal   = "Stochastic Processes and Their Applications",
    year      = "1996",
    volume    = "62",
    number    = "1",
    pages     = "139-168",
    publisher = "Elsevier"
}

@BOOK {lutkepohl2005new,
    author    = "L{\"u}tkepohl, Helmut",
    title     = "New Introduction to Multiple Time Series Analysis",
    publisher = "Springer Science \& Business Media",
    year      = "2005"
}

@ARTICLE {dahlhaus2006statistical,
    author    = "Dahlhaus, Rainer and Rao, Suhasini Subba",
    title     = "Statistical inference for time-varying {ARCH} processes",
    journal   = "The Annals of Statistics",
    year      = "2006",
    volume    = "34",
    number    = "3",
    pages     = "1075-1114",
    publisher = "Institute of Mathematical Statistics"
}

@ARTICLE {dahlhaus2009empirical,
    author    = "Dahlhaus, Rainer and Polonik, Wolfgang",
    title     = "Empirical spectral processes for locally stationary time series",
    journal   = "Bernoulli",
    year      = "2009",
    volume    = "15",
    number    = "1",
    pages     = "1-39",
    publisher = "Bernoulli"
}

@ARTICLE {richter2019cross,
    author    = "Richter, Stefan and Dahlhaus, Rainer",
    title     = "Cross validation for locally stationary processes",
    journal   = "The Annals of Statistics",
    year      = "2019",
    volume    = "47",
    number    = "4",
    pages     = "2145-2173",
    publisher = "Institute of Mathematical Statistics"
}

@ARTICLE {truquet2017parameter,
    author    = "Truquet, Lionel",
    title     = "Parameter stability and semiparametric inference in time varying auto-regressive conditional heteroscedasticity models",
    journal   = "Journal of the Royal Statistical Society: Series B",
    year      = "2017",
    volume    = "79",
    number    = "5",
    pages     = "1391-1414",
    publisher = "Wiley Online Library"
}

@ARTICLE {li2019kernel,
    author  = "Li, Degui and Phillips, Peter C. B. and Gao, Jiti",
    title   = "Kernel-based inference in time-varying coefficient cointegrating regression",
    journal = "Journal of Econometrics",
    year    = "2019",
    volume  = "215",
    number  = "2",
    pages   = "607-632"
}

@ARTICLE {phillips2017estimating,
    author    = "Phillips, Peter C. B. and Li, Degui and Gao, Jiti",
    title     = "Estimating smooth structural change in cointegration models",
    journal   = "Journal of Econometrics",
    year      = "2017",
    volume    = "196",
    number    = "1",
    pages     = "180-195",
    publisher = "Elsevier"
}

@ARTICLE {fan2005profile,
    author    = "Fan, Jianqing and Huang, Tao",
    title     = "Profile likelihood inferences on semiparametric varying-coefficient partially linear models",
    journal   = "Bernoulli",
    year      = "2005",
    volume    = "11",
    number    = "6",
    pages     = "1031-1057",
    publisher = "Bernoulli Society for Mathematical Statistics and Probability"
}

@ARTICLE {vogt2012nonparametric,
    author    = "Vogt, Michael",
    title     = "Nonparametric regression for locally stationary time series",
    journal   = "The Annals of Statistics",
    year      = "2012",
    volume    = "40",
    number    = "5",
    pages     = "2601-2633",
    publisher = "Institute of Mathematical Statistics"
}

@ARTICLE {zhang2012inference,
    author    = "Zhang, Ting and Wu, Wei Biao",
    title     = "Inference of time-varying regression models",
    journal   = "The Annals of Statistics",
    year      = "2012",
    volume    = "40",
    number    = "3",
    pages     = "1376-1402",
    publisher = "Institute of Mathematical Statistics"
}

@ARTICLE {freedman1975tail,
    author    = "Freedman, David A",
    title     = "On tail probabilities for martingales",
    journal   = "The Annals of Probability",
    year      = "1975",
    volume    = "3",
    number    = "1",
    pages     = "100-118",
    publisher = "Institute of Mathematical Statistics"
}

@BOOK {hall2014martingale,
    author    = "Hall, Peter and Heyde, Christopher C",
    title     = "Martingale Limit Theory and Its Application",
    publisher = "Academic Press",
    year      = "1980"
}

}

 \newpage
 
\begin{center}
\textbf{Table 1.}	The percentages of $\widehat{p} < 2$, $\widehat{p} = 2$, and $\widehat{p} > 2$ \\
\begin{tabular}{c l ccc c l ccc}
				\hline
				DGPs	& $T$ & $\widehat{p} < 2$ & $\widehat{p} = 2$ & $\widehat{p} > 2$ \\
				\hline
				\multirow{3}{*}{DGP 1}
				& 200   & 0.065 & 0.935 & 0.000 \\
				& 400   & 0.000 & 1.000 & 0.000 \\
				& 800   & 0.000 & 1.000 & 0.000 \\
				\hline
				\multirow{3}{*}{DGP 2}
				& 200   & 0.045 & 0.955 & 0.000 \\
				& 400   & 0.000 & 1.000 & 0.000 \\
				& 800   & 0.000 & 1.000 & 0.000 \\
				\hline
\end{tabular}
\end{center}

\bigskip

\begin{center} 
		\textbf{Table 2.} The percentages of $\widehat{r} = 0$ and $\widehat{r} = 1$ \\
			\begin{tabular}{c l ccc c l cc}
				\hline
				DGPs	& $T$ & $\widehat{r} = 0$ & $\widehat{r} = 1$ \\
				\hline
				\multirow{3}{*}{DGP 1}
				& 200   & 0.000 & 1.000   \\
				& 400   & 0.000 & 1.000  \\
				& 800   & 0.000 & 1.000   \\
				\hline
				\multirow{3}{*}{DGP 2}
				& 200   & 0.865 & 0.135  \\
				& 400   & 0.985 & 0.015  \\
				& 800   & 1.000 & 0.000  \\
				\hline
			\end{tabular}
	\end{center}
	
	\bigskip
	
	\begin{center} 
		\textbf{Table 3.} The RMSEs and the empirical coverage probabilities at the 95\% nominal level
		\begin{tabular}{lcc c cc c cc}
			\hline
			& \multicolumn{2}{c}{$\bm{\alpha}(\tau)$}& & \multicolumn{2}{c}{$\bm{\beta}$} & & \multicolumn{2}{c}{$\bm{\Gamma}_1(\tau)$}\\
			\cline{2-3} \cline{5-6} \cline{8-9}
			$ T$     & \text{RMSE} & \text{Coverage rate}& & \text{RMSE} & \text{Coverage rate} & & \text{RMSE} & \text{Coverage rate}\\
			\hline
			200   & 0.146 & 0.927& & 0.017 & 0.885 & & 0.225 & 0.930\\
			400   & 0.099 & 0.940& & 0.007 & 0.940 & & 0.156 & 0.938\\
			800   & 0.070 & 0.947& & 0.003 & 0.936 & & 0.109 & 0.945\\
			\hline
		\end{tabular}
	\end{center}
	
 \newpage
	
	\begin{center} 
\textbf{Table 4.}  Size and power evaluation 
			\begin{tabular}{c c c ccc c ccc}
				\hline
				& & &\multicolumn{3}{c}{$5\%$}& &\multicolumn{3}{c}{$10\%$}\\
				\cline{4-6} \cline{8-10}
				&\text{Bandwidth}&T &200 &400 &800  &   &200 &400 &800 \\
				\hline
				\multirow{9}{*}{\shortstack{$b=0$\\(size)}}
				&$0.6T^{-1/5}$ & &  0.092 & 0.065  & 0.065&  &  0.170 & 0.138  &  0.113 \\
				&$0.8T^{-1/5}$ & &  0.092 & 0.058  & 0.054 &  &  0.145 & 0.125  & 0.117 \\
				&$1.0T^{-1/5}$ & &  0.078 & 0.048 & 0.049&  &  0.124 & 0.095  & 0.103 \\
				&$1.2T^{-1/5}$ & &  0.080 & 0.045  & 0.055 &  &  0.148 & 0.113  & 0.102 \\
				&$1.4T^{-1/5}$ & &  0.070 & 0.045  & 0.053 &  &  0.116 & 0.115  & 0.113 \\
				&$1.6T^{-1/5}$ & &  0.093  & 0.038  & 0.060&  &  0.130 & 0.120  & 0.120 \\
				&$1.8T^{-1/5}$ & &  0.074 & 0.045  & 0.050&  &  0.114 & 0.108  & 0.103 \\
				&$2.0T^{-1/5}$ & &  0.073 & 0.043  & 0.048&  &  0.123 & 0.095	  & 0.100 \\
				\hline
				\multirow{9}{*}{\shortstack{$b=1$\\(local power)}}
				&$0.6T^{-1/5}$ & & 0.209  & 0.171  & 0.166&  & 0.306  & 0.264  & 0.264 \\
				&$0.8T^{-1/5}$ & & 0.177  & 0.148  & 0.165&  & 0.264  & 0.240  & 0.264 \\
				&$1.0T^{-1/5}$ & & 0.159  & 0.130  & 0.164 &  & 0.239  & 0.230  & 0.241 \\
				&$1.2T^{-1/5}$ & & 0.149 & 0.123  & 0.154&  & 0.218  & 0.210  & 0.230 \\
				&$1.4T^{-1/5}$ & & 0.149  & 0.128  & 0.145&  & 0.215  & 0.208  & 0.234 \\
				&$1.6T^{-1/5}$ & & 0.152  & 0.127  & 0.138&  & 0.204  & 0.208  & 0.232 \\
				&$1.8T^{-1/5}$ & & 0.144  & 0.126  & 0.135&  & 0.206  & 0.211  & 0.224 \\
				&$2.0T^{-1/5}$ & & 0.141  & 0.130  & 0.137&  & 0.205  & 0.214  & 0.220 \\
				\hline
				\multirow{9}{*}{\shortstack{$b=2$\\(local power)}}
				&$0.6T^{-1/5}$ & & 0.572  & 0.597  & 0.652& & 0.712  & 0.708  & 0.749 \\
				&$0.8T^{-1/5}$ & & 0.522  & 0.547  & 0.615&  & 0.659  & 0.673  & 0.723 \\
				&$1.0T^{-1/5}$ & & 0.470  & 0.517  & 0.577&  & 0.604  & 0.647  & 0.693 \\
				&$1.2T^{-1/5}$ & & 0.454  & 0.488  & 0.559&  & 0.574 & 0.624  & 0.667 \\
				&$1.4T^{-1/5}$ & & 0.453  & 0.488  & 0.537&  & 0.566  & 0.613  & 0.657 \\
				&$1.6T^{-1/5}$ & & 0.468  & 0.473  & 0.525&  & 0.566  & 0.605  & 0.644 \\
				&$1.8T^{-1/5}$ & & 0.446  & 0.459  & 0.511&  & 0.573  & 0.603  & 0.625 \\
				&$2.0T^{-1/5}$ & & 0.426  & 0.458  & 0.503&  & 0.569  & 0.590  & 0.613 \\
				\hline
			\end{tabular}
	\end{center}	
	
	\bigskip
	
	\begin{center} 
	\textbf{Table 5.} The Johansen's LR tests for constant parameter VECM models ($p$-values)  
	\begin{tabular}{ll ccc c ccc}
		\hline
	&	& \multicolumn{3}{c}{$H_0:\ r=0$}& & \multicolumn{3}{c}{$H_0:\ r=1$}\\
		\cline{3-5} \cline{7-9}
Short Rate	&	Long Rate     & $p=1$ & $p=2$& $p=3$ &   & $p=1$ & $p=2$& $p=3$\\
		\hline
1-year	&	3-year   & 0.001 & 0.001 & 0.001 & & 0.103 & 0.058 & 0.081\\
1-year	&	4-year   & 0.003 & 0.001 & 0.002 & & 0.108 & 0.066 & 0.086\\
1-year	&	5-year   & 0.005 & 0.001 & 0.003 & & 0.105 & 0.069 & 0.087\\
2-year	&	3-year   & 0.003 & 0.001 & 0.003 & & 0.108 & 0.074 & 0.091\\
2-year	&	4-year   & 0.004 & 0.001 & 0.004 & & 0.105 & 0.078 & 0.093\\
2-year	&	5-year   & 0.005 & 0.001 & 0.006 & & 0.099 & 0.079 & 0.092\\
		\hline
	\end{tabular}
\end{center}

\bigskip
\bigskip

	\begin{center} 
	\textbf{Table 6.} Parameter estimates for constant parameter VECM models 
	\begin{tabular}{l l c c c c c}
		\hline
Short Rate	&	Long Rate	& $\beta^*$& &$\alpha_1$ & &$\alpha_2$\\
		\hline
	1-year	&3-year   & -0.988 &  & 0.050 & & -0.018 \\
		&         & [-1.041,-0.934] &  & [-0.016,0.115] & & [-0.074,0.039] \\
	1-year	&4-year   & -0.976 &  & 0.038 & & -0.023 \\
		&    & [-1.046,-0.905] &  & [-0.140,0.089]& & [-0.065,0.019]\\
	1-year	&5-year   & -0.964 &  & 0.029 & & -0.028 \\
		&   & [-1.048,-0.879] &  &[-0.015,0.074] & & [-0.062,0,008] \\
	2-year	&3-year   & -0.991 &  & 0.008 & & 0.069 \\
		& & [-1.014, -0.967]  &  & [-0.128,0.143] & & [-0.077,0.214] \\
	2-year	&4-year   & -0.979 &  & -0.020 & & 0.036 \\
		&         & [-1.050,-0.834] &  & [-0.094,0.054] & & [-0.048,0.120] \\
	2-year	&5-year   & -0.966 &  & -0.028 & & 0.0221 \\
		&         & [-1.025,-0.907] &  & [-0.082,0.024] & & [-0.041,0.085] \\
		\hline
	\end{tabular}
\end{center}

\bigskip

\begin{center} 
	\textbf{Table 7.} Estimation and testing results of time-varying VECM models 
	
	\begin{tabular}{l l  c c c c c c}
		\hline
Short Rate	&	Long Rate	&$\widehat{p}$ &Constancy &Breusch-Godfrey& $\widehat{r}$ & $\beta^*$ & $\bm{\alpha}(\tau)=\mathbf{0}$ \\
		\hline
1-year	&	3-year   & 3  & 0.000 & 0.943 & 1 &  -1.026  & 0.000\\
	    &            &    &       &       &   & [-1.055,-0.997] &\\
1-year	&	4-year   & 3  & 0.000 & 0.945 & 1 & -1.022 &0.000\\
	    &		     &    &       &       &   & [-1.063,-0.981] &\\
1-year	&	5-year   & 2  & 0.001 & 0.564 & 1 & -0.954 & 0.000\\
	    &		     &    &       &       &   & [-1.012,-0.896]& \\
	    2-year	&	3-year   & 2  & 0.001 & 0.747 & 1 & -0.991& 0.000\\
	    &		     &    &       &       &   & [-1.007,-0.974]& \\
	    2-year	&	4-year   & 2  & 0.000 & 0.862 & 1 & -0.974& 0.000\\
	    &		     &    &       &       &   & [-1.003,-0.945]& \\
	    2-year	&	5-year   & 2  & 0.000 & 0.921 & 1 & -0.957&0.000\\
	    &		     &    &       &       &   & [-0.994,-0.920]& \\
		\hline
	\end{tabular}
\end{center}

			\begin{figure}[H]
						\centering
			{\includegraphics[width=14cm]{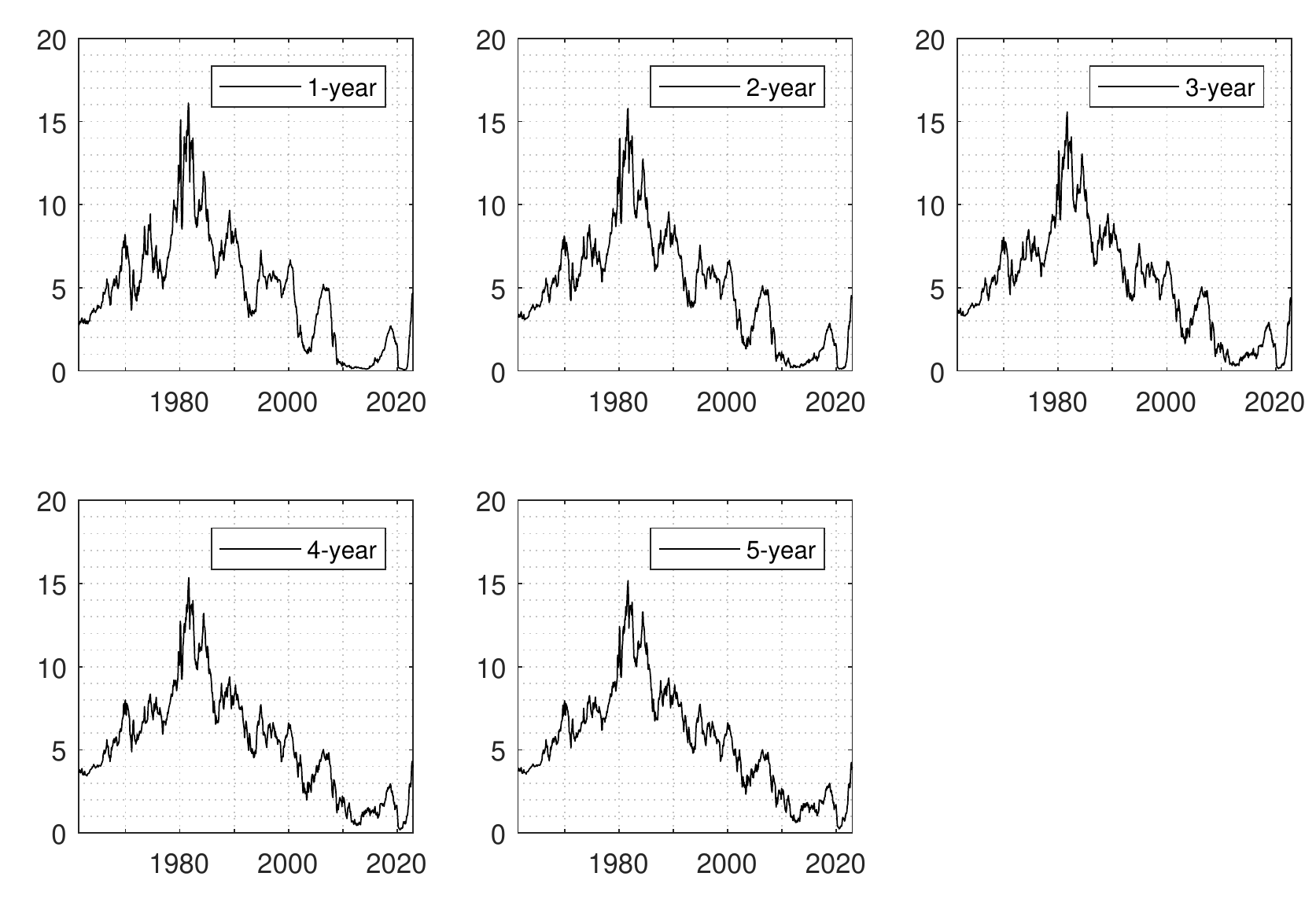}}
			\caption{Plots of short rate and long rate.}\label{Fg1}
		\end{figure}

 	\begin{figure}[H]
 				\centering
		{\includegraphics[width=4.17cm]{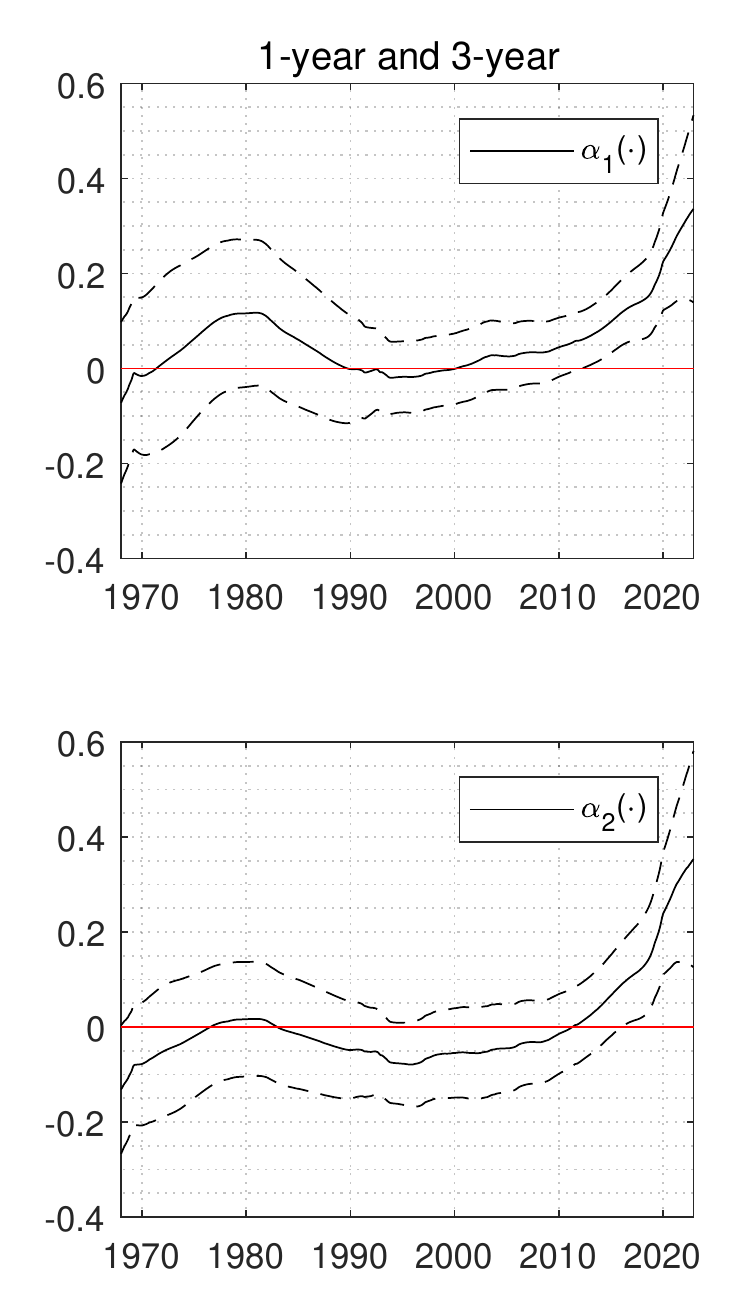}}
		{\includegraphics[width=4cm]{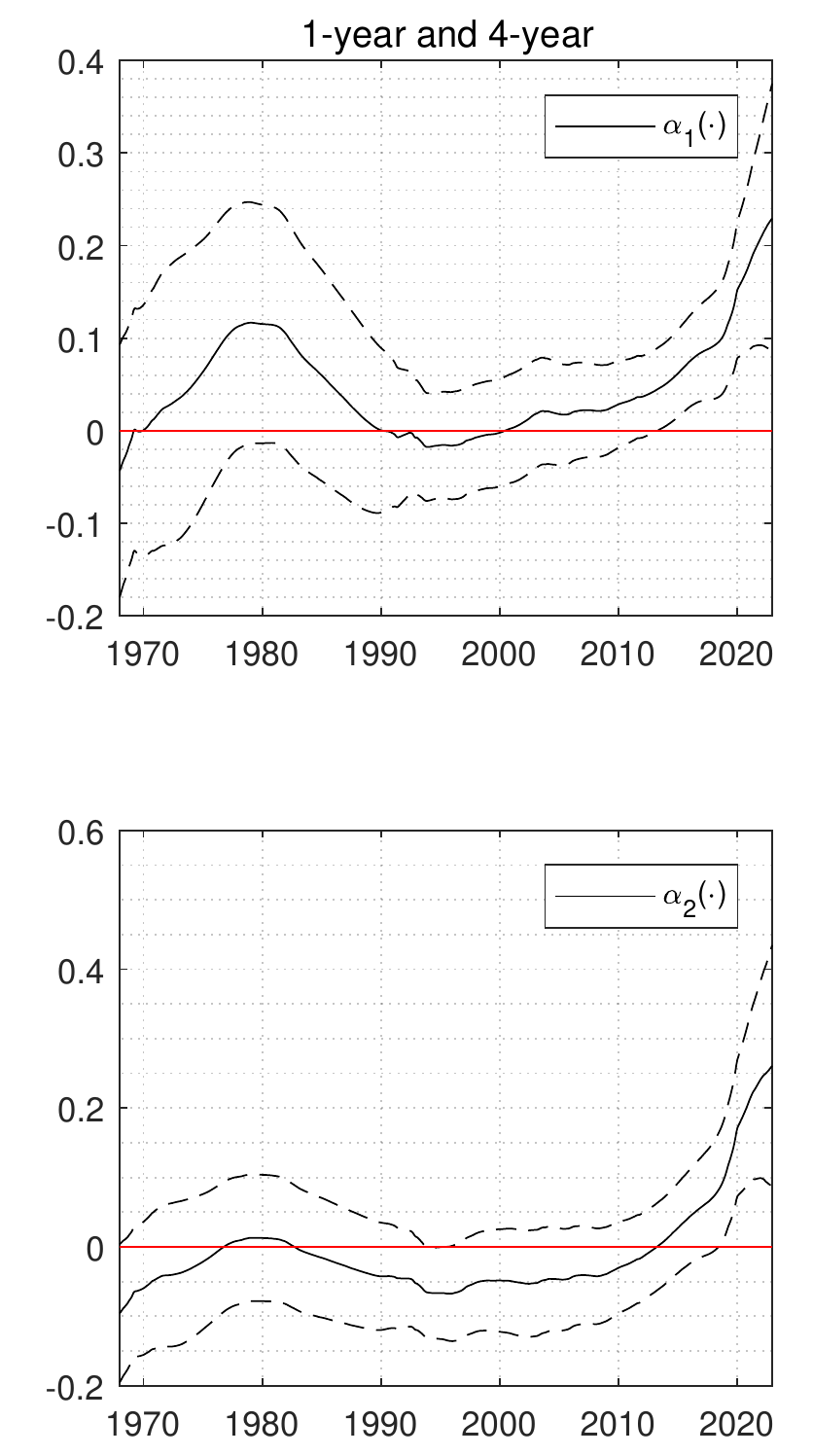}}
		{\includegraphics[width=4cm]{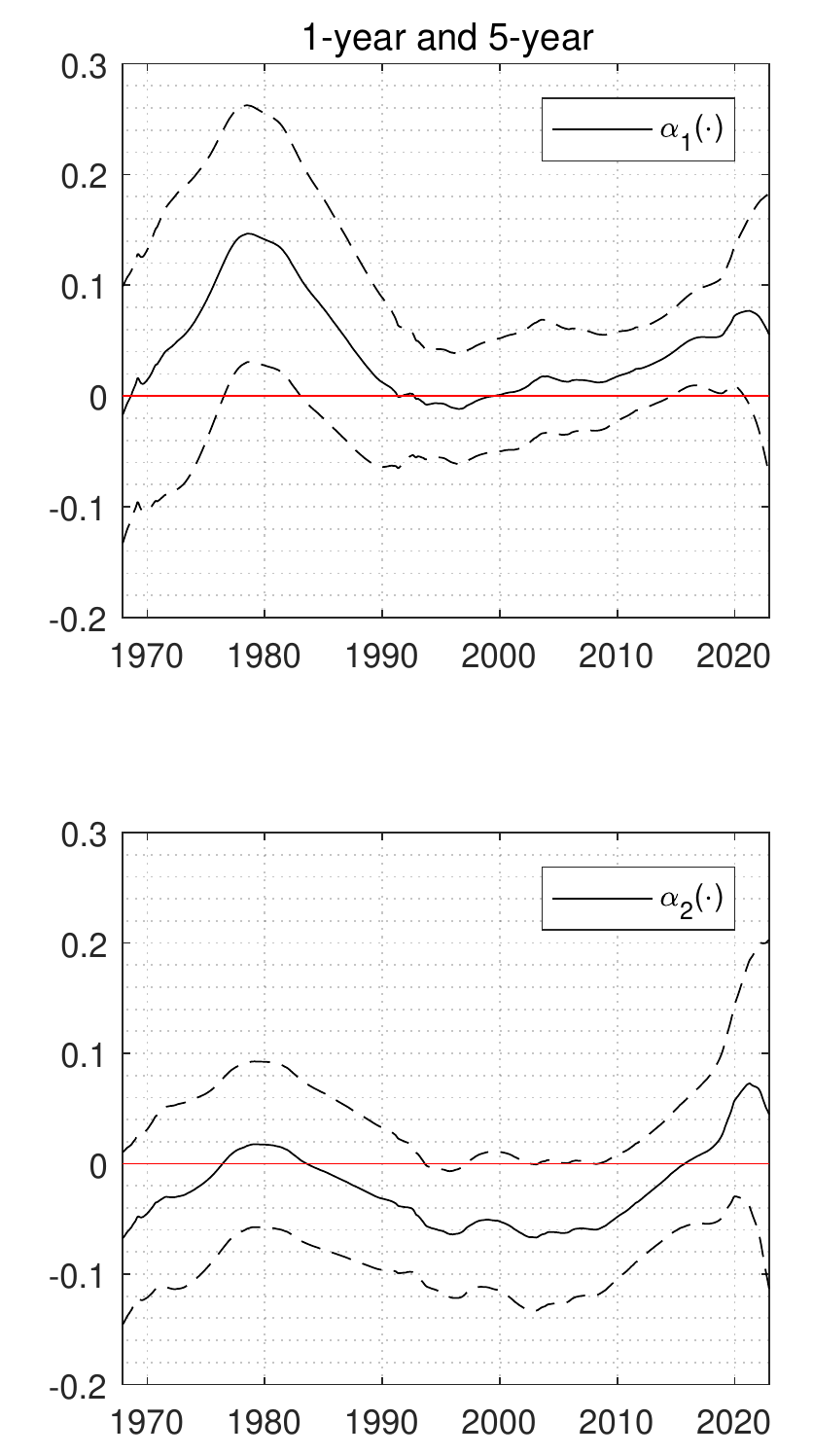}}\\
		{\includegraphics[width=4cm]{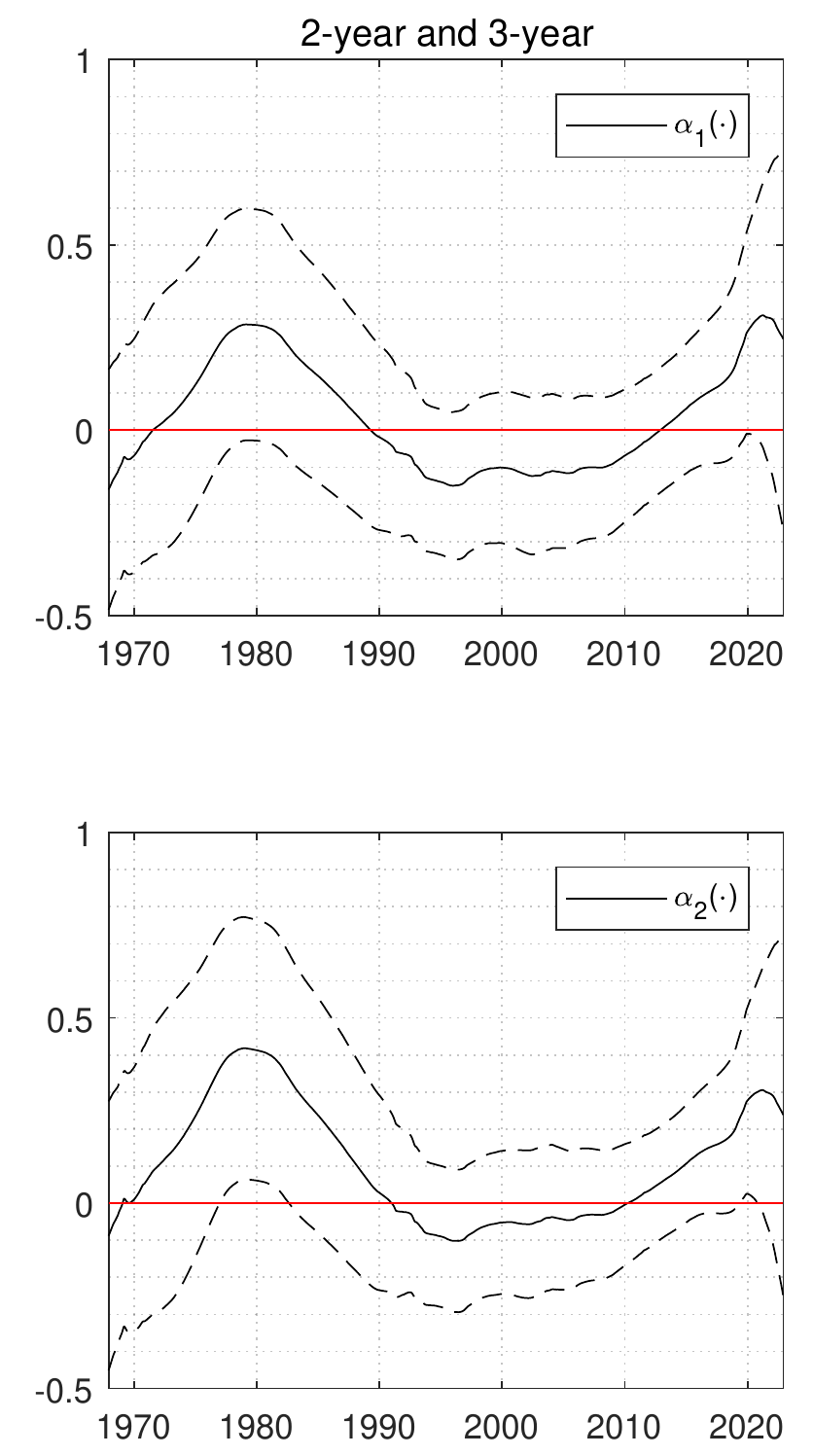}}
		{\includegraphics[width=4cm]{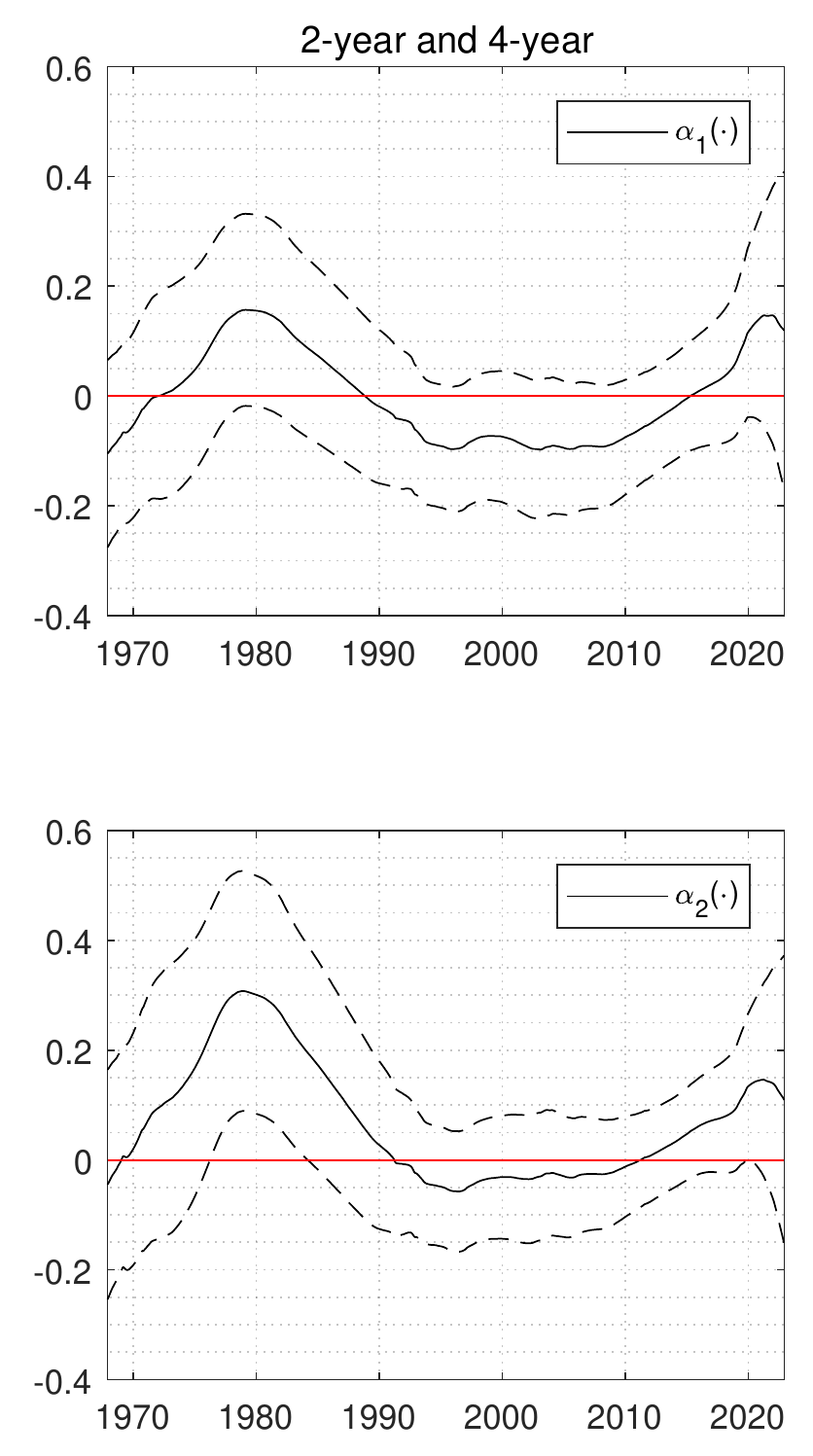}}
		{\includegraphics[width=4cm]{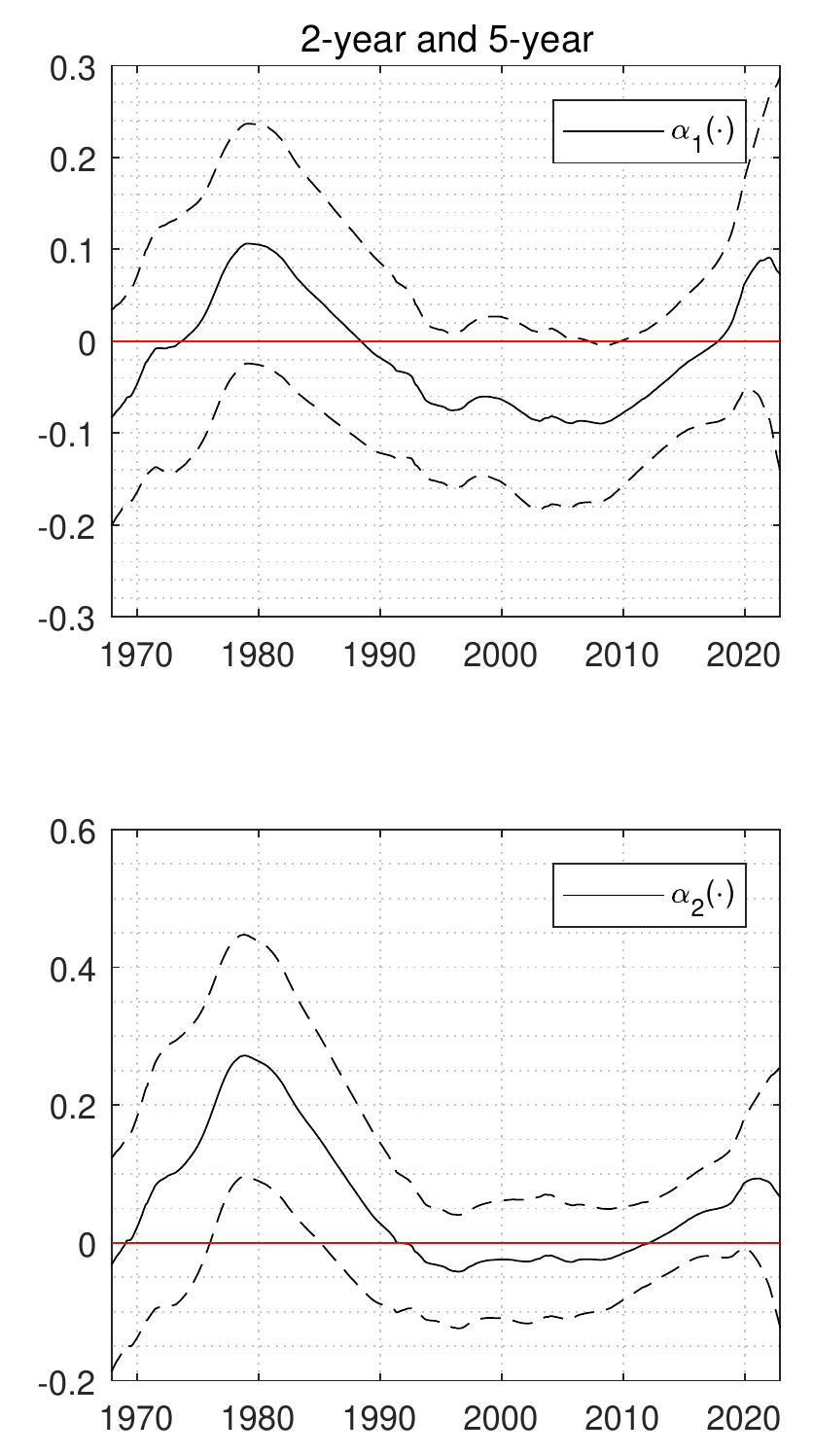}}\\
		{\includegraphics[width=4cm]{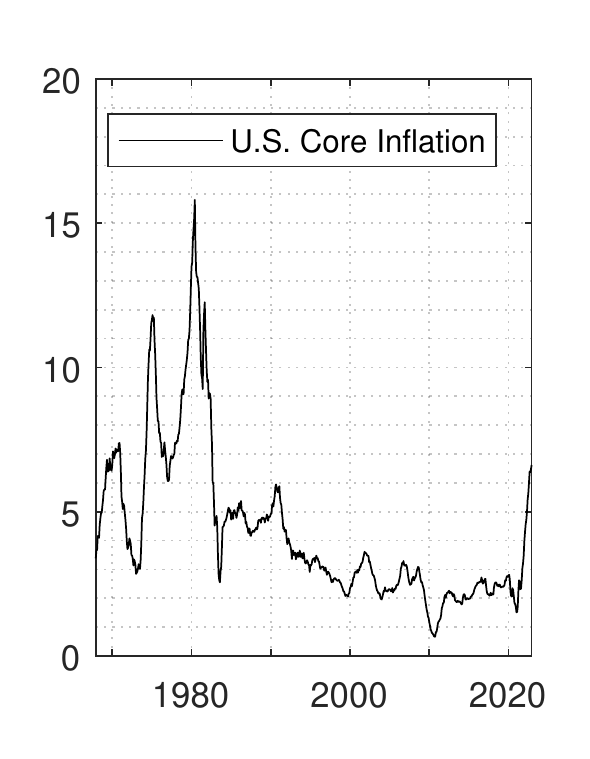}}\\
		\caption{\small Estimates of time-varying adjustment coefficients and their 95\% point-wise confidence intervals, as well as the U.S. core inflation rate.}\label{Fg2}
	\end{figure}

	\newpage
	
	\setcounter{page}{1}
	\linespread{1.1}
	\small
	
	\begin{center}
\bf \Large

Online Supplementary Appendix A to \\``Time-Varying Vector Error-Correction Models: \\ Estimation and Inference"

\end{center}	

\begin{center}
$^{\ast}${\sc Jiti Gao} and $^{\ast}${\sc Bin Peng} and $^{\ddag}${\sc Yayi Yan}

\medskip

    $^{\ast}$Department of Econometrics and Business Statistics,  Monash University\\
    
\medskip   
    
    $^{\dag}$School of Statistics and Management, Shanghai University of Finance and Economics
\end{center}	
	
This file is organised as follows. In Appendix \ref{App0}, we present some extra results on parameter stability test, which include the study on a sequence of local alternatives, and a simulation-assisted numerical procedure to improve finite sample performance. Appendix \ref{AppNot} includes the notation and mathematical symbols, which will be repeatedly used throughout the development. In Appendix \ref{AppLem}, we provide the preliminary lemmas to facilitate the development of the main theorems of the paper. Appendix \ref{AppPMain} includes the proofs of the main results, while Appendix \ref{AppPPre} covers the proofs of the preliminary lemmas.

	\section*{Appendix A} 
	
	\renewcommand{\theequation}{A.\arabic{equation}}
	\renewcommand{\thesection}{A.\arabic{section}}
	\renewcommand{\thefigure}{A.\arabic{figure}}
	\renewcommand{\thetable}{A.\arabic{table}}
	\renewcommand{\thelemma}{A.\arabic{lemma}}
	\renewcommand{\theassumption}{A.\arabic{assumption}}
	\renewcommand{\thetheorem}{A.\arabic{theorem}}
	\renewcommand{\theproposition}{A.\arabic{proposition}}
	
	\setcounter{equation}{0}
	\setcounter{lemma}{0}
	\setcounter{section}{0}
	\setcounter{table}{0}
	\setcounter{figure}{0}
	\setcounter{assumption}{0}
	\setcounter{proposition}{0}
	\numberwithin{equation}{section}
	
	\section{Extra Results on Parameter Stability Test}\label{App0}
	
		Of interest, we also consider a sequence of local alternatives of the form:
	\begin{eqnarray}\label{Eq2.9}
		\mathbb{H}_1: \mathbf{C}\mathbf{b}(\tau) = \mathbf{c} + d_T \mathbf{f}(\tau),
	\end{eqnarray}
	where $\mathbf{f}(\tau)$ is a twice continuously differentiable vector of functions, and $d_T \to 0$. The term $d_T \mathbf{f}(\tau)$ characterizes the departure of the time--varying coefficient $\mathbf{C}\mathbf{b}(\tau)$ from the constant $\mathbf{c}$. According to the development of Theorem \ref{Thm4}, it is straightforward to obtain the following corollary.
	
	\begin{corollary}\label{Coro3}
		Let the conditions of Theorem \ref{Thm4} hold. Under $\mathbb{H}_1$ of \eqref{Eq2.9}, if $d_T = T^{-1/2}h^{-1/4}$, 
		\begin{enumerate}
			\item $T\sqrt{h}\left(\widehat{Q}_{\mathbf{C}, \widehat{\mathbf{H}}} -\frac{1}{Th}s\widetilde{v}_0\right)\to_D N\left(\delta_1,4s C_B\right),$
			
			\item $\Pr\left(\widehat{Q}_{\mathbf{C}, \widehat{\mathbf{H}}}^*>q_{1-\alpha}\right) \to \Phi\left(q_{\alpha} + \frac{\delta_1}{2\sqrt{ s C_B}} \right),$
		\end{enumerate}
		where $\delta_1 = \int_{0}^{1}\mathbf{f}(\tau)^\top (\mathbf{C} \mathbf{V}_{\mathbf{b}}(\tau)\mathbf{C}^\top )^{-1}\mathbf{f}(\tau) \mathrm{d}\tau$, $q_{\alpha}$ stands for the $\alpha^{th}$ quantile of standard normal distribution and $\widehat{Q}_{\mathbf{C}, \widehat{\mathbf{H}}}^*=\frac{T\sqrt{h}\left(\widehat{Q}_{\mathbf{C}, \widehat{\mathbf{H}}} -\frac{1}{Th}s \widetilde{v}_0\right)}{\sqrt{4 s C_B}}$ is the normalized test statistic.
	\end{corollary}
	
	Corollary \ref{Coro3} shows that the test has a non-trivial power against $\mathbb{H}_1$ when $d_T = T^{-1/2}h^{-1/4}$. If $T^{-1/2}h^{-1/4} = o(d_T)$, the power of the test converges to 1, i.e., $\Pr(\widehat{Q}_{\mathbf{C}, \widehat{\mathbf{H}}}^*>q_{1-\alpha})\to 1$.
	
	\medskip
	
	To improve the finite sample performance of the test, we next propose a simulation-assisted testing procedure. A similar procedure has also been adopted by \cite{zhang2012inference}, and \cite{truquet2017parameter} for the same purpose in the context of time--varying models.
	
	\medskip
	
	{\bf Algorithm --- a simulation--assisted testing procedure}
	\begin{itemize}		
		\item[] Step 1: Use the sample $\{\mathbf{y}_t\}$ to estimate the unrestricted model, and then compute $\widehat{Q}_{\mathbf{C}, \widehat{\mathbf{H}}}$ based on \eqref{Eq2.8}.
		
		Step 2: Generate i.i.d. $d$-dimensional standard normal random vectors $\{\Delta\mathbf{y}_t^*\}$ and generate i.i.d. $r_0$-dimensional standard normal random vectors $\{\mathbf{z}_t^*\}$.
		
		Step 3: Compute the bootstrap statistic $\widetilde{Q}_{\mathbf{C}, \widehat{\mathbf{H}}}^b$ in the same way as $\widehat{Q}_{\mathbf{C}, \widehat{\mathbf{H}}}$ based on the following regression model
		\begin{equation}\label{Eq2.10}
			\Delta\mathbf{y}_t^* = \bm{\alpha}^*(\tau_t)\mathbf{z}_t^* + \sum_{j=1}^{p_0-1} \bm{\Gamma}_j^*(\tau_t)	\Delta\mathbf{y}_{t-j}^* + \bm{\omega}^*(\tau_t) \bm{\varepsilon}_t^*.
		\end{equation}

		Step 4: Repeat Steps 2--3 $B$ times to obtain $B$ bootstrap test statistics $\{\widetilde{Q}_{\mathbf{C}, \widehat{\mathbf{H}}}^b\}_{b=1}^{B}$, as well as its empirical quantile $\widehat{q}_{1-\alpha}$. We reject the null hypothesis \eqref{Eq3.1} at level $\alpha$ if $\widehat{Q}_{\mathbf{C}, \widehat{\mathbf{H}}}>\widehat{q}_{1-\alpha}$.
	\end{itemize}
	
	Note that we can regard $\Delta\mathbf{y}_t^*$ as from the data generating process \eqref{Eq2.10} with $\bm{\alpha}^*(\tau) = \mathbf{0}_{d\times r_0}$, $\bm{\Gamma}_j^*(\tau) = \mathbf{0}_d$, $\bm{\omega}^*(\tau_t) = \mathbf{I}_d$, and further we replace the cointegrated component $\bm{\beta}^\top \mathbf{y}_{t-1}$ by $\mathbf{z}_t^*$. We then have the following asymptotic property of the simulated test statistic $\widetilde{Q}_{\mathbf{C}, \widehat{\mathbf{H}}}^b$.
	
	\begin{corollary}\label{Coro4}
		Let Assumption \ref{Ass2} hold. Suppose that $\{\Delta\mathbf{y}_t^*\}$ and $\{\mathbf{z}_t^*\}$ are generated from i.i.d. standard multivariate normal distributions, then we have
		$$
		T\sqrt{h}\left(\widetilde{Q}_{\mathbf{C}, \widehat{\mathbf{H}}}^b -\frac{1}{Th}s\widetilde{v}_0\right)\to_D N(0,4s C_B).
		$$
	\end{corollary}
	By Theorem \ref{Thm4} and Corollary \ref{Coro4}, $\widehat{Q}_{\mathbf{C}, \widehat{\mathbf{H}}}$ and $\widetilde{Q}_{\mathbf{C}, \widehat{\mathbf{H}}}^b$ have the same asymptotic distribution under the null, which suggests that instead of using the central limit theorem, the critical value of $\widehat{Q}_{\mathbf{C}, \widehat{\mathbf{H}}}$ can be obtained by simulating $\widetilde{Q}_{\mathbf{C}, \widehat{\mathbf{H}}}^b$.

	\section{Extra Notation and Mathematical Symbols}\label{AppNot}
	
 Before proceeding further, we denote a few mathematical symbols to felicitate the development. Let $\mathbf{W}_d^{*}(\cdot,\cdot)$ be a $d$-dimensional Brownian motion and independent of $\mathbf{W}_d(\cdot,\cdot)$. In view of the facts that $[\bm\beta, \bm\beta_\perp]$ is nonsingular and  $\bm\beta^\top \bm\beta_\perp=0$, we have $\mathbf{I}_d =\mathbf{P}_{\bm\beta} + \mathbf{P}_{\bm\beta_\perp}$, so it is easy to see that $\mathbf{J}^{-1} = [\overline{\bm{\beta}},\bm{\beta}_{\perp}]$.

Let

\begin{eqnarray*}
&&\mathbf{D}_T = \mathrm{diag}\{T\sqrt{h}, Th \mathbf{I}_{d-r_0-1}\} ,\quad \widetilde{\mathbf{D}}_T = \mathrm{diag}\{\mathbf{D}_T, \sqrt{Th} \mathbf{I}_{r_0}\},\quad \mathbf{D}_T^{*}=\mathrm{diag}\{\widetilde{\mathbf{D}}_T,\sqrt{Th}\mathbf{I}_{d(p_0-1)}\},\nonumber\\
&& \mathbf{Q}(\tau) = \left[\bm{\alpha}_{\perp}(\tau),\bm{\beta} \right]^\top,\quad \mathbf{Q}^*(\tau) = \mathrm{diag}(\mathbf{Q}(\tau),\mathbf{I}_{d(p_0-1)}), \nonumber\\
&&\bm{\Sigma}_{\mathbf{z},j}(\tau) = E(\widetilde{\mathbf{z}}_{t}(\tau)\widetilde{\mathbf{z}}_{t-j}^\top(\tau)) =  \sum_{k=0}^{\infty}\bm{\Psi}_{k+j}(\tau)\bm{\Omega}(\tau)\bm{\Psi}_k^\top(\tau),\nonumber\\
&&\bm{\Sigma}_{\mathbf{zx}}(\tau) = E(\widetilde{\mathbf{z}}_t(\tau)\Delta\widetilde{\mathbf{x}}_t^\top(\tau)),\quad \bm{\Sigma}_{\mathbf{x}}(\tau) = E(\Delta \widetilde{\mathbf{x}}_{t}(\tau)\Delta \widetilde{\mathbf{x}}_{t}^\top(\tau)),\nonumber
\end{eqnarray*}	
where  $\widetilde{\mathbf{z}}_t(\tau) = \sum_{i=1}^{p}\mathbf{B}_i(\tau)\widetilde{\mathbf{z}}_{t-i}(\tau) + \widetilde{\mathbf{u}}_t(\tau)$.

	 Let
	 
	 \begin{eqnarray*}
&&	 \mathbf{q}(\tau) = \bm{\alpha}_{\perp}^\top(\tau) \mathbf{W}_d(\tau,\bm{\Sigma}_{\mathbf{y}}(\tau)),\quad \bm{\xi}(\tau) = \mathbf{q}(\tau)/\sqrt{\mathbf{q}^\top(\tau)\mathbf{q}(\tau)} ,\nonumber\\
&&\mathbf{q}_T(\tau) = T^{-1/2}\bm{\alpha}_{\perp}^\top(\tau)\mathbf{y}_{\delta_T},\quad  \bm{\xi}_T(\tau) = \mathbf{q}_T(\tau)/\sqrt{\mathbf{q}_T^\top(\tau)\mathbf{q}_T(\tau)},\nonumber
	 \end{eqnarray*}
	 where $\delta_T = \lfloor T(\tau-h) \rfloor$. Accordingly, let $\bm{\xi}_{\perp}(\tau)$ be the $(d-r_0)\times(d-r_0-1)$ orthogonal complement matrix such that $\bm{\Xi}(\tau) = [\bm{\xi}(\tau),\bm{\xi}_{\perp}(\tau)]$ and $\bm{\Xi}^\top(\tau)\bm{\Xi}(\tau)=\mathbf{I}_{d-r_0}$, and let $\bm{\Xi}_T(\tau) = [\bm{\xi}_T(\tau),\bm{\xi}_{T,\perp}(\tau)]$. In addition, let $\widetilde{\bm{\Xi}}_T(\tau) = \mathrm{diag}\{\bm{\Xi}_T(\tau),  \mathbf{I}_{r_0}\}$ and $ \bm{\Xi}_T^*(\tau) = \mathrm{diag}\{\widetilde{\bm{\Xi}}_T(\tau),\mathbf{I}_{d(p_0-1)} \}$.
	
 Define 
\begin{eqnarray*}
	\bm{\Delta}_{l}(\tau) = \left[\begin{matrix}
		\bm{\Delta}_{l,1}(\tau)      & \bm{\Delta}_{l,2}(\tau) \\
		\bm{\Delta}_{l,2}^\top(\tau) & \bm{\Delta}_{l,3}(\tau) \\
	\end{matrix} \right],\nonumber
\end{eqnarray*}
	where 
	
	\begin{eqnarray*}
	\bm{\Delta}_{l,1}(\tau) &=& \widetilde{c}_l\mathbf{q}(\tau)^\top\mathbf{q}(\tau),\nonumber\\
	\bm{\Delta}_{l,2}(\tau) &=& \sqrt{2\mathbf{q}(\tau)^\top\mathbf{q}(\tau)} \int_{-1}^{1}\mathbf{W}_d^{*,\top}((u+1)/2,\bm{\Sigma}_{\bm{\alpha}}(\tau))u^lK(u)\mathrm{d}u  \, \bm{\xi}_{\perp}(\tau) ,\nonumber\\
	\bm{\Delta}_{l,3}(\tau) &=& 2 \bm{\xi}_{\perp}^\top(\tau) \int_{-1}^{1}\mathbf{W}_d^{*}((u+1)/2,\bm{\Sigma}_{\bm{\alpha}}(\tau))\mathbf{W}_d^{*,\top}((u+1)/2,\bm{\Sigma}_{\bm{\alpha}}(\tau))u^lK(u)\mathrm{d}u \,  \bm{\xi}_{\perp}(\tau),\nonumber
	\end{eqnarray*}
and $\bm{\Sigma}_{\bm{\alpha}}(\tau) = \bm{\alpha}_{\perp}^\top(\tau) \mathbf{P}_{\bm{\beta}_{\perp}} \bm{\Psi}_{\tau}(1) \bm{\Omega}(\tau) \bm{\Psi}_{\tau}^\top(1) \mathbf{P}_{\bm{\beta}_{\perp}} \bm{\alpha}_{\perp}(\tau)$. By construction, $\bm{\Delta}_l(\tau)$ is nonsingular with probability $1$ for $l=0,2$.

The sample version of $\bm{\Delta}_{l}(\tau) $ is defined as

\begin{eqnarray*}
\bm{\Delta}_{T,l}(\tau) &=& \left[\begin{matrix}
					\bm{\Delta}_{T,l,1}(\tau) & \bm{\Delta}_{T,l,2}(\tau)\\
					\bm{\Delta}_{T,l,2}^\top(\tau) & \bm{\Delta}_{T,l,3}(\tau)\\
				\end{matrix} \right],\nonumber
\end{eqnarray*}
where 

\begin{eqnarray*}
				\bm{\Delta}_{T,l,1}(\tau) &=&  \widetilde{c}_l\mathbf{q}_T^\top(\tau)\mathbf{q}_T(\tau), \nonumber\\ 
				\bm{\Delta}_{T,l,2}(\tau) &=&  \bm{\xi}_T^\top(\tau)\mathbf{q}_T(\tau) \frac{\sqrt{2}}{Th} \sum_{t=1}^{T}\left(\frac{1}{\sqrt{2Th}}\sum_{j=\delta_T+1}^{t-1} \mathbf{P}_{\bm{\beta}_{\perp}}\bm{\Psi}_{\tau}(1)\bm{\omega}(\tau)\bm{\varepsilon}_j\right)^\top\nonumber\\ 
				&& \times \left(\frac{\tau_t-\tau}{h}\right)^lK\left(\frac{\tau_t-\tau}{h}\right)\bm{\alpha}_{\perp}(\tau)\bm{\xi}_{T,\perp}(\tau),\nonumber\\
				\bm{\Delta}_{T,l,3}(\tau)
				&= &\bm{\xi}_{T,\perp}^\top(\tau)\bm{\alpha}_{\perp}^\top(\tau) \frac{2}{Th} \sum_{t=1}^{T}\left(\frac{1}{\sqrt{2Th}}\sum_{j=\delta_T+1}^{t-1} \mathbf{P}_{\bm{\beta}_{\perp}}\bm{\Psi}_{\tau}(1)\bm{\omega}(\tau)\bm{\varepsilon}_j\right) \nonumber\\
				&&\times\left(\frac{1}{\sqrt{2Th}}\sum_{j=\delta_T+1}^{t-1} \mathbf{P}_{\bm{\beta}_{\perp}}\bm{\Psi}_{\tau}(1)\bm{\omega}(\tau)\bm{\varepsilon}_j\right)^\top \left(\frac{\tau_t-\tau}{h}\right)^lK\left(\frac{\tau_t-\tau}{h}\right)\bm{\alpha}_{\perp}(\tau)\bm{\xi}_{T,\perp}(\tau).\nonumber
			\end{eqnarray*}

	\section{Preliminary Lemmas}\label{AppLem}

\begin{lemma}\label{L4}

Under Assumptions \ref{Ass1} and \ref{Ass2}.1, for any given $\tau \in (0,1)$, 

\begin{enumerate}[wide, labelwidth=!, labelindent=0pt]
\item 	$\frac{1}{T^2h}\sum_{t=1}^{T}\mathbf{y}_{t-1}\mathbf{y}_{t-1}^\top K\left(\frac{\tau_t-\tau}{h}\right)\to_D \mathbf{W}_d(\tau,\bm{\Sigma}_{\mathbf{y}}(\tau)) \mathbf{W}_d^\top(\tau,\bm{\Sigma}_{\mathbf{y}}(\tau))$;
			
\item $\mathbf{D}_T^{-1} \bm{\Xi}_T^\top(\tau)\left[ \sum_{t=1}^{T}\bm{\alpha}_{\perp}^\top(\tau)\mathbf{y}_{t-1}\mathbf{y}_{t-1}^\top\bm{\alpha}_{\perp}(\tau)\left(\frac{\tau_t-\tau}{h}\right)^l K\left(\frac{\tau_t-\tau}{h}\right)\right] \bm{\Xi}_T(\tau) \mathbf{D}_T^{-1}  \to_D \bm{\Delta}_{l}(\tau)$;
			
\item $\frac{1}{\sqrt{Th}}\mathbf{D}_T^{-1} \bm{\Xi}_T^\top(\tau) \sum_{t=1}^{T}\bm{\alpha}_{\perp}^\top(\tau)\mathbf{y}_{t-1}\mathbf{y}_{t-1}^\top\bm{\beta} \left(\frac{\tau_t-\tau}{h}\right)^l K\left(\frac{\tau_t-\tau}{h}\right) = O_P(1/\sqrt{Th})$;

\item $\frac{1}{Th}\sum_{t=1}^{T}\bm{\beta}^\top \mathbf{y}_{t-1} \mathbf{y}_{t-1}^\top \bm{\beta}\left(\frac{\tau_t-\tau}{h}\right)^lK\left(\frac{\tau_t-\tau}{h}\right) = \left\{ \begin{matrix}
					\widetilde{c}_l\bm{\beta}^\top\bm{\Sigma}_{\mathbf{z},0}(\tau)\bm{\beta}+O_P(h^2+ 1/\sqrt{Th}) \quad \text{if}\ $l$\ \text{is even}\\
					\widetilde{c}_l\bm{\beta}^\top\bm{\Sigma}_{\mathbf{z},0}(\tau)\bm{\beta}+O_P(h+ 1/\sqrt{Th}) \quad \text{if}\ $l$\ \text{is odd}
				\end{matrix}  \right.$;
			
			\item $ \widetilde{\mathbf{D}}_T^{+} \widetilde{\bm{\Xi}}_T^{\top}(\tau)\left[ \sum_{t=1}^{T}\mathbf{Q}(\tau)\mathbf{y}_{t-1}\mathbf{y}_{t-1}^\top\mathbf{Q}^\top(\tau) \left(\frac{\tau_t-\tau}{h}\right)^lK\left(\frac{\tau_t-\tau}{h}\right)\right] \widetilde{\bm{\Xi}}_T(\tau) \widetilde{\mathbf{D}}_T^{+}   \\ \to_D  \left[\begin{matrix}
					\bm{\Delta}_{l}(\tau) & \mathbf{0}_{(d-r_0)\times r_0}\\
			\mathbf{0}_{r_0\times(d-r_0)}	& \widetilde{c}_l\bm{\beta}^\top\bm{\Sigma}_{\mathbf{z},0}(\tau)\bm{\beta} \\
				\end{matrix} \right]$;

			\item $\frac{1}{\sqrt{Th}}\widetilde{\mathbf{D}}_T^{+} \widetilde{\bm{\Xi}}_T^\top(\tau) \sum_{t=1}^{T}\mathbf{Q}(\tau)\mathbf{y}_{t-1} \Delta \mathbf{x}_{t-1}^\top\left(\frac{\tau_t-\tau}{h}\right)^l K\left(\frac{\tau_t-\tau}{h}\right) \\
				=\left\{\begin{matrix}
					\left[\begin{matrix}
						\mathbf{0}_{(d-r_0)\times d(p_0-1)}+O_P(1/\sqrt{Th})  \\ \widetilde{c}_l\bm{\beta}^\top\bm{\Sigma}_{\mathbf{zx}}(\tau) + O_P(h^2 + 1/\sqrt{Th})\\
					\end{matrix} \right]  \quad \text{if}\ $l$\ \text{is even}\\
					\left[\begin{matrix}
						\mathbf{0}_{(d-r_0)\times d(p_0-1)}+O_P(1/\sqrt{Th})  \\ \widetilde{c}_l\bm{\beta}^\top\bm{\Sigma}_{\mathbf{zx}}(\tau) + O_P(h + 1/\sqrt{Th}) \\
					\end{matrix} \right]  \quad \text{if}\ $l$\ \text{is odd}\\
				\end{matrix} \right.,$	\\			
			
			\item $\mathbf{D}_T^{*,-1} \bm{\Xi}_T^{*,\top}(\tau) \mathbf{Q}^*(\tau)\mathbf{S}_{T,l}(\tau) \mathbf{Q}^{*,\top}(\tau)\bm{\Xi}_T^{*}(\tau) \mathbf{D}_T^{*,-1} \\
				\to_D  \left[\begin{matrix}
					\bm{\Delta}_{l}(\tau) & \mathbf{0}_{(d-r_0)\times r_0} &	\mathbf{0}_{(d-r_0)\times d(p_0-1)} \\ \mathbf{0}_{r_0\times(d-r_0)}& \widetilde{c}_l\bm{\beta}^\top\bm{\Sigma}_{\mathbf{z,0}}(\tau)\bm{\beta} & \widetilde{c}_l\bm{\beta}^\top\bm{\Sigma}_{\mathbf{zx}}(\tau)\\
					\mathbf{0}_{d(p_0-1)\times(d-r_0)}
					 & \widetilde{c}_l\bm{\Sigma}_{\mathbf{zx}}^\top(\tau)\bm{\beta} & \widetilde{c}_l\bm{\Sigma}_{\mathbf{x}}(\tau)\\
				\end{matrix} \right].$ 
		\end{enumerate}
		
	\end{lemma}

	\begin{lemma}\label{L5}
		Let Assumptions \ref{Ass1} and \ref{Ass2}.1 hold. Then if $\frac{T^{1-\frac{2}{\delta}}h}{(\log T)^{1-\frac{2}{\delta}}} \to \infty$,
		\begin{enumerate}[wide, labelwidth=!, labelindent=0pt]
			\item [1.] $\sup_{\tau\in[0,1]}\left|	\frac{1}{Th}\sum_{t=1}^{T}\widetilde{\mathbf{z}}_{t}(\tau_t)\left(\frac{\tau_t-\tau}{h}\right)^lK\left(\frac{\tau_t-\tau}{h}\right)\right| = O_P\left(\sqrt{\log T / (Th)}\right)$;
		\end{enumerate}
		
		\medskip
		
\noindent 		In addition, if $\frac{T^{1-\frac{4}{\delta}}h}{(\log T)^{1-\frac{4}{\delta}}} \to \infty$ and $E\left(\left|\bm{\varepsilon}_t\right|^4\mid \mathcal{F}_{t-1}\right) < \infty$ a.s. Then,
		\begin{enumerate}[wide, labelwidth=!, labelindent=0pt]
			\item [2.] for any fixed integer $p \geq 0$,
			\begin{eqnarray*}
				&&\sup_{\tau\in[0,1]}\left|	\frac{1}{Th}\sum_{t=1}^{T-p}\left(\widetilde{\mathbf{z}}_{t}(\tau_t)\widetilde{\mathbf{z}}_{t+p}^\top(\tau_{t+p})-E(\widetilde{\mathbf{z}}_{t}(\tau_t)\widetilde{\mathbf{z}}_{t+p}^\top(\tau_{t+p}))\right)\left(\frac{\tau_t-\tau}{h}\right)^l K\left(\frac{\tau_t-\tau}{h}\right)\right| =O_P\left(\sqrt{\log T / (Th)}\right).\nonumber
			\end{eqnarray*}
			
			\item [3.] $\sup_{\tau\in[0,1]}\left|	\frac{1}{Th}\sum_{t=1}^{T}\widetilde{\mathbf{z}}_{t-p}(\tau_t)\mathbf{u}_t\left(\frac{\tau_t-\tau}{h}\right)^lK\left(\frac{\tau_t-\tau}{h}\right)\right| = O_P\left(\sqrt{\log T / (Th)}\right)$ for any fixed integer $p \geq 1$.
		\end{enumerate}
		
	\end{lemma}
	
	\begin{lemma}\label{L6}
		Let Assumptions \ref{Ass1}--\ref{Ass2} hold.  Define $S_T^* = \max_{t \leq T}|\sum_{j=1}^{t}\widetilde{\mathbf{z}}_j(\tau_j)|$. Then for $x\geq c\sqrt{T\log T}$ and some constant $c>0$, we have $\Pr\left( S_T^* \geq x\right) = O\left(\frac{T}{x^\delta}\right)$.
		
	\end{lemma}

	\begin{lemma}\label{L7}
		Let Assumptions \ref{Ass1} and \ref{Ass2}.1 hold.  
\begin{enumerate}[wide, labelwidth=!, labelindent=0pt]
\item[1.]  $\sup_{u\in[0,1]}\left|T^{-1/2}\mathbf{y}_{\lfloor Tu \rfloor}-\mathbf{W}_d(u,\bm{\Sigma}_{\mathbf{y}}(u))\right|=o_P(1)$.
\end{enumerate}		
\medskip

\noindent		Suppose further $T^{1-\frac{2}{\delta}}h\to \infty$, we have
		\begin{enumerate}[wide, labelwidth=!, labelindent=0pt]
			\item [2.] $\sup_{\tau\in[h,1-h]}\left|	\frac{1}{T^2h}\sum_{t=1}^{T}\mathbf{y}_{t-1}\mathbf{y}_{t-1}^\top K\left(\frac{\tau_t-\tau}{h}\right)-\mathbf{W}_d(\tau,\bm{\Sigma}_{\mathbf{y}}(\tau))\mathbf{W}_d^\top(\tau,\bm{\Sigma}_{\mathbf{y}}(\tau))\right| = o_P(1)$;
			
			\item [3.] $\sup_{\tau\in[h,1-h]}\left|\mathbf{D}_T^{-1} \bm{\Xi}_T^\top(\tau)\left[ \sum_{t=1}^{T}\bm{\alpha}_{\perp}^\top(\tau)\mathbf{y}_{t-1}\mathbf{y}_{t-1}^\top\bm{\alpha}_{\perp}(\tau)\left(\frac{\tau_t-\tau}{h}\right)^l K\left(\frac{\tau_t-\tau}{h}\right)\right] \bm{\Xi}_T(\tau) \mathbf{D}_T^{-1} - \bm{\Delta}_{l}(\tau)\right| = o_P(1)$;
		\end{enumerate}
		
\medskip

\noindent	Suppose further $\frac{T^{1-\frac{4}{\delta+1}}h }{(\log T)^{(\delta-1)/(\delta+1)}} \to \infty$ and $\{\bm{\varepsilon}_t\}$ is a sequence of independent random variables, we have	
		\begin{enumerate}[wide, labelwidth=!, labelindent=0pt]
			\item [4.] $\sup_{0\leq \tau\leq 1}\left|\frac{1}{\sqrt{Th}}\mathbf{D}_T^{-1} \bm{\Xi}_T^\top(\tau) \sum_{t=1}^{T}\bm{\alpha}_{\perp}^\top(\tau)\mathbf{y}_{t-1}\mathbf{y}_{t-1}^\top\bm{\beta} \left(\frac{\tau_t-\tau}{h}\right)^l K\left(\frac{\tau_t-\tau}{h}\right)\right|= O_P(\sqrt{\log T/(Th)});$
			
			\item [5.] $\sup_{0\leq \tau\leq 1}\left|\frac{1}{\sqrt{Th}}\mathbf{D}_T^{-1} \bm{\Xi}_T^\top(\tau) \sum_{t=1}^{T}\bm{\alpha}_{\perp}^\top(\tau)\mathbf{y}_{t-1}\mathbf{u}_{t}^\top \left(\frac{\tau_t-\tau}{h}\right)^l K\left(\frac{\tau_t-\tau}{h}\right)\right|= O_P(\sqrt{\log T/(Th)})$.
		\end{enumerate}
	\end{lemma}
	
	\begin{lemma}\label{L8}
		Let Assumptions \ref{Ass1}--\ref{Ass2} hold. In addition, Suppose further $\frac{Th^2}{(\log T)^2} \to \infty$ and $Th^6 \to 0$,  we have
		\begin{enumerate}[wide, labelwidth=!, labelindent=0pt]
			\item $\sup_{\tau\in[0,1]}\left|\frac{1}{T^{3/2}h}\sum_{t=1}^{T}(\Delta\mathbf{x}_{t-1}\mathbf{y}_{t-1}^\top) \otimes \widehat{\bm{\alpha}}(\tau_t)\left(\frac{\tau_t-\tau}{h}\right)^l K\left(\frac{\tau_t-\tau}{h}\right)\right|= O_P(\sqrt{\log T/(Th)})$;
			
			\item $\frac{1}{T^2}\sum_{t=1}^{T}\widetilde{\mathbf{R}}_t(\widehat{\bm{\alpha}})\widehat{\bm{\Omega}}^{-1}(\tau_t)\widetilde{\mathbf{R}}_t^\top(\widehat{\bm{\alpha}}) \to_D \int_{0}^{1}\mathbf{W}_{d-r_0}(u) \mathbf{W}_{d-r_0}^\top(u) \otimes \bm{\alpha}^\top(u)\bm{\Omega}^{-1}(u)\bm{\alpha}(u)\mathrm{d}u$, where $\mathbf{W}_{d-r_0}(u) = \left[\mathbf{0}_{(d-r_0)\times r_0},\mathbf{I}_{d-r_0}\right]\mathbf{W}_d(u,\bm{\Sigma}_{\mathbf{y}}(u))$;
			
			\item $\frac{1}{T}\sum_{t=1}^{T}\widetilde{\mathbf{R}}_t(\widehat{\bm{\alpha}})\widehat{\bm{\Omega}}^{-1}(\tau_t)\mathbf{u}_t \to_D \int_{0}^{1}\mathbf{W}_{d-r_0}(u)\otimes \mathrm{d}\mathbf{W}_{r_0}(u)$, where $\mathbf{W}_{r_0}(u) = \mathbf{W}_{r_0}(u,\bm{\alpha}^\top(u)\bm{\Omega}^{-1}(u)\bm{\alpha}(u))$ is independent of $\mathbf{W}_{d-r_0}(u)$;
			
			\item $\frac{1}{T}\sum_{t=1}^{T}\widetilde{\mathbf{R}}_t(\widehat{\bm{\alpha}})\widehat{\bm{\Omega}}^{-1}(\tau_t)\left(\bm{\alpha}(\tau_t) - \widehat{\bm{\alpha}}(\tau_t)\right)\bm{\beta}^{\top}\mathbf{y}_{t-1} = o_P(1)$;
			
			\item $\frac{1}{T}\sum_{t=1}^{T}\widetilde{\mathbf{R}}_t(\widehat{\bm{\alpha}})\widehat{\bm{\Omega}}^{-1}(\tau_t)\left(\bm{\Gamma}(\tau_t) - \widehat{\bm{\Gamma}}(\tau_t,\widehat{\bm{\alpha}}\bm{\beta}^\top)\right)\Delta \mathbf{x}_{t-1}=o_P(1)$, where
			{\small
			\begin{eqnarray*}
				\widehat{\bm{\Gamma}}(\tau,\widehat{\bm{\alpha}}\bm{\beta}^\top) =  \sum_{s=1}^{T}(\Delta\mathbf{y}_{s}-\widehat{\bm{\alpha}}(\tau_s)\bm{\beta}^\top\mathbf{y}_{s-1})\Delta\mathbf{x}_{s-1}^{*,\top} K\left(\frac{\tau_s-\tau}{h}\right) \cdot \left(\sum_{s=1}^{T}\Delta\mathbf{x}_{s-1}^*\Delta\mathbf{x}_{s-1}^{*,\top} K\left(\frac{\tau_s-\tau}{h}\right)\right)^{-1}\left[\begin{matrix}
					\mathbf{I}_{d(p_0-1)} \\
					\mathbf{0}_{d(p_0-1)}
				\end{matrix} \right].\nonumber
			\end{eqnarray*}}
		\end{enumerate}
	\end{lemma}
	
	\begin{lemma}\label{L9}
		Let Assumptions \ref{Ass1}--\ref{Ass2} hold. Then
		\begin{enumerate}[wide, labelwidth=!, labelindent=0pt]
			\item if $p\geq p_0$, then $\text{RSS}(p)=\frac{1}{T}\sum_{t=1}^{T}\mathbf{u}_t^\top\mathbf{u}_t + O_P\left(c_T^2\right)$ with $c_T= h^2 + \sqrt{\frac{\log T}{Th}}$;
			\item if $p < p_0$, then $\text{RSS}(p)=\frac{1}{T}\sum_{t=1}^{T}\mathbf{u}_t^\top\mathbf{u}_t + c + o_P\left(1\right)$ with some constant $c >0$.
		\end{enumerate}
	\end{lemma}

	\begin{lemma}\label{L10}
		Let Assumptions \ref{Ass1}--\ref{Ass2} hold. 
		
		Define $w_{s,v}=\frac{1}{T\sqrt{h}}\int_{-1}^{1}K\left(u\right)K\left(u+\frac{s-v}{Th}\right)\mathrm{d}u$ and $\widetilde{U} = \sum_{s=2}^{T}\sum_{v=1}^{s-1}\mathbf{u}_s^\top\mathbf{W}_{s-1}^\top\mathbf{H}_s\mathbf{W}_{v-1}\mathbf{u}_v w_{s,v}$ with some deterministic unknown matrices $\{\mathbf{H}_s\}$ and $\mathbf{W}_{s} = \mathbf{w}_s\otimes \mathbf{I}_d$. Then
		$$
		\widetilde{U} \to_D N\left(0,\sigma_{\widetilde{U}}^2\right),
		$$
		where $ \sigma_{\widetilde{U}}^2 = \lim_{T \to \infty}\sum_{s=2}^{T}\mathrm{tr}\left\{E\left(\mathbf{H}_s^\top\mathbf{W}_{s-1}\mathbf{u}_s\mathbf{u}_s^\top\mathbf{W}_{s-1}^\top\mathbf{H}_s\right)E\left(\sum_{v=1}^{s-1}\mathbf{W}_{v-1}\mathbf{u}_v\mathbf{u}_v^\top\mathbf{W}_{v-1}^\top\right)w_{s,v}^2\right\}$.
	\end{lemma}

	\section{Proofs of the Main Results}\label{AppPMain}
	
	\begin{proof}[Proof of Lemma \ref{L1}]
	\item
	\noindent (1)--(2). Write
	\begin{eqnarray*}
		\mathbf{C}_{\tau}(L)&=& [\bm{\Gamma}_\tau(L)(1-L) - \bm{\alpha}(\tau)\bm{\beta}^\top L] \mathbf{J}^{-1}\mathbf{J}\nonumber \\
		&=&[\bm{\Gamma}_\tau(L)\overline{\bm{\beta}}(1-L) - \bm{\alpha}(\tau)L,\bm{\Gamma}_\tau(L)\bm{\beta}_{\perp}(1-L)] \cdot[\bm{\beta},\bm{\beta}_{\perp}]^\top \nonumber\\
		&=&\mathbf{J}^{-1}\mathbf{J}[\bm{\Gamma}_\tau(L)\overline{\bm{\beta}}(1-L) - \bm{\alpha}(\tau)L,\bm{\Gamma}_\tau(L)\bm{\beta}_{\perp}] \cdot[\bm{\beta},\bm{\beta}_{\perp}(1-L)]^\top\nonumber\\
		&=&\mathbf{J}^{-1} \mathbf{B}_\tau^*(L) \mathbf{P}(L),\nonumber
	\end{eqnarray*}
	where $\mathbf{P}(L) = [\bm{\beta},\bm{\beta}_{\perp}(1-L)]^\top$. Clearly, $\mathrm{det}(\mathbf{P}(L))$ has exactly $d-r_0$ unit roots and, thus by Assumption \ref{Ass1}, $\mathbf{B}_\tau^*(L) \neq 0$ for all $|L|\leq 1$ must hold. We now can conclude that $\mathbf{B}_\tau^*(L)$ is an invertible operator.
	
	Define $\mathbf{z}_t:= \mathbf{J}^{-1}\mathbf{P}(L)\mathbf{y}_t= \mathbf{P}_{\bm{\beta}} \mathbf{y}_t + \mathbf{P}_{\bm{\beta_\perp}} \Delta \mathbf{y}_t$ and thus $\mathbf{J}^{-1} \mathbf{B}_{\tau_t}^*(L)\mathbf{J}\mathbf{z}_t = \mathbf{u}_t$. Note that  
	\begin{eqnarray*}
		\mathbf{B}_\tau(0) = \mathbf{J}^{-1}\mathbf{B}_\tau^*(0)\mathbf{J} = \bm{\Gamma}_{\tau}(0)\mathbf{P}_{\bm{\beta}}+\bm{\Gamma}_{\tau}(0)\mathbf{P}_{\bm{\beta}_\perp}= \bm{\Gamma}_{\tau}(0) = \mathbf{I}_d\nonumber
	\end{eqnarray*}		
	and thus $\mathbf{B}_\tau(L)$ has the representation $\mathbf{I}_d - \sum_{i=1}^{p}\mathbf{B}_i(\tau)L^i$. In addition, we have $\mathrm{det}(\mathbf{B}_\tau(L)) \neq 0$ for all $|L|\leq 1$ since the roots of $\mathrm{det}(\mathbf{B}_\tau^*(L))$  all lie outside the unit circle.
	
		Consider $\mathbf{P}_{\bm{\beta}_\perp}\bm{\Psi}_\tau(1)$. We have
	\begin{eqnarray*}
		\mathbf{P}_{\bm{\beta}_\perp}\bm{\Psi}_\tau(1) &=& \mathbf{P}_{\bm{\beta}_\perp}\mathbf{J}^{-1}\mathbf{B}_\tau^{*,-1}(1)\mathbf{J}=\bm{\beta}_{\perp}[\mathbf{0}_{d-r_0\times r_0}, \mathbf{I}_{d-r_0}] (\mathbf{J}^{-1}\mathbf{B}_\tau^{*}(1))^{-1}\nonumber\\
		&=&\bm{\beta}_{\perp}[\mathbf{0}_{d-r_0\times r_0}, \mathbf{I}_{d-r_0}] [-\bm{\alpha}(\tau),\bm{\Gamma}_{\tau}(1)\bm{\beta}_{\perp}]^{-1}\nonumber\\
		&=&\bm{\beta}_{\perp} [\bm{\alpha}_{\perp}^\top(\tau)\bm{\Gamma}_{\tau}(1) \bm{\beta}_{\perp}]^{-1} \bm{\alpha}_{\perp}^\top(\tau),\nonumber
	\end{eqnarray*}
where the last step follows from
	
	\begin{eqnarray*}
	[-\bm{\alpha}(\tau),\bm{\Gamma}_{\tau}(1)\bm{\beta}_{\perp}]^{-1} = \left[\begin{matrix}
		(\bm{\alpha}^\top(\tau)\bm{\alpha}(\tau))^{-1}\bm{\alpha}^\top(\tau)\left\{\bm{\Gamma}_{\tau}(1)\bm{\beta}_{\perp}[\bm{\alpha}_{\perp}^\top(\tau) \bm{\Gamma}_{\tau}(1)\bm{\beta}_{\perp}]^{-1}\bm{\alpha}_{\perp}^\top(\tau)-\mathbf{I}_d\right\}  \\
		[\bm{\alpha}_{\perp}^{\top}(\tau)\bm{\Gamma}_{\tau}(1)\bm{\beta}_{\perp}]^{-1}\bm{\alpha}_{\perp}^{\top}(\tau)
	\end{matrix} \right].\nonumber
	\end{eqnarray*}
	
Recall that $\mathbf{I}_d =  \mathbf{P}_{\bm{\beta}} + \mathbf{P}_{\bm{\beta}_\perp}$, so $\Delta \mathbf{y}_t = \mathbf{P}_{\bm{\beta}}\Delta \mathbf{y}_t + \mathbf{P}_{\bm{\beta}_\perp}\Delta \mathbf{y}_t$. In connection with $\mathbf{z}_t = \mathbf{P}_{\bm{\beta}} \mathbf{y}_t + \mathbf{P}_{\bm{\beta}_\perp} \Delta \mathbf{y}_t$, we have $\Delta \mathbf{y}_t = \mathbf{z}_t - \mathbf{P}_{\bm{\beta}}\mathbf{z}_{t-1}$, $\overline{\bm{\beta}}_{\perp}^\top \Delta \mathbf{y}_t =\overline{\bm{\beta}}_{\perp}^\top \mathbf{z}_t$ and $\Delta \mathbf{y}_t = \mathbf{P}_{\bm{\beta}_\perp}  \mathbf{z}_t + \mathbf{P}_{\bm{\beta}}\Delta \mathbf{y}_t$. Hence,
	\begin{eqnarray*}
		\mathbf{y}_t = \mathbf{y}_0 + \sum_{i=1}^{t}\Delta \mathbf{y}_i = \mathbf{y}_0 +  \mathbf{P}_{\bm{\beta}_\perp}\sum_{j=1}^{t}\mathbf{z}_j + \mathbf{P}_{\bm{\beta}}\mathbf{y}_t - \mathbf{P}_{\bm{\beta}} \mathbf{y}_0 = \mathbf{P}_{\bm{\beta}_\perp}\sum_{j=1}^{t}\mathbf{z}_j + \mathbf{P}_{\bm{\beta}}\mathbf{z}_t + \mathbf{P}_{\bm{\beta}_\perp} \mathbf{y}_0.\nonumber
	\end{eqnarray*}
	The proof of parts (1)-(2) is now completed.
	
	\medskip
	
	\noindent (3). Define $\widetilde{\mathbf{z}}_t(\tau):= \mathbf{J}^{-1}\mathbf{P}(L)\widetilde{\mathbf{y}}_t(\tau) = \mathbf{P}_{\bm{\beta}} \widetilde{\mathbf{y}}_t(\tau) + \mathbf{P}_{\bm{\beta}_\perp} \Delta \widetilde{\mathbf{y}}_t(\tau)$. Write
	\begin{eqnarray*}
		\mathbf{C}_{\tau}(L)\widetilde{\mathbf{y}}_t(\tau) &=& \mathbf{v}(\tau) + \widetilde{\mathbf{u}}_t(\tau), \nonumber\\
		\mathbf{J}^{-1} \mathbf{B}_\tau^*(L) \mathbf{P}_\tau(L)\widetilde{\mathbf{y}}_t(\tau) &=& \widetilde{\mathbf{u}}_t(\tau), \nonumber\\
		\mathbf{J}^{-1}\mathbf{B}_\tau^*(L)\mathbf{J}\widetilde{\mathbf{z}}_t(\tau)&=& \widetilde{\mathbf{u}}_t(\tau),\nonumber \\
		\widetilde{\mathbf{z}}_t(\tau)=\mathbf{B}_\tau^{-1}(L)\widetilde{\mathbf{u}}_t(\tau)&=&\sum_{j=0}^{\infty}\bm{\Psi}_j(\tau)\widetilde{\mathbf{u}}_{t-j}(\tau).\nonumber
	\end{eqnarray*}
	Hence, $\widetilde{\mathbf{z}}_t(\tau)$ is a stationary VAR$(p)$ process which has an MA$(\infty)$ representation. Similar to the proof of part (1) and additionally by BN decomposition, we have
	\begin{eqnarray*}
		\Delta \widetilde{\mathbf{y}}_t(\tau) &=& \mathbf{P}_{\bm{\beta}}\Delta \widetilde{\mathbf{y}}_t(\tau) + \mathbf{P}_{\bm{\beta}_\perp}\Delta \widetilde{\mathbf{y}}_t(\tau)= \mathbf{P}_{\bm{\beta}_\perp}\widetilde{\mathbf{z}}_t(\tau) + \mathbf{P}_{\bm{\beta}}\Delta \widetilde{\mathbf{y}}_t(\tau)\nonumber\\
		&=& \mathbf{P}_{\bm{\beta}_\perp}(\bm{\Psi}_\tau(1)\widetilde{\mathbf{u}}_t(\tau) - \bm{\Psi}_\tau^*(L) \Delta\widetilde{\mathbf{u}}_t(\tau)) + \mathbf{P}_{\bm{\beta}}\Delta \widetilde{\mathbf{y}}_t(\tau),\nonumber
	\end{eqnarray*}
	where $\bm{\Psi}_\tau(L):= \sum_{j=0}^{\infty}\bm{\Psi}_j(\tau) L^j =\mathbf{B}_\tau^{-1}(L)$, $\bm{\Psi}_\tau^*(L) = \sum_{j=0}^\infty \bm{\Psi}_j^*(\tau)L^j$ and $\bm{\Psi}_j^*(\tau) = \sum_{i=j+1}^{\infty}\bm{\Psi}_i(\tau)$.
	
	Hence, we have
	\begin{eqnarray*}
		\widetilde{\mathbf{y}}_t(\tau) &=& \mathbf{P}_{\bm{\beta}_\perp}\left[\bm{\Psi}_\tau(1)\sum_{i=1}^{t}\widetilde{\mathbf{u}}_i(\tau) - \bm{\Psi}_\tau^*(L)\widetilde{\mathbf{u}}_t(\tau) +\bm{\Psi}_\tau^*(L)\widetilde{\mathbf{u}}_0(\tau)\right]  + \mathbf{P}_{\bm{\beta}} \widetilde{\mathbf{z}}_t(\tau)+\mathbf{P}_{\bm{\beta}_\perp} \widetilde{\mathbf{y}}_0(\tau)\nonumber\\
		&=&\mathbf{P}_{\bm{\beta}_\perp}\bm{\Psi}_\tau(1)\sum_{i=1}^{t}\widetilde{\mathbf{u}}_i(\tau) +  \mathbf{P}_{\bm{\beta}} \bm{\Psi}_\tau(1)\widetilde{\mathbf{u}}_t(\tau) - \bm{\Psi}_\tau^*(L) \widetilde{\mathbf{u}}_t  + \mathbf{P}_{\bm{\beta}}\bm{\Psi}_\tau^*(L)\widetilde{\mathbf{u}}_{t-1}(\tau)+\widetilde{\mathbf{y}}_0^*(\tau),\nonumber
	\end{eqnarray*}
	where $\widetilde{\mathbf{y}}_0^*(\tau) =\mathbf{P}_{\bm{\beta}_\perp}\bm{\Psi}_\tau^*(L)\widetilde{\mathbf{u}}_0(\tau) +\mathbf{P}_{\bm{\beta}_\perp} \widetilde{\mathbf{y}}_0(\tau)$.

	By the proof of   (1)-(2), we have $\Delta \mathbf{y}_t =\mathbf{z}_t - \mathbf{P}_{\bm{\beta}}\mathbf{z}_{t-1}$ and $\mathbf{z}_t =  \sum_{i=1}^{p}\mathbf{B}_i(\tau_t)\mathbf{z}_{t-i} + \mathbf{u}_t$. Let $\mathbf{B}(\tau)$ be the companion matrix of $\mathbf{z}_t$ such that 
	\begin{eqnarray*}
	\mathbf{B}(\tau)=\left[\begin{matrix}
		\mathbf{B}_{1}(\tau) & \cdots & \mathbf{B}_{p_0-1}(\tau) & \mathbf{B}_{p_0}(\tau)  \\
		\mathbf{I}_d & \cdots& \mathbf{0}_d & \mathbf{0}_d\\
		\vdots & \ddots&\vdots & \vdots\\
		\mathbf{0}_d &\cdots &\mathbf{I}_d & \mathbf{0}_d\\
	\end{matrix} \right]\nonumber
	\end{eqnarray*}
	and thus $\mathbf{z}_t = \sum_{j=0}^{\infty}\overline{\mathbf{J}}\prod_{i=0}^{j-1}\mathbf{B}(\tau_{t-m})\overline{\mathbf{J}}^\top \mathbf{v}(\tau_{t-j})+ \sum_{j=0}^{\infty}\overline{\mathbf{J}}\prod_{i=0}^{j-1}\mathbf{B}(\tau_{t-m})\overline{\mathbf{J}}^\top \bm{\omega}(\tau_{t-j})\bm{\varepsilon}_{t-j}$, where $\overline{\mathbf{J}}=[\mathbf{I}_{d},\mathbf{0}_{d\times d(p_0-1)}]$.
	
	Define $\rho_B = \max\{|\rho|: \rho\ \text{is the eigenvalue of}\ \mathbf{B}(\tau)\}$. Since $\mathrm{det}(\mathbf{B}_\tau(L)) \neq 0$ for all $|L|\leq 1$, we have $\rho_{B} < 1$. Similar to the proof of Proposition 2.4 in \cite{dahlhaus2009empirical}, we have $\max_{t\geq 1}|\prod_{m=0}^{j-1}\mathbf{B}(\tau_{t-m}) | \leq M \rho_B^j$. Note that $\Delta \widetilde{\mathbf{y}}_t(\tau) = \widetilde{\mathbf{z}}_t(\tau) - \mathbf{P}_{\bm{\beta}}\widetilde{\mathbf{z}}_{t-1}(\tau)$ and $\widetilde{\mathbf{z}}_t(\tau) = \sum_{i=1}^{p}\mathbf{B}_i(\tau)\widetilde{\mathbf{z}}_{t-i}(\tau) + \widetilde{\mathbf{u}}_t(\tau) = \sum_{j=0}^{\infty}\overline{\mathbf{J}}\mathbf{B}^j(\tau)\overline{\mathbf{J}}^\top\bm{\omega}(\tau)\bm{\varepsilon}_{t-j}$. 
	For any conformable matrices $\{\mathbf{A}_i\}$ and $\{\mathbf{B}_i\}$, since
\begin{eqnarray*}
\prod_{i=1}^{r}\mathbf{A}_i-\prod_{i=1}^{r}\mathbf{B}_i=\sum_{j=1}^{r}\left( \prod_{k=1}^{j-1}\mathbf{A}_k \right)\left(\mathbf{A}_j-\mathbf{B}_j\right)\left(\prod_{k=j+1}^{r}\mathbf{B}_k\right),\nonumber
\end{eqnarray*}	
we then obtain 
	\begin{eqnarray*}
		&& \left|\overline{\mathbf{J}}\prod_{i=0}^{j-1}\mathbf{B}(\tau_{t-i})\overline{\mathbf{J}}^\top\bm{\omega}(\tau_{t-j})-\overline{\mathbf{J}}\mathbf{B}^j(\tau_t)\overline{\mathbf{J}}^\top\bm{\omega}(\tau_t)\right|\nonumber\\
		& = & \left| \left(\overline{\mathbf{J}}\prod_{i=0}^{j-1}\mathbf{B}(\tau_{t-i})\overline{\mathbf{J}}^\top-\overline{\mathbf{J}}\mathbf{B}^j(\tau_t)\overline{\mathbf{J}}^\top\right)\bm{\omega}(\tau_t)+\overline{\mathbf{J}}\prod_{i=0}^{j-1}\mathbf{B}(\tau_{t-i})\overline{\mathbf{J}}^\top\left( \bm{\omega}(\tau_{t-j})-\bm{\omega}(\tau_t)\right)\right|\nonumber\\
		&\leq& M\sum_{i=1}^{j-1}\left|\mathbf{B}^i(\tau_t)(\mathbf{B}(\tau_{t-i})-\mathbf{B}(\tau_t))\prod_{k=i+1}^{j-1}\mathbf{B}(\tau_{t-k}) \right|+M \rho_B^j \frac{j}{T}\nonumber\\
		&\leq& M\sum_{i=1}^{j-1}\frac{i}{T}\rho_B^{j-1}+M \rho_B^j \frac{j}{T}.\nonumber
	\end{eqnarray*}
	In addition, by Minkowski inequality, we have
	\begin{eqnarray*}
		\left\|\mathbf{z}_t-\widetilde{\mathbf{z}}_t(\tau_t)\right\|_\delta&\leq & \sum_{j=1}^{\infty}\left|\overline{\mathbf{J}}\prod_{i=0}^{j-1}\mathbf{B}(\tau_{t-i})\overline{\mathbf{J}}^\top\bm{\omega}(\tau_{t-j})-\overline{\mathbf{J}}\mathbf{B}^j(\tau_t)\overline{\mathbf{J}}^\top\bm{\omega}(\tau_t)\right| \cdot \left\|\bm{\varepsilon}_t\right\|_\delta \nonumber\\
		&\leq& M\sum_{j=1}^{\infty} \left(\sum_{i=1}^{j-1}\frac{i}{T}\rho_B^{j-1}+ \rho_B^j \frac{j}{T}\right) =O\left(T^{-1}\right).\nonumber
	\end{eqnarray*}
	Similarly, we have $ \| \widetilde{\mathbf{z}}_t(\tau)- \widetilde{\mathbf{z}}_t(\tau')\|_\delta = O(|\tau-\tau'|)$.
	
	By the proof of (1)-(2), we have $\Delta \mathbf{y}_t = \mathbf{z}_t - \mathbf{P}_{\bm{\beta}}\mathbf{z}_{t-1}$ and $\Delta \widetilde{\mathbf{y}}_t(\tau) = \widetilde{\mathbf{z}}_t(\tau) - \mathbf{P}_{\bm{\beta}}\widetilde{\mathbf{z}}_{t-1}(\tau)$. Then, since $\max_t \|\mathbf{z}_t-\widetilde{\mathbf{z}}_t(\tau_t)\|_\delta  = O(1/T)$ and $\|\widetilde{\mathbf{z}}_t(\tau)- \widetilde{\mathbf{z}}_t(\tau')\|_\delta = O(|\tau-\tau'|)$ for any $\tau,\tau' \in [0,1]$, we have $\max_t \|\Delta\mathbf{y}_t-\Delta\widetilde{\mathbf{y}}_t(\tau_t)\|_\delta = O(1/T)$ and $\| \Delta\widetilde{\mathbf{y}}_t(\tau)-\Delta \widetilde{\mathbf{y}}_t(\tau')\|_\delta = O(|\tau-\tau'|)$.
	
Since $\Delta \mathbf{y}_t = \mathbf{P}_{\bm{\beta}_\perp}  \mathbf{z}_t + \mathbf{P}_{\bm{\beta}}\Delta \mathbf{y}_t$ and $\bm{\beta}^\top \mathbf{y}_t=\bm{\beta}^\top\mathbf{z}_t$, we have
	\begin{eqnarray*}
	\mathbf{y}_t = \mathbf{y}_0 +  \mathbf{P}_{\bm{\beta}_\perp}\sum_{j=1}^{t}\mathbf{z}_j + \mathbf{P}_{\bm{\beta}}\mathbf{z}_t -\mathbf{P}_{\bm{\beta}} \mathbf{z}_0\nonumber
	\end{eqnarray*}
	and 
\begin{eqnarray*}
\sum_{j=1}^{t}\Delta\widetilde{\mathbf{y}}_j(\tau_j) =\mathbf{P}_{\bm{\beta}_\perp}\sum_{j=1}^{t}\widetilde{\mathbf{z}}_j(\tau_j) + \mathbf{P}_{\bm{\beta}}\widetilde{\mathbf{z}}_t(\tau_t) - \mathbf{P}_{\bm{\beta}}\widetilde{\mathbf{z}}_0(\tau_1)+ \mathbf{P}_{\bm{\beta}}\sum_{j=1}^{t-1}(\widetilde{\mathbf{z}}_j(\tau_j)-\widetilde{\mathbf{z}}_j(\tau_{j+1})).\nonumber
\end{eqnarray*}

	Hence, by part (1) we have
	\begin{eqnarray*}
	\|\mathbf{y}_t - \mathbf{y}_0 - \sum_{j=1}^{t} \Delta\widetilde{\mathbf{y}}_t(\tau_j)\|_\delta &\leq& \| \mathbf{P}_{\bm{\beta}_\perp}\sum_{j=1}^{t}(\mathbf{z}_j-\widetilde{\mathbf{z}}_j(\tau_j))\|_\delta + \|\mathbf{P}_{\bm{\beta}}\sum_{j=1}^{t-1}(\widetilde{\mathbf{z}}_j(\tau_j)-\widetilde{\mathbf{z}}_j(\tau_{j+1}))\|_\delta + O(1/T)\nonumber\\
	&=&I_{T,1} + I_{T,2} + O(1/T).\nonumber
	\end{eqnarray*}

For $I_{T,1}$, as $\{(\mathbf{J}\prod_{i=0}^{j-1}\mathbf{B}(\tau_{t-i})\mathbf{J}^\top \bm{\omega}(\tau_{t-j}) - \mathbf{J}\mathbf{B}^j(\tau_t)\mathbf{J}^\top\bm{\omega}(\tau_t))\bm{\varepsilon}_{t-j}\}_{t}$ is a sequence of martingale differences, by Burkholder's inequality, we have
\begin{eqnarray*}
\left|\left|\sum_{j=1}^{t} (\mathbf{z}_j - \widetilde{\mathbf{z}}_j(\tau_j)) \right|\right|_\delta 	&=&\left|\left|\sum_{j=1}^{t}\sum_{k=0}^{\infty} (\mathbf{J}\prod_{i=0}^{k-1}\mathbf{B}(\tau_{j-i})\mathbf{J}^\top \bm{\omega}(\tau_{j-k}) - \mathbf{J}\mathbf{B}^k(\tau_j)\mathbf{J}^\top\bm{\omega}(\tau_j))\bm{\varepsilon}_{j-k} \right|\right|_\delta \nonumber\\
	&\leq &\sum_{k=0}^{\infty} \left|\left|\sum_{j=1}^{t} (\mathbf{J}\prod_{i=0}^{k-1}\mathbf{B}(\tau_{j-i})\mathbf{J}^\top \bm{\omega}(\tau_{j-k}) - \mathbf{J}\mathbf{B}^k(\tau_j)\mathbf{J}^\top\bm{\omega}(\tau_j))\bm{\varepsilon}_{j-k} \right|\right|_\delta \nonumber \\
	&\leq &\sum_{k=0}^{\infty}O(\sqrt{T}) \max_j \left|\mathbf{J}\prod_{i=0}^{k-1}\mathbf{B}(\tau_{j-i})\mathbf{J}^\top \bm{\omega}(\tau_{j-k}) - \mathbf{J}\mathbf{B}^k(\tau_j)\mathbf{J}^\top\bm{\omega}(\tau_j)\right| = O(1/\sqrt{T}).\nonumber
\end{eqnarray*}

Hence, $I_{T,1} = O(1/\sqrt{T})$. Similarly, we have $I_{T,2} =O(1/\sqrt{T})$. The proof of part (3) is now completed.

\medskip

\noindent (4). By (1)-(3) of this lemma, we have $\sup_{u\in[0,1]}|\mathbf{z}_{\lfloor Tu \rfloor}|\leq \left\{\sum_{t=1}^{T}|\mathbf{z}_t|^{\delta}\right\}^{1/\delta}= O_P(T^{1/\delta})$ and
\begin{eqnarray*}
\sup_{u\in[0,1]} \left|\sum_{t=1}^{\lfloor Tu \rfloor}(\mathbf{z}_t-\widetilde{\mathbf{z}}_t(\tau_t))|\leq \sum_{t=1}^{T}|\mathbf{z}_t-\widetilde{\mathbf{z}}_t(\tau_t)\right| =O_P(1).\nonumber
\end{eqnarray*}
Then, we have uniformly over $u\in [0,1]$
\begin{eqnarray*}
	T^{-1/2} \mathbf{y}_{\lfloor Tu \rfloor} &=&\mathbf{P}_{\bm{\beta}_\perp} T^{-1/2}\sum_{t=1}^{\lfloor Tu \rfloor}\mathbf{z}_t + T^{-1/2}\mathbf{P}_{\bm{\beta}}\mathbf{z}_{\lfloor Tu \rfloor} + T^{-1/2}\mathbf{P}_{\bm{\beta}_\perp} \mathbf{y}_0\nonumber \\
	&=& \mathbf{P}_{\bm{\beta}_\perp} T^{-1/2}\sum_{t=1}^{\lfloor Tu \rfloor} \sum_{j=0}^{\infty}\bm{\Psi}_j(\tau_t)\bm{\omega}(\tau_t) \bm{\varepsilon}_{t-j} + O_P(T^{1/\delta-1/2})\nonumber
\end{eqnarray*}
with $|\bm{\Psi}_j(\tau)|$ converging to zero exponentially. By BN decomposition, we have $\bm{\Psi}_\tau(L):= \sum_{j=0}^{\infty}\bm{\Psi}_j(\tau) L^j = \bm{\Psi}_\tau(1) - (1-L) \bm{\Psi}_\tau^*(L)$ with $\bm{\Psi}_\tau^*(L) = \sum_{j=0}^\infty \bm{\Psi}_j^*(\tau) L^j$ and $\bm{\Psi}_j^*(\tau) = \sum_{i=j+1}^{\infty}\bm{\Psi}_i(\tau)$. Hence,
\begin{eqnarray*}
	&&T^{-1/2}\sum_{t=1}^{\lfloor Tu\rfloor} \sum_{j=0}^{\infty}\bm{\Psi}_j(\tau_t)\bm{\omega}(\tau_t) \bm{\varepsilon}_{t-j}\nonumber\\
	&=&T^{-1/2}\sum_{t=1}^{\lfloor Tu \rfloor}\bm{\Psi}_{\tau_t}(1)\bm{\omega}(\tau_t)\bm{\varepsilon}_t+T^{-1/2}\bm{\Psi}_{\tau_1}^*(L)\bm{\omega}(\tau_1)\bm{\varepsilon}_{0}-T^{-1/2}\bm{\Psi}_{\tau_{\lfloor Tu \rfloor}}^*(L)\bm{\omega}(\tau_{\lfloor Tu \rfloor})\bm{\varepsilon}_{\lfloor Tu \rfloor} \nonumber\\
	&&+ T^{-1/2} \sum_{t=1}^{\lfloor Tu \rfloor-1}\left(\bm{\Psi}_{\tau_{t+1}}(L)\bm{\omega}(\tau_{t+1})-\bm{\Psi}_{\tau_t}(L)\bm{\omega}(\tau_t)\right)\bm{\varepsilon}_t \nonumber\\
	&=& \mathbf{I}_{T,3} + \mathbf{I}_{T,4} + \mathbf{I}_{T,5} + \mathbf{I}_{T,6}.\nonumber
\end{eqnarray*}

By the usual functional central limit theory for martingale difference sequences (e.g., Theorem 4.1 in \citealp{hall2014martingale}), we have 
\begin{eqnarray*}
\mathbf{I}_{T,3} \Rightarrow \int_{0}^{u}\bm{\Psi}_s(1) \bm{\omega}(s)\mathrm{d}\mathbf{W}_d(s).\nonumber
\end{eqnarray*}
Note that the Lindeberg condition can be verified as $\|\bm{\Psi}_{\tau_t}(1)\bm{\omega}(\tau_t)\bm{\varepsilon}_t\|_\delta <\infty$ for some $\delta >2$, and the convergence of conditional variance can be verified as $E(\bm{\varepsilon}_t\bm{\varepsilon}_t^\top \mid \mathcal{F}_{t-1}) = \mathbf{I}_d$ a.s.

In addition, we have $\mathbf{I}_{T,4} = O_P(T^{-1/2})$ uniformly over $u \in [0,1]$ as 
\begin{eqnarray*}
\left\|\bm{\Psi}_{\tau_1}^*(L)\bm{\omega}(\tau_1)\bm{\varepsilon}_{0}\right\|_\delta \leq O(1) \sup_{\tau\in[0,1]}\sum_{j=1}^\infty j|\bm{\Psi}_j(\tau)| < \infty.\nonumber
\end{eqnarray*}
For $\mathbf{I}_{T,5}$, we have
\begin{eqnarray*}
\sup_{u\in[0,1]}\left|T^{-1/2}\bm{\Psi}_{\tau_{\lfloor Tu \rfloor}}^*(L)\bm{\omega}(\tau_{\lfloor Tu\rfloor})\bm{\varepsilon}_{\lfloor Tu \rfloor}\right| \leq T^{-1/2}\sum_{t=1}^{T}|\bm{\Psi}_{\tau_{t}}^*(L)\bm{\omega}(\tau_{t})\bm{\varepsilon}_{t}|=O_P(T^{1/\delta-1/2}).\nonumber
\end{eqnarray*}

Finally, for $\mathbf{I}_{T,6}$, as $\sum_{t=1}^{T-1}\left\|\left(\bm{\Psi}_{\tau_{t+1}}(L)\bm{\omega}(\tau_{t+1})-\bm{\Psi}_{\tau_t}(L)\bm{\omega}(\tau_t)\right)\bm{\varepsilon}_t\right\|_\delta< \infty$, we have
\begin{eqnarray*}
	&& \sup_{u\in[0,1]}\|\mathbf{I}_{T,6}\| \leq \sup_{u\in[0,1]}\frac{1}{\sqrt{T}}\sum_{t=1}^{\lfloor Tu \rfloor-1}\left\|\left(\widetilde{\mathbb{B}}_{t+1}(L)-\widetilde{\mathbb{B}}_{t}(L)\right)\bm{\epsilon}_t\right\|\nonumber\\
	&&\leq \frac{1}{\sqrt{T}}\sum_{t=1}^{T-1}\left\|\left(\bm{\Psi}_{\tau_{t+1}}(L)\bm{\omega}(\tau_{t+1})-\bm{\Psi}_{\tau_t}(L)\bm{\omega}(\tau_t)\right)\bm{\varepsilon}_t\right\| = O(1/\sqrt{T}).\nonumber
\end{eqnarray*}

The proof of part (4) is now completed.
	\end{proof}

	\begin{proof}[Proof of Theorem \ref{Thm1}]
		\item
		\noindent (1). Recall we have defined many relevant notations in Appendix \ref{AppNot}. Then it is easy to know that
		$$
		\mathbf{Q}^{-1}(\tau) = \left[\bm{\beta}_{\perp}(\bm{\alpha}_{\perp}^\top(\tau)\bm{\beta}_{\perp})^{-1},\bm{\alpha}(\tau)(\bm{\beta}^\top\bm{\alpha}(\tau))^{-1}\right].
		$$
Let
		$$
		\mathbf{M}(\tau_t) = [\bm{\Pi}(\tau_t),\bm{\Gamma}(\tau_t)] - [\bm{\Pi}(\tau),\bm{\Gamma}(\tau)] - [\bm{\Pi}^{(1)}(\tau),\bm{\Gamma}^{(1)}(\tau)](\tau_t-\tau) -\frac{1}{2}[\bm{\Pi}^{(2)}(\tau),\bm{\Gamma}^{(2)}(\tau)](\tau_t-\tau)^2.
		$$
Also, define $\mathbf{h}_t^* = \mathbf{h}_t \otimes [1,\frac{\tau_{t+1}-\tau}{h}]^\top$, $\mathbf{V}_{T}(\tau) = \left[\mathbf{V}_{T,0}(\tau), \mathbf{V}_{T,1}(\tau)\right]$ and 
		$$
		\mathbf{S}_T(\tau) = \left[\begin{matrix}
			\mathbf{S}_{T,0}(\tau) & \mathbf{S}_{T,1}(\tau)\\ \mathbf{S}_{T,1}(\tau) & \mathbf{S}_{T,2}(\tau)
		\end{matrix} \right].
		$$

We now proceed, and write
		\begin{eqnarray*}
			&&[\widehat{\bm{\Pi}}(\tau), \widehat{\bm{\Gamma}}(\tau)]\mathbf{Q}^{*,-1}(\tau)\bm{\Xi}_T^*(\tau)\mathbf{D}_T^{*}\nonumber\\ 
			&=&\mathbf{V}_{T}(\tau) \left(\mathbf{I}_2\otimes\mathbf{Q}^{*,\top}(\tau)\bm{\Xi}_T^*(\tau)\mathbf{D}_T^{*,-1}\right)\nonumber\\
			&&\times\left(\mathbf{I}_2\otimes\mathbf{D}_T^{*}\bm{\Xi}_T^{*,\top}(\tau)\mathbf{Q}^{*,-1,\top}(\tau)\right) \mathbf{S}_{T}^{+}(\tau) \left(\mathbf{I}_2\otimes\mathbf{Q}^{*,-1}(\tau)\bm{\Xi}_T^*(\tau)\mathbf{D}_T^{*}\right)\times [\mathbf{I}_{dp_0},\mathbf{0}_{dp_0}]^\top\nonumber\\
			&=& [\bm{\Pi}(\tau), \bm{\Gamma}(\tau),h\bm{\Pi}^{(1)}(\tau), h\bm{\Gamma}^{(1)}(\tau)]\left(\mathbf{I}_2\otimes\mathbf{Q}^{*,-1}(\tau)\bm{\Xi}_T^*(\tau)\mathbf{D}_T^{*}\right)\nonumber\\
			&&\times\left(\mathbf{I}_2\otimes\mathbf{D}_T^{*,-1}\bm{\Xi}_T^{*,\top}(\tau)\mathbf{Q}^{*}(\tau)\right)\mathbf{S}_{T}(\tau)\left(\mathbf{I}_2\otimes \mathbf{Q}^{*,\top}(\tau)\bm{\Xi}_T^*(\tau)\mathbf{D}_T^{*,-1}\right)\nonumber\\
			&&\times\left(\mathbf{I}_2\otimes\mathbf{D}_T^{*}\bm{\Xi}_T^{*,\top}(\tau)\mathbf{Q}^{*,-1,\top}(\tau)\right) \mathbf{S}_{T}^{+}(\tau) \left(\mathbf{I}_2\otimes\mathbf{Q}^{*,-1}(\tau)\bm{\Xi}_T^*(\tau)\mathbf{D}_T^{*}\right)[\mathbf{I}_{dp_0},\mathbf{0}_{dp_0}]^\top\nonumber\\
			&&+\frac{1}{2}h^2[\bm{\Pi}^{(2)}(\tau), \bm{\Gamma}^{(2)}(\tau)]\mathbf{Q}^{*,-1}(\tau)\bm{\Xi}_T^*(\tau)\mathbf{D}_T^{*}\nonumber\\
			&&\times\mathbf{D}_T^{*,-1}\bm{\Xi}_T^{*,\top}(\tau)\mathbf{Q}^{*}(\tau)[\mathbf{S}_{T,2}(\tau),\mathbf{S}_{T,3}(\tau)]\left(\mathbf{I}_2\otimes\mathbf{Q}^{*,\top}(\tau)\bm{\Xi}_T^*(\tau)\mathbf{D}_T^{*,-1}\right)\nonumber\\
			&&\times\left(\mathbf{I}_2\otimes\mathbf{D}_T^{*}\bm{\Xi}_T^{*,\top}(\tau)\mathbf{Q}^{*,-1,\top}(\tau)\right) \mathbf{S}_{T}^{+}(\tau) \left(\mathbf{I}_2\otimes\mathbf{Q}^{*,-1}(\tau)\bm{\Xi}_T^*(\tau)\mathbf{D}_T^{*}\right)[\mathbf{I}_{dp_0},\mathbf{0}_{dp_0}]^\top\nonumber\\
			&& +
			\sum_{t=1}^{T}\mathbf{M}(\tau_t)\mathbf{h}_{t-1}\mathbf{h}_{t-1}^{*,\top} K\left(\frac{\tau_t-\tau}{h}\right) \left(\mathbf{I}_2\otimes\mathbf{Q}^{*,\top}(\tau)\bm{\Xi}_T^*(\tau)\mathbf{D}_T^{*,-1}\right)\nonumber\\
			&& \times \left(\mathbf{I}_2\otimes\mathbf{D}_T^{*}\bm{\Xi}_T^{*,\top}(\tau)\mathbf{Q}^{*,-1,\top}(\tau)\right) \mathbf{S}_{T}^{+}(\tau) \left(\mathbf{I}_2\otimes\mathbf{Q}^{*,-1}(\tau)\bm{\Xi}_T^*(\tau)\mathbf{D}_T^{*}\right)[\mathbf{I}_{dp_0},\mathbf{0}_{dp_0}]^\top\\
			&& + 
			\sum_{t=1}^{T}\mathbf{u}_{t}\mathbf{h}_{t-1}^{*,\top} K\left(\frac{\tau_t-\tau}{h}\right) \left(\mathbf{I}_2\otimes\mathbf{Q}^{*,\top}(\tau)\bm{\Xi}_T^*(\tau)\mathbf{D}_T^{*,-1}\right)\nonumber\\
			&&\times  \left(\mathbf{I}_2\otimes\mathbf{D}_T^{*}\bm{\Xi}_T^{*,\top}(\tau)\mathbf{Q}^{*,-1,\top}(\tau)\right) \mathbf{S}_{T}^{+}(\tau) \left(\mathbf{I}_2\otimes\mathbf{Q}^{*,-1}(\tau)\bm{\Xi}_T^*(\tau)\mathbf{D}_T^{*}\right)[\mathbf{I}_{dp_0},\mathbf{0}_{dp_0}]^\top\nonumber\\
			&=& \mathbf{J}_{T,1} + \mathbf{J}_{T,2} + \mathbf{J}_{T,3} + \mathbf{J}_{T,4}.
		\end{eqnarray*}

		Consider $\mathbf{J}_{T,1}$. By Lemma \ref{L4}, 
		 $\left(\mathbf{I}_2\otimes\mathbf{D}_T^{*,-1}\bm{\Xi}_T^{*,\top}(\tau)\mathbf{Q}^{*}(\tau)\right)\mathbf{S}_{T}(\tau)\left(\mathbf{I}_2\otimes \mathbf{Q}^{*,\top}(\tau)\bm{\Xi}_T^*(\tau)\mathbf{D}_T^{*,-1}\right)$ is asymptotically non-singular and thus
		$$
		\mathbf{J}_{T,1} = [\bm{\Pi}(\tau), \bm{\Gamma}(\tau)]\mathbf{Q}^{*,-1}(\tau)\bm{\Xi}_T^*(\tau)\mathbf{D}_T^{*}.
		$$
		
		Consider $\mathbf{J}_{T,2}$. Note that $\bm{\Pi}^{(j)}(\tau) = \bm{\alpha}^{(j)}(\tau)\bm{\beta}^\top$ and thus 
		$$
		[\bm{\Pi}^{(j)}(\tau), \bm{\Gamma}^{(j)}(\tau)]\mathbf{Q}^{*,-1}(\tau) = [\bm{\alpha}^{(j)}(\tau)\bm{\beta}^\top\mathbf{Q}^{-1}(\tau), \bm{\Gamma}^{(j)}(\tau)] = [\mathbf{0}_{d\times (d-r_0)},\bm{\alpha}^{(j)}(\tau),\bm{\Gamma}^{(j)}(\tau)].
		$$
		By Lemma \ref{L4}, we have
		\begin{eqnarray*}
			&&\mathbf{D}_T^{*,-1}\bm{\Xi}_T^{*,\top}(\tau)\mathbf{Q}^{*}(\tau)[\mathbf{S}_{T,2}(\tau),\mathbf{S}_{T,3}(\tau)]\left(\mathbf{I}_2\otimes\mathbf{Q}^{*,\top}(\tau)\bm{\Xi}_T^*(\tau)\mathbf{D}_T^{*,-1}\right)\nonumber\\
			&&\times\left(\mathbf{I}_2\otimes\mathbf{D}_T^{*}\bm{\Xi}_T^{*,\top}(\tau)\mathbf{Q}^{*,-1,\top}(\tau)\right) \mathbf{S}_{T}^{+}(\tau) \left(\mathbf{I}_2\otimes\mathbf{Q}^{*,-1}(\tau)\bm{\Xi}_T^*(\tau)\mathbf{D}_T^{*}\right)[\mathbf{I}_{dp_0+1},\mathbf{0}_{dp_0+1}]^\top\nonumber\\
			&=&\left[\begin{matrix}
				O_P(1) & \mathbf{0}_{(d-r_0)\times (r_0+d(p_0-1))}+ O_P(1/\sqrt{Th})&
				\\ \mathbf{0}_{ (r_0+d(p_0-1))\times(d-r_0)}+O_P(1/\sqrt{Th}) & \widetilde{c}_2\mathbf{I}_{ (r_0+d(p_0-1))}+O_P(h^2+1/\sqrt{Th})\\
			\end{matrix} \right]\nonumber
		\end{eqnarray*}
		and thus
		\begin{eqnarray*}
			\mathbf{J}_{T,2} &=& \frac{1}{2} h^2[\mathbf{0}_{d\times (d-r_0)},\sqrt{Th}\bm{\alpha}^{(2)}(\tau),\sqrt{Th}\bm{\Gamma}^{(2)}(\tau)] \nonumber\\
			&&\times \left[\begin{matrix}
				O_P(1) & \mathbf{0}_{(d-r_0)\times (r_0+d(p_0-1))}+ O_P(1/\sqrt{Th})&
				\\ \mathbf{0}_{ (r_0+d(p_0-1))\times(d-r_0)}+O_P(1/\sqrt{Th}) & \widetilde{c}_2\mathbf{I}_{ (r_0+d(p_0-1))}+O_P(h^2+1/\sqrt{Th})\\
			\end{matrix} \right]\nonumber \\
			&=& \frac{1}{2}\sqrt{Th}h^2\widetilde{c}_2[\mathbf{0}_{d\times (d-r_0)},\bm{\alpha}^{(2)}(\tau),\bm{\Gamma}^{(2)}(\tau)] + O_P(h^2)\nonumber \\
			&=& \frac{1}{2}\sqrt{Th}h^2\widetilde{c}_2[\bm{\alpha}^{(2)}(\tau)\bm{\beta}^\top,\bm{\Gamma}^{(2)}(\tau)] \mathbf{Q}^{*,-1}(\tau)+ O_P(h^2).\nonumber
		\end{eqnarray*}
		
		Consider $\mathbf{J}_{T,3}$. Note that $\bm{\Pi}^{(j)}(\tau) = \bm{\alpha}^{(j)}(\tau)\bm{\beta}^\top$ and thus $\bm{\Pi}^{(j)}(\tau)\mathbf{Q}^{-1}(\tau) = [\mathbf{0}_{d\times (d-r_0)},\bm{\mathbf{\alpha}}^{(j)}(\tau)]$. Let $\mathbf{M}(\tau_t)\mathbf{Q}^{*,-1}(\tau) = [\mathbf{0}_{d\times (d-r_0)},\mathbf{M}_1(\tau_t),\mathbf{M}_2(\tau_t)]$ such that $\mathbf{M}_1(\tau_t)$ is $d\times r_0$ and $\mathbf{M}_2(\tau_t)$ is $d\times d(p_0-1)$. Hence, by using similar arguments as in the proof of Lemma \ref{L4}, $\mathbf{M}_1(\tau_t) = O(h^3)$ and $\mathbf{M}_2(\tau_t) = O(h^3)$, we have
		\begin{eqnarray*}
			\mathbf{J}_{T,3}&=&
			\sum_{t=1}^{T}\mathbf{M}(\tau_t)\mathbf{Q}^{*,-1}(\tau)\mathbf{Q}^{*}(\tau)\mathbf{h}_{t-1}\mathbf{h}_{t-1}^{*,\top} K\left(\frac{\tau_t-\tau}{h}\right) \left(\mathbf{I}_2\otimes\mathbf{Q}^{*,\top}(\tau)\bm{\Xi}_T^*(\tau)\mathbf{D}_T^{*,-1}\right)\nonumber\\
			&&\times \left(\mathbf{I}_2\otimes\mathbf{D}_T^{*}\bm{\Xi}_T^{*,\top}(\tau)\mathbf{Q}^{*,-1,\top}(\tau)\right) \mathbf{S}_{T}^{+}(\tau) \left(\mathbf{I}_2\otimes\mathbf{Q}^{*,-1}(\tau)\bm{\Xi}_T^*(\tau)\mathbf{D}_T^{*}\right)[\mathbf{I}_{dp_0},\mathbf{0}_{dp_0}]^\top\nonumber\\
			&=& 
			\sum_{t=1}^{T}\left( \mathbf{M}_1(\tau_t)\bm{\beta}^\top\mathbf{y}_{t-1}+ \mathbf{M}_2(\tau_t)\Delta\mathbf{x}_{t-1}\right)\nonumber\\
			&&\times\left[\mathbf{y}_{t-1}^\top\bm{\alpha}_{\perp}(\tau)\bm{\Xi}_T(\tau)\mathbf{D}_T^{-1},\frac{1}{\sqrt{Th}}\mathbf{w}_{t-1}\right]\otimes \left[1,\frac{\tau_t-\tau}{h}\right] K\left(\frac{\tau_t-\tau}{h}\right) \cdot O_P(1)\nonumber\\
			&=&[O_P(h^3),O_P(h^3\sqrt{Th})].\nonumber
		\end{eqnarray*}

		Consider $\mathbf{J}_{T,4}$. Write
		\begin{eqnarray*}
			&&\left[\begin{matrix}
				\sum_{t=1}^{T}\mathbf{u}_{t}\mathbf{y}_{t-1}^\top\left(\frac{\tau_t-\tau}{h}\right)^l K\left(\frac{\tau_t-\tau}{h}\right) & \sum_{t=1}^{T}\mathbf{u}_{t}\Delta\mathbf{x}_{t-1}^\top\left(\frac{\tau_t-\tau}{h}\right)^l K\left(\frac{\tau_t-\tau}{h}\right)
			\end{matrix} \right] \mathbf{Q}^{*,\top}(\tau)\bm{\Xi}_T^*(\tau)\mathbf{D}_T^{*,-1}\nonumber\\
			&=&\left[\begin{matrix}
				\sum_{t=1}^{T}\mathbf{u}_{t}\mathbf{y}_{t-1}^\top\bm{\alpha}_{\perp}(\tau)\bm{\Xi}_T(\tau)\mathbf{D}_T^{-1}\left(\frac{\tau_t-\tau}{h}\right)^l K\left(\frac{\tau_t-\tau}{h}\right) & \frac{1}{\sqrt{Th}}\sum_{t=1}^{T}\mathbf{u}_{t}\mathbf{w}_{t-1}^\top\left(\frac{\tau_t-\tau}{h}\right)^l K\left(\frac{\tau_t-\tau}{h}\right)
			\end{matrix} \right]\nonumber\\
			& =& [\mathbf{J}_{T,41}, \mathbf{J}_{T,42}],\nonumber
		\end{eqnarray*}
		where $\mathbf{w}_t = [\mathbf{z}_t^\top\bm{\beta} ,\Delta\mathbf{x}_t^\top]^\top$. For $\mathbf{J}_{T,42}$, note that $\{\mathrm{vec}\left(\mathbf{u}_{t}\mathbf{w}_{t-1}^\top\right)\}_{t=1}^T$ is a sequence of martingale differences, the Lindeberg's condition can be verified as $\|\mathbf{u}_{t}\mathbf{w}_{t-1}^\top\|_{\delta/2}\leq \|\mathbf{u}_{t}\|_\delta\|\mathbf{w}_{t-1}\|_{\delta}$ for some $\delta > 4$ and the convergence of conditional variance can be verified by Lemma \ref{L5} (2). We then have
		$$
		\frac{1}{\sqrt{Th}}\sum_{t=1}^{T}\mathrm{vec}\left(\mathbf{u}_{t}\mathbf{w}_{t-1}^\top\right)\left(\frac{\tau_t-\tau}{h}\right)^l K\left(\frac{\tau_t-\tau}{h}\right) \to_D N\left(\mathbf{0}, \widetilde{v}_{2l}\bm{\Sigma}_{\mathbf{w}}(\tau)\otimes \bm{\Omega}(\tau)\right)
		$$
		and thus $\mathbf{J}_{T,42} = O_P(1)$. Next, consider $\mathbf{J}_{T,41}$,
		\begin{eqnarray*}
			\mathbf{J}_{T,41} &=& \left[\begin{matrix}
				\frac{1}{T\sqrt{h}}\sum_{t=1}^{T}\mathbf{u}_t\mathbf{y}_{t-1}^\top\bm{\alpha}_{\perp}(\tau)\bm{\xi}_T(\tau)\left(\frac{\tau_t-\tau}{h}\right)^lK\left(\frac{\tau_t-\tau}{h}\right)  & \frac{1}{Th}\sum_{t=1}^{T}\mathbf{u}_t\mathbf{y}_{t-1}^\top\bm{\alpha}_{\perp}(\tau)\bm{\xi}_{T,\perp}(\tau)\left(\frac{\tau_t-\tau}{h}\right)^lK\left(\frac{\tau_t-\tau}{h}\right)\nonumber\\
			\end{matrix}  \right] \\
			&=& \left[\mathbf{J}_{T,411}, \mathbf{J}_{T,412} \right].\nonumber
		\end{eqnarray*}
		
		For $\mathbf{J}_{T,411}$, similar to the proof of Lemma \ref{L4} (1), by the continuous mapping theory, we have
		\begin{eqnarray*}
			\mathbf{J}_{T,411} &=& \frac{1}{\sqrt{Th}}\sum_{t=1}^{T}\mathbf{u}_t\left(\frac{\tau_t-\tau}{h}\right)^lK\left(\frac{\tau_t-\tau}{h}\right) \left(\frac{1}{\sqrt{T}}\mathbf{y}_{\delta_T}^\top\bm{\alpha}_{\perp}(\tau)\bm{\xi}_T(\tau)\right) \nonumber\\
			&&+ \frac{1}{\sqrt{T}}\sum_{t=1}^{T}\mathbf{u}_t\frac{1}{\sqrt{Th}}\left(\mathbf{y}_{t-1}-\mathbf{y}_{\delta_T}\right)^\top\bm{\alpha}_{\perp}(\tau)\bm{\xi}_T(\tau)\left(\frac{\tau_t-\tau}{h}\right)^lK\left(\frac{\tau_t-\tau}{h}\right) \nonumber\\
			&=& \frac{1}{\sqrt{Th}}\sum_{t=1}^{T}\mathbf{u}_t\left(\frac{\tau_t-\tau}{h}\right)^lK\left(\frac{\tau_t-\tau}{h}\right) \left(\frac{1}{\sqrt{T}}\mathbf{y}_{\delta_T}^\top\bm{\alpha}_{\perp}(\tau)\bm{\xi}_T(\tau)\right) + O_P(\sqrt{h})\nonumber\\
			&\to_D &  \sqrt{2\mathbf{q}(\tau)^\top\mathbf{q}(\tau)} \int_{-1}^{1}u^lK(u)\mathrm{d}\mathbf{W}_d^{*}((u+1)/2,\bm{\Omega}(\tau)).\nonumber
		\end{eqnarray*}
		
		Similarly, for $\mathbf{J}_{T,412}$, we have
		\begin{eqnarray*}
			\mathbf{J}_{T,412} &=& \frac{2}{\sqrt{2Th}}\sum_{t=1}^{T}\mathbf{u}_t\frac{1}{\sqrt{2Th}}\left(\mathbf{y}_{t-1}-\mathbf{y}_{\delta_T}\right)^\top\bm{\alpha}_{\perp}(\tau)\bm{\xi}_{T,\perp}(\tau)\left(\frac{\tau_t-\tau}{h}\right)^lK\left(\frac{\tau_t-\tau}{h}\right) \nonumber\\
			&\to_D &  \left(\bm{\xi}_{\perp}^\top(\tau)\int_{-1}^{1}\mathbf{W}_d^{*}((u+1)/2,\bm{\Sigma}_{\bm{\alpha}}(\tau))u^lK(u)\mathrm{d}\mathbf{W}_d^{*,\top}((u+1)/2,\bm{\Omega}(\tau))\right)^\top.\nonumber
		\end{eqnarray*}
		Therefore, we have $\mathbf{J}_{T,41} = O_P(1)$. Hence, by Lemma \ref{L4} (6), we have
		\begin{eqnarray*}
			&&\mathbf{J}_{T,4}=\left[O_P(1), \frac{1}{\sqrt{Th}}\sum_{t=1}^{T}\mathbf{u}_{t}\mathbf{w}_{t-1}^\top K\left(\frac{\tau_t-\tau}{h}\right),O_P(1),\frac{1}{\sqrt{Th}}\sum_{t=1}^{T}\mathbf{u}_{t}\mathbf{w}_{t-1}^\top\left(\frac{\tau_t-\tau}{h}\right)  K\left(\frac{\tau_t-\tau}{h}\right) \right]\nonumber\\
			&&\times \left[\begin{matrix}
				O_P(1) & O_P(1/\sqrt{Th})\\
				O_P(1/\sqrt{Th}) & \bm{\Sigma}_{\mathbf{w}}^{-1}(\tau)+O_P(h^2+1/\sqrt{Th})\\
				O_P(1) & O_P(1/\sqrt{Th})\\
				O_P(1/\sqrt{Th}) & O_P(h+1/\sqrt{Th})\\
			\end{matrix} \right]\nonumber\\
			&=&\left[O_P(1), \frac{1}{\sqrt{Th}}\sum_{t=1}^{T}\mathbf{u}_{t}\mathbf{w}_{t-1}^\top \bm{\Sigma}_{\mathbf{w}}^{-1}(\tau)K\left(\frac{\tau_t-\tau}{h}\right) \right] +o_P(1).\nonumber
		\end{eqnarray*}
		
		Combining the above derivations, as $Th^5 = O(1)$, we have
		\begin{eqnarray*}
			&& \sqrt{Th}\mathrm{vec}\left(\left([\widehat{\bm{\Pi}}(\tau), \widehat{\bm{\Gamma}}(\tau)] - [\bm{\Pi}(\tau), \bm{\Gamma}(\tau)]- \frac{1}{2}\widetilde{c}_2h^2[\bm{\Pi}^{(2)}(\tau), \bm{\Gamma}^{(2)}(\tau)]\right)\mathbf{Q}^{*,-1}(\tau)\right)\nonumber\\
			& = & \left(\mathbf{Q}^{*,-1,\top}(\tau)\otimes \mathbf{I}_d \right) \sqrt{Th}\left(\mathrm{vec}\left( [\widehat{\bm{\Pi}}(\tau), \widehat{\bm{\Gamma}}(\tau)] - [\bm{\Pi}(\tau), \bm{\Gamma}(\tau)]- \frac{1}{2}\widetilde{c}_2h^2[\bm{\Pi}^{(2)}(\tau), \bm{\Gamma}^{(2)}(\tau)]\right)\right)\nonumber\\
			&=& \mathrm{vec}\left[\begin{matrix}
				\mathbf{0}_{d\times (d-r_0)}& \frac{1}{\sqrt{Th}}\sum_{t=1}^{T}\mathbf{u}_{t}\mathbf{w}_{t-1}^\top K\left(\frac{\tau_t-\tau}{h}\right) \bm{\Sigma}_{\mathbf{w}}^{-1}(\tau)
			\end{matrix} \right] + o_P(1)\nonumber\\
			&\to_D & \left[ \begin{matrix}
				\mathbf{0}_{d(d-r_0)\times 1}\\
				N\left(\mathbf{0}, \widetilde{v}_0\bm{\Sigma}_{\mathbf{w}}^{-1}(\tau)\otimes \bm{\Omega}(\tau) \right)
			\end{matrix}\right],\nonumber
		\end{eqnarray*}
		which implies
		\begin{eqnarray*}
			&&\sqrt{Th}\mathrm{vec}\left([\widehat{\bm{\Pi}}(\tau), \widehat{\bm{\Gamma}}(\tau)] - [\bm{\Pi}(\tau), \bm{\Gamma}(\tau)] -\frac{1}{2}h^2\widetilde{c}_2[\bm{\Pi}^{(2)}(\tau), \bm{\Gamma}^{(2)}(\tau)]\right)\nonumber\\
			&\to_D& N\left(\mathbf{0}, \widetilde{v}_0\left(\left[ \begin{matrix}
				\bm{\beta} & \mathbf{0}\\
				\mathbf{0} & \mathbf{I}_{d(p_0-1)}\\
			\end{matrix}\right]\bm{\Sigma}_{\mathbf{w}}^{-1}(\tau)\left[ \begin{matrix}
				\bm{\beta}^\top & \mathbf{0}\\
				\mathbf{0}& \mathbf{I}_{d(p_0-1)}\\
			\end{matrix}\right]\right)\otimes \bm{\Omega}(\tau) \right).\nonumber
		\end{eqnarray*}
		
		\medskip
		
		\noindent (2). Applying the uniform convergence results in Lemmas \ref{L5} and \ref{L7} to terms $\mathbf{J}_{T,1}$--$\mathbf{J}_{T,4}$, which are defined in the proof of part (1), we have 
		$$
		\sup_{\tau \in [0,1]}\left|[\widehat{\bm{\Pi}}(\tau), \widehat{\bm{\Gamma}}(\tau)] - [\bm{\Pi}(\tau), \bm{\Gamma}(\tau)]\right| = O_P\left(h^2 + \sqrt{\log T/(Th)} \right).
		$$
		
		\medskip
		
		\noindent (3). Consider $\widehat{\bm{\Omega}}(\tau)$ first. Note that by the proof of part (1) and the uniform convergence results in Lemmas \ref{L5} and \ref{L7}, we have
		$$
		\sup_{0\leq \tau \leq 1}\left|\left([\widehat{\bm{\Pi}}(\tau), \widehat{\bm{\Gamma}}(\tau)] - [\bm{\Pi}(\tau), \bm{\Gamma}(\tau)]\right)\mathbf{Q}^{*,-1}(\tau)\bm{\Xi}_T^*(\tau)\mathbf{D}_T^{*} \right| = O_P\left([\log T, \sqrt{Th}h^2 + \sqrt{\log T}]\right).
		$$
		
		Note that
		$$
		\widehat{\mathbf{u}}_t = \mathbf{u}_t + \left(\bm{\Pi}(\tau_t)-\widehat{\bm{\Pi}}(\tau_t)\right)\mathbf{y}_{t-1} + \left(\bm{\Gamma}(\tau_t)-\widehat{\bm{\Gamma}}(\tau_t)\right)\Delta \mathbf{x}_{t-1},
		$$
		which implies that
		\begin{eqnarray*}
			&&\frac{1}{T}\sum_{t=1}^{T}\left(\mathbf{u}_t-\widehat{\mathbf{u}}_t\right)\mathbf{u}_t^\top w_t(\tau)\nonumber\\
			&=& \frac{1}{T}\sum_{t=1}^{T}\left([\widehat{\bm{\Pi}}(\tau_t), \widehat{\bm{\Gamma}}(\tau_t)] - [\bm{\Pi}(\tau_t), \bm{\Gamma}(\tau_t)]\right)\mathbf{h}_{t-1}\mathbf{u}_t^\top w_t(\tau)\nonumber\\
			&=& \frac{1}{T}\sum_{t=1}^{T}\left([\widehat{\bm{\Pi}}(\tau_t), \widehat{\bm{\Gamma}}(\tau_t)] - [\bm{\Pi}(\tau_t), \bm{\Gamma}(\tau_t)]\right)\mathbf{Q}^{*,-1}(\tau)\bm{\Xi}_T^*(\tau)\mathbf{D}_T^{*}\left[\begin{matrix}
				\frac{1}{T\sqrt{h}}\bm{\xi}_T^\top(\tau_t)\bm{\alpha}_{\perp}^\top(\tau_t)\mathbf{y}_{t-1}\nonumber\\
				\frac{1}{Th}\bm{\xi}_{T,\perp}^\top(\tau_t)\bm{\alpha}_{\perp}^\top(\tau_t)(\mathbf{y}_{t-1}-\mathbf{y}_{\delta_T})\nonumber\\
				\frac{1}{\sqrt{Th}}\bm{\beta}^\top\mathbf{y}_{t-1}\nonumber\\
				\frac{1}{\sqrt{Th}}\Delta\mathbf{x}_{t-1}
			\end{matrix}\right]\mathbf{u}_t^\top w_t(\tau)\nonumber\\
			&=&O_P(h^2 + \sqrt{\log T/(Th)})\nonumber
		\end{eqnarray*}
		uniformly over $\tau \in [0,1]$. Similarly, it is easy to show that 
		$$
		\sup_{\tau\in[0,1]}\left|\frac{1}{T}\sum_{t=1}^{T}\left(\mathbf{u}_t-\widehat{\mathbf{u}}_t\right)\left(\mathbf{u}_t-\widehat{\mathbf{u}}_t\right)^\top w_t(\tau)\right| = O_P \left(h^4 + \log T/(Th)\right).
		$$
		Hence, we have
		\begin{eqnarray*}
			\sup_{\tau \in [0,1]}\left|\widehat{\bm{\Omega}}(\tau) - \bm{\Omega}(\tau)\right| =\sup_{\tau \in [0,1]}\left|T^{-1}\sum_{t=1}^{T}\mathbf{u}_t\mathbf{u}_t^\top w_t\left(\tau\right) - \bm{\Omega}(\tau)\right|+ O_P\left(h^2 + \sqrt{\log T/(Th)} \right).\nonumber
		\end{eqnarray*}
		
		In addition, by using similar arguments to the proof of Lemma \ref{L5}, we have
		$$
		\sup_{\tau \in [0,1]}\left|T^{-1}\sum_{t=1}^{T}\mathbf{u}_t\mathbf{u}_t^\top w_t\left(\tau\right) - \bm{\Omega}(\tau)\right| =  O_P\left(h^2 + \sqrt{\log T/(Th)}\right).
		$$
		
		Next, consider $\left[\sum_{t=1}^{T}K\left(\frac{\tau_t-\tau}{h}\right) \right]\mathbf{S}_{T,0}^+(\tau)$. By Lemmas \ref{L5} and \ref{L7}, we have uniformly over $\tau \in [h,1-h]$,
		\begin{eqnarray*}
			&&\left[\sum_{t=1}^{T}K\left(\frac{\tau_t-\tau}{h}\right) \right]\mathbf{S}_{T,0}^+(\tau)\nonumber\\
			&=& Th(\widetilde{c}_0+O(1/(Th)))\mathbf{Q}^{*,\top}(\tau)\bm{\Xi}_T^*(\tau)\mathbf{D}_T^{*,-1}\nonumber\\
			&&\left[\begin{matrix}
				\bm{\Delta}_{0}^{-1}(\tau)+o_P(1) & O_P(\sqrt{\log T /(Th)}) &
				\\ O_P(\sqrt{\log T /(Th)})& \widetilde{c}_0^{-1}\bm{\Sigma}_{\mathbf{w}}^{-1}(\tau) + O_P(h^2+\sqrt{\log T /(Th)}) \\
			\end{matrix} \right]\cdot  \mathbf{D}_T^{*,-1}\bm{\Xi}_T^{*,\top}(\tau)\mathbf{Q}^{*}(\tau)\nonumber\\
			&=& \mathbf{Q}^{*,\top}(\tau)\left[\begin{matrix}
				\mathbf{0} & \mathbf{0} &
				\\ \mathbf{0}& \widetilde{c}_0^{-1}\bm{\Sigma}_{\mathbf{w}}^{-1}(\tau) \\
			\end{matrix} \right]\mathbf{Q}^{*}(\tau)+ O_P(h^2+\sqrt{\log T /(Th)})\nonumber\\
			&=&\left[ \begin{matrix}
				\bm{\beta} & \mathbf{0}\\
				\mathbf{0} & \mathbf{I}_{d(p_0-1)}\\
			\end{matrix}\right]\bm{\Sigma}_{\mathbf{w}}^{-1}(\tau)\left[ \begin{matrix}
				\bm{\beta}^\top & \mathbf{0}\\
				\mathbf{0}& \mathbf{I}_{d(p_0-1)}\\
			\end{matrix}\right] + O_P(h^2+\sqrt{\log T /(Th)}).\nonumber
		\end{eqnarray*}

		The proof is now completed.

	\end{proof}

\begin{proof}[Proof of Theorem \ref{Thm3}]
	\item
	Note that $\widetilde{\mathbf{r}}_t(\widehat{\bm{\alpha}}) = \widetilde{\mathbf{R}}_t^\top(\widehat{\bm{\alpha}}) \mathrm{vec}\left(\bm{\beta}^{*,\top}\right) + \mathbf{u}_t^*$, where 
	$$
	\mathbf{u}_t^* = \mathbf{u}_t + \left(\bm{\alpha}(\tau_t) - \widehat{\bm{\alpha}}(\tau_t)\right)\bm{\beta}^{\top}\mathbf{y}_{t-1} + \left(\bm{\Pi}(\tau_t) - \widehat{\bm{\Pi}}(\tau_t,\widehat{\bm{\alpha}}\bm{\beta}^\top)\right)\Delta \mathbf{x}_{t-1}
	$$
	and
	\begin{eqnarray*}
		\widehat{\bm{\Gamma}}(\tau,\widehat{\bm{\alpha}}\bm{\beta}^\top) &=& \sum_{s=1}^{T}(\Delta\mathbf{y}_{s}-\widehat{\bm{\alpha}}(\tau_s)\bm{\beta}^\top\mathbf{y}_{s-1})\Delta\mathbf{x}_{s-1}^{*,\top} K\left(\frac{\tau_s-\tau}{h}\right)\nonumber\\
		&&\times\left(\sum_{s=1}^{T}\Delta\mathbf{x}_{s-1}^*\Delta\mathbf{x}_{s-1}^{*,\top} K\left(\frac{\tau_s-\tau}{h}\right)\right)^{-1}\left[\begin{matrix}
			\mathbf{I}_{d(p_0-1)} \\
			\mathbf{0}_{d(p_0-1)}
		\end{matrix} \right],\nonumber
	\end{eqnarray*}
	which implies that
	\begin{eqnarray*}
		&&T\mathrm{vec}\left[\widehat{\bm{\beta}}^{*,\top}-\bm{\beta}^{*,\top}\right]\nonumber\\
		&=& \left(\frac{1}{T^2}\sum_{t=1}^{T}\widetilde{\mathbf{R}}_t(\widehat{\bm{\alpha}})\widehat{\bm{\Omega}}^{-1}(\tau_t)\widetilde{\mathbf{R}}_t^\top(\widehat{\bm{\alpha}})\right)^{-1}\frac{1}{T}\sum_{t=1}^{T}\widetilde{\mathbf{R}}_t(\widehat{\bm{\alpha}})\widehat{\bm{\Omega}}^{-1}(\tau_t)\mathbf{u}_t^*\nonumber\\
		&=& \left(\frac{1}{T^2}\sum_{t=1}^{T}\widetilde{\mathbf{R}}_t(\widehat{\bm{\alpha}})\widehat{\bm{\Omega}}^{-1}(\tau_t)\widetilde{\mathbf{R}}_t^\top(\widehat{\bm{\alpha}})\right)^{-1}\frac{1}{T}\sum_{t=1}^{T}\widetilde{\mathbf{R}}_t(\widehat{\bm{\alpha}})\widehat{\bm{\Omega}}^{-1}(\tau_t)\mathbf{u}_t\nonumber\\
		&& + \left(\frac{1}{T^2}\sum_{t=1}^{T}\widetilde{\mathbf{R}}_t(\widehat{\bm{\alpha}})\widehat{\bm{\Omega}}^{-1}(\tau_t)\widetilde{\mathbf{R}}_t^\top(\widehat{\bm{\alpha}})\right)^{-1}\frac{1}{T}\sum_{t=1}^{T}\widetilde{\mathbf{R}}_t(\widehat{\bm{\alpha}})\widehat{\bm{\Omega}}^{-1}(\tau_t)\left(\bm{\alpha}(\tau_t) - \widehat{\bm{\alpha}}(\tau_t)\right)\bm{\beta}^{\top}\mathbf{y}_{t-1}\nonumber\\
		&& + \left(\frac{1}{T^2}\sum_{t=1}^{T}\widetilde{\mathbf{R}}_t(\widehat{\bm{\alpha}})\widehat{\bm{\Omega}}^{-1}(\tau_t)\widetilde{\mathbf{R}}_t^\top(\widehat{\bm{\alpha}})\right)^{-1}\frac{1}{T}\sum_{t=1}^{T}\widetilde{\mathbf{R}}_t(\widehat{\bm{\alpha}})\widehat{\bm{\Omega}}^{-1}(\tau_t)\left(\bm{\Pi}(\tau_t) - \widehat{\bm{\Pi}}(\tau_t,\widehat{\bm{\alpha}}\bm{\beta}^\top)\right)\Delta \mathbf{x}_{t-1}\nonumber\\
		&=&\mathbf{J}_{T,5} + \mathbf{J}_{T,6} + \mathbf{J}_{T,7}. \nonumber
	\end{eqnarray*}
	
	By Lemma \ref{L8}, $\mathbf{J}_{T,6}$ and $\mathbf{J}_{T,7}$ are both $o_P(1)$. Next, consider $\mathbf{J}_{T,5}$. By Lemma \ref{L8}, we have
	$$
	\mathbf{J}_{T,5} \to_D \left(\int_{0}^{1}\mathbf{W}_{d-r_0}(u) \mathbf{W}_{d-r_0}^\top(u) \otimes \bm{\alpha}^\top(u)\bm{\Omega}^{-1}(u)\bm{\alpha}(u)\mathrm{d}u\right)^{-1} \int_{0}^{1}\mathbf{W}_{d-r_0}(u)\otimes \mathrm{d}\mathbf{W}_{r_0}(u),
	$$
	where $\mathbf{W}_{d-r_0}(u)$ and $\mathbf{W}_{r_0}(u)$ are mutually independent.
	
	Since $\mathbf{W}_{d-r_0}(u)$ and $\mathbf{W}_{r_0}(u)$ are mutually independent, the conditional distribution of $\mathbf{J}_{T,5}$, given $\mathbf{W}_{d-r_0}(u)$, is
	$$
	\mathbf{J}_{T,5} \to_D N\left(\mathbf{0}, \left(\int_{0}^{1}\mathbf{W}_{d-r_0}(u) \mathbf{W}_{d-r_0}^\top(u) \otimes \bm{\alpha}^\top(u)\bm{\Omega}^{-1}(u)\bm{\alpha}(u)\mathrm{d}u\right)^{-1}\right),
	$$
	which implies that
	$$
	\left(\frac{1}{T^2}\sum_{t=1}^{T}\mathbf{y}_{t-1}^{(2)}\mathbf{y}_{t-1}^{(2),\top}  \otimes \widehat{\bm{\alpha}}^\top(\tau_t)\widehat{\bm{\Omega}}^{-1}(\tau_t) \widehat{\bm{\alpha}}(\tau_t)\right)^{1/2}\mathbf{J}_{T,1} \to_D N\left(\mathbf{0}, \mathbf{I}_{(d-r_0)r_0}\right).
	$$
	
	The proof is now completed.
\end{proof}

	\begin{proof}[Proof of Theorem \ref{Prop2.1}]
		\item
		We need to prove that $\lim_{T\to \infty}\Pr\left(\text{IC}(p) < \text{IC}(p_0)\right)=0$ for all $p\neq p_0$ and $p\leq P$. Note that
		$$
		\text{IC}(p)-\text{IC}(p_0)=\log[\text{RSS}(p)/\text{RSS}(p_0)]+(p-p_0)\chi_T.
		$$
		
		(i) For $p < p_0$, Lemma \ref{L9} implies that $\text{RSS}(p)/\text{RSS}(p_0) > 1 + \nu$ for some $\nu > 0 $ with large probability for all large $T$. Thus, $\log[\text{RSS}(p)/\text{RSS}(p_0)] \geq \nu/2$ for large $T$. Because $\chi_T\to 0$, we have $\text{IC}(p)-\text{IC}(p_0)\geq \nu/2-(p_0-p)\chi_T \geq \nu/3$ for large $T$ with large probability. Thus $\Pr\left(\text{IC}(p) < \text{IC}(p_0)\right) \to 0$ for $p < p_0$.
		
		(ii) We then consider $p > p_0$.  Lemma \ref{L9} implies that $\text{RSS}(p)/\text{RSS}(p_0)=1+O_P(c_T^2)$ with $c_T= h^2 + \sqrt{\frac{\log(T)}{Th}}$. Hence, $\log[\text{RSS}(p)/\text{RSS}(p_0)]=O_P(c_T^2)$. Because $(p-p_0) \chi_T \geq \chi_T $, which converges to zero at a slower rate than $c_T^2$, it implies that
		$$
		\Pr\left(\text{IC}(p) < \text{IC}(p_0)\right)\leq \Pr\left(\log[\text{RSS}(p)/\text{RSS}(p_0)] + \chi_T < 0\right) \to 0.
		$$

		The proof is now completed.		
	\end{proof}

	\begin{proof}[Proof of Theorem \ref{Thm2}]
		\item
		By the proof of Theorem \ref{Thm1} and using the uniform convergence results in Lemmas \ref{L5} and \ref{L7}, we have uniformly over $\tau \in [0,1]$
		\begin{eqnarray*}
			&&\left[ \widehat{\bm{\Pi}}(\tau) - {\bm{\Pi}}(\tau) , \widehat{\bm{\Gamma}}(\tau) - {\bm{\Gamma}}(\tau)\right] \mathbf{Q}^{*,-1}(\tau) \bm{\Xi}_T^*(\tau) \mathbf{D}_T^*\nonumber\\
			&=& \left[\mathbf{0}_{d\times (d-r_0)} , \frac{1}{2}\sqrt{Th^5}\widetilde{c}_2(\tau)\bm{\alpha}^{(2)}(\tau) ,\frac{1}{2}\sqrt{Th^5}\widetilde{c}_2(\tau)\bm{\Gamma}^{(2)}(\tau)\right] + O_P([\log T,\sqrt{\log T}]),\nonumber
		\end{eqnarray*}
		where $\tilde{c}_k(\tau) =\int_{-\tau/h}^{(1-\tau)/h} u^k K(u) \mathrm{d}u$. Hence, we have
		$$
		\int_{0}^{1}\widehat{\bm{\Pi}}(\tau)\mathrm{d}\tau - \int_{0}^{1}\bm{\alpha}(\tau)\mathrm{d}\tau\bm{\beta}^\top = \frac{1}{2}h^2\int_{0}^{1}\tilde{c}_2(\tau)\bm{\alpha}^{(2)}(\tau)\mathrm{d}\tau\bm{\beta}^\top + \left[\begin{matrix}
			O_P\left(\log T/(Th)\right)\int_{0}^{1}\bm{\alpha}_{\perp}^\top(\tau)\mathrm{d}\tau\\
			O_P\left(\sqrt{\log T/(Th)}\right)\bm{\beta}^\top
		\end{matrix}\right],
		$$
		which implies that $\bm{\beta}_{\perp}^\top\int_{0}^{1}\widehat{\bm{\Pi}}^\top(\tau)\mathrm{d}\tau = O_P\left(\log T/(Th)\right)$ and
		$$
		\bm{\beta}^\top\int_{0}^{1}\widehat{\bm{\Pi}}^\top(\tau)\mathrm{d}\tau = \bm{\beta}^\top \bm{\beta}\int_{0}^{1}\bm{\alpha}^\top(\tau)\mathrm{d}\tau + O_P\left(h^2+\sqrt{\log T/(Th)}\right).
		$$
		Let $\bm{\beta}_0 = [\bm{\beta}(\bm{\beta}^\top \bm{\beta})^{-1/2},\bm{\beta}_{\perp}]^\top$. Since $\bm{\beta}_0^\top\bm{\beta}_0=\mathbf{I}_d$, by the unitary invariance properties of singular values, we have
		$$
		\sigma_j\left(\bm{\beta}_0 \int_{0}^{1}\widehat{\bm{\Pi}}^\top(\tau)\mathrm{d}\tau\right) = \sigma_j\left(\widehat{\mathbf{S}} \int_{0}^{1}\widehat{\bm{\Pi}}^\top(\tau)\mathrm{d}\tau\right) = \sigma_j(\widehat{\mathbf{R}})
		$$
		for all $j=1,\ldots,d$. Since $
		(\bm{\beta}^\top \bm{\beta})^{-1/2}\bm{\beta}^\top\int_{0}^{1}\widehat{\bm{\Pi}}^\top(\tau)\mathrm{d}\tau = (\bm{\beta}^\top \bm{\beta})^{1/2}\int_{0}^{1}\bm{\alpha}^\top(\tau)\mathrm{d}\tau + O_P\left(h^2+\sqrt{\log T/(Th)}\right)$, by matrix perturbation theory, we have $\sigma_j\left(\widehat{\mathbf{R}} \right) - \sigma_j\left((\bm{\beta}^\top \bm{\beta})^{1/2}\int_{0}^{1}\bm{\alpha}(\tau)\mathrm{d}\tau\right) = O_P\left(h^2+\sqrt{\log T/(Th)}\right)$ for $j=1,\ldots,r_0$. Similarly, we have $\sigma_j\left(\widehat{\mathbf{R}} \right) = O_P\left( \log T/(Th) \right)$ for $j=r_0+1,\ldots,d$.
		
		\medskip
		
		By the column-pivoting step in the QR decomposition, we can conclude that
		$$
		\sigma_{r_0}(\widehat{\mathbf{R}}) \leq \widehat{\mu}_k = \sqrt{\sum_{j=k}^{d} \widehat{\mathbf{R}}^2(k,j)} \leq \sigma_{1}(\widehat{\mathbf{R}}) \quad \text{for}\quad  k=1,\ldots,r_0
		$$
		and
		$$
		\widehat{\mu}_k = \sqrt{\sum_{j=k}^{d} \widehat{\mathbf{R}}^2(k,j)} = O_P(\log T/(Th)) \quad \text{for}\quad  k=r_0+1,\ldots,d.
		$$ 
		
		When $r_0 = 0$, we have $\widehat{\mu}_k = O_P(\log T/(Th))$ for $k=1,\ldots,d$, which implies that $\widehat{\mu}_0/\widehat{\mu}_1 \to \infty$, while $\widehat{\mu}_r/\widehat{\mu}_{r+1} = 1$ for $r \geq 1$ with probability approaching to 1. Similarly, when $0< r_0 < d$, $\widehat{\mu}_r/\widehat{\mu}_{r+1} \asymp 1$ for $r=1,\ldots,r_0-1$, $\widehat{\mu}_{r_0}/\widehat{\mu}_{r_0+1} \to \infty$, and $\widehat{\mu}_r/\widehat{\mu}_{r+1} = 1$ for $r \geq r_0+ 1$ with probability approaching to 1. The proof is now completed.
		
	\end{proof}

	\begin{proof}[Proof of Theorem \ref{Thm4}]
		\item
		\noindent (1). By part (2) of Theorem \ref{Thm1} , we have 
		$$
		\sup_{\tau \in [0,1]}\left|[\widehat{\bm{\Pi}}(\tau), \widehat{\bm{\Gamma}}(\tau)] - [\bm{\Pi}(\tau), \bm{\Gamma}(\tau)]\right| = O_P\left(h^2 + \sqrt{\log T/(Th)} \right)
		$$
		and thus by $Th^6\to 0$, we have
		$$
		\sqrt{T}\left(\widehat{\mathbf{c}} - \mathbf{c}\right) = \sqrt{T}\frac{1}{T}\sum_{t=1}^{T}\mathbf{C}\left(\widehat{\mathbf{b}}(\tau_t)-\mathbf{b}(\tau_t)\right) = \sqrt{T}\frac{1}{T}\sum_{t=\lfloor Th\rfloor+1}^{\lfloor T(1-h)\rfloor}\mathbf{C}\left(\widehat{\mathbf{b}}(\tau_t)-\mathbf{b}(\tau_t)\right) + o_P(1).
		$$
		
		In addition, by the proof of part(1) of Theorem \ref{Thm1} and the uniform convergence results in Lemmas \ref{L5} and \ref{L7}, we have uniformly over $\tau \in [h,1-h]$,
		\begin{eqnarray*}
			&& [\widehat{\bm{\Pi}}(\tau)-\bm{\Pi}(\tau), \widehat{\bm{\Gamma}}(\tau)-\bm{\Gamma}(\tau)]\nonumber\\
			&=&\frac{1}{2}h^2\widetilde{c}_2[\bm{\alpha}^{(2)}(\tau)\bm{\beta}^\top,\bm{\Gamma}^{(2)}(\tau)]+ O(h^3) + O_P\left( h^2 \sqrt{\log T/(Th)}\right)\nonumber\\
			&&+ \frac{1}{Th}\sum_{t=1}^{T}\mathbf{u}_{t}\mathbf{w}_{t-1}^\top\bm{\Sigma}_{\mathbf{w}}^{-1}(\tau)  K\left(\frac{\tau_t-\tau}{h}\right) \left[ \begin{matrix}
				\bm{\beta}^\top & \mathbf{0}_{r_0\times d(p_0-1)}\\
				\mathbf{0}_{d(p_0-1)\times r_0}& \mathbf{I}_{d(p_0-1)}\\
			\end{matrix}\right] + O_P\left(\log T /(Th) \right).\nonumber
		\end{eqnarray*}
		Hence, we have
		\begin{eqnarray*}
			\widehat{\mathbf{b}}(\tau)-\mathbf{b}(\tau) &=& \frac{1}{2}h^2\widetilde{c}_2 \mathbf{b}^{(2)}(\tau) + \bm{\Sigma}_{\mathbf{w}}^{-1}(\tau)\otimes \mathbf{I}_d\left(\frac{1}{Th}\sum_{t=1}^{T}\mathbf{w}_{t-1}\otimes \mathbf{u}_t K\left(\frac{\tau_t-\tau}{h}\right)\right)\nonumber\\
			&& + O_P(h^2\sqrt{\log T/(Th)}) + O(h^3)+O_P(\log T/(Th))\nonumber
		\end{eqnarray*}
		uniformly over $\tau \in [h,1-h]$.
		
		Combining the above two results, as $Th^6\to0$ and $Th^2/(\log T)^2 \to \infty$, we have
		\begin{eqnarray*}
			&&\sqrt{T}\left(\widehat{\mathbf{c}} - \mathbf{c} -\frac{1}{2}h^2\widetilde{c}_2\int_{0}^{1}\mathbf{C}\mathbf{b}^{(2)}(\tau)\mathrm{d}\tau\right) \nonumber\\
			&=& \mathbf{C}\frac{1}{\sqrt{T}}\sum_{t=1}^{T}\left(\frac{1}{Th}\sum_{s=1}^{T}\bm{\Sigma}_{\mathbf{w}}^{-1}(\tau_s)\otimes \mathbf{I}_dK\left(\frac{\tau_s-\tau_t}{h}\right)\right)\mathbf{w}_{t-1} \otimes \mathbf{u}_t +o_P(1)\nonumber\\
			&=& \mathbf{C}\frac{1}{\sqrt{T}}\sum_{t=1}^{T}\left(\bm{\Sigma}_{\mathbf{w}}^{-1}(\tau_t)\otimes \mathbf{I}_d\right)\mathbf{w}_{t-1}\otimes \mathbf{u}_t +o_P(1)\nonumber\\
			&\to_D& N\left(\bm{0},\int_{0}^{1}\mathbf{C}\left(\mathbf{\Sigma}_{\mathbf{w}}^{-1}(\tau)\otimes \mathbf{\Omega}(\tau)\right) \mathbf{C}^\top\mathrm{d}\tau\right)\nonumber
		\end{eqnarray*}
		by the conventional martingale central limit theory. Note that the Lindeberg's condition can be verified as $\|\mathbf{w}_{t-1}\otimes \mathbf{u}_{t}\|_{\delta/2}\leq \|\mathbf{u}_{t}\|_\delta\|\mathbf{w}_{t-1}\|_{\delta}$ for some $\delta > 4$ and the convergence of conditional variance can be verified by Lemma \ref{L5} (2).
		
		\medskip
		
		\noindent (2). By part (1), we have $\sup_{\tau\in [0,1]} \|\widehat{\mathbf{b}}(\tau)-\mathbf{b}(\tau) \|=O_P\left(h^2 + \sqrt{\log T/(Th)}\right)$. Then we can conclude that as $Th^{11/2} = o(1)$,
		\begin{eqnarray}\label{EqA.2.1}
			&&\frac{1}{T}\sum_{t\in\mathcal{B}_T}\left[\mathbf{C}\widehat{\mathbf{b}}(\tau_t)-\mathbf{c}\right]^\top\mathbf{H}(\tau_t)
			\left[\mathbf{C}\widehat{\mathbf{b}}(\tau_t)-\mathbf{c}\right] = T^{-1} \times (Th)\times O_P\left(\log T / (Th) +h^4\right)  \nonumber\\
			& = &  O_P\left(\log T / T +h^5\right) =o_P\left(T^{-1}h^{-1/2}\right),
		\end{eqnarray}
		where $\mathcal{B}_T = \{1,\ldots,\lfloor Th \rfloor\}\cup \{\lfloor T(1-h) \rfloor, \ldots, T\}$.
		
		In addition, by  Lemma \ref{L5}, we have
		\begin{eqnarray}\label{Rrate}
			&&\sup_{\tau\in[h,1-h]}\left|\frac{1}{Th}\sum_{t=1}^{T}\mathbf{w}_{t-1}\mathbf{w}_{t-1}^\top K\left(\frac{\tau_t-\tau}{h}\right)-\bm{\Sigma}_{\mathbf{w}}(\tau)\right| =  O_P\left(h^2 + \sqrt{\frac{\log T}{Th}}\right),\nonumber \\
			&&\sup_{\tau\in[0,1]}\left| \mathbf{R}_T(\tau) \right|= O_P\left(\sqrt{\frac{\log T}{Th}}\right),
		\end{eqnarray}
		where $\mathbf{R}_T(\tau)=\frac{1}{Th}\sum_{t=1}^{T}\mathbf{W}_{t-1}\mathbf{u}_tK\left(\frac{\tau_t-\tau}{h}\right)$ and $\mathbf{W}_t = \mathbf{w}_t \otimes \mathbf{I}_d$. Hence, under the null hypothesis, we have the following Bahadur representation
		\begin{eqnarray}\label{Rrate2}
			\sup_{[h,1-h]}\left|\mathbf{C}\widehat{\mathbf{b}}(\tau)- \mathbf{c}- \mathbf{C}\left(\bm{\Sigma}_{\mathbf{w}}^{-1}(\tau) \otimes \mathbf{I}_d\right) \mathbf{R}_T(\tau)\right|=O_P(\rho_T^2)
		\end{eqnarray}
		with $\rho_T = h^2 + \sqrt{{\log T}/{(Th)}}$.
		
		By \eqref{EqA.2.1}-\eqref{Rrate2} and $Th^{11/2} \to 0$, under the null we have
		\begin{eqnarray*}
			&&\frac{1}{T}\sum_{t=1}^{T}\left[\mathbf{C}\widehat{\mathbf{b}}(\tau_t)-\mathbf{c}\right]^\top\mathbf{H}(\tau_t)
			\left[\mathbf{C}\widehat{\mathbf{b}}(\tau_t)-\mathbf{c}\right]\nonumber\\
			&=& \frac{1}{T}\sum_{t=\lfloor Th\rfloor+1}^{\lfloor T(1-h)\rfloor}\mathbf{R}_T^\top(\tau_t) \mathbf{H}_0 (\tau_t)\mathbf{R}_T(\tau_t) + O_P\left(\log T / T +h^5\right) + O_P\left(\rho_T^2 \sqrt{\frac{\log T}{Th}}\right) \nonumber\\
			&=& \frac{1}{T}\sum_{t=1}^{T}\mathbf{R}_T^\top(\tau_t) \mathbf{H}_{0}(\tau_t)\mathbf{R}_T(\tau_t) + o_P(T^{-1}h^{-1/2}),\nonumber
		\end{eqnarray*}
		where $\mathbf{H}_{0}(\tau)=\bm{\Sigma}_{\mathbf{W}}^{-1}(\tau)\mathbf{C}^\top\mathbf{H}(\tau)\mathbf{C}\bm{\Sigma}_{\mathbf{W}}^{-1}(\tau)$ and $\bm{\Sigma}_{\mathbf{W}}(\tau) = \bm{\Sigma}_{\mathbf{w}}(\tau)\otimes \mathbf{I}_d$.
		
		Consider $ \frac{1}{T}\sum_{t=1}^{T}\mathbf{R}_T^\top(\tau_t) \mathbf{H}_{0}(\tau_t)\mathbf{R}_T(\tau_t)$, and write
		\begin{eqnarray*}
			&&  \frac{1}{T}\sum_{t=1}^{T}\mathbf{R}_T^\top(\tau_t) \mathbf{H}_{0}(\tau_t)\mathbf{R}_T(\tau_t) =  \frac{1}{T^2h^2}\sum_{s=1}^{T}\mathbf{u}_s^\top\mathbf{W}_{s-1}^\top \left\{\frac{1}{T}\sum_{t=1}^{T}\mathbf{H}_0(\tau_t)K^2\left(\frac{\tau_s-\tau_t}{h}\right)\right\} \mathbf{W}_{s-1}\mathbf{u}_s\nonumber\\
			&& + \frac{1}{T^2h^2}\sum_{s=1}^{T}\sum_{v=1,\neq s}^{T}\mathbf{u}_s^\top\mathbf{W}_{s-1}^\top \left\{\frac{1}{T}\sum_{t=1}^{T}\mathbf{H}_0(\tau_t)K\left(\frac{\tau_s-\tau_t}{h}\right)K\left(\frac{\tau_v-\tau_t}{h}\right) \right\}\mathbf{W}_{v-1} \mathbf{u}_v\nonumber\\
			&:=&J_{T,8}+J_{T,9},\nonumber
		\end{eqnarray*}
		where the definitions of $J_{T,8}$ and $J_{T,9}$ should be obvious.
		
		Consider $J_{T,8}$. We first introduce some additional notation to facilitate the development. Let $A_Q=\widetilde{v}_0\cdot \mathrm{tr}\left\{\int_{0}^{1} \bm{\Sigma}_Q(\tau) \mathrm{d}\tau\right\}$ and $B_Q= 4 C_B\cdot \mathrm{tr}\left\{\int_{0}^{1}\bm{\Sigma}_Q^2(\tau) \mathrm{d}\tau\right\}$, where 
		$$
		\bm{\Sigma}_Q(\tau)=\mathbf{H}(\tau)^{1/2}\mathbf{C}\mathbf{V}_{\mathbf{b}}(\tau)\mathbf{C}^\top \mathbf{H}(\tau)^{1/2}. 
		$$
		Simple algebra shows that
		\begin{eqnarray*}
			&&\mathbf{u}_t^\top\mathbf{W}_{t-1}^\top\bm{\Sigma}_{\mathbf{W}}^{-1}(\tau)\mathbf{C}^\top\mathbf{H}(\tau)\mathbf{C}\bm{\Sigma}_{\mathbf{W}}^{-1}(\tau)\mathbf{W}_{t-1}\mathbf{u}_t \nonumber\\
			&=& \mathrm{tr}\left\{\left[\left(\bm{\Sigma}_{\mathbf{w}}^{-1}(\tau)\mathbf{w}_{t-1}\mathbf{w}_{t-1}^\top\bm{\Sigma}_{\mathbf{w}}^{-1}(\tau)\right)\otimes  \mathbf{u}_t\mathbf{u}_t^\top \right]\cdot\mathbf{C}^\top\mathbf{H}(\tau)\mathbf{C} \right\},\nonumber
		\end{eqnarray*}
		which implies that
		\begin{eqnarray*}
			J_{T,8} &=& \widetilde{v}_0\frac{1}{T^2h}\sum_{s=1}^{T}\mathbf{u}_s^\top\mathbf{W}_{s-1}^\top\left[\mathbf{H}_0(\tau_s)+O(h)\right]\mathbf{W}_{s-1}\mathbf{u}_s\nonumber\\
			&=& \widetilde{v}_0\frac{1}{T^2h}\sum_{s=1}^{T}\mathrm{tr}\left\{\left[\left(\bm{\Sigma}_{\mathbf{w}}^{-1}(\tau_s)\mathbf{w}_{s-1}\mathbf{w}_{s-1}^\top\bm{\Sigma}_{\mathbf{w}}^{-1}(\tau_s)\right)\otimes  \mathbf{u}_s\mathbf{u}_s^\top \right]\cdot \mathbf{C}^\top\mathbf{H}(\tau_s)\mathbf{C}\right\}+O_P(T^{-1})\nonumber\\
			&=& \widetilde{v}_0\frac{1}{T^2h}\sum_{t=1}^{T}\mathrm{tr}\left\{\left[\bm{\Sigma}_{\mathbf{w}}^{-1}(\tau_t)\otimes  \bm{\Omega}(\tau_t)\right]\cdot\mathbf{C}^\top\mathbf{H}(\tau_t)\mathbf{C}\right\}+O_P(T^{-1}+T^{-3/2}h^{-1})\nonumber\\
			&=&(Th)^{-1}A_Q + o_P\left(T^{-1}h^{-1/2}\right).\nonumber
		\end{eqnarray*}
		
		Consider $J_{T,9}$. Let $w_{s,v}=\frac{1}{T\sqrt{h}}\int_{-1}^{1}K\left(u\right)K\left(u+\frac{s-v}{Th}\right)\mathrm{d}u$. Since
		\begin{eqnarray*}
			&&\frac{1}{T}\sum_{t=1}^{T}\mathbf{H}_0(\tau_t)K\left(\frac{\tau_s-\tau_t}{h}\right)K\left(\frac{\tau_v-\tau_t}{h}\right)\nonumber\\
			&=&h\int_{-1}^{1}\mathbf{H}_0(\tau_s+uh)K\left(u\right)K\left(u+\frac{s-v}{Th}\right)\mathrm{d}u + O(1/(Th)),\nonumber
		\end{eqnarray*}
		we have
		$$
		T\sqrt{h}J_{T,9}= 2\sum_{s=2}^{T}\sum_{v=1}^{s-1}\mathbf{u}_s^\top\mathbf{W}_{s-1}^\top\mathbf{H}_0(\tau_s)\mathbf{W}_{v-1}\mathbf{u}_v w_{s,v}(1+o(1))= 2 \widetilde{U}+o_P(1),
		$$
		where the definition of $\widetilde{U}$ is obvious. By Lemma \ref{L10}, we have $$\widetilde{U} \to_D N\left(0,\sigma_{\widetilde{U}}^2\right),$$ where $ \sigma_{\widetilde{U}}^2 = \lim_{T \to \infty}\sum_{s=2}^{T}\mathrm{tr}\left\{E\left(\mathbf{H}_0^\top(\tau_s)\mathbf{W}_{s-1}\mathbf{u}_s\mathbf{u}_s^\top\mathbf{W}_{s-1}^\top\mathbf{H}_0(\tau_s)\right)E\left(\sum_{v=1}^{s-1}\mathbf{W}_{v-1}\mathbf{u}_v\mathbf{u}_v^\top\mathbf{W}_{v-1}^\top\right)w_{s,v}^2\right\}$. We then show that $\sigma_{\widetilde{U}}^2 = C_B \mathrm{tr}\left\{\int_{0}^{1}\bm{\Sigma}_Q(\tau)^2\mathrm{d}\tau\right\}$. Let $\mathbf{V}_1 (\tau)=\mathbf{H}_0(\tau)\mathbf{V}_2(\tau)\mathbf{H}_0(\tau)$ and $\mathbf{V}_2(\tau) = \bm{\Sigma}_{\mathbf{w}}(\tau) \otimes \bm{\Omega}(\tau)$. Write
		\begin{eqnarray*}
			&&\sum_{t=2}^{T}E\left(\mathbf{H}_0^\top(\tau_t)\mathbf{W}_{t-1}\bm{\eta}_t\mathbf{u}_t^\top\mathbf{W}_{t-1}^\top\mathbf{H}_0(\tau_t)\right)E\left( \sum_{s=1}^{t-1}\mathbf{W}_{s-1}\mathbf{u}_s\mathbf{u}_s^\top\mathbf{W}_{s-1}^\top\right)w_{s,t}^2\nonumber\\
			&=&\sum_{t=2}^{T}\sum_{s=1}^{t-1}\mathbf{V}_1(\tau_t)\mathbf{V}_2(\tau_s)w_{s,t}^2 =  \frac{1}{T^2h}\sum_{t=2}^{T}\sum_{s=1}^{t-1}\mathbf{V}_1(\tau_t)\mathbf{V}_2(\tau_s) \left[\int_{-1}^{1}K\left(u\right)K\left(u+\frac{t-s}{Th}\right)\mathrm{d}u\right]^2\nonumber\\
			&=&\frac{1}{T^2h}\sum_{s=1}^{T-1}\sum_{j=1}^{T-s}\mathbf{V}_1(\tau_s+j/T)\mathbf{V}_2(\tau_s)\left[\int_{-1}^{1}K\left(u\right)K\left(u+\frac{t-s}{Th}\right)\mathrm{d}u\right]^2\nonumber\\
			&=&\frac{1}{T^2h}\sum_{s=1}^{T-1}\sum_{j=1}^{T-s}\mathbf{V}_1(\tau_s+j/T)\mathbf{V}_2(\tau_s)\left[\int_{-1}^{1}K\left(u\right)K\left(u+\frac{j}{Th}\right)\mathrm{d}u\right]^2\nonumber\\
			&=&\frac{1}{T^2h}\sum_{s=1}^{T-1}\sum_{j=1}^{T-s}\mathbf{V}_1(\tau_s)\mathbf{V}_2(\tau_s)\left[\int_{-1}^{1}K\left(u\right)K\left(u+\frac{j}{Th}\right)\mathrm{d}u\right]^2\nonumber\\
			&&+\frac{1}{T^2h}\sum_{s=1}^{T-1}\sum_{j=1}^{T-s}O(j/T)\mathbf{V}_2(\tau_s)\left[\int_{-1}^{1}K\left(u\right)K\left(u+\frac{j}{Th}\right)\mathrm{d}u\right]^2 := J_{T,9} + J_{T,10},\nonumber
		\end{eqnarray*}
		where the definitions of $J_{T,9} $ and $J_{T,10}$ are obvious.
		
		It is easy to verify $\mathrm{tr}\left\{J_{T,9}\right\} \to C_B \mathrm{tr}\left\{\int_{0}^{1}\bm{\Sigma}_Q(\tau)^2\mathrm{d}\tau\right\}$. For $J_{T,10}$, we have
		\begin{eqnarray*}
			\left\|J_{T,10}\right\|&\leq&M \frac{1}{Th}\sum_{j=1}^{T}j/T\left[\int_{-1}^{1}K\left(u\right)K\left(u+\frac{j}{Th}\right)\mathrm{d}u\right]^2\nonumber\\
			&=&Mh  \int_{0}^{2}v \left[\int_{-1}^{1}K\left(u\right)K\left(u+v\right)\mathrm{d}u\right]^2\mathrm{d}v + o(1) =o(1).\nonumber
		\end{eqnarray*}
		
		Combining the above results, we have proved
		\begin{eqnarray*}
			T\sqrt{h}\left[\frac{1}{T}\sum_{t=1}^{T}\mathbf{R}_T^\top(\tau_t) \mathbf{H}_{0}(\tau_t)\mathbf{R}_T(\tau_t)-(Th)^{-1}A_Q\right]\to_D N\left(0,B_Q\right).\nonumber
		\end{eqnarray*}
		
		Note that
		\begin{eqnarray*}
			&&\frac{1}{T}\sum_{t=1}^{T}\left[\mathbf{C}\widehat{\mathbf{b}}(\tau_t)-\widehat{\mathbf{c}}\right]^\top\mathbf{H}(\tau_t)\left[\mathbf{C}\widehat{\mathbf{b}}(\tau_t)-\widehat{\mathbf{c}}\right]-\frac{1}{T}\sum_{t=1}^{T}\left[\mathbf{C}\widehat{\mathbf{b}}(\tau_t)-\mathbf{c}\right]^\top\mathbf{H}(\tau_t)
			\left[\mathbf{C}\widehat{\mathbf{b}}(\tau_t)-\mathbf{c}\right]\nonumber\\
			&=&\left( \widehat{\mathbf{c}}-\mathbf{c}\right)^\top\frac{1}{T}\sum_{t=1}^{T} \mathbf{H}(\tau_t)\left(\widehat{\mathbf{c}}-\mathbf{c}\right) - 2\left( \widehat{\mathbf{c}}-\mathbf{c}\right)^\top\frac{1}{T}\sum_{t=1}^{T} \mathbf{H}(\tau_t)\left(\mathbf{C}\widehat{\mathbf{b}}(\tau_t)-\mathbf{c}\right) := J_{T,11} - 2 J_{T,12}.\nonumber
		\end{eqnarray*}
		
		Since $\widehat{\mathbf{c}} = \mathbf{c} + O_P\left(T^{-1/2}\right)$ by part (1), we have $J_{T,11} = O_P(T^{-1}) = o_P(T^{-1}h^{-1/2})$. For $J_{T,12}$, by \eqref{EqA.2.1} and \eqref{Rrate2}, we have
		\begin{eqnarray*}
			J_{T,12} &=& \left(\widehat{\mathbf{c}}-\mathbf{c}\right)^\top\frac{1}{T}\sum_{t=\lfloor Th\rfloor+1}^{\lfloor T(1-h)\rfloor}\mathbf{H}(\tau_t)\mathbf{C}\bm{\Sigma}_{\mathbf{W}}^{-1}(\tau_t)\mathbf{R}_T(\tau_t) + o_P\left(T^{-1}h^{-1/2}\right)\nonumber\\
			&=&  O_P(T^{-1})+o_P\left(T^{-1}h^{-1/2}\right)=o_P\left(T^{-1}h^{-1/2}\right)\nonumber
		\end{eqnarray*}
		provided that
		\begin{eqnarray*}
			&& \frac{1}{T}\sum_{t=\lfloor Th\rfloor+1}^{\lfloor T(1-h)\rfloor}\mathbf{H}(\tau_t)\mathbf{C}\bm{\Sigma}_{\mathbf{W}}^{-1}(\tau_t)\mathbf{R}_T(\tau_t) \nonumber\\
			&=& \frac{1}{T}\sum_{s=1}^{T}\left\{\frac{1}{Th}\sum_{t=\lfloor Th\rfloor+1}^{\lfloor T(1-h)\rfloor}\mathbf{H}(\tau_t)\mathbf{C}\bm{\Sigma}_{\mathbf{W}}^{-1}(\tau_t)K\left(\frac{\tau_s-\tau_t}{h}\right)\right\} \mathbf{W}_{s-1}\mathbf{u}_s = O_P(T^{-1/2}).\nonumber
		\end{eqnarray*}
		
		We then conclude that $T\sqrt{h}\left[\widehat{Q}_{\mathbf{C},\mathbf{H}}-(Th)^{-1}A_Q\right]\to_D N\left(0,B_Q\right)$. We next complete this proof by showing $T\sqrt{h}\left[\widehat{Q}_{\mathbf{C},\mathbf{H}} - \widehat{Q}_{\mathbf{C},\widehat{\mathbf{H}}}\right] = o_P(1)$.
		
		Observe that
		\begin{eqnarray*}
			&&T\sqrt{h}\left[\widehat{Q}_{\mathbf{C},\mathbf{H}} - \widehat{Q}_{\mathbf{C},\widehat{\mathbf{H}}}\right]\nonumber\\ 
			&=& T\sqrt{h}\frac{1}{T}\sum_{t=\lfloor Th\rfloor+1}^{\lfloor T(1-h)\rfloor}\left(\left[\mathbf{C}\widehat{\mathbf{b}}(\tau_t)-\widehat{\mathbf{c}}\right]^\top\mathbf{H}(\tau_t)\left[\mathbf{C}\widehat{\mathbf{b}}(\tau_t)-\widehat{\mathbf{c}}\right]-\left[\mathbf{C}\widehat{\mathbf{b}}(\tau_t)-\widehat{\mathbf{c}}\right]^\top\widehat{\mathbf{H}}(\tau_t)
			\left[\mathbf{C}\widehat{\mathbf{b}}(\tau_t)-\widehat{\mathbf{c}}\right]\right)\nonumber\\
			&& + O_P(Th^{5.5} + \sqrt{h}\log T)\nonumber\\
			&=& T\sqrt{h} \times O_P(\rho_T^3) + O_P(Th^{5.5} + \sqrt{h}\log T) = o_P(1)\nonumber
		\end{eqnarray*}
		using $\sup_{\tau\in[h,1-h]}\left|\mathbf{H}(\tau) - \widehat{\mathbf{H}}(\tau) \right| = O_P(\rho_T)$.
			\end{proof}

		\begin{proof}[Proof of Corollary \ref{Coro3}]
			\item
			\noindent (1). Under the local alternative (9), we have $\mathbf{C}\mathbf{b}(\tau) = \mathbf{c} + d_T \mathbf{f}(\tau)$ and thus
			\begin{eqnarray*}
				&&\widehat{Q}_{\mathbf{C},\mathbf{H}} - \frac{1}{T}\sum_{t=1}^{T}\mathbf{R}_T^\top(\tau_t) \mathbf{H}_{0}(\tau_t)\mathbf{R}_T(\tau_t) \nonumber\\
				&=& d_T^2 \frac{1}{T}\sum_{t=1}^{T}\mathbf{f}(\tau_t)^\top \mathbf{H}(\tau_t)\mathbf{f}(\tau_t) + 2d_T \frac{1}{T}\sum_{t=1}^{T}\mathbf{f}(\tau_t)^\top \mathbf{H}(\tau_t)\left(\mathbf{C}\widehat{\mathbf{b}}(\tau_t)-\mathbf{C}\mathbf{b}(\tau_t)\right)\nonumber\\
				&& + \left[\frac{1}{T}\sum_{t=1}^{T}\left(\mathbf{C}\widehat{\mathbf{b}}(\tau_t)-\mathbf{C}\mathbf{b}(\tau_t)\right)^\top \mathbf{H}(\tau_t)\left(\mathbf{C}\widehat{\mathbf{b}}(\tau_t)-\mathbf{C}\mathbf{b}(\tau_t)\right) - \frac{1}{T}\sum_{t=1}^{T}\mathbf{R}_T^\top(\tau_t) \mathbf{H}_{0}(\tau_t)\mathbf{R}_T(\tau_t) \right]\nonumber\\
				&=& d_T^2 \frac{1}{T}\sum_{t=1}^{T}\mathbf{f}(\tau_t)^\top \mathbf{H}(\tau_t)\mathbf{f}(\tau_t) + I_{T,7} + I_{T,8}.\nonumber
			\end{eqnarray*}
			
			Since $\mathbf{C}\widehat{\mathbf{b}}(\tau)-\mathbf{C}\mathbf{b}(\tau) = O_P\left(d_T \rho_T+\sqrt{\frac{\log T}{Th}}\rho_T\right) + \mathbf{C}\bm{\Sigma}_{\mathbf{W}}^{-1}(\tau)\mathbf{R}_T(\tau)$ uniformly over $\tau\in[h,1-h]$ and
			$$
			\frac{1}{T}\sum_{t=1}^{T}\mathbf{f}(\tau_t)^\top \mathbf{H}(\tau_t) \mathbf{C}\bm{\Sigma}_{\mathbf{W}}^{-1}(\tau_t)\mathbf{R}_T(\tau_t)=O_P(T^{-1/2}),
			$$
			we have $I_{T,7} = O_P\left(d_T (d_T \rho_T + \sqrt{\frac{\log T}{Th}}\rho_T + T^{-1/2} )\right)=o_P(T^{-1}h^{-1/2}) $.
			
			For $I_{T,8}$, since $\sup_{\tau \in [0,1]}\left\|\mathbf{R}_T(\tau)\right\| = O_P\left(\sqrt{\frac{\log T}{Th}} \right)$, we have
			$$
			I_{T,8} = O_P\left(d_T^2\rho_T^2 + d_T\rho_T\sqrt{\frac{\log T}{Th}} \right) = o_P(T^{-1}h^{-1/2}).
			$$
			
			As $T\sqrt{h}\left(\frac{1}{T}\sum_{t=1}^{T}\mathbf{R}_T^\top(\tau_t) \mathbf{H}_{0}(\tau_t)\mathbf{R}_T(\tau_t)- (Th)^{-1}A_Q\right) \to N(0,B_Q)$,
			we have
			$$
			T\sqrt{h}\left(\widehat{Q}_{\mathbf{C},\mathbf{H}}- (Th)^{-1}A_Q\right) \to N(\delta_1,B_Q).
			$$
			
			In addition, similar to the proof of Theorem \ref{Thm4}, we have $T\sqrt{h}\left(\widehat{Q}_{\mathbf{C},\mathbf{H}}-\widehat{Q}_{\mathbf{C},\widehat{\mathbf{H}}}\right)=o_P(1)$. The proof is now completed.
			
			\medskip
			
			\noindent (2). Part (2) follows directly from part (1).
			
		\end{proof}

	\begin{proof}[Proof of Corollary \ref{Coro4}]
		\item
	Indeed,	we can regard $\Delta\mathbf{y}_t^*$ as from the data generating process \eqref{Eq2.10} with $\bm{\alpha}^*(\tau) = \mathbf{0}_{d\times r_0}$, $\bm{\Gamma}_j^*(\tau) = \mathbf{0}_d$, $\bm{\omega}^*(\tau_t) = \mathbf{I}_d$, and further we replace the cointegrated component $\bm{\beta}^\top \mathbf{y}_{t-1}$ by $\mathbf{z}_t^*$. By using similar arguments to those used in the proof of Theorem \ref{Thm4}, we can show that 
			$$
	T\sqrt{h}\left(\widetilde{Q}_{\mathbf{C}, \widehat{\mathbf{H}}}^b -\frac{1}{Th}s\widetilde{v}_0\right)\to_D N(0,4s C_B).
	$$
	
	\end{proof}

	\section{Proofs of the Preliminary Lemmas}\label{AppPPre}

	\begin{proof}[Proof of Lemma \ref{L4}]
		\item
		
		\noindent (1). Define $\delta_T = \lfloor T(\tau-h) \rfloor$. Write
		\begin{eqnarray*}
			&&T^{-2}h^{-1}\sum_{t=1}^{T}\mathbf{y}_{t-1}\mathbf{y}_{t-1}^\top K\left(\frac{\tau_t-\tau}{h}\right) \nonumber\\
			&=& (T^{-1/2}\mathbf{y}_{\delta_T})(T^{-1/2}\mathbf{y}_{\delta_T})^\top \cdot (Th)^{-1}\sum_{t=1}^{T}K\left(\frac{\tau_t-\tau}{h}\right)\nonumber\\
			&&+ (T^{-1/2}\mathbf{y}_{\delta_T})\cdot T^{-3/2}h^{-1}\sum_{t=1}^{T} (\mathbf{y}_{t-1}-\mathbf{y}_{\delta_T})^\top K\left(\frac{\tau_t-\tau}{h}\right)\nonumber\\
			&& +  T^{-3/2}h^{-1}\sum_{t=1}^{T} (\mathbf{y}_{t-1}-\mathbf{y}_{\delta_T}) K\left(\frac{\tau_t-\tau}{h}\right)\cdot (T^{-1/2}\mathbf{y}_{\delta_T})^\top\nonumber\\
			&& + T^{-2}h^{-1}\sum_{t=1}^{T} (\mathbf{y}_{t-1}-\mathbf{y}_{\delta_T}) (\mathbf{y}_{t-1}-\mathbf{y}_{\delta_T})^\top K\left(\frac{\tau_t-\tau}{h}\right)\nonumber\\
			&=&\mathbf{K}_{T,1}+ \mathbf{K}_{T,2} + \mathbf{K}_{T,3} + \mathbf{K}_{T,4}.\nonumber
		\end{eqnarray*}
		
		Consider $\mathbf{K}_{T,1}$. By part (1) and the continuous mapping theorem, we have
		\begin{eqnarray*}
			\mathbf{K}_{T,1} &=& \left(T^{-1/2}\mathbf{P}_{\bm{\beta}_\perp}\sum_{t=1}^{\delta_T}\bm{\Psi}_{\tau_t}(1)\bm{\omega}(\tau_t)\bm{\varepsilon}_t\right)\left(T^{-1/2}\mathbf{P}_{\bm{\beta}_\perp}\sum_{t=1}^{\delta_T}\bm{\Psi}_{\tau_t}(1)\bm{\omega}(\tau_t)\bm{\varepsilon}_t\right)^\top + O_P(T^{1/\delta-1/2})\nonumber\\
			&\to_D & \mathbf{P}_{\bm{\beta}_\perp}\int_{0}^{\tau}\bm{\Psi}_u(1) \bm{\omega}(u)\mathrm{d} \mathbf{W}_d(u) \left(\int_{0}^{\tau}\bm{\Psi}_u(1) \bm{\omega}(u)\mathrm{d} \mathbf{W}_d(u)\right)^\top\overline{\bm{\beta}}_{\perp}\bm{\beta}_{\perp}^\top.\nonumber
		\end{eqnarray*}

		By the proof of part (1), we have $\|\mathbf{y}_{t-1}-\mathbf{y}_{\delta_T}\|_\delta = O(\sqrt{Th})$ for $t \in \left[\lfloor T(\tau-h) \rfloor,\lfloor T(\tau+h)\rfloor \right]$. Hence, both $\mathbf{K}_{T,2}$ and $\mathbf{K}_{T,3}$ are $O_P(\sqrt{h})$. Similarly, $\mathbf{K}_{T,4} = O_P(h)$.
		
		The proof of part (1) is now completed.
		
		\medskip

		\noindent (2). Write
		\begin{eqnarray*}
			\mathbf{D}_T^{-1} \bm{\Xi}_T^\top(\tau)\left[ \sum_{t=1}^{T}\bm{\alpha}_{\perp}^\top(\tau)\mathbf{y}_{t-1}\mathbf{y}_{t-1}^\top\bm{\alpha}_{\perp}(\tau)\left(\frac{\tau_t-\tau}{h}\right)^l K\left(\frac{\tau_t-\tau}{h}\right)\right] \bm{\Xi}_T(\tau) \mathbf{D}_T^{-1}=\left[\begin{matrix}
				\mathbf{K}_{T,l}(1) & \mathbf{K}_{T,l}(2)\\
				\mathbf{K}_{T,l}^\top(2) & \mathbf{K}_{T,l}(3)\\
			\end{matrix} \right],\nonumber
		\end{eqnarray*}
		where
		\begin{eqnarray*}
			\mathbf{K}_{T,l}(1) &=& \frac{1}{T^2h}\bm{\xi}_T^\top(\tau)\sum_{t=1}^{T}\bm{\alpha}_{\perp}^\top(\tau)\mathbf{y}_{t-1}\mathbf{y}_{t-1}^\top\bm{\alpha}_{\perp}(\tau)\bm{\xi}_T(\tau)\left(\frac{\tau_t-\tau}{h}\right)^l K\left(\frac{\tau_t-\tau}{h}\right),\nonumber\\
			\mathbf{K}_{T,l}(2) &=& \frac{1}{T^2h^{3/2}}\bm{\xi}_T^\top(\tau)\sum_{t=1}^{T}\bm{\alpha}_{\perp}^\top(\tau)\mathbf{y}_{t-1}\mathbf{y}_{t-1}^\top\bm{\alpha}_{\perp}(\tau)\bm{\xi}_{T,\perp}(\tau)\left(\frac{\tau_t-\tau}{h}\right)^l K\left(\frac{\tau_t-\tau}{h}\right),\nonumber\\
			\mathbf{K}_{T,l}(3) &=& \frac{1}{T^2h^{2}}\bm{\xi}_{T,\perp}^\top(\tau)\sum_{t=1}^{T}\bm{\alpha}_{\perp}^\top(\tau)\mathbf{y}_{t-1}\mathbf{y}_{t-1}^\top\bm{\alpha}_{\perp}(\tau)\bm{\xi}_{T,\perp}(\tau)\left(\frac{\tau_t-\tau}{h}\right)^l K\left(\frac{\tau_t-\tau}{h}\right).\nonumber
		\end{eqnarray*}
		
		Let $\delta_T = \lfloor T(\tau-h) \rfloor$. Consider $\mathbf{K}_{T,l}(1)$. For $\lfloor T(\tau-h) \rfloor + 2 \leq t \leq \lfloor T(\tau+h) \rfloor+1$, we have
		$$
		\max_t\|1/\sqrt{T}\sum_{j=\delta_T+1}^{t-1}\Delta \mathbf{y}_j\|_\delta = O(\sqrt{h})
		$$
		and thus
		\begin{eqnarray*}
			\mathbf{K}_{T,l}(1) &=& \bm{\xi}_T^\top(\tau)\bm{\alpha}_{\perp}^\top(\tau)\mathbf{y}_{\delta_T}/\sqrt{T} (\bm{\xi}_T^\top(\tau)\bm{\alpha}_{\perp}^\top(\tau)\mathbf{y}_{\delta_T}/\sqrt{T})^\top \widetilde{c}_l + O_P(\sqrt{h})\nonumber\\
			& =& \widetilde{c}_l\mathbf{q}_T^\top(\tau)\mathbf{q}_T(\tau) + O_P(\sqrt{h}).\nonumber
		\end{eqnarray*}

		Consider $\mathbf{K}_{T,l}(2)$. Since $\mathbf{q}_T^\top(\tau) \bm{\xi}_{T,\perp}(\tau) = 0$, we have
		\begin{eqnarray*}
			&&\mathbf{K}_{T,l}(2)\nonumber\\
			&=& \bm{\xi}_T^\top(\tau)\mathbf{q}_T(\tau) \frac{1}{Th^{3/2}} \sum_{t=1}^{T}\frac{1}{\sqrt{T}}(\mathbf{y}_{t-1} - \mathbf{y}_{\delta_T})^\top\left(\frac{\tau_t-\tau}{h}\right)^l K\left(\frac{\tau_t-\tau}{h}\right)\bm{\alpha}_{\perp}(\tau)\bm{\xi}_{T,\perp}(\tau)\nonumber\\
			&& + \bm{\xi}_T^\top(\tau)\bm{\alpha}_{\perp}^\top(\tau) \frac{1}{Th^{3/2}} \sum_{t=1}^{T}\frac{1}{\sqrt{T}}(\mathbf{y}_{t-1} - \mathbf{y}_{\delta_T})\frac{1}{\sqrt{T}}(\mathbf{y}_{t-1} - \mathbf{y}_{\delta_T})^\top \left(\frac{\tau_t-\tau}{h}\right)^lK\left(\frac{\tau_t-\tau}{h}\right)\bm{\alpha}_{\perp}(\tau)\bm{\xi}_{T,\perp}(\tau)\nonumber\\
			&=&\mathbf{K}_{T,l}(2,1) + \mathbf{K}_{T,l}(2,2).\nonumber
		\end{eqnarray*}
		
		Again, for $\lfloor T(\tau-h) \rfloor + 2 \leq t \leq \lfloor T(\tau+h) \rfloor+1$, we have
		$$
		\max_t\|1/\sqrt{T}\sum_{j=\delta_T+1}^{t-1}\Delta \mathbf{y}_j\|_\delta = O(\sqrt{h})
		$$
		and thus $\mathbf{K}_{T,l}(2,2) = O_P(\sqrt{h})$. For $\mathbf{K}_{T,l}(2,1)$, by the proof of Lemma \ref{L1} and $|\tau_j-\tau|\leq h$, we have
		\begin{eqnarray*}
			&&\frac{1}{\sqrt{2Th}}(\mathbf{y}_{t-1} - \mathbf{y}_{\delta_T})^\top \nonumber\\
			&=& \frac{1}{\sqrt{2Th}}\sum_{j=\delta_T+1}^{t-1}\Delta\widetilde{\mathbf{y}}_j(\tau_j) + O_P(\sqrt{h/T})\nonumber \\
			&=& \frac{1}{\sqrt{2Th}}\sum_{j=\delta_T+1}^{t-1} \mathbf{P}_{\bm{\beta}_\perp}\bm{\Psi}_{\tau_j}(1)\bm{\omega}(\tau_j)\bm{\varepsilon}_j +  O_P(\sqrt{h/T} + 1/\sqrt{Th})\nonumber\\
			&=& \frac{1}{\sqrt{2Th}}\sum_{j=\delta_T+1}^{t-1} \mathbf{P}_{\bm{\beta}_\perp}\bm{\Psi}_{\tau}(1)\bm{\omega}(\tau)\bm{\varepsilon}_j +  O_P(\sqrt{h/T} + 1/\sqrt{Th}+h),\nonumber
		\end{eqnarray*}
		which implies
		\begin{eqnarray*}
			\mathbf{K}_{T,l}(2,1) &=& \bm{\xi}_T^\top(\tau)\mathbf{q}_T(\tau) \frac{\sqrt{2}}{Th} \sum_{t=1}^{T}\left(\frac{1}{\sqrt{2Th}}\sum_{j=\delta_T+1}^{t-1} \mathbf{P}_{\bm{\beta}_\perp}\bm{\Psi}_{\tau}(1)\bm{\omega}(\tau)\bm{\varepsilon}_j\right)^\top\nonumber\\ 
			&& \times \left(\frac{\tau_t-\tau}{h}\right)^lK\left(\frac{\tau_t-\tau}{h}\right)\bm{\alpha}_{\perp}(\tau)\bm{\xi}_{T,\perp}(\tau) + O_P(\sqrt{h}).\nonumber
		\end{eqnarray*}
		Similarly, we have
		\begin{eqnarray*}
\mathbf{K}_{T,l}(3) &= &\bm{\xi}_{T,\perp}^\top(\tau)\bm{\alpha}_{\perp}^\top(\tau) \frac{2}{Th} \sum_{t=1}^{T}\left(\frac{1}{\sqrt{2Th}}\sum_{j=\delta_T+1}^{t-1} \mathbf{P}_{\bm{\beta}_\perp}\bm{\Psi}_{\tau}(1)\bm{\omega}(\tau)\bm{\varepsilon}_j\right)\nonumber \\
			&&\times\left(\frac{1}{\sqrt{2Th}}\sum_{j=\delta_T+1}^{t-1} \mathbf{P}_{\bm{\beta}_\perp}\bm{\Psi}_{\tau}(1)\bm{\omega}(\tau)\bm{\varepsilon}_j\right)^\top \left(\frac{\tau_t-\tau}{h}\right)^lK\left(\frac{\tau_t-\tau}{h}\right)\bm{\alpha}_{\perp}(\tau)\bm{\xi}_{T,\perp}(\tau) + O_P(h).\nonumber
		\end{eqnarray*}
		
		Furthermore, by the usual functional central limit theory (e.g., Theorem 4.1 in \citealp{hall2014martingale}), we have
		\begin{eqnarray*}
			&&\left(\frac{1}{\sqrt{T}}\sum_{t=1}^{\delta_T}\mathbf{P}_{\bm{\beta}_\perp}\bm{\Psi}_{\tau_t}(1)\bm{\omega}(\tau_t)\bm{\varepsilon}_t, \frac{1}{\sqrt{2Th}}\sum_{t=\delta_T+1}^{\delta_T(u)}\bm{\alpha}_{\perp}^\top(\tau)\mathbf{P}_{\bm{\beta}_\perp}\bm{\Psi}_{\tau}(1)\bm{\omega}(\tau)\bm{\varepsilon}_t \right)\nonumber\\ 
			&\Rightarrow& \left(\mathbf{W}_d(\tau,\bm{\Sigma}_{\mathbf{y}}(\tau)), \mathbf{W}_d^{*}(u,\bm{\Sigma}_{\bm{\alpha}}(\tau)) \right),\nonumber
		\end{eqnarray*}
		where $\delta_T(u) = \delta_T + \lfloor2uTh \rfloor+ 1$ and $\mathbf{W}_d^{*}(\cdot,\cdot)$ is independent of $\mathbf{W}_d(\cdot,\cdot)$. Note that the Lindeberg condition can be verified as $\|\bm{\Psi}_{\tau_t}(1)\bm{\omega}(\tau_t)\bm{\varepsilon}_t\|_\delta <\infty$ for some $\delta >2$, and the convergence of conditional variance can be verified as $E(\bm{\varepsilon}_t\bm{\varepsilon}_t^\top \mid \mathcal{F}_{t-1}) = \mathbf{I}_d$ a.s. Then part (2) follows directly from the continuous mapping theorem.
		
		\medskip
		
		\noindent (3). Write
		\begin{eqnarray*}
			&&\frac{1}{\sqrt{Th}}\mathbf{D}_T^{-1} \bm{\Xi}_T^\top(\tau) \sum_{t=1}^{T}\bm{\alpha}_{\perp}^\top(\tau)\mathbf{y}_{t-1}\mathbf{y}_{t-1}^\top\bm{\beta}\left(\frac{\tau_t-\tau}{h}\right)^l K\left(\frac{\tau_t-\tau}{h}\right)\nonumber\\
			&=&\left[\begin{matrix}
				\frac{1}{T^{3/2}h}\sum_{t=1}^{T}\bm{\xi}_T^\top(\tau)\bm{\alpha}_{\perp}^\top(\tau)\mathbf{y}_{t-1}\mathbf{z}_{t-1}^\top\bm{\beta} \left(\frac{\tau_t-\tau}{h}\right)^lK\left(\frac{\tau_t-\tau}{h}\right)\nonumber\\
				\frac{1}{T^{3/2}h^{3/2}}\sum_{t=1}^{T}\bm{\xi}_{T,\perp}^\top(\tau)\bm{\alpha}_{\perp}^\top(\tau)\mathbf{y}_{t-1}\mathbf{z}_{t-1}^\top\bm{\beta} \left(\frac{\tau_t-\tau}{h}\right)^lK\left(\frac{\tau_t-\tau}{h}\right)\nonumber\\
			\end{matrix} \right] =\left[\begin{matrix}
				\mathbf{K}_{T,l}(4)\\
				\mathbf{K}_{T,l}(5)\\
			\end{matrix} \right].\nonumber
		\end{eqnarray*}
		
		Consider $\mathbf{K}_{T,l}(4)$, 
		\begin{eqnarray*}
			&&\frac{1}{T^{3/2}h}\sum_{t=1}^{T}\bm{\xi}_T^\top(\tau)\bm{\alpha}_{\perp}^\top(\tau)\mathbf{y}_{t-1}\mathbf{z}_{t-1}^\top \left(\frac{\tau_t-\tau}{h}\right)^lK\left(\frac{\tau_t-\tau}{h}\right)\nonumber\\
			&=& \frac{1}{Th}\bm{\xi}_T^\top(\tau)\sum_{t=1}^{T}\left(\frac{1}{\sqrt{T}}\bm{\alpha}_{\perp}^\top(\tau)\mathbf{y}_{\delta_T}+\frac{1}{\sqrt{T}}\bm{\alpha}_{\perp}^\top(\tau)(\mathbf{y}_{t-1} - \mathbf{y}_{\delta_T})\right) \mathbf{z}_{t-1}^\top \left(\frac{\tau_t-\tau}{h}\right)^lK\left(\frac{\tau_t-\tau}{h}\right)\nonumber\\
			&=&\sqrt{\mathbf{q}_T^\top(\tau)\mathbf{q}_T(\tau)}\frac{1}{Th}\sum_{t=1}^{T}\mathbf{z}_{t-1}^\top \left(\frac{\tau_t-\tau}{h}\right)^lK\left(\frac{\tau_t-\tau}{h}\right)\nonumber\\
			&& + \frac{\sqrt{h}}{Th}\bm{\xi}_T^\top(\tau)\sum_{t=1}^{T}\frac{1}{\sqrt{Th}}\bm{\alpha}_{\perp}^\top(\tau)(\mathbf{y}_{t-1} - \mathbf{y}_{\delta_T})\mathbf{z}_{t-1}^\top\left(\frac{\tau_t-\tau}{h}\right)^l K\left(\frac{\tau_t-\tau}{h}\right)\nonumber\\
			&=& \mathbf{K}_{T,l}(4,1) + \mathbf{K}_{T,l}(4,2).\nonumber
		\end{eqnarray*} 
		
		Similar to the proof of Lemma \ref{L1}, by $\max_{t}\|\mathbf{z}_t - \widetilde{\mathbf{z}}_{t}(\tau_{t})\|_\delta = O(1/T)$ and using BN decomposition, we have
		\begin{eqnarray*}
			&&\frac{1}{Th}\sum_{t=1}^{T}\mathbf{z}_{t-1}\left(\frac{\tau_t-\tau}{h}\right)^l K\left(\frac{\tau_t-\tau}{h}\right)\nonumber\\
			&=&\frac{1}{Th}\sum_{t=1}^{T}\widetilde{\mathbf{z}}_{t-1}(\tau_{t-1})\left(\frac{\tau_t-\tau}{h}\right)^l K\left(\frac{\tau_t-\tau}{h}\right) + O_P(1/T) = O_P(1/\sqrt{Th})\nonumber
		\end{eqnarray*}
		and thus $\mathbf{K}_{T,l}(4,1) = O_P(1/\sqrt{Th})$. For $\mathbf{K}_{T,l}(4,2)$, by Lemma \ref{L1}(1), we have
		\begin{eqnarray*}
			&&\frac{1}{Th}\sum_{t=1}^{T}\frac{1}{\sqrt{Th}}\sum_{j=\delta_T+1}^{t-1}\Delta \mathbf{y}_{j}\mathbf{z}_{t-1}^\top \left(\frac{\tau_t-\tau}{h}\right)^l K\left(\frac{\tau_t-\tau}{h}\right) \nonumber\\
			&=& \frac{1}{Th}\sum_{t=1}^{T}\frac{1}{\sqrt{Th}}\sum_{j=\delta_T+1}^{t-1}\Delta \widetilde{\mathbf{y}}_{j}(\tau_j)\widetilde{\mathbf{z}}_{t-1}^\top(\tau_{t-1})\left(\frac{\tau_t-\tau}{h}\right)^l K\left(\frac{\tau_t-\tau}{h}\right) + O_P(\sqrt{h/T}).\nonumber
		\end{eqnarray*}
		Note that $\frac{1}{\sqrt{Th}}\sum_{j=\delta_T+1}^{t-1}\Delta \widetilde{\mathbf{y}}_{j}(\tau_{j}) = \mathbf{P}_{\bm{\beta}_\perp}\frac{1}{\sqrt{Th}}  \sum_{j=\delta_T+1}^{t-1}\widetilde{\mathbf{z}}_{j}(\tau_j) + \frac{1}{\sqrt{Th}}\mathbf{P}_{\bm{\beta}}\sum_{j=\delta_T+1}^{t-1}\Delta \widetilde{\mathbf{z}}_{j}(\tau_{j})$. Similarly, by Lemma \ref{L1}, we have
		\begin{eqnarray*}
			\|\frac{1}{\sqrt{Th}}\mathbf{P}_{\bm{\beta}}\sum_{j=\delta_T+1}^{t-1}\Delta \widetilde{\mathbf{z}}_{j}(\tau_{j})\|_{\delta/2}
			&\leq&  \frac{1}{\sqrt{Th}}\|\mathbf{P}_{\bm{\beta}}\widetilde{\mathbf{z}}_{\delta_T}(\tau_{\delta_T+1})\|_{\delta/2} + \frac{1}{\sqrt{Th}}\|\mathbf{P}_{\bm{\beta}}\widetilde{\mathbf{z}}_{t-1}(\tau_{t-1})\|_{\delta/2} \nonumber\\
			&&+ \frac{1}{\sqrt{Th}}\|\mathbf{P}_{\bm{\beta}}\sum_{j=\delta_T+1}^{t-1}(\widetilde{\mathbf{z}}_{j}(\tau_{j})-\widetilde{\mathbf{z}}_{j}(\tau_{j+1}))\|_{\delta/2}\nonumber\\
			& =& O(1/\sqrt{Th}).\nonumber
		\end{eqnarray*}
		In addition, since $\widetilde{\mathbf{z}}_{t}(\tau) = \bm{\Psi}_{\tau}(1)\bm{\omega}(\tau)\bm{\varepsilon}_t+\bm{\Psi}_{\tau}^*(L)\bm{\omega}(\tau)\bm{\varepsilon}_{t-1}-\bm{\Psi}_{\tau}^*(L)\bm{\omega}(\tau)\bm{\varepsilon}_{t}$,
		we have
		\begin{eqnarray*}
			&&\frac{1}{Th}\sum_{t=1}^{T}\frac{1}{\sqrt{Th}}\sum_{j=\delta_T+1}^{t-1}\Delta \widetilde{\mathbf{y}}_{j}(\tau_j)\widetilde{\mathbf{z}}_{t-1}^\top(\tau_{t-1})\left(\frac{\tau_t-\tau}{h}\right)^l K\left(\frac{\tau_t-\tau}{h}\right)\nonumber\\
			&=&\frac{1}{Th}\sum_{t=1}^{T}\mathbf{P}_{\bm{\beta}_\perp}\frac{1}{\sqrt{Th}}  \sum_{j=\delta_T+1}^{t-1}\widetilde{\mathbf{z}}_{j}(\tau_j)\widetilde{\mathbf{z}}_{t-1}^\top(\tau_{t-1})\left(\frac{\tau_t-\tau}{h}\right)^l K\left(\frac{\tau_t-\tau}{h}\right) + O_P(1/\sqrt{Th})\nonumber\\
			&=& \mathbf{P}_{\bm{\beta}_\perp}\frac{1}{Th}\sum_{t=1}^{T}\frac{1}{\sqrt{Th}}\sum_{j=\delta_T+1}^{t-1}\bm{\Psi}_{\tau_j}(1)\bm{\omega}(\tau_j)\bm{\varepsilon}_j (\bm{\Psi}_{\tau_{t-1}}(1)\bm{\omega}(\tau_{t-1})\bm{\varepsilon}_{t-1})^\top \left(\frac{\tau_t-\tau}{h}\right)^lK\left(\frac{\tau_t-\tau}{h}\right) \nonumber\\
			&& + \mathbf{P}_{\bm{\beta}_\perp}\frac{1}{Th}\sum_{t=1}^{T}\frac{1}{\sqrt{Th}}\sum_{j=\delta_T+1}^{t-1}\bm{\Psi}_{\tau_j}(1)\bm{\omega}(\tau_j)\bm{\varepsilon}_j (\bm{\Psi}_{\tau_{t-1}}^*(L)\bm{\omega}(\tau_{t-1})\bm{\varepsilon}_{t-2}-\bm{\Psi}_{\tau_{t-1}}^*(L)\bm{\omega}(\tau_{t-1})\bm{\varepsilon}_{t-1})^\top\\ &&\times\left(\frac{\tau_t-\tau}{h}\right)^lK\left(\frac{\tau_t-\tau}{h}\right) \nonumber\\
			&&+ \mathbf{P}_{\bm{\beta}_\perp}\frac{1}{Th}\sum_{t=1}^{T}\frac{1}{\sqrt{Th}}\bm{\Psi}_{\tau_{\delta_T+1}}^*(L)\bm{\omega}(\tau_{\delta_T+1})\bm{\varepsilon}_{\delta_T} \widetilde{\mathbf{z}}_{t-1}^\top(\tau_{t-1})\left(\frac{\tau_t-\tau}{h}\right)^l K\left(\frac{\tau_t-\tau}{h}\right)\nonumber\\
			&& + \mathbf{P}_{\bm{\beta}_\perp}\frac{1}{Th}\sum_{t=1}^{T}\frac{1}{\sqrt{Th}}\bm{\Psi}_{\tau_{t-1}}^*(L)\bm{\omega}(\tau_{t-1})\bm{\varepsilon}_{t-1} \widetilde{\mathbf{z}}_{t-1}^\top(\tau_{t-1})\left(\frac{\tau_t-\tau}{h}\right)^l K\left(\frac{\tau_t-\tau}{h}\right)\nonumber\\
			&& + \mathbf{P}_{\bm{\beta}_\perp}\frac{1}{Th}\sum_{t=1}^{T}\frac{1}{\sqrt{Th}}\sum_{j=\delta_T+1}^{t-2}\left(\bm{\Psi}_{\tau_{j+1}}^*(L)\bm{\omega}(\tau_{j+1})-\bm{\Psi}_{\tau_{j}}^*(L)\bm{\omega}(\tau_{j})\right)\bm{\varepsilon}_{j} \widetilde{\mathbf{z}}_{t-1}^\top(\tau_{t-1})\left(\frac{\tau_t-\tau}{h}\right)^l K\left(\frac{\tau_t-\tau}{h}\right)\nonumber\\
			&& =\mathbf{K}_{T,l}(4,21)+ \mathbf{K}_{T,l}(4,22) + O_P(1/\sqrt{Th}) .\nonumber
		\end{eqnarray*}
		Since $\sum_{j=\delta_T+1}^{t-2}\bm{\varepsilon}_j\bm{\varepsilon}_{t-1}^\top$ is a sequence of martingale differences, by using Burkholder's inequality, we have
		\begin{eqnarray*}
			&&\left\|\sum_{t=1}^{T}\sum_{j=\delta_T+1}^{t-2}\bm{\varepsilon}_j\bm{\varepsilon}_{t-1}^\top \left(\frac{\tau_t-\tau}{h}\right)^lK\left(\frac{\tau_t-\tau}{h}\right)\right\|_{\delta/2}^2\nonumber\\
			&\leq& O(1) \sum_{t=1}^{T}\left\|\sum_{j=\delta_T+1}^{t-2}\bm{\varepsilon}_j\bm{\varepsilon}_{t-1}^\top \left(\frac{\tau_t-\tau}{h}\right)^lK\left(\frac{\tau_t-\tau}{h}\right)\right\|_{\delta/2}^2 \nonumber\\
			&\leq&O(1) \sum_{t=1}^{T}\left\|\sum_{j=\delta_T+1}^{t-2}\bm{\varepsilon}_j\right\|_{\delta}^2\left\|\bm{\varepsilon}_{t-1}\right\|_{\delta}^2 \left(\frac{\tau_t-\tau}{h}\right)^{2l}K\left(\frac{\tau_t-\tau}{h}\right)^2 =O(T^2h^2).\nonumber
		\end{eqnarray*}
		Therefore, $\mathbf{K}_{T,l}(4,21) = O_P(1/\sqrt{Th})$. For $\mathbf{K}_{T,l}(4,22)$, let $\mathbf{p}_t = \sum_{j=\delta_T+1}^{t}\bm{\Psi}_{\tau_j}(1)\bm{\omega}(\tau_j)\bm{\varepsilon}_j$, $\mathbf{w}_t = \bm{\Psi}_{\tau}^*(L)\bm{\omega}(\tau)\bm{\varepsilon}_{t}$, we have
		\begin{eqnarray*}
			&&\sum_{t=1}^{T} \mathbf{p}_{t-1}(\mathbf{w}_{t-2} - \mathbf{w}_{t-1})^\top \left(\frac{\tau_t-\tau}{h}\right)^l K\left(\frac{\tau_t-\tau}{h}\right)\nonumber\\
			&=& \sum_{t=1}^{T} \mathbf{p}_{t-1}\mathbf{w}_{t-2}^\top \left(\frac{\tau_t-\tau}{h}\right)^lK\left(\frac{\tau_t-\tau}{h}\right) - \sum_{t=1}^{T} \mathbf{p}_{t-1}\mathbf{w}_{t-1}^\top \left(\frac{\tau_t-\tau}{h}\right)^lK\left(\frac{\tau_t-\tau}{h}\right)\nonumber\\
			&=& \sum_{t=1}^{T} \mathbf{p}_{t-2}\mathbf{w}_{t-2}^\top \left(\frac{\tau_t-\tau}{h}\right)^lK\left(\frac{\tau_t-\tau}{h}\right) - \sum_{t=1}^{T} \mathbf{p}_{t-1}\mathbf{w}_{t-1}^\top \left(\frac{\tau_t-\tau}{h}\right)^lK\left(\frac{\tau_t-\tau}{h}\right) + O_P(\sqrt{Th})\nonumber\\
			&=& \sum_{t=1}^{T} \mathbf{p}_{t-2}\mathbf{w}_{t-2}^\top \left(\frac{\tau_{t-1}-\tau}{h}\right)^lK\left(\frac{\tau_{t-1}-\tau}{h}\right) - \sum_{t=1}^{T} \mathbf{p}_{t-1}\mathbf{w}_{t-1}^\top \left(\frac{\tau_t-\tau}{h}\right)^lK\left(\frac{\tau_t-\tau}{h}\right) + O_P(\sqrt{Th})\nonumber\\
			&=& \mathbf{p}_{-1}\mathbf{w}_{-1}^\top \left(\frac{\tau_0-\tau}{h}\right)^lK\left(\frac{\tau_{0}-\tau}{h}\right)-\mathbf{p}_{T-1}\mathbf{w}_{T-1}^\top \left(\frac{\tau_T-\tau}{h}\right)^lK\left(\frac{\tau_{T}-\tau}{h}\right)+ O_P(\sqrt{Th})\nonumber\\
			&=& O_P(\sqrt{Th}).\nonumber
		\end{eqnarray*}
		Hence, $\mathbf{K}_{T,l}(4,22) = O_P(1/(Th))$. Therefore, $\mathbf{K}_{T,l}(4,2) = O_P(1/\sqrt{T})$ and thus $\mathbf{K}_{T,l}(4) = O_P(1/\sqrt{Th})$.

		Consider $\mathbf{K}_{T,l}(5)$. As $\bm{\xi}_{T,\perp}^\top(\tau)\bm{\alpha}_{\perp}^\top(\tau)\mathbf{y}_{\delta_T} = 0$, we have
		\begin{eqnarray*}
			&&\frac{1}{T^{3/2}h^{3/2}}\sum_{t=1}^{T}\bm{\xi}_{T,\perp}^\top(\tau)\bm{\alpha}_{\perp}^\top(\tau)\mathbf{y}_{t-1}\mathbf{z}_{t-1}^\top \left(\frac{\tau_t-\tau}{h}\right)^l K\left(\frac{\tau_t-\tau}{h}\right) \nonumber\\
			&=&\frac{1}{Th^{3/2}}\bm{\xi}_{T,\perp}^\top(\tau)\frac{1}{\sqrt{T}}\bm{\alpha}_{\perp}^\top(\tau)\sum_{t=1}^{T}\sum_{j=\delta_T+1}^{t-1}\Delta \mathbf{y}_{j}\mathbf{z}_{t-1}^\top \left(\frac{\tau_t-\tau}{h}\right)^lK\left(\frac{\tau_t-\tau}{h}\right)\nonumber
		\end{eqnarray*}
		
		Similar to the proof of $\mathbf{K}_{T,l}(4,2)$, we have 
		$$
		\frac{1}{Th^{3/2}}\bm{\xi}_{T,\perp}^\top(\tau)\frac{1}{\sqrt{T}}\bm{\alpha}_{\perp}^\top(\tau)\sum_{t=1}^{T}\sum_{j=\delta_T+1}^{t-1}\Delta \mathbf{y}_{j}\mathbf{z}_{t-1}^\top \left(\frac{\tau_t-\tau}{h}\right)^lK\left(\frac{\tau_t-\tau}{h}\right) = O_P(1/\sqrt{Th}).
		$$

The proof of part (3) is now completed.
		
\medskip
		
\noindent (4). Since $\bm{\beta}^\top \mathbf{y}_t = \bm{\beta}^\top \mathbf{z}_t$, by Lemma \ref{L1} and using BN decomposition developed in the proof of Lemma \ref{L1}, we have
		\begin{eqnarray*}
			&&\frac{1}{Th}\sum_{t=1}^{T}\bm{\beta}^\top \mathbf{y}_{t-1} \mathbf{y}_{t-1}^\top \bm{\beta}\left(\frac{\tau_t-\tau}{h}\right)^lK\left(\frac{\tau_t-\tau}{h}\right) \nonumber\\
			&=& \frac{1}{Th}\sum_{t=1}^{T}\bm{\beta}^\top \mathbf{z}_{t-1} \mathbf{z}_{t-1}^\top \bm{\beta}\left(\frac{\tau_t-\tau}{h}\right)^lK\left(\frac{\tau_t-\tau}{h}\right)\nonumber\\
			&=& \frac{1}{Th}\sum_{t=1}^{T}\bm{\beta}^\top \widetilde{\mathbf{z}}_{t-1}(\tau_{t-1}) \widetilde{\mathbf{z}}_{t-1}(\tau_{t-1})^\top \bm{\beta}\left(\frac{\tau_t-\tau}{h}\right)^lK\left(\frac{\tau_t-\tau}{h}\right)+O_P(1/T)\nonumber\\
			&=& \frac{1}{Th}\sum_{t=1}^{T}\bm{\beta}^\top E\left(\widetilde{\mathbf{z}}_{t-1}(\tau_{t-1}) \widetilde{\mathbf{z}}_{t-1}(\tau_{t-1})^\top\right) \bm{\beta}\left(\frac{\tau_t-\tau}{h}\right)^lK\left(\frac{\tau_t-\tau}{h}\right) + O_P(1/\sqrt{Th}).\nonumber
		\end{eqnarray*}
		
		Note that for $\tau \in [h,1-h]$, by the definition of Riemann integral, we have
		\begin{eqnarray*}
			&&\frac{1}{Th}\sum_{t=1}^{T}\bm{\Sigma}_{\mathbf{z},0}(\tau_t)\left(\frac{\tau_t-\tau}{h}\right)^lK\left(\frac{\tau_t-\tau}{h}\right)\nonumber\\
			&=& \frac{1}{h}\int_{0}^{1}\bm{\Sigma}_{\mathbf{z},0}(\tau_t)\left(\frac{\tau_t-\tau}{h}\right)^lK\left(\frac{\tau_t-\tau}{h}\right) +O(1/(Th)) \nonumber\\
			&=&\int_{-\tau/h}^{(1-\tau)/h}\bm{\Sigma}_{\mathbf{z},0}(\tau+uh)u^l K(u)\mathrm{d}u+O(1/(Th)) \nonumber\\
			&=&\int_{-\tau/h}^{(1-\tau)/h}\left(\bm{\Sigma}_{\mathbf{z},0}(\tau)+\bm{\Sigma}_{\mathbf{z},0}^{(1)}(\tau)uh+\frac{1}{2}\bm{\Sigma}_{\mathbf{z},0}^{(2)}(\tau)u^2h^2\right)u^l K(u)\mathrm{d}u+O(1/(Th)) + o(h^2) \nonumber\\
			&=&\widetilde{c}_l\bm{\Sigma}_{\mathbf{z},0}(\tau) + h \widetilde{c}_{l+1}\bm{\Sigma}_{\mathbf{z},0}^{(1)}(\tau) + O(h^2) + O(1/(Th)).\nonumber
		\end{eqnarray*}
		
		Note that $\widetilde{c}_{l+1} = 0$ if $l$ is even. The proof of part (4) is now complete.
		
		\medskip
		
		\noindent (5). Part (5) follows directly from parts (2)--(4).
		
		\medskip
		
		\noindent (6). Write
		\begin{eqnarray*}
			&&\frac{1}{\sqrt{Th}}\widetilde{\mathbf{D}}_T^{+} \widetilde{\bm{\Xi}}_T^\top(\tau) \sum_{t=1}^{T}\mathbf{Q}^\top(\tau)\mathbf{y}_{t-1} \Delta \mathbf{x}_{t-1}^\top \left(\frac{\tau_t-\tau}{h}\right)^l K\left(\frac{\tau_t-\tau}{h}\right)\nonumber\\
			&=&	\left[\begin{matrix}
				\frac{1}{\sqrt{Th}} \mathbf{D}_T^{-1} \bm{\Xi}_T^\top(\tau) \sum_{t=1}^{T}\bm{\alpha_{\perp}}^\top(\tau)\mathbf{y}_{t-1} \Delta \mathbf{x}_{t-1}^\top \left(\frac{\tau_t-\tau}{h}\right)^lK\left(\frac{\tau_t-\tau}{h}\right)\\
				\frac{1}{Th} \sum_{t=1}^{T}\bm{\beta}^\top\mathbf{y}_{t-1} \Delta \mathbf{x}_{t-1}^\top \left(\frac{\tau_t-\tau}{h}\right)^lK\left(\frac{\tau_t-\tau}{h}\right)\\
			\end{matrix}  \right].\nonumber
		\end{eqnarray*}
		
		Similar to the proof of part (2), we have 
		$$
		\frac{1}{\sqrt{Th}} \mathbf{D}_T^{-1} \bm{\Xi}_T^\top(\tau) \sum_{t=1}^{T}\bm{\alpha_{\perp}}^\top(\tau)\mathbf{y}_{t-1} [\Delta\mathbf{y}_{t-1}^\top,\ldots,\Delta\mathbf{y}_{t-p_0+1}^\top]\left(\frac{\tau_t-\tau}{h}\right)^l K\left(\frac{\tau_t-\tau}{h}\right) = O_P(1/\sqrt{Th}).
		$$
		
		In addition, similar to the proof of part (3), we have
		\begin{eqnarray*}
			&&\frac{1}{Th} \sum_{t=1}^{T}\bm{\beta}^\top\mathbf{y}_{t-1} \Delta \mathbf{x}_{t-1}^\top \left(\frac{\tau_t-\tau}{h}\right)^l K\left(\frac{\tau_t-\tau}{h}\right)\nonumber\\
			&= & \bm{\beta}^\top\frac{1}{Th} \sum_{t=1}^{T}\mathbf{z}_{t-1} \Delta \mathbf{x}_{t-1}^\top \left(\frac{\tau_t-\tau}{h}\right)^lK\left(\frac{\tau_t-\tau}{h}\right)\nonumber\\
			&=& \bm{\beta}^\top\frac{1}{Th} \sum_{t=1}^{T}\widetilde{\mathbf{z}}_{t-1}(\tau_{t-1}) \Delta \widetilde{\mathbf{x}}_{t-1}^\top(\tau_{t-1})\left(\frac{\tau_t-\tau}{h}\right)^l K\left(\frac{\tau_t-\tau}{h}\right) + O_P(1/T)\nonumber\\
			&=&  \bm{\beta}^\top\frac{1}{Th} \sum_{t=1}^{T}E\left(\widetilde{\mathbf{z}}_{t-1}(\tau) \Delta \widetilde{\mathbf{x}}_{t-1}^\top(\tau) \right)\left(\frac{\tau_t-\tau}{h}\right)^l K\left(\frac{\tau_t-\tau}{h}\right) + O_P(1/\sqrt{Th}).\nonumber
		\end{eqnarray*}
		Since $\Delta\mathbf{y}_t(\tau) = \mathbf{z}_t(\tau) - \bm{\overline{\beta}}\bm{\beta}^\top \mathbf{z}_{t-1}(\tau)$, the result follows.
		
		\medskip
		
		\noindent (7). Write
		$$
		\mathbf{D}_T^{*,-1} \bm{\Xi}_T^{*,\top}(\tau) \mathbf{Q}^*(\tau)\mathbf{S}_{T,l}(\tau) \mathbf{Q}^{*,\top}(\tau)\bm{\Xi}_T^{*}(\tau) \widetilde{\mathbf{D}}_T^{*,+}
		= \left[\begin{matrix}
			\bm{\Delta}_{T,l}(1) & \bm{\Delta}_{T,l}(2) \\
			\bm{\Delta}_{T,l}^\top(2) & \bm{\Delta}_{T,l}(3) \\
		\end{matrix} \right],
		$$
		where
		\begin{eqnarray*}
			\bm{\Delta}_{T,l}(1) &=& \widetilde{\mathbf{D}}_T^{+} \widetilde{\bm{\Xi}}_T^{\top}(\tau)\left[ \sum_{t=1}^{T}\mathbf{Q}(\tau)\mathbf{y}_{t-1}\mathbf{y}_{t-1}^\top\mathbf{Q}^\top(\tau)\left(\frac{\tau_t-\tau}{h}\right)^l K\left(\frac{\tau_t-\tau}{h}\right)\right] \widetilde{\bm{\Xi}}_T(\tau) \widetilde{\mathbf{D}}_T^{+} ,\nonumber\\
			\bm{\Delta}_{T,l}(2) &=& \frac{1}{\sqrt{Th}}\widetilde{\mathbf{D}}_T^{+} \widetilde{\bm{\Xi}}_T^\top(\tau) \sum_{t=1}^{T}\mathbf{Q}(\tau)\mathbf{y}_{t-1} \Delta \mathbf{x}_{t-1}^\top\left(\frac{\tau_t-\tau}{h}\right)^l K\left(\frac{\tau_t-\tau}{h}\right), \nonumber\\
			\bm{\Delta}_{T,l}(3) &=& \frac{1}{Th} \sum_{t=1}^{T}\Delta\mathbf{x}_{t-1} \Delta \mathbf{x}_{t-1}^\top \left(\frac{\tau_t-\tau}{h}\right)^lK\left(\frac{\tau_t-\tau}{h}\right).\nonumber
		\end{eqnarray*}
		Then part (7) follows directly from parts (2)--(6).
		
	\end{proof}

	\begin{proof}[Proof of Lemma \ref{L5}]
		\item 
		\noindent (1). Let $w_t(\tau) = \frac{1}{Th}\left(\frac{\tau_t-\tau}{h}\right)^l K\left(\frac{\tau_t-\tau}{h}\right)$. By using BN decomposition, we have
		\begin{eqnarray*}
			&&\frac{1}{Th}\sum_{t=1}^{T} \widetilde{\mathbf{z}}_t(\tau_t) \left(\frac{\tau_t-\tau}{h}\right)^l K\left(\frac{\tau_t-\tau}{h}\right)\nonumber\\
			&=&\sum_{t=1}^{T}\bm{\Psi}_{\tau_t}(1)\bm{\omega}(\tau_t)\bm{\varepsilon}_t w_t(\tau) +\bm{\Psi}_{\tau_1}^*(L)\bm{\omega}(\tau_1)\bm{\varepsilon}_{0}w_1(\tau)-\bm{\Psi}_{\tau_{T}}^*(L)\bm{\omega}(\tau_{T})\bm{\varepsilon}_{T} w_T(\tau)\nonumber\\
			&&+\sum_{t=1}^{T}\left(\bm{\Psi}_{\tau_{t+1}}(L)\bm{\omega}(\tau_{t+1})w_{t+1}(\tau)-\bm{\Psi}_{\tau_t}(L)\bm{\omega}(\tau_t)w_t(\tau)\right)\bm{\varepsilon}_t.\nonumber
		\end{eqnarray*}
		
		It is easy to see the second term and the third term are $ O_P(1/(Th))$ uniformly over $\tau \in [0,1]$. By the Lipchitz continuity of $K(\cdot)$, $\bm{\omega}(\cdot)$, $\bm{\Psi}_{\tau}(1)$ and $\sup_{\tau \in [0,1]}|w_{t+1}(\tau)-w_t(\tau)| = O(1/(Th)^2)$, the fourth term is $O_P(1/(Th))$ uniformly over $\tau \in [0,1]$.
		
		Next, we consider the first term. let $\{S_l\}$ be a finite number of sub-intervals covering the interval $[0,1]$, which are centered at $s_l$ with the length $\xi_T = o(h^2)$. Denote the number of these intervals by $N_T$ then $N_T = O(\xi_T^{-1})$. Hence, we have
		\begin{eqnarray*}
			\sup_{\tau\in[0,1]}\left|\sum_{t=1}^{T}\mathbf{v}_tw_t(\tau)\right|&\leq&\max_{1\leq l\leq N_T}\left|\sum_{t=1}^{T}\mathbf{v}_t w_t(s_l)\right| + \max_{1\leq l\leq N_T}\sup_{\tau\in S_l}\left|\sum_{t=1}^{T}\mathbf{v}_t\left( w_t(\tau)- w_t(s_l)\right)\right|\\
			&=& K_{T,5} + K_{T,6}.\nonumber
		\end{eqnarray*}
		where $\mathbf{v}_t = \bm{\Psi}_{\tau_t}(1)\bm{\omega}(\tau_t)\bm{\varepsilon}_t$.
		
		By the continuity of kernel function $K(\cdot)$ and taking $\xi_T = O(\gamma_Th^2)$ with $\gamma_T = \sqrt{\log T /(Th)}$, $K_{T,6}$ is bounded by $O(1)\frac{\xi_T}{h^2}E|\mathbf{v}_{t}| = O(\gamma_T)$. 
		
		We then apply the truncation method to prove term $K_{T,5}$. Define $\mathbf{v}_t' = \mathbf{v}_t' I(|\mathbf{v}_t|\leq (Th/\log T)^{1/2})$ and $\mathbf{v}_t'' = \mathbf{v}_t - \mathbf{v}_t'$. Then, we have
		\begin{eqnarray*}
			K_{T,5} &=& \max_{1\leq l\leq N_T}\left|\sum_{t=1}^{T}\left(\mathbf{v}_t' + \mathbf{v}_t'' - E(\mathbf{v}_t' + \mathbf{v}_t'' \mid \mathcal{F}_{t-1})\right) w_t(s_l)\right|\nonumber\\
			&\leq & \max_{1\leq l\leq N_T}\left|\sum_{t=1}^{T}\left(\mathbf{v}_t'  - E(\mathbf{v}_t' \mid \mathcal{F}_{t-1})\right) w_t(s_l)\right|+\max_{1\leq l\leq N_T}\left|\sum_{t=1}^{T} \mathbf{v}_t''w_t(s_l)\right| \nonumber\\
			&&+\max_{1\leq l\leq N_T}\left|\sum_{t=1}^{T} E(\mathbf{v}_t'' \mid \mathcal{F}_{t-1}) w_t(s_l)\right|\nonumber\\
			&=& K_{T,51} + K_{T,52} + K_{T,53}.\nonumber
		\end{eqnarray*}
		For $K_{T,52}$, by H\"older's inequality and Markov's inequality,
		\begin{eqnarray*}
			E|K_{T,52}|&\leq& O(1/(Th)) \sum_{t=1}^{T}E|\mathbf{v}_t''|\leq O(1/(Th)) \sum_{t=1}^{T} \|\mathbf{v}_t\|_\delta (E|I(|\mathbf{v}_t| > (Th/\log T)^{1/2})|)^{(\delta-1)/\delta} \nonumber\\
			&\leq &(1/(Th)) \sum_{t=1}^{T}E|\mathbf{v}_t''|\leq O(1/(Th)) \sum_{t=1}^{T} \|\mathbf{v}_t\|_\delta^\delta \left(\frac{Th}{\log T}\right)^{-(\delta-1)/2} \nonumber\\
			&=&o(\sqrt{\log T/Th})\nonumber
		\end{eqnarray*}
		if $\frac{T^{1-\frac{2}{\delta}}h}{(\log T)^{1-\frac{2}{\delta}}} \to \infty$. Similarly, $K_{T,53} = o(\sqrt{\log T/Th})$.
		
		We finally deal with $K_{T,51}$. For any fixed $1\leq l\leq N_T$, let 
		$$\mathbf{Y}_t = \left(\mathbf{v}_t'  - E(\mathbf{v}_t' \mid \mathcal{F}_{t-1})\right) w_t(s_l),
		$$
		then $E(\mathbf{Y}_t\mid \mathcal{F}_{t-1}) = 0$ and $|\mathbf{Y}_t| \leq 2\frac{\sqrt{Th}}{\sqrt{\log T}}d_T$ with $d_T = \max_{1\leq t\leq s_l}w_t(s_l)$. In addition, as $E(\bm{\varepsilon}_{t}\bm{\varepsilon}_{t}^\top \mid \mathcal{F}_{t-1}) = \mathbf{I}_d$ a.s., we have
		$$
		\max_{1\leq l\leq N_T}\left|\sum_{t=1}^{T}E\left(\mathbf{Y}_t\mathbf{Y}_t^\top \mid \mathcal{F}_{t-1} \right) \right| = O_{a.s.}(1/(Th)). 
		$$
		Then by Proposition 2.1 in \cite{freedman1975tail}, we have
		\begin{eqnarray*}
			\Pr\left(K_{T,51} \geq \sqrt{8M}\gamma_T \right)&\leq& \Pr\left(K_{T,51} \geq \sqrt{8M}\gamma_T,\max_{1\leq l\leq N_T}\left|\sum_{t=1}^{T}E\left(\mathbf{Y}_t\mathbf{Y}_t^\top \mid \mathcal{F}_{t-1} \right) \right|\leq \frac{M}{Th} \right)\nonumber\\
			&& + \Pr\left(\max_{1\leq l\leq N_T}\left|\sum_{t=1}^{T}E\left(\mathbf{Y}_t\mathbf{Y}_t^\top \mid \mathcal{F}_{t-1} \right) \right|> \frac{M}{Th} \right)\nonumber\\
			&\leq & N_T \exp\left( - \frac{8M\gamma_T^2}{\frac{M}{Th}} \right) = o(1).\nonumber
		\end{eqnarray*}
		Then, we have
		$$
		\sup_{\tau\in[0,1]}\left|\frac{1}{Th}\sum_{t=1}^{T} \widetilde{\mathbf{z}}_t(\tau_t) \left(\frac{\tau_t-\tau}{h}\right)^l K\left(\frac{\tau_t-\tau}{h}\right)\right| = O_P(\sqrt{\log T/(Th)}).
		$$
		
		\medskip
		
		\noindent (2). We start from $p=0$, the other cases with fixed $p\geq 1$ can be verified in a similar manner. Write
		\begin{eqnarray*}
			&&\mathrm{vec}\left[\widetilde{\mathbf{z}}_t(\tau_t)\widetilde{\mathbf{z}}_t^\top(\tau_t) - E(\widetilde{\mathbf{z}}_t(\tau_t)\widetilde{\mathbf{z}}_t^\top(\tau_t))\right]\nonumber\\
			&=& \sum_{j=0}^{\infty}\left(\mathbf{B}_j(\tau_t) \otimes \mathbf{B}_j(\tau_t) \right) \mathrm{vec}\left[\bm{\varepsilon}_{t-j}\bm{\varepsilon}_{t-j}^\top - \mathbf{I}_d\right] + \sum_{r=1}^{\infty}\sum_{j=0}^{\infty}\left(\mathbf{B}_{j+r}(\tau_t) \otimes \mathbf{B}_j(\tau_t) \right) \mathrm{vec}\left[\bm{\varepsilon}_{t-j}\bm{\varepsilon}_{t-j-r}^\top\right]\nonumber\\
			&& + \sum_{r=1}^{\infty}\sum_{j=0}^{\infty}\left(\mathbf{B}_{j}(\tau_t) \otimes \mathbf{B}_{j+r}(\tau_t) \right) \mathrm{vec}\left[\bm{\varepsilon}_{t-j-r}\bm{\varepsilon}_{t-j}^\top\right]\nonumber
		\end{eqnarray*}
	with $\mathbf{B}_j(\tau) = \bm{\Psi}_j(\tau)\bm{\omega}(\tau)$.
		Let $w_t(\tau) = \frac{1}{Th}\left(\frac{\tau_t-\tau}{h}\right)^l K\left(\frac{\tau_t-\tau}{h}\right)$, we have
		\begin{eqnarray*}
			&&\sup_{\tau\in[0,1]}\left|	\frac{1}{Th}\sum_{t=1}^{T}\left(\widetilde{\mathbf{z}}_{t}(\tau_t)\widetilde{\mathbf{z}}_{t}^\top(\tau_{t})-E(\widetilde{\mathbf{z}}_{t}(\tau_t)\widetilde{\mathbf{z}}_{t}^\top(\tau_{t}))\right)\left(\frac{\tau_t-\tau}{h}\right)^l K\left(\frac{\tau_t-\tau}{h}\right)\right|\nonumber\\
			&\leq&  \sup_{\tau\in[0,1]}\left|\sum_{t=1}^{T}\sum_{j=0}^{\infty}\left(\mathbf{B}_j(\tau_t) \otimes \mathbf{B}_j(\tau_t) \right) \mathrm{vec}\left[\bm{\varepsilon}_{t-j}\bm{\varepsilon}_{t-j}^\top - \mathbf{I}_d\right]w_t(\tau)\right|\nonumber\\
			&& + 2\sup_{\tau\in[0,1]}\left|\sum_{t=1}^{T}\sum_{r=1}^{\infty}\sum_{j=0}^{\infty}\left(\mathbf{B}_{j}(\tau_t) \otimes \mathbf{B}_{j+r}(\tau_t) \right) \mathrm{vec}\left[\bm{\varepsilon}_{t-j-r}\bm{\varepsilon}_{t-j}^\top\right]w_t(\tau)\right|\nonumber\\
			&=& K_{T,7} + K_{T,8}.\nonumber
		\end{eqnarray*}
		
		Similarly to the proof of part (1), if $\frac{T^{1-\frac{4}{\delta}}h}{(\log T)^{1-\frac{4}{\delta}}} \to \infty$ and $E\left(\left|\bm{\varepsilon}_t\right|^4\mid \mathcal{F}_{t-1}\right) < \infty$ a.s., we have $K_{T,7} = O_P\left(\sqrt{\log T/(Th)}\right)$.
		
		We now consider $K_{T,8}$. Define $\mathbb{B}_t^r(L) = \sum_{j=0}^{\infty}\left(\mathbf{B}_{j+r}(\tau_t) \otimes \mathbf{B}_{j}(\tau_t) \right)L^j$, which can be decomposed as
		$$
		\mathbb{B}_t^r(L) = \mathbb{B}_t^r(1) - (1-L)\widetilde{\mathbb{B}}_t^r(L),
		$$
		where $\widetilde{\mathbb{B}}_t^r(L) = \sum_{j=0}^{\infty}\widetilde{\mathbf{B}}_{j,t}^r$ and $\widetilde{\mathbf{B}}_{j,t}^r = \sum_{k=j+1}^{\infty}(\mathbf{B}_{k+r}(\tau_t) \otimes \mathbf{B}_{k}(\tau_t))$. By using the above BN decomposition, we have
		\begin{eqnarray*}
			K_{T,8} &\leq & \sup_{0\leq\tau\leq1} \left|\sum_{t=1}^{T}\xi_t\bm{\varepsilon}_tw_t(\tau)\right|+\sup_{0\leq\tau\leq1} \left|\sum_{r=1}^{\infty}\widetilde{\mathbb{B}}_1^r(L)\mathrm{vec}\left(\bm{\varepsilon}_0\bm{\varepsilon}_{-r}\right)w_1(\tau)\right|\nonumber\\
			&&+\sup_{0\leq\tau\leq1} \left|\sum_{r=1}^{\infty}\widetilde{\mathbb{B}}_T^r(L)\mathrm{vec}\left(\bm{\varepsilon}_T\bm{\varepsilon}_{T-r}\right)w_T(\tau)\right|\nonumber\\
			&&+\sup_{0\leq\tau\leq1} \left|\sum_{t=1}^{T}\sum_{r=1}^{\infty}\left(\widetilde{\mathbb{B}}_{t+1}^r(L)w_{t+1}(\tau)-\widetilde{\mathbb{B}}_t^r(L)w_t(\tau)\right)\mathrm{vec}\left(\bm{\varepsilon}_t\bm{\varepsilon}_{t-r}\right)\right|\nonumber\\
			&=& K_{T,81}+ K_{T,82}+K_{T,83}+K_{T,84},\nonumber
		\end{eqnarray*}
		where $\bm{\xi}_t = \sum_{r=1}^{\infty}\sum_{s=0}^{\infty}(\mathbf{B}_{s+r}(\tau_t)\bm{\varepsilon}_{t-r})\otimes \mathbf{B}_s(\tau_t)$. Similar to the proof of part (1), the terms $K_{T,82}$--$K_{T,84}$ are $O_P(1/(Th))$.
		
		For term $K_{T,81}$, similarly let $\{S_l\}$ be a finite number of sub-intervals covering the interval $[0,1]$, which are centered at $s_l$ with the length $\xi_T = o(h^2)$. Denote the number of these intervals by $N_T$ then $N_T = O(\xi_T^{-1})$. Hence, we have
		\begin{eqnarray*}
			\sup_{\tau\in[0,1]}\left|\sum_{t=1}^{T}\mathbf{v}_tw_t(\tau)\right|&\leq&\max_{1\leq l\leq N_T}\left|\sum_{t=1}^{T}\mathbf{v}_t w_t(s_l)\right| + \max_{1\leq l\leq N_T}\sup_{\tau\in S_l}\left|\sum_{t=1}^{T}\mathbf{v}_t\left( w_t(\tau)- w_t(s_l)\right)\right|\nonumber\\
			&=& K_{T,811} + K_{T,812}.\nonumber
		\end{eqnarray*}
		where $\mathbf{v}_t = \xi_t\bm{\varepsilon}_t$. By the continuity of kernel function $K(\cdot)$ and taking $\xi_T = O(\gamma_Th^2)$ with $\gamma_T = \sqrt{\log T /(Th)}$, $K_{T,812}$ is bounded by $O(1)\frac{\xi_T}{h^2}E|\mathbf{v}_{t}| = O(\gamma_T)$.

		We then apply the truncation method to prove term $K_{T,81}$. Define $\mathbf{v}_t' = \mathbf{v}_t' I(|\mathbf{v}_t|\leq (Th/\log T)^{1/2})$ and $\mathbf{v}_t'' = \mathbf{v}_t - \mathbf{v}_t'$. Then, we have
		\begin{eqnarray*}
			K_{T,811} &=& \max_{1\leq l\leq N_T}\left|\sum_{t=1}^{T}\left(\mathbf{v}_t' + \mathbf{v}_t'' - E(\mathbf{v}_t' + \mathbf{v}_t'' \mid \mathcal{F}_{t-1})\right) w_t(s_l)\right|\nonumber\\
			&\leq & \max_{1\leq l\leq N_T}\left|\sum_{t=1}^{T}\left(\mathbf{v}_t'  - E(\mathbf{v}_t'' \mid \mathcal{F}_{t-1})\right) w_t(s_l)\right|+\max_{1\leq l\leq N_T}\left|\sum_{t=1}^{T} \mathbf{v}_t''w_t(s_l)\right|\nonumber \\
			&&+\max_{1\leq l\leq N_T}\left|\sum_{t=1}^{T} E(\mathbf{v}_t'' \mid \mathcal{F}_{t-1}) w_t(s_l)\right|\nonumber\\
			&=& K_{T,8111} + K_{T,8112} + K_{T,8113}.\nonumber
		\end{eqnarray*}
		For $K_{T,8112}$, by H\"older's inequality and Markov's inequality,
		\begin{eqnarray*}
			E|K_{T,8112}|&\leq& O(1/(Th)) \sum_{t=1}^{T}E|\mathbf{v}_t''|\leq O(1/(Th)) \sum_{t=1}^{T} \|\mathbf{v}_t\|_{\delta/2} (E|I(|\mathbf{v}_t| > (Th/\log T)^{1/2})|)^{(\delta-2)/\delta}\nonumber \\
			&\leq &(1/(Th)) \sum_{t=1}^{T}E|\mathbf{v}_t''|\leq O(1/(Th)) \sum_{t=1}^{T} \|\mathbf{v}_t\|_{\delta/2}^{\delta/2} \left(\frac{Th}{\log T}\right)^{-(\delta-2)/4}\nonumber \\
			&=&o(\sqrt{\log T/Th})\nonumber
		\end{eqnarray*}
		if $\frac{T^{1-\frac{4}{\delta}}h}{(\log T)^{1-\frac{4}{\delta}}} \to \infty$. Similarly, $K_{T,8113} = o(\sqrt{\log T/Th})$.

		In the next step, we deal with $K_{T,8111}$. For any fixed $1\leq l\leq N_T$, let 
		$$\mathbf{Y}_t = \left(\mathbf{v}_t'  - E(\mathbf{v}_t' \mid \mathcal{F}_{t-1})\right) w_t(s_l),
		$$
		then $E(\mathbf{Y}_t\mid \mathcal{F}_{t-1}) = 0$ and $|\mathbf{Y}_t| \leq 2\frac{\sqrt{Th}}{\sqrt{\log T}}d_T$ with $d_T = \max_{1\leq t\leq s_l}w_t(s_l)$. Note that 
		\begin{eqnarray*}
			\max_{1\leq l\leq N_T}\left|\sum_{t=1}^{T}E\left(\mathbf{Y}_t\mathbf{Y}_t^\top \mid \mathcal{F}_{t-1} \right)\right| &\leq& \max_{1\leq l\leq N_T} \sum_{t=1}^{T}E\left(|\mathbf{Y}_t|^2 \mid \mathcal{F}_{t-1} \right)\nonumber\\
			&\leq & O(1/(Th))\max_{1\leq l\leq N_T} \sum_{t=1}^{T}E\left(\mathbf{v}_t^\top\mathbf{v}_t \mid \mathcal{F}_{t-1} \right)w_t(s_l).\nonumber
		\end{eqnarray*}
		
		In addition, using the summation by parts formula, we have
		\begin{eqnarray*}
			&& \max_{1\leq l\leq N_T}\left|\sum_{t=1}^{T}\left[E\left(\mathbf{v}_t^\top\mathbf{v}_t \mid \mathcal{F}_{t-1} \right) - E\left(\mathbf{v}_t^\top\mathbf{v}_t \right)\right] w_t(s_l)\right|\nonumber\\
			&\leq & (2\sup_{u}|u^lK(u)| + \int_{-1}^{1}|lu^{l-1}K(u) + u^{l}K^{(1)}(u)|\mathrm{d}u) \nonumber \\
			&&\times \left| \frac{1}{Th}\sum_{t=1}^{T} \left[E\left(\mathbf{v}_t^\top\mathbf{v}_t \mid \mathcal{F}_{t-1} \right) - E\left(\mathbf{v}_t^\top\mathbf{v}_t \right)\right] \right| = O_P(1/(\sqrt{T}h))=o_P(1).\nonumber
		\end{eqnarray*}
		
		We therefore have $\max_{1\leq l\leq N_T}\left|\sum_{t=1}^{T}E\left(\mathbf{Y}_t\mathbf{Y}_t^\top \mid \mathcal{F}_{t-1} \right)\right| = O_P(1/(Th))$.
		
		Finally, by Proposition 2.1 in \cite{freedman1975tail}, we have
		\begin{eqnarray*}
			\Pr\left(K_{T,8111} \geq \sqrt{8M}\gamma_T \right)&\leq& \Pr\left(K_{T,8111} \geq \sqrt{8M}\gamma_T,\max_{1\leq l\leq N_T}\left|\sum_{t=1}^{T}E\left(\mathbf{Y}_t\mathbf{Y}_t^\top \mid \mathcal{F}_{t-1} \right) \right|\leq \frac{M}{Th} \right)\nonumber\\
			&& + \Pr\left(\max_{1\leq l\leq N_T}\left|\sum_{t=1}^{T}E\left(\mathbf{Y}_t\mathbf{Y}_t^\top \mid \mathcal{F}_{t-1} \right) \right|> \frac{M}{Th} \right)\nonumber\\
			&\leq & N_T \exp\left( - \frac{8M\gamma_T^2}{\frac{M}{Th}} \right) = o(1).\nonumber
		\end{eqnarray*}
		
		\medskip
		
		\noindent (3). Part (3) can be proved in a similar way as that of part (2), so we omit the details.

	\end{proof}

	\begin{proof}[Proof of Lemma \ref{L6}]
		\item
		
		Write
		\begin{eqnarray*}
			\max_{t\leq T}\left|\sum_{j=1}^{t}\widetilde{\mathbf{z}}_j(\tau_j)\right|
			&\leq &\max_{t\leq T}\left|\sum_{j=1}^{t}\bm{\Psi}_{\tau_j}(1)\bm{\omega}(\tau_j)\bm{\varepsilon}_j\right|+\left|\bm{\Psi}_{\tau_1}^*(L)\bm{\omega}(\tau_1)\bm{\varepsilon}_{0}\right| + \max_{t\leq T}\left|\bm{\Psi}_{\tau_t}^*(L)\bm{\omega}(\tau_t)\bm{\varepsilon}_t\right|\nonumber\\
			&&+\max_{t\leq T}\left|\sum_{j=1}^{t-1}\left(\bm{\Psi}_{\tau_{j+1}}(L)\bm{\omega}(\tau_{j+1})-\bm{\Psi}_{\tau_j}(L)\bm{\omega}(\tau_j)\right)\bm{\varepsilon}_j\right|\nonumber\\
			&=& \mathbf{K}_{T,9} + \mathbf{K}_{T,10} + \mathbf{K}_{T,11} + \mathbf{K}_{T,12}.\nonumber
		\end{eqnarray*}
		
		For $\mathbf{K}_{T,10}$, by Chebyshev's inequality and $\|\bm{\Psi}_{\tau_1}^*(L)\bm{\omega}(\tau_1)\bm{\varepsilon}_{0} \|_\delta < \infty$, we have
		$$
		\Pr\left(\left|\bm{\Psi}_{\tau_1}^*(L)\bm{\omega}(\tau_1)\bm{\varepsilon}_{0}\right| \geq x\right)= O(1/x^\delta).
		$$
		Similarly, since
		$$
		\Pr\left(\max_{t\leq T}\left|\bm{\Psi}_{\tau_t}^*(L)\bm{\omega}(\tau_t)\bm{\varepsilon}_t\right| \geq x\right)\leq \sum_{t=1}^{T}\Pr\left(\left|\bm{\Psi}_{\tau_t}^*(L)\bm{\omega}(\tau_t)\bm{\varepsilon}_t\right| \geq x\right),
		$$ 
		we have $\Pr\left(\mathbf{K}_{T,11} \geq x\right)= O(T/x^\delta)$.
		
		Since
		$$
		\mathbf{K}_{T,12}\leq \sum_{j=1}^{T-1}\left|\left(\bm{\Psi}_{\tau_{j+1}}(L)\bm{\omega}(\tau_{j+1})-\bm{\Psi}_{\tau_j}(L)\bm{\omega}(\tau_j)\right)\bm{\varepsilon}_j\right|
		$$
		and
		\begin{eqnarray*}
			&&\|\sum_{j=1}^{T-1}\left|\left(\bm{\Psi}_{\tau_{j+1}}(L)\bm{\omega}(\tau_{j+1})-\bm{\Psi}_{\tau_j}(L)\bm{\omega}(\tau_j)\right)\bm{\varepsilon}_j\right| \|_\delta\nonumber\\
			&\leq& \sum_{j=1}^{T-1}\left|\bm{\Psi}_{\tau_{j+1}}(1)\bm{\omega}(\tau_{j+1})-\bm{\Psi}_{\tau_j}(1)\bm{\omega}(\tau_j)\right| \max_t\|\bm{\varepsilon}_t\|_\delta =O(1),\nonumber
		\end{eqnarray*}
		we have
		\begin{eqnarray*}
			&&\Pr\left(\mathbf{K}_{T,12} \geq x\right)\leq \Pr\left(\sum_{j=1}^{T-1}\left|\left(\bm{\Psi}_{\tau_{j+1}}(L)\bm{\omega}(\tau_{j+1})-\bm{\Psi}_{\tau_j}(L)\bm{\omega}(\tau_j)\right)\bm{\varepsilon}_j\right| \geq x\right)\nonumber\\
			&\leq& \frac{\|\sum_{j=1}^{T-1}\left|\left(\bm{\Psi}_{\tau_{j+1}}(L)\bm{\omega}(\tau_{j+1})-\bm{\Psi}_{\tau_j}(L)\bm{\omega}(\tau_j)\right)\bm{\varepsilon}_j\right|\|_\delta^\delta}{x^\delta} = O(1/x^\delta).\nonumber
		\end{eqnarray*}
		
		Then by the Nagaev inequality (\citealp{borovkov1973notes}) for the maximum of absolute partial sums of independent variables, we have if $x\geq c \sqrt{T \log T}$ for some constant $c\geq 0$
		$$
		\Pr\left(\mathbf{K}_{T,9}\geq x \right) \leq \left(1 + 2/\delta\right)^{\delta} \frac{\mu_{T,\delta}}{x^\delta} + 2\exp\left(-\frac{2x^2}{e^{\delta}(\delta+2)^2\mu_{T,2}}\right) = O\left(\frac{T}{x^\delta}\right),
		$$
		where $\mu_{T,\delta} = \sum_{t=1}^{T}E|\bm{\Psi}_{\tau_t}(1)\bm{\omega}(\tau_t)\bm{\varepsilon}_t|^\delta$.
		
		Combing the above results, we complete the proof.
		
	\end{proof}

	\begin{proof}[Proof of Lemma \ref{L7}]
		\item
		\noindent (1). By the proof of Lemma \ref{L1}, we have
		$$
		\sup_{u\in[0,1]}\left|T^{-1/2}\mathbf{y}_{\lfloor Tu \rfloor}- \mathbf{P}_{\bm{\beta}_\perp} T^{-1/2}\sum_{t=1}^{\lfloor Tu \rfloor}\bm{\Psi}_{\tau_t}(1)\bm{\omega}(\tau_t)\bm{\varepsilon}_t \right| = o_P(1).
		$$
		In addition, by Gaussian approximations for vector martingale differences (e.g., Theorem 1 in \citealp{eberlein1986strong}), we have
		$$
		\sup_{u\in[0,1]}\left| \mathbf{P}_{\bm{\beta}_\perp} T^{-1/2}\sum_{t=1}^{\lfloor Tu \rfloor}\bm{\Psi}_{\tau_t}(1)\bm{\omega}(\tau_t)\bm{\varepsilon}_t - \mathbf{W}_d(u,\bm{\Sigma}_{\mathbf{y}}(u))\right| = o_{a.s.}(1).
		$$
		The proof of part (1) is now completed.
		
		\medskip
		
		\noindent (2). Write
		\begin{eqnarray*}
			&&T^{-2}h^{-1}\sum_{t=1}^{T}\mathbf{y}_{t-1}\mathbf{y}_{t-1}^\top K\left(\frac{\tau_t-\tau}{h}\right) \nonumber\\
			&=& (T^{-1/2}\mathbf{y}_{\delta_T})(T^{-1/2}\mathbf{y}_{\delta_T})^\top \cdot (Th)^{-1}\sum_{t=1}^{T}K\left(\frac{\tau_t-\tau}{h}\right)\nonumber\\
			&&+ (T^{-1/2}\mathbf{y}_{\delta_T})\cdot T^{-3/2}h^{-1}\sum_{t=1}^{T} (\mathbf{y}_{t-1}-\mathbf{y}_{\delta_T})^\top K\left(\frac{\tau_t-\tau}{h}\right)\nonumber\\
			&& +  T^{-3/2}h^{-1}\sum_{t=1}^{T} (\mathbf{y}_{t-1}-\mathbf{y}_{\delta_T}) K\left(\frac{\tau_t-\tau}{h}\right)\cdot (T^{-1/2}\mathbf{y}_{\delta_T})^\top\nonumber\\
			&& + T^{-2}h^{-1}\sum_{t=1}^{T} (\mathbf{y}_{t-1}-\mathbf{y}_{\delta_T}) (\mathbf{y}_{t-1}-\mathbf{y}_{\delta_T})^\top K\left(\frac{\tau_t-\tau}{h}\right)\nonumber\\
			&=&\mathbf{K}_{T,13}+ \mathbf{K}_{T,14} + \mathbf{K}_{T,15} + \mathbf{K}_{T,16}.\nonumber
		\end{eqnarray*}
		
		Consider $\mathbf{K}_{T,13}$. By part (1) of this lemma, we have $\sup_{\tau\in[0,1]}|T^{-1/2}\mathbf{y}_{\delta_T} - \mathbf{W}_d(\tau,\bm{\Sigma}_{\mathbf{y}}(\tau))| = o_P(1)$. As $\sup_{\tau\in[0,1]}|\mathbf{W}_d(\tau,\bm{\Sigma}_{\mathbf{y}}(\tau))| = O_P(1)$, we have $\sup_{\tau\in[0,1]}|T^{-1/2}\mathbf{y}_{\delta_T}|=O_P(1)$ and thus
		$$
		\sup_{\tau\in[h,1-h]}|\mathbf{J}_{T,1} - \mathbf{W}_d(\tau,\bm{\Sigma}_{\mathbf{y}}(\tau))\mathbf{W}_d^\top(\tau,\bm{\Sigma}_{\mathbf{y}}(\tau))| = o_P(1).
		$$
		
		Consider $\mathbf{K}_{T,14}$. We first need to prove 
		$$
		\sup_{\tau \in[0,1],u\in[0,1]}\left|\frac{1}{\sqrt{2Th}}\sum_{t=\delta_T+1}^{\delta_T(u)}\Delta\mathbf{y}_t - \mathbf{W}_\tau(u)\right| = o_P(1),
		$$
		where $\delta_T(u) = \delta_T +\lfloor2uTh\rfloor+1$ and $\mathbf{W}_\tau(u)$ is a Brownian motion. Write
		\begin{eqnarray*}
			\frac{1}{\sqrt{2Th}}\sum_{t=\delta_T+1}^{\delta_T(u)}\Delta\mathbf{y}_t&=&\mathbf{P}_{\bm{\beta}_\perp} (2Th)^{-1/2}\sum_{t=\delta_T+1}^{\delta_T(u)}\mathbf{z}_t + (2Th)^{-1/2}\mathbf{P}_{\bm{\beta}}\mathbf{z}_{\delta_T(u)} + (2Th)^{-1/2}\mathbf{P}_{\bm{\beta}_\perp} \mathbf{z}_{\delta_T}.\nonumber
		\end{eqnarray*} 
		
		Since $\max_{1\leq t \leq T}|\mathbf{z}_{t}|\leq \left\{\sum_{t=1}^{T}|\mathbf{z}_t|^{\delta}\right\}^{1/\delta}= O_P(T^{1/\delta})$, the second term and the third term are $o_P(1)$ if $T^{1-\frac{2}{\delta}}h\to \infty$. Similar to the proof of Lemma \ref{L1}, if $T^{1-\frac{2}{\delta}}h\to \infty$, by using BN decomposition, we have
		$$
		\sup_{u\in[0,1]}\left|(Th)^{-1/2}\sum_{t=\delta_T+1}^{\delta_T(u)}\mathbf{z}_t- (Th)^{-1/2}\sum_{t=\delta_T+1}^{\delta_T(u)}\bm{\Psi}_{\tau_t}(1)\bm{\omega}(\tau_t)\bm{\varepsilon}_t \right| = o_P(1).
		$$
		For any fix $\tau \in [0,1]$, by strong invariance principle, there exists a Brownian motion $\mathbf{W}_\tau(u)$ such that 
		$$
		\sup_{u\in[0,1]}\left|(Th)^{-1/2}\sum_{t=\delta_T+1}^{\delta_T(u)}\bm{\Psi}_{\tau_t}(1)\bm{\omega}(\tau_t)\bm{\varepsilon}_t - \mathbf{W}_\tau(u) \right| = o_{a.s.}(1).
		$$
		
		Since this convergence is almost surely, for any $\epsilon>0$, we have
		\begin{eqnarray*}
			&&\Pr\left(\sup_{\tau\in[0,1]}\sup_{u\in[0,1]}\left|(Th)^{-1/2}\sum_{t=\delta_T+1}^{\delta_T(u)}\bm{\Psi}_{\tau_t}(1)\bm{\omega}(\tau_t)\bm{\varepsilon}_t - \mathbf{W}_{\tau}(u) \right| >\epsilon \right)\nonumber\\
			&\leq &\sum_{j=1}^{T}\Pr\left(\sup_{u\in[0,1]}\left|(Th)^{-1/2}\sum_{t=j+1}^{j+\lfloor2uTh\rfloor+1}\bm{\Psi}_{\tau_t}(1)\bm{\omega}(\tau_t)\bm{\varepsilon}_t - \mathbf{W}_{\tau}(u) \right| >\epsilon \right)\to 0.\nonumber
		\end{eqnarray*}
		Hence, we can conclude that $\sup_{\tau \in[0,1],u\in[0,1]}\left|\frac{1}{\sqrt{2Th}}\sum_{t=\delta_T+1}^{\delta_T(u)}\Delta\mathbf{y}_t - \mathbf{W}_\tau(u)\right| = o_P(1)$. In addition, since the order of maximum of absolute value of $T$ normal variables is $O_P(\sqrt{\log T})$, then we have $$\sup_{u\in[0,1]}\left|\frac{1}{\sqrt{2Th}}\sum_{t=\delta_T+1}^{\delta_T(u)}\Delta\mathbf{y}_t\right| = O_P(\sqrt{\log T}).$$
		
		Next, write
		\begin{eqnarray*}
			&&\sup_{\tau\in[0,1]}\left|(T^{-1/2}\mathbf{y}_{\delta_T})\cdot T^{-3/2}h^{-1}\sum_{t=1}^{T} (\mathbf{y}_{t-1}-\mathbf{y}_{\delta_T})^\top K\left(\frac{\tau_t-\tau}{h}\right)\right|\nonumber\\
			&\leq &O_P(1)\sqrt{h}\sup_{\tau \in[0,1],u\in[0,1]}\left|\frac{1}{\sqrt{2Th}}\sum_{t=\delta_T+1}^{\delta_T(u)}\Delta\mathbf{y}_t\right|\sup_{\tau\in[0,1]}\left|(Th)^{-1}\sum_{t=1}^{T} K\left(\frac{\tau_t-\tau}{h}\right)\right|\nonumber\\
			&=&O_P(\sqrt{h\log T}) =o_P(1).\nonumber
		\end{eqnarray*}
		
		Similarly, $\sup_{0\leq\tau\leq1}|\mathbf{K}_{T,15}| = o_P(1)$ and $\sup_{0\leq\tau\leq1}|\mathbf{K}_{T,16} | = o_P(1)$. The proof of part (2) is now completed.
		
		\medskip
		
		\noindent (3). Write
		\begin{eqnarray*}
			\mathbf{D}_T^{-1} \bm{\Xi}_T^\top(\tau)\left[ \sum_{t=1}^{T}\bm{\alpha}_{\perp}^\top(\tau)\mathbf{y}_{t-1}\mathbf{y}_{t-1}^\top\bm{\alpha}_{\perp}(\tau)\left(\frac{\tau_t-\tau}{h}\right)^l K\left(\frac{\tau_t-\tau}{h}\right)\right] \bm{\Xi}_T(\tau) \mathbf{D}_T^{-1}=\left[\begin{matrix}
				\mathbf{K}_{T,l}(6) & \mathbf{K}_{T,l}(7)\\
				\mathbf{K}_{T,l}^\top(7) & \mathbf{K}_{T,l}(8)\\
			\end{matrix} \right],\nonumber
		\end{eqnarray*}
		where
		\begin{eqnarray*}
			\mathbf{K}_{T,l}(6) &=& \frac{1}{T^2h}\bm{\xi}_T^\top(\tau)\sum_{t=1}^{T}\bm{\alpha}_{\perp}^\top(\tau)\mathbf{y}_{t-1}\mathbf{y}_{t-1}^\top\bm{\alpha}_{\perp}(\tau)\bm{\xi}_T(\tau)\left(\frac{\tau_t-\tau}{h}\right)^l K\left(\frac{\tau_t-\tau}{h}\right),\nonumber\\
			\mathbf{K}_{T,l}(7) &=& \frac{1}{T^2h^{3/2}}\bm{\xi}_T^\top(\tau)\sum_{t=1}^{T}\bm{\alpha}_{\perp}^\top(\tau)\mathbf{y}_{t-1}\mathbf{y}_{t-1}^\top\bm{\alpha}_{\perp}(\tau)\bm{\xi}_{T,\perp}(\tau)\left(\frac{\tau_t-\tau}{h}\right)^l K\left(\frac{\tau_t-\tau}{h}\right),\nonumber\\
			\mathbf{K}_{T,l}(8) &=& \frac{1}{T^2h^{2}}\bm{\xi}_{T,\perp}^\top(\tau)\sum_{t=1}^{T}\bm{\alpha}_{\perp}^\top(\tau)\mathbf{y}_{t-1}\mathbf{y}_{t-1}^\top\bm{\alpha}_{\perp}(\tau)\bm{\xi}_{T,\perp}(\tau)\left(\frac{\tau_t-\tau}{h}\right)^l K\left(\frac{\tau_t-\tau}{h}\right).\nonumber
		\end{eqnarray*}
		
		Consider $\mathbf{K}_{T,l}(6)$. Similar to the proof of part (2), by $\sup_{\tau\in[0,1]}|\bm{\xi}_T(\tau)-\bm{\xi}(\tau)|=o_P(1)$, we have $\sup_{\tau\in[h,1-h]}\left|\mathbf{K}_{T,l}(6) - \bm{\Delta}_{l,1}(\tau)\right| = o_P(1)$.
		
		Consider $\mathbf{K}_{T,l}(7)$. Since $\mathbf{q}_T^\top(\tau) \bm{\xi}_{T,\perp}(\tau) = 0$, we have
		\begin{eqnarray*}
			&&\mathbf{K}_{T,l}(7) = \bm{\xi}_T^\top(\tau)\mathbf{q}_T(\tau) \frac{1}{Th^{3/2}} \sum_{t=1}^{T}\frac{1}{\sqrt{T}}(\mathbf{y}_{t-1} - \mathbf{y}_{\delta_T})^\top\left(\frac{\tau_t-\tau}{h}\right)^l K\left(\frac{\tau_t-\tau}{h}\right)\bm{\alpha}_{\perp}(\tau)\bm{\xi}_{T,\perp}(\tau)\nonumber\\
			&& + \bm{\xi}_T^\top(\tau)\bm{\alpha}_{\perp}^\top(\tau) \frac{1}{Th^{3/2}} \sum_{t=1}^{T}\frac{1}{\sqrt{T}}(\mathbf{y}_{t-1} - \mathbf{y}_{\delta_T})\frac{1}{\sqrt{T}}(\mathbf{y}_{t-1} - \mathbf{y}_{\delta_T})^\top \left(\frac{\tau_t-\tau}{h}\right)^lK\left(\frac{\tau_t-\tau}{h}\right)\bm{\alpha}_{\perp}(\tau)\bm{\xi}_{T,\perp}(\tau)\nonumber\\
			&=&\mathbf{K}_{T,l}(7,1) + \mathbf{K}_{T,l}(7,2).\nonumber
		\end{eqnarray*}
		For $\mathbf{K}_{T,l}(7,2)$, similar to the proof of part (2), we have $\sup_{\tau\in[0,1]}|\mathbf{K}_{T,l}(7,2)| = o_P(1)$. In addition, by part (1), we have uniformly over $\tau\in[h,1-h]$
		\begin{eqnarray*}
			&&\sum_{t=1}^{T}\frac{1}{\sqrt{T}}(\mathbf{y}_{t-1} - \mathbf{y}_{\delta_T})\left(\frac{\tau_t-\tau}{h}\right)^l K\left(\frac{\tau_t-\tau}{h}\right) \nonumber\\
			&=&\sqrt{2} \int_{-1}^{1}\mathbf{P}_{\bm{\beta}_\perp}\left(\frac{1}{\sqrt{2h}}\int_{\tau-h}^{\tau-h+(u+1)h}\bm{\Psi}_{r}(1)\bm{\omega}(r)\mathrm{d}\mathbf{W}_d^*(r)\right)u^{l}K(u)\mathrm{d}u + o_P(1)\nonumber\\
			&=&\sqrt{2} \int_{-1}^{1}\mathbf{P}_{\bm{\beta}_\perp}\bm{\Psi}_{\tau}(1)\bm{\omega}(\tau)\frac{1}{\sqrt{2h}}\mathbf{W}_d^*((u+1)h)u^{l}K(u)\mathrm{d}u + o_P(1)\nonumber\\
			&=&\sqrt{2}\int_{-1}^{1}\mathbf{W}_d^{*}((u+1)/2,\bm{\Sigma}_{\bm{\alpha}}(\tau))u^lK(u)\mathrm{d}u + o_P(1).\nonumber
		\end{eqnarray*}
		Similarly, $\sup_{\tau\in[h,1-h]}\left|\mathbf{K}_{T,l}(8)-\bm{\Delta}_{l,3}(\tau)\right| = o_P(1)$.
		
		The proof of part (3) is now completed.
		
		\medskip
		
		\noindent (4). Write
		\begin{eqnarray*}
			&&\frac{1}{\sqrt{Th}}\mathbf{D}_T^{-1} \bm{\Xi}_T^\top(\tau) \sum_{t=1}^{T}\bm{\alpha}_{\perp}^\top(\tau)\mathbf{y}_{t-1}\mathbf{y}_{t-1}^\top\bm{\beta}\left(\frac{\tau_t-\tau}{h}\right)^l K\left(\frac{\tau_t-\tau}{h}\right)\nonumber\\
			&=&\left[\begin{matrix}
				\frac{1}{T^{3/2}h}\sum_{t=1}^{T}\bm{\xi}_T^\top(\tau)\bm{\alpha}_{\perp}^\top(\tau)\mathbf{y}_{t-1}\mathbf{z}_{t-1}^\top \left(\frac{\tau_t-\tau}{h}\right)^lK\left(\frac{\tau_t-\tau}{h}\right)\\
				\frac{1}{T^{3/2}h^{3/2}}\sum_{t=1}^{T}\bm{\xi}_{T,\perp}^\top(\tau)\bm{\alpha}_{\perp}^\top(\tau)\mathbf{y}_{t-1}\mathbf{z}_{t-1}^\top \left(\frac{\tau_t-\tau}{h}\right)^lK\left(\frac{\tau_t-\tau}{h}\right)\\
			\end{matrix} \right] =\left[\begin{matrix}
				\mathbf{K}_{T,l}(9)\\
				\mathbf{K}_{T,l}(10)\\
			\end{matrix} \right].\nonumber
		\end{eqnarray*}
		
		Consider $\mathbf{K}_{T,l}(9)$, 
		\begin{eqnarray*}
			&&\frac{1}{T^{3/2}h}\sum_{t=1}^{T}\bm{\xi}_T^\top(\tau)\bm{\alpha}_{\perp}^\top(\tau)\mathbf{y}_{t-1}\mathbf{z}_{t-1}^\top \left(\frac{\tau_t-\tau}{h}\right)^lK\left(\frac{\tau_t-\tau}{h}\right)\nonumber\\
			&=& \frac{1}{Th}\bm{\xi}_T^\top(\tau)\sum_{t=1}^{T}\left(\frac{1}{\sqrt{T}}\bm{\alpha}_{\perp}^\top(\tau)\mathbf{y}_{\delta_T}+\frac{1}{\sqrt{T}}\bm{\alpha}_{\perp}^\top(\tau)(\mathbf{y}_{t-1} - \mathbf{y}_{\delta_T})\right)\mathbf{z}_{t-1}^\top \left(\frac{\tau_t-\tau}{h}\right)^lK\left(\frac{\tau_t-\tau}{h}\right)\nonumber\\
			&=&\sqrt{\mathbf{q}_T^\top(\tau)\mathbf{q}_T(\tau)}\frac{1}{Th}\sum_{t=1}^{T}\mathbf{z}_{t-1}^\top \left(\frac{\tau_t-\tau}{h}\right)^lK\left(\frac{\tau_t-\tau}{h}\right)\nonumber\\
			&& + \frac{\sqrt{h}}{Th}\bm{\xi}_T^\top(\tau)\sum_{t=1}^{T}\frac{1}{\sqrt{Th}}\bm{\alpha}_{\perp}^\top(\tau)(\mathbf{y}_{t-1} - \mathbf{y}_{\delta_T})\mathbf{z}_{t-1}^\top\left(\frac{\tau_t-\tau}{h}\right)^l K\left(\frac{\tau_t-\tau}{h}\right)\nonumber\\
			&=& \mathbf{K}_{T,l}(9,1) + \mathbf{K}_{T,l}(9,2).\nonumber
		\end{eqnarray*}

		As $\sup_{\tau \in [0,1]} |\mathbf{q}_T(\tau)| = O_P(1)$ and 
		$$
		\sup_{\tau \in [0,1]} \left|\frac{1}{Th}\sum_{t=1}^{T}\mathbf{z}_{t-1} \left(\frac{\tau_t-\tau}{h}\right)^lK\left(\frac{\tau_t-\tau}{h}\right)\right| = O_P\left(\sqrt{\log T /(Th)}\right),
		$$
		we have $\mathbf{K}_{T,l}(9,1) = O_P\left(\sqrt{\log T/(Th)}\right)$. For $\mathbf{K}_{T,l}(9,2)$, write
		\begin{eqnarray*}
			&&\frac{1}{T^{3/2}h}\sum_{t=1}^{T}(\mathbf{y}_{t-1} - \mathbf{y}_{\delta_T})\mathbf{z}_{t-1}^\top\left(\frac{\tau_t-\tau}{h}\right)^l K\left(\frac{\tau_t-\tau}{h}\right)\nonumber\\
			&=&\frac{1}{T^{3/2}h}\sum_{t=1}^{T}(\mathbf{y}_{t-1} - \mathbf{y}_{\delta_T})(\mathbf{z}_{t-1}-\widetilde{\mathbf{z}}_{t-1}(\tau_{t-1}))^\top\left(\frac{\tau_t-\tau}{h}\right)^l K\left(\frac{\tau_t-\tau}{h}\right)\nonumber\\
			&& + \mathbf{P}_{\bm{\beta}_\perp}\frac{1}{T^{3/2}h}\sum_{t=1}^{T}\sum_{j=\delta_T+1}^{t-1}(\mathbf{z}_{j} - \widetilde{\mathbf{z}}_{j}(\tau_j))\widetilde{\mathbf{z}}_{t-1}^\top(\tau_{t-1})\left(\frac{\tau_t-\tau}{h}\right)^l K\left(\frac{\tau_t-\tau}{h}\right)\nonumber\\ 
			&& +\mathbf{P}_{\bm{\beta}}\frac{1}{T^{3/2}h}\sum_{t=1}^{T}(\mathbf{z}_{t-1} - \mathbf{z}_{\delta_T}) \widetilde{\mathbf{z}}_{t-1}^\top(\tau_{t-1})\left(\frac{\tau_t-\tau}{h}\right)^l K\left(\frac{\tau_t-\tau}{h}\right)\nonumber\\ 
			&& + \mathbf{P}_{\bm{\beta}_\perp}\frac{1}{T^{3/2}h}\sum_{t=1}^{T}\sum_{j=\delta_T+1}^{t-1} \widetilde{\mathbf{z}}_{j}(\tau_j)\widetilde{\mathbf{z}}_{t-1}^\top(\tau_{t-1})\left(\frac{\tau_t-\tau}{h}\right)^l K\left(\frac{\tau_t-\tau}{h}\right)\nonumber\\ 
			&=& \mathbf{K}_{T,l}(9,21) + \mathbf{K}_{T,l}(9,22) + \mathbf{K}_{T,l}(9,23)+ \mathbf{K}_{T,l}(9,24).\nonumber
		\end{eqnarray*}
		
		Note that $\sup_{\tau \in[0,1],u\in[0,1]}\left|\frac{1}{\sqrt{2Th}}\sum_{t=\delta_T+1}^{\delta_T(u)}\Delta\mathbf{y}_t\right| = O_P(\sqrt{\log T})$, and $\max_t \|\mathbf{z}_t-\widetilde{\mathbf{z}}_t(\tau_t)\|_\delta  = O(1/T)$. Thus 
		\begin{eqnarray*}
			\sup_{0\leq \tau\leq 1}\left|\mathbf{K}_{T,l}(9,21)\right| &\leq& \sup_{\tau \in[0,1],u\in[0,1]}\left|\frac{1}{\sqrt{Th}}\sum_{t=\delta_T+1}^{\delta_T(u)}\Delta\mathbf{y}_t\right|\times \sup_{-1\leq u\leq 1}|u^l K(u)|\times \frac{1}{T\sqrt{h}}\sum_{t=1}^{T}|\mathbf{z}_{t-1}-\widetilde{\mathbf{z}}_{t-1}(\tau_{t-1})|\nonumber \\
			&=&O_P(\sqrt{\log T}/(T\sqrt{h})).\nonumber
		\end{eqnarray*}
		
		For $\mathbf{K}_{T,l}(9,22)$, since $\sup_{\tau \in[0,1]}\left|\frac{1}{Th}\sum_{t=1}^{T}|\widetilde{\mathbf{z}}_{t-1}(\tau_{t-1})|\left(\frac{\tau_t-\tau}{h}\right)^l K\left(\frac{\tau_t-\tau}{h}\right)\right| = O_P(1)$, we have
		\begin{eqnarray*}
			\sup_{0\leq \tau\leq 1}\left|\mathbf{K}_{T,l}(9,22)\right| &\leq& O(1)\frac{1}{\sqrt{T}}\sum_{j=1}^{T}|\mathbf{z}_{j} - \widetilde{\mathbf{z}}_{j}(\tau_j)| \sup_{\tau \in[0,1]}\left|\frac{1}{Th}\sum_{t=1}^{T}|\widetilde{\mathbf{z}}_{t-1}(\tau_{t-1})|\left(\frac{\tau_t-\tau}{h}\right)^l K\left(\frac{\tau_t-\tau}{h}\right)\right|\nonumber\\
			&=& O_P(1/\sqrt{T}).\nonumber
		\end{eqnarray*}
		
		Similarly, we have $\sup_{0\leq \tau\leq 1}\left|\mathbf{K}_{T,l}(9,23)\right| = O_P(1/\sqrt{T})$. In addition, since $\widetilde{\mathbf{z}}_{t}(\tau) = \bm{\Psi}_{\tau}(1)\bm{\omega}(\tau)\bm{\varepsilon}_t+\bm{\Psi}_{\tau}^*(L)\bm{\omega}(\tau)\bm{\varepsilon}_{t-1}-\bm{\Psi}_{\tau}^*(L)\bm{\omega}(\tau)\bm{\varepsilon}_{t}$, let $\mathbf{p}_t(\tau) = \sum_{j=\delta_T+1}^{t} \widetilde{\mathbf{z}}_{j}(\tau_j)$, $\mathbf{w}_t = \bm{\Psi}_{\tau_{t-1}}^*(L)\bm{\omega}(\tau_{t-1})\bm{\varepsilon}_{t}$ and $\mathbf{v}_t = \bm{\Psi}_{\tau_{t-1}}(1)\bm{\omega}(\tau_{t-1})\bm{\varepsilon}_t$, we have
		\begin{eqnarray*}
			\mathbf{K}_{T,l}(9,24) &=&\frac{1}{T^{3/2}h}\sum_{t=1}^{T}\mathbf{p}_{t-1}(\tau)(\mathbf{w}_{t-2}-\mathbf{w}_{t-1})^\top\left(\frac{\tau_t-\tau}{h}\right)^l K\left(\frac{\tau_t-\tau}{h}\right)\nonumber\\
			&&+\frac{1}{T^{3/2}h}\sum_{t=1}^{T}\widetilde{\mathbf{z}}_{t-1}(\tau_{t-1})\mathbf{v}_{t-1}^\top\left(\frac{\tau_t-\tau}{h}\right)^l K\left(\frac{\tau_t-\tau}{h}\right)\nonumber\\
			&& + \frac{1}{T^{3/2}h}\sum_{t=1}^{T}\mathbf{p}_{t-2}(\tau)\mathbf{v}_{t-1}^\top\left(\frac{\tau_t-\tau}{h}\right)^l K\left(\frac{\tau_t-\tau}{h}\right)\nonumber\\
			&=&\mathbf{K}_{T,l}(9,241)+\mathbf{K}_{T,l}(9,242)+\mathbf{K}_{T,l}(9,243).\nonumber
		\end{eqnarray*}
		
		For $\mathbf{K}_{T,l}(9,241)$, write
		\begin{eqnarray*}
			\mathbf{K}_{T,l}(9,241) &=& \frac{1}{T^{3/2}h}\sum_{t=1}^{T}\mathbf{p}_{t-1}(\tau)\mathbf{w}_{t-2}^\top\left(\frac{\tau_t-\tau}{h}\right)^l K\left(\frac{\tau_t-\tau}{h}\right) \nonumber\\
			&& - \frac{1}{T^{3/2}h}\sum_{t=1}^{T}\mathbf{p}_{t-1}(\tau)\mathbf{w}_{t-1}^\top\left(\frac{\tau_t-\tau}{h}\right)^l K\left(\frac{\tau_t-\tau}{h}\right)\nonumber\\
			&=& \frac{1}{T^{3/2}h}\sum_{t=1}^{T}\widetilde{\mathbf{z}}_{t-1}(\tau_{t-1})\mathbf{w}_{t-2}^\top\left(\frac{\tau_t-\tau}{h}\right)^l K\left(\frac{\tau_t-\tau}{h}\right)\nonumber\\ 
			&&+\frac{1}{T^{3/2}h}\sum_{t=1}^{T}\mathbf{p}_{t-2}(\tau)\mathbf{w}_{t-2}^\top\left(\frac{\tau_t-\tau}{h}\right)^l K\left(\frac{\tau_t-\tau}{h}\right) \nonumber\\
			&& - \frac{1}{T^{3/2}h}\sum_{t=1}^{T}\mathbf{p}_{t-1}(\tau)\mathbf{w}_{t-1}^\top\left(\frac{\tau_t-\tau}{h}\right)^l K\left(\frac{\tau_t-\tau}{h}\right).\nonumber
		\end{eqnarray*}
		
		Note that $\sup_{\tau\in[0,1]}\left|\frac{1}{Th}\sum_{t=1}^{T}\widetilde{\mathbf{z}}_{t-1}(\tau_{t-1})\mathbf{w}_{t-2}^\top\left(\frac{\tau_t-\tau}{h}\right)^l K\left(\frac{\tau_t-\tau}{h}\right)\right|=O_P(1)$, we have
		$$
		\sup_{\tau\in[0,1]}\left|\frac{1}{T^{3/2}h}\sum_{t=1}^{T}\widetilde{\mathbf{z}}_{t-1}(\tau_{t-1})\mathbf{w}_{t-2}^\top\left(\frac{\tau_t-\tau}{h}\right)^l K\left(\frac{\tau_t-\tau}{h}\right)\right|=O_P(1/\sqrt{T}).
		$$ 
		In addition,
		\begin{eqnarray*}
			&&\frac{1}{T^{3/2}h}\sum_{t=1}^{T}\mathbf{p}_{t-2}(\tau)\mathbf{w}_{t-2}^\top\left(\frac{\tau_t-\tau}{h}\right)^l K\left(\frac{\tau_t-\tau}{h}\right) - \frac{1}{T^{3/2}h}\sum_{t=1}^{T}\mathbf{p}_{t-1}(\tau)\mathbf{w}_{t-1}^\top\left(\frac{\tau_t-\tau}{h}\right)^l K\left(\frac{\tau_t-\tau}{h}\right)\nonumber\\
			&=& - \frac{1}{T^{3/2}h} \mathbf{p}_{T-1}(\tau)\mathbf{w}_{T-1}^\top\left(\frac{1-\tau}{h}\right)^l K\left(\frac{1-\tau}{h}\right)\nonumber\\
			&& - \frac{1}{T^{3/2}h} \sum_{t=1}^{T-1}\mathbf{p}_{t-1}(\tau)\mathbf{w}_{t-1}^\top\left[\left(\frac{\tau_t-\tau}{h}\right)^l K\left(\frac{\tau_t-\tau}{h}\right)-\left(\frac{\tau_{t+1}-\tau}{h}\right)^l K\left(\frac{\tau_{t+1}-\tau}{h}\right) \right].\nonumber
		\end{eqnarray*}
		
		It is easy to see that the first term is $O_P(\sqrt{\log T}/(T\sqrt{h}))$ uniformly over $\tau \in [0,1]$ since 
		$$
		\sup_{\tau \in[0,1],u\in[0,1]}\left|\frac{1}{\sqrt{2Th}}\sum_{t=\delta_T+1}^{\delta_T(u)}\Delta\mathbf{y}_t\right| = O_P(\sqrt{\log T}).
		$$
		In addition, by Cauchy-Schwarz inequality, we have uniformly over $\tau \in [0,1]$
		\begin{eqnarray*}
			&&\frac{1}{T^{3/2}h} \sum_{t=1}^{T-1}\mathbf{p}_{t-1}(\tau)\mathbf{w}_{t-1}^\top\left[\left(\frac{\tau_t-\tau}{h}\right)^l K\left(\frac{\tau_t-\tau}{h}\right)-\left(\frac{\tau_{t+1}-\tau}{h}\right)^l K\left(\frac{\tau_{t+1}-\tau}{h}\right) \right]\nonumber\\
			&\leq &\frac{O_P(\sqrt{\log T})}{T\sqrt{h}} \left(\sum_{t=1}^{T-1}|\mathbf{w}_{t-1}^\top|^2\right)^{1/2} \left(\sum_{t=1}^{T-1}\left|\left(\frac{\tau_t-\tau}{h}\right)^l K\left(\frac{\tau_t-\tau}{h}\right)-\left(\frac{\tau_{t+1}-\tau}{h}\right)^l K\left(\frac{\tau_{t+1}-\tau}{h}\right) \right|^2\right)^{1/2}\nonumber\\
			&=& \frac{\sqrt{\log T}}{T\sqrt{h}} \times O_P(\sqrt{T})\times O(1/\sqrt{Th}) = O_P\left(\sqrt{\log T}/(Th)\right) .\nonumber
		\end{eqnarray*}
		
		As $\sup_{\tau\in[0,1]}\left|\frac{1}{Th}\sum_{t=1}^{T}\widetilde{\mathbf{z}}_{t-1}(\tau_{t-1})\mathbf{v}_{t-1}^\top\left(\frac{\tau_t-\tau}{h}\right)^l K\left(\frac{\tau_t-\tau}{h}\right)\right|=O_P(1)$, we have $$\sup_{\tau\in[0,1]}\left|\mathbf{K}_{T,l}(9,242)\right| = O_P(1/\sqrt{T}).$$ For $\mathbf{K}_{T,l}(9,243)$, note that $\left\{\frac{1}{\sqrt{Th}}\mathbf{p}_{t-2}(\tau)\mathbf{v}_{t-1}^\top\right\}_{t=1}^{T}$ is a sequence of squared integrable martingale differences, in the next we will show $\sup_{\tau \in[0,1]}|\mathbf{K}_{T,l}(9,243)| = O_P(\sqrt{\log T/T})$ and thus $\sup_{\tau \in[0,1]}|\mathbf{K}_{T,l}(9,2)| = O_P(1/\sqrt{T}) + O_P(\sqrt{\log T}/(Th)) + O_P(\sqrt{\log T/T})$.
		
		For $\mathbf{K}_{T,l}(9,243)$, let $\{S_l\}$ be a finite number of sub-intervals covering the interval $[0,1]$, which are centered at $s_l$ with the length $\xi_T = o(h^2)$. Denote the number of these intervals by $N_T$ then $N_T = O(\xi_T^{-1})$. Hence, we have 
		\begin{eqnarray*}
			&&\sup_{\tau\in[0,1]}\left|\frac{1}{Th}\sum_{t=1}^{T}\frac{1}{\sqrt{Th}}\mathbf{p}_{t-2}(\tau)\mathbf{v}_{t-1}^\top K\left(\frac{\tau_t-\tau}{h}\right)\right|\nonumber\\
			&\leq&\max_{1\leq l\leq N_T}\left|\frac{1}{Th}\sum_{t=1}^{T}\frac{1}{\sqrt{Th}}\mathbf{p}_{t-2}(s_l)\mathbf{v}_{t-1}^\top K\left(\frac{\tau_t-s_l}{h}\right)\right|\nonumber\\
			&& + \max_{1\leq l\leq N_T}\sup_{\tau\in S_l}\left|\frac{1}{Th}\sum_{t=1}^{T}\left(\frac{1}{\sqrt{Th}}\mathbf{p}_{t-2}(\tau) K\left(\frac{\tau_t-\tau}{h}\right)-\frac{1}{\sqrt{Th}}\mathbf{p}_{t-2}(s_l) K\left(\frac{\tau_t-s_l}{h}\right)\right)\mathbf{v}_{t-1}^\top\right|\nonumber\\
			&=& K_{T,l}(9,243,1) + K_{T,l}(9,243,2).\nonumber
		\end{eqnarray*}
		
		For $K_{T,l}(9,243,2)$, we have
		\begin{eqnarray*}
			&&\left|\frac{1}{Th}\sum_{t=1}^{T}\left(\frac{1}{\sqrt{Th}}\mathbf{p}_{t-2}(\tau) K\left(\frac{\tau_t-\tau}{h}\right)-\frac{1}{\sqrt{Th}}\mathbf{p}_{t-2}(s_l) K\left(\frac{\tau_t-s_l}{h}\right)\right)\mathbf{v}_{t-1}^\top\right|\nonumber\\
			&\leq&\frac{1}{Th}\sum_{t=1}^{T}\left|\frac{1}{\sqrt{Th}}\mathbf{p}_{t-2}(\tau) -\frac{1}{\sqrt{Th}}\mathbf{p}_{t-2}(s_l) \right|\cdot\left|\mathbf{v}_{t-1}\right|K\left(\frac{\tau_t-\tau}{h}\right)\nonumber\\
			&& + \frac{1}{Th}\sum_{t=1}^{T}\left|\frac{1}{\sqrt{Th}}\mathbf{p}_{t-2}(s_l)\right|\cdot \left|K\left(\frac{\tau_t-\tau}{h}\right)- K\left(\frac{\tau_t-s_l}{h}\right)\right| \left|\mathbf{v}_{t-1}\right|.\nonumber
		\end{eqnarray*}
		
		By the continuity of kernel function $K(\cdot)$, $\sup_{\tau \in[0,1],r\in[0,1]}\left|\frac{1}{\sqrt{2Th}}\sum_{t=\delta_T+1}^{\delta_T(r)}\widetilde{\mathbf{z}}_t(\tau_t)\right|=O_P(\sqrt{\log T})$ and taking $\xi_T = O(\gamma_Th^2/\sqrt{\log T})$ with $\gamma_T = \sqrt{\log T /(Th)}$, then the second term is uniformly bounded by $O(\sqrt{\log T})\frac{\xi_T}{h^2}E|\mathbf{v}_{t}| = O(\gamma_T)$. In addition, as 
		$$
		\sup_{\tau \in[0,1],|s_l-\tau|\leq\xi_T}\left|\frac{1}{\sqrt{T\xi_T}}\sum_{t=\lfloor T(\tau-h)+1 \rfloor }^{\lfloor T(s_l-h)+1 \rfloor}\widetilde{\mathbf{z}}_t(\tau_t)\right|=O_P(\sqrt{\log T}),
		$$
		the first term is uniformly $O_P(\gamma_T)$.
		
		We then apply the truncation method to prove term $K_{T,l}(9,243,1)$. Define $\bm{\eta}_t(s_l) = \frac{1}{\sqrt{Th}}\mathbf{p}_{t-1}(s_l)\mathbf{v}_{t}^\top$, $\bm{\eta}_t'(s_l) = \bm{\eta}_t(s_l) I(|\frac{1}{\sqrt{Th}}\mathbf{p}_{t-1}(s_l)|\leq (Th/\log T)^{1/4},|\mathbf{v}_t|\leq (Th/\log T)^{1/4})$ and $\bm{\eta}_t''(s_l) = \bm{\eta}_t(s_l) - \bm{\eta}_t'(s_l) $. Then, we have
		\begin{eqnarray*}
			K_{T,l}(9,243,1) &=& \max_{1\leq l\leq N_T}\left|\frac{1}{Th}\sum_{t=1}^{T}\left(\bm{\eta}_t'(s_l) + \bm{\eta}_t''(s_l) - E(\bm{\eta}_t'(s_l) + \bm{\eta}_t''(s_l) \mid \mathcal{F}_{t-1})\right) K\left(\frac{\tau_t-s_l}{h}\right)\right|\nonumber\\
			&\leq & \max_{1\leq l\leq N_T}\left|\frac{1}{Th}\sum_{t=1}^{T}\left(\bm{\eta}_t'(s_l)  - E(\bm{\eta}_t'(s_l) \mid \mathcal{F}_{t-1})\right) K\left(\frac{\tau_t-s_l}{h}\right)\right|\nonumber\\
			&&+\max_{1\leq l\leq N_T}\left|\frac{1}{Th}\sum_{t=1}^{T} \bm{\eta}_t''(s_l) K\left(\frac{\tau_t-s_l}{h}\right)\right| +\max_{1\leq l\leq N_T}\left|\frac{1}{Th}\sum_{t=1}^{T} E(\bm{\eta}_t''(s_l) \mid \mathcal{F}_{t-1}) K\left(\frac{\tau_t-s_l}{h}\right)\right|\nonumber\\
			&=& K_{T,l}(9,243,11) +K_{T,l}(9,243,12) + K_{T,l}(9,243,13).\nonumber
		\end{eqnarray*}
		
		For $K_{T,l}(9,243,12)$, write
		\begin{eqnarray*}
			K_{T,l}(9,243,12)&\leq& \max_{1\leq l\leq N_T} \frac{1}{Th}\sum_{t=1}^{T} |\frac{1}{\sqrt{Th}}\mathbf{p}_{t-1}(s_l)|\cdot|\mathbf{v}_{t}|\cdot I(|\frac{1}{\sqrt{Th}}\mathbf{p}_{t-1}(s_l)|> (Th/\log T)^{1/4})K\left(\frac{\tau_t-s_l}{h}\right) \nonumber\\
			&&+\max_{1\leq l\leq N_T} \frac{1}{Th}\sum_{t=1}^{T} \left|\frac{1}{\sqrt{Th}}\mathbf{p}_{t-1}(s_l)\right|\cdot|\mathbf{v}_{t}|\cdot I(|\mathbf{v}_t|> (Th/\log T)^{1/4})K\left(\frac{\tau_t-s_l}{h}\right).\nonumber
		\end{eqnarray*}
		
		For the first term, by Lemma \ref{L6} and $\max_{1\leq l \leq N_T , \lfloor T(s_l-h)\rfloor+1\leq t\leq \lfloor T(s_l+h)\rfloor+1} |\frac{1}{\sqrt{Th}}\mathbf{p}_{t-1}(s_l)| = O_P(\sqrt{\log T})$, we have
		\begin{eqnarray*}
			&& \max_{1\leq l\leq N_T} \frac{1}{Th}\sum_{t=1}^{T} |\frac{1}{\sqrt{Th}}\mathbf{p}_{t-1}(s_l)|\cdot|\mathbf{v}_{t}|\cdot I(|\frac{1}{\sqrt{Th}}\mathbf{p}_{t-1}(s_l)|> (Th/\log T)^{1/4})K\left(\frac{\tau_t-s_l}{h}\right)\nonumber\\
			&\leq & \sup_{-1\leq u\leq 1}|K(u)| \times \max_{1\leq l \leq N_T , \lfloor T(s_l-h)\rfloor+1\leq t\leq \lfloor T(s_l+h)\rfloor+1} |\frac{1}{\sqrt{Th}}\mathbf{p}_{t-1}(s_l)|\nonumber\\
			&& \times \sum_{l=1}^{N_T}\max_{\lfloor T(s_l-h)\rfloor+1\leq t\leq \lfloor T(s_l+h)\rfloor+1} I\left(\left|\frac{1}{\sqrt{Th}}\mathbf{p}_{t-1}(s_l)\right|> (Th/\log T)^{1/4}\right) \times \frac{1}{Th} \sum_{t=1}^{T}|\mathbf{v}_t|\nonumber\\
			&=& O_P\left( \frac{\sqrt{\log T}N_T\cdot Th }{h(\sqrt{Th} (Th/\log T)^{1/4})^{\delta} } \right) = o_P\left( \frac{\sqrt{\log T} }{\sqrt{Th}} \right),\nonumber
		\end{eqnarray*}
		if $\frac{T^{1-\frac{6}{3\delta + 4}}h}{(\log T)^{\delta/(3\delta +4)}} \to \infty$ which holds since $T^{1-2/\delta} h \to \infty$. For the second term, we have
		\begin{eqnarray*}
			&&\max_{1\leq l\leq N_T} \frac{1}{Th}\sum_{t=1}^{T} \left|\frac{1}{\sqrt{Th}}\mathbf{p}_{t-1}(s_l)\right|\cdot|\mathbf{v}_{t}|\cdot I(|\mathbf{v}_t|> (Th/\log T)^{1/4})K\left(\frac{\tau_t-s_l}{h}\right)\nonumber\\
			&\leq & O_P(\sqrt{\log T}/ (Th)) \times \sum_{t=1}^{T}|\mathbf{v}_{t}|\cdot I(|\mathbf{v}_t|> (Th/\log T)^{1/4})\nonumber\\
			&=&O_P(\sqrt{\log T}/ (Th)) \times O_P(T \times (Th/\log T)^{-(\delta-1)/4})= o_P\left( \frac{\sqrt{\log T} }{\sqrt{Th}} \right)\nonumber
		\end{eqnarray*}
		if $\frac{T^{1-\frac{4}{\delta+1}}h }{(\log T)^{(\delta-1)/(\delta+1)}} \to \infty$. Similarly, we have $K_{T,l}(9,243,13) = o_P\left( \frac{\sqrt{\log T} }{\sqrt{Th}} \right)$.
		
		In the next stage, we deal with $K_{T,l}(9,243,11)$. For any fixed $1\leq l\leq N_T$, let 
		$$\mathbf{Y}_t = \frac{1}{Th}\left(\bm{\eta}_t'(s_l)  - E(\bm{\eta}_t'(s_l) \mid \mathcal{F}_{t-1})\right) K\left(\frac{\tau_t-s_l}{h}\right),
		$$
		then $E(\mathbf{Y}_t\mid \mathcal{F}_{t-1}) = 0$ and $|\mathbf{Y}_t| \leq 2\frac{\sqrt{Th}}{\sqrt{\log T}}d_T$ with $d_T = \max_{1\leq t\leq s_l}K\left(\frac{\tau_t-s_l}{h}\right)/(Th)$. In addition, by the proof of part (3), we have
		\begin{eqnarray*}
			&&\max_{1\leq l\leq N_T}\left|\sum_{t=1}^{T}E\left(\mathbf{Y}_t\mathbf{Y}_t^\top \mid \mathcal{F}_{t=1} \right) \right| \nonumber\\
			&\leq& 4d_T\max_{1\leq l\leq N_T} \frac{1}{Th}\sum_{t=1}^{T} \left| \frac{1}{\sqrt{Th}}\mathbf{p}_{t-1}(s_l)E(\mathbf{v}_t\mathbf{v}_t^\top) \frac{1}{\sqrt{Th}}\mathbf{p}_{t-1}^\top(s_l) \right| K\left(\frac{\tau_t-s_l}{h}\right) = O_P(d_T). \nonumber
		\end{eqnarray*}
		
		Then by Proposition 2.1 in \cite{freedman1975tail}, we have
		\begin{eqnarray*}
			&&\Pr\left(K_{T,l}(9,243,11) \geq \sqrt{8M}\gamma_T \right)\nonumber\\
			&\leq& \Pr\left(K_{T,l}(9,243,11) \geq \sqrt{8M}\gamma_T,\max_{1\leq l\leq N_T}\left|\sum_{t=1}^{T}E\left(\mathbf{Y}_t\mathbf{Y}_t^\top \mid \mathcal{F}_{t=1} \right) \right|\leq \frac{M}{Th} \right)\nonumber\\
			&& + \Pr\left(\max_{1\leq l\leq N_T}\left|\sum_{t=1}^{T}E\left(\mathbf{Y}_t\mathbf{Y}_t^\top \mid \mathcal{F}_{t=1} \right) \right|> \frac{M}{Th} \right)\nonumber\\
			&\leq & N_T \exp\left( - \frac{8M\gamma_T^2}{\frac{M}{Th}} \right) + o(1) = o(1).\nonumber
		\end{eqnarray*}
		
		Then, we have
		$$
		\sup_{\tau\in[0,1]}\left|\frac{1}{Th}\sum_{t=1}^{T}\frac{1}{\sqrt{Th}}\mathbf{p}_{t-2}(\tau)\mathbf{v}_{t-1}^\top K\left(\frac{\tau_t-\tau}{h}\right)\right| = O_P(\sqrt{\log T/(Th)}).
		$$
		
		Finally, consider $\mathbf{K}_{T,l}(10)$,
		\begin{eqnarray*}
			&&\frac{1}{T^{3/2}h^{3/2}}\sum_{t=1}^{T}\bm{\xi}_{T,\perp}^\top(\tau)\bm{\alpha}_{\perp}^\top(\tau)\mathbf{y}_{t-1}\mathbf{z}_{t-1}^\top \left(\frac{\tau_t-\tau}{h}\right)^l K\left(\frac{\tau_t-\tau}{h}\right)\nonumber \\
			&=&\frac{1}{Th^{3/2}}\sum_{t=1}^{T}\bm{\xi}_{T,\perp}^\top(\tau)\left(\frac{1}{\sqrt{T}}\bm{\alpha}_{\perp}^\top(\tau)\mathbf{y}_{\delta_T}+\frac{1}{\sqrt{T}}\bm{\alpha}_{\perp}^\top(\tau)(\mathbf{y}_{\delta_T}-\mathbf{y}_{t-1}) \right)\mathbf{z}_{t-1}^\top \left(\frac{\tau_t-\tau}{h}\right)^lK\left(\frac{\tau_t-\tau}{h}\right)\nonumber\\
			&=&\frac{1}{Th^{3/2}}\bm{\xi}_{T,\perp}^\top(\tau)\frac{1}{\sqrt{T}}\bm{\alpha}_{\perp}^\top(\tau)\sum_{t=1}^{T}\sum_{j=\delta_T+1}^{t-1}\Delta \mathbf{y}_{j}\mathbf{z}_{t-1}^\top \left(\frac{\tau_t-\tau}{h}\right)^lK\left(\frac{\tau_t-\tau}{h}\right).\nonumber
		\end{eqnarray*}
		
		To complete the proof, it suffices to show
		$$
		\sup_{\tau \in [0,1]}\left|\frac{1}{Th}\sum_{t=1}^{T}\frac{1}{\sqrt{Th}}\sum_{j=\delta_T+1}^{t-1}\Delta \mathbf{y}_{j}\mathbf{z}_{t-1}^\top \left(\frac{\tau_t-\tau}{h}\right)^lK\left(\frac{\tau_t-\tau}{h}\right) \right| = O_P(\sqrt{\log T/(Th)}).
		$$
		
		Similar to the proof of term $\mathbf{K}_{T,l}(9,2)$, we have as $T h^2\rightarrow \infty$
		$$
		\sup_{\tau \in [0,1]}\left|\mathbf{K}_{T,l}(10)\right| = O_P(1/\sqrt{Th}) + O_P(\sqrt{\log T}/(Th^{3/2})) + O_P(\sqrt{\log T/(Th)}) = O_P(\sqrt{\log T/(Th)}).
		$$
		
		The proof is now completed.
		
		\medskip
		\noindent (5). The proof of part (5) is identical to that of part (4), so is omitted here.
	\end{proof}

	\begin{proof}[Proof of Lemma \ref{L8}]
		\item
		\noindent (1). By Theorem \ref{Thm1} (2), we have $\widehat{\bm{\alpha}}(\tau) - \bm{\alpha}(\tau) = O(h^2) + O_P(\sqrt{\log T /(Th)})$ uniformly over $\tau \in [0,1]$. Then, similar to the proof of Lemma \ref{L7} (4), we have
		\begin{eqnarray*}
			&&\sup_{\tau\in[0,1]}\left|\frac{1}{T^{3/2}h}\sum_{t=1}^{T}(\Delta\mathbf{x}_{t-1}\mathbf{y}_{t-1}^\top) \otimes \widehat{\bm{\alpha}}(\tau_t) K\left(\frac{\tau_t-\tau}{h}\right)\right|\nonumber\\
			&\leq & \sup_{\tau\in[0,1]}\left|\frac{1}{T^{3/2}h}\sum_{t=1}^{T}(\Delta\mathbf{x}_{t-1}\mathbf{y}_{t-1}^\top) \otimes \bm{\alpha}(\tau_t) K\left(\frac{\tau_t-\tau}{h}\right)\right|\nonumber\\
			&&+\sup_{\tau\in[0,1]}\left|\frac{1}{T^{3/2}h}\sum_{t=1}^{T}(\Delta\mathbf{x}_{t-1}\mathbf{y}_{t-1}^\top) \otimes O(h^2) K\left(\frac{\tau_t-\tau}{h}\right)\right|+ O_P\left(\sqrt{\log T /(Th)} \right)\nonumber\\
			&=&O_P\left(\sqrt{\log T /(Th)} \right).\nonumber
		\end{eqnarray*}
		
		\medskip
		\noindent (2). By part (1) and Lemma \ref{L5}, we have $\sup_{0\leq \tau \leq 1}\left|\frac{1}{T^{3/2}h}\mathbf{R}^\top(\widehat{\bm{\alpha}})\mathbf{K}(\tau)\Delta\mathbf{X}^* \right| = O_P\left(\sqrt{\log T /(Th)} \right)$ and $\sup_{0\leq \tau \leq 1}\left|\frac{1}{Th}\Delta\mathbf{X}^{*,\top}\mathbf{K}(\tau)\Delta\mathbf{X}^* \right| = O_P\left(1\right)$. Hence, by simple calculations, we have
		\begin{eqnarray*}
			&&\frac{1}{T^2}\sum_{t=1}^{T}\widetilde{\mathbf{R}}_t(\widehat{\bm{\alpha}})\widehat{\bm{\Omega}}^{-1}(\tau_t)\widetilde{\mathbf{R}}_t^\top(\widehat{\bm{\alpha}}) \nonumber\\
			&=& \frac{1}{T^2}\sum_{t=1}^{T}\mathbf{R}_t(\widehat{\bm{\alpha}})\widehat{\bm{\Omega}}^{-1}(\tau_t)\mathbf{R}_t^\top(\widehat{\bm{\alpha}}) + O_P\left(\sqrt{\log T /(Th)} \right)\nonumber\\
			&=& \frac{1}{T^2}\sum_{t=1}^{T}(\mathbf{y}_{t-1}^{(2)}\mathbf{y}_{t-1}^{(2),\top})\otimes(\widehat{\bm{\alpha}}^\top(\tau_t)\widehat{\bm{\Omega}}^{-1}(\tau_t)\widehat{\bm{\alpha}}(\tau_t)) + O_P\left(\sqrt{\log T /(Th)} \right).\nonumber
		\end{eqnarray*}
		
		Noting that uniformly over $\tau \in [0,1]$, we have $\widehat{\bm{\alpha}}(\tau) - \bm{\alpha}(\tau) = O(h^2) + O_P(\sqrt{\log T /(Th)}) $ and $\widehat{\bm{\Omega}}(\tau) - \bm{\Omega}(\tau) = O(h^2) + O_P(\sqrt{\log T /(Th)})$, by Lemma \ref{L1} and the continuous mapping theorem, we have
		\begin{eqnarray*}
			&&\frac{1}{T^2}\sum_{t=1}^{T}(\mathbf{y}_{t-1}^{(2)}\mathbf{y}_{t-1}^{(2),\top})\otimes(\widehat{\bm{\alpha}}^\top(\tau_t)\widehat{\bm{\Omega}}^{-1}(\tau_t)\widehat{\bm{\alpha}}(\tau_t))\nonumber\\
			&=& \frac{1}{T^2}\sum_{t=1}^{T}(\mathbf{y}_{t-1}^{(2)}\mathbf{y}_{t-1}^{(2),\top})\otimes(\bm{\alpha}^\top(\tau_t)\bm{\Omega}^{-1}(\tau_t)\bm{\alpha}(\tau_t)) +o_P(1)\nonumber\\
			&\to_D&\int_{0}^{1}\mathbf{W}_{d-r_0}(u) \mathbf{W}_{d-r_0}^\top(u) \otimes \bm{\alpha}^\top(u)\bm{\Omega}^{-1}(u)\bm{\alpha}(u)\mathrm{d}u.\nonumber
		\end{eqnarray*}
		
		\medskip

		\noindent (3). Write
		\begin{eqnarray*}
			&&\frac{1}{T}\sum_{t=1}^{T}\widetilde{\mathbf{R}}_t(\widehat{\bm{\alpha}})\widehat{\bm{\Omega}}^{-1}(\tau_t)\mathbf{u}_t\nonumber\\
			&=& \frac{1}{T}\sum_{t=1}^{T}\mathbf{R}_t(\widehat{\bm{\alpha}})\widehat{\bm{\Omega}}^{-1}(\tau_t)\mathbf{u}_t -\frac{1}{T}\sum_{t=1}^{T}\Delta \mathbf{X}_{t-1}^\top(\Delta \mathbf{X}^\top\mathbf{K}(\tau_t)\Delta \mathbf{X})^{-1}\Delta \mathbf{X}^\top\mathbf{K}(\tau_t)\mathbf{R}(\widehat{\bm{\alpha}})\widehat{\bm{\Omega}}^{-1}(\tau_t)\mathbf{u}_t\nonumber\\
			&=& \mathbf{K}_{T,17} + \mathbf{K}_{T,18}.\nonumber
		\end{eqnarray*}
		
		Consider $\mathbf{K}_{T,17}$,
		\begin{eqnarray*}
			\mathbf{K}_{T,17} &=& \frac{1}{T}\sum_{t=1}^{T} \mathbf{y}_{t-1}^{(2)}\otimes \bm{\alpha}^\top(\tau_t)\bm{\Omega}^{-1}(\tau_t)\mathbf{u}_t +\frac{1}{T}\sum_{t=1}^{T} \mathbf{y}_{t-1}^{(2)}\otimes \bm{\alpha}^\top(\tau_t)\left(\widehat{\bm{\Omega}}^{-1}(\tau_t)-\bm{\Omega}^{-1}(\tau_t)\right)\mathbf{u}_t  \nonumber\\ 
			&& + \frac{1}{T}\sum_{t=1}^{T} \mathbf{y}_{t-1}^{(2)}\otimes (\widehat{\bm{\alpha}}^\top(\tau_t)-\bm{\alpha}^\top(\tau_t))\left(\widehat{\bm{\Omega}}^{-1}(\tau_t)-\bm{\Omega}^{-1}(\tau_t)\right)\mathbf{u}_t\nonumber\\
			&& + \frac{1}{T}\sum_{t=1}^{T} \mathbf{y}_{t-1}^{(2)}\otimes (\widehat{\bm{\alpha}}^\top(\tau_t)-\bm{\alpha}^\top(\tau_t))\bm{\Omega}^{-1}(\tau_t)\mathbf{u}_t\nonumber\\
			&=&\mathbf{K}_{T,17,1}+\mathbf{K}_{T,17,2}+\mathbf{K}_{T,17,3}+\mathbf{K}_{T,17,4}.\nonumber
		\end{eqnarray*}
		
		Consider $\mathbf{K}_{T,17,1}$ first. Note that $\mathbf{y}_{t-1} = \mathbf{P}_{\bm{\beta}_\perp} \sum_{j=1}^{t-1}\mathbf{z}_j + \mathbf{P}_{\bm{\beta}}\mathbf{z}_{t-1} + \mathbf{P}_{\bm{\beta}_\perp} \mathbf{y}_0$, since $\mathbf{y}_0 = O_P(1)$ and $\mathbf{z}_{t}$ is a time-varying VAR process, we have
		$$
		\frac{1}{T}\sum_{t=1}^{T} \mathbf{z}_{t-1} \otimes \bm{\alpha}^\top(\tau_t)\bm{\Omega}^{-1}(\tau_t)\mathbf{u}_t =O_P(1/\sqrt{T})
		$$
		and
		$$
		\frac{1}{T}\sum_{t=1}^{T} \mathbf{y}_{0} \otimes \bm{\alpha}^\top(\tau_t)\bm{\Omega}^{-1}(\tau_t)\mathbf{u}_t =O_P(1/\sqrt{T}).
		$$
		Similarly, by using BN decomposition and the proof of Lemma \ref{L1}, we have
		\begin{eqnarray*}
			\mathbf{K}_{T,17,1} &=& \frac{1}{T}\sum_{t=1}^{T}\left(\left[\mathbf{0}_{(d-r_0)\times r_0},\mathbf{I}_{d-r_0}\right]\mathbf{P}_{\bm{\beta}_\perp} \sum_{s=1}^{t-1}\bm{\Psi}_{\tau_s}(1)\mathbf{u}_s\right) \otimes \bm{\alpha}^\top(\tau_t)\bm{\Omega}^{-1}(\tau_t)\mathbf{u}_t + O_P(1/\sqrt{T})\nonumber\\
			&\to_D& \int_{0}^{1}\mathbf{W}_{d-r_0}(u)\otimes \mathrm{d}\mathbf{W}_{r_0}(u).\nonumber
		\end{eqnarray*}
		
		In addition, by Lemma \ref{L1}, we have 
		$$
		\mathbf{P}_{\bm{\beta}_\perp} \bm{\Psi}_{\tau_t}(1)\mathbf{u}_t = \bm{\beta}_{\perp} [\bm{\alpha}_{\perp}^\top(\tau_t)\bm{\Gamma}_{\tau_t}(1) \bm{\beta}_{\perp}]^{-1} \bm{\alpha}_{\perp}^\top(\tau_t)\mathbf{u}_t
		$$
		and thus
		\begin{eqnarray*}
			&& E\left( \mathbf{P}_{\bm{\beta}_\perp} \bm{\Psi}_{\tau_t}(1)\mathbf{u}_t\mathbf{u}_t^\top\bm{\Omega}^{-1}(\tau_t) \bm{\alpha}(\tau_t)\right) \nonumber\\
			&= &\bm{\beta}_{\perp} [\bm{\alpha}_{\perp}^\top(\tau_t)\bm{\Gamma}_{\tau_t}(1) \bm{\beta}_{\perp}]^{-1} \bm{\alpha}_{\perp}^\top(\tau_t)\bm{\alpha}(\tau_t) =0,\nonumber
		\end{eqnarray*}
		which implies that $\mathbf{W}_{d-r_0}(\cdot)$ and $\mathbf{W}_{r_0}(\cdot)$ are mutually independent Brownian motions.
		
		For $\mathbf{K}_{T,17,2}$, we have
		\begin{eqnarray*}
			&& \frac{1}{T}\sum_{t=1}^{T} \mathbf{y}_{t-1}^{(2)}\otimes \bm{\alpha}^\top(\tau_t)\left(\widehat{\bm{\Omega}}^{-1}(\tau_t)-\bm{\Omega}^{-1}(\tau_t)\right)\mathbf{u}_t\nonumber\\
			&=&  \frac{1}{T}\sum_{t=1}^{T} \mathbf{y}_{t-1}^{(2)}\otimes \bm{\alpha}^\top(\tau_t)\bm{\Omega}^{-1}(\tau_t)\left(\bm{\Omega}(\tau_t)-\widehat{\bm{\Omega}}(\tau_t)\right)\widehat{\bm{\Omega}}^{-1}(\tau_t)\mathbf{u}_t\nonumber\\
			&=& \frac{1}{T}\sum_{t=1}^{T} \mathbf{y}_{t-1}^{(2)}\otimes \bm{\alpha}^\top(\tau_t)\bm{\Omega}^{-1}(\tau_t)\left(\bm{\Omega}(\tau_t)-\widehat{\bm{\Omega}}(\tau_t)\right)\bm{\Omega}^{-1}(\tau_t)\mathbf{u}_t\nonumber\\
			&& + \frac{1}{T}\sum_{t=1}^{T} \mathbf{y}_{t-1}^{(2)}\otimes \bm{\alpha}^\top(\tau_t)\bm{\Omega}^{-1}(\tau_t)\left(\bm{\Omega}(\tau_t)-\widehat{\bm{\Omega}}(\tau_t)\right)\left(\widehat{\bm{\Omega}}^{-1}(\tau_t)-\bm{\Omega}^{-1}(\tau_t)\right)\mathbf{u}_t\nonumber
		\end{eqnarray*}
		
		By $\sup_{\tau \in [0,1]}\left|\widehat{\bm{\Omega}}(\tau)-\bm{\Omega}(\tau) \right| = O_P(h^2+\sqrt{\log T/(Th)})$, the second term is bounded by
		$$
		\sqrt{T}O_P(h^2+\sqrt{\log T/(Th)})^2 = o_P(1)
		$$
		if $Th^8 \to 0$. For the first term, write
		\begin{eqnarray*}
			&&\frac{1}{T}\sum_{t=1}^{T} \mathbf{y}_{t-1}^{(2)}\otimes \bm{\alpha}^\top(\tau_t)\bm{\Omega}^{-1}(\tau_t)\left(\bm{\Omega}(\tau_t)-\widehat{\bm{\Omega}}(\tau_t)\right)\bm{\Omega}^{-1}(\tau_t)\mathbf{u}_t\nonumber\\
			&=& \frac{1}{T}\sum_{t=1}^{T} \mathbf{y}_{t-1}^{(2)}\otimes \bm{\alpha}^\top(\tau_t)\bm{\Omega}^{-1}(\tau_t)\left(\bm{\Omega}(\tau_t)-\frac{1}{T}\sum_{s=1}^{T}\mathbf{u}_s\mathbf{u}_s^\top w_s(\tau_t)\right)\bm{\Omega}^{-1}(\tau_t)\mathbf{u}_t\nonumber\\
			&& + \frac{1}{T}\sum_{t=1}^{T} \mathbf{y}_{t-1}^{(2)}\otimes \bm{\alpha}^\top(\tau_t)\bm{\Omega}^{-1}(\tau_t)\left(\frac{1}{T}\sum_{s=1}^{T}(\mathbf{u}_s-\widehat{\mathbf{u}}_s)\mathbf{u}_s^\top w_s(\tau_t)\right)\bm{\Omega}^{-1}(\tau_t)\mathbf{u}_t\nonumber\\
			&& + \frac{1}{T}\sum_{t=1}^{T} \mathbf{y}_{t-1}^{(2)}\otimes \bm{\alpha}^\top(\tau_t)\bm{\Omega}^{-1}(\tau_t)\left(\frac{1}{T}\sum_{s=1}^{T}\mathbf{u}_s(\mathbf{u}_s-\widehat{\mathbf{u}}_s)^\top w_s(\tau_t)\right)\bm{\Omega}^{-1}(\tau_t)\mathbf{u}_t\nonumber\\
			&& + \frac{1}{T}\sum_{t=1}^{T} \mathbf{y}_{t-1}^{(2)}\otimes \bm{\alpha}^\top(\tau_t)\bm{\Omega}^{-1}(\tau_t)\left(\frac{1}{T}\sum_{s=1}^{T}(\widehat{\mathbf{u}}_s-\mathbf{u}_s)(\widehat{\mathbf{u}}_s-\mathbf{u}_s)^\top w_s(\tau_t)\right)\bm{\Omega}^{-1}(\tau_t)\mathbf{u}_t\nonumber\\
			&=&\mathbf{J}_1 + \mathbf{J}_2 + \mathbf{J}_3 + \mathbf{J}_4,\nonumber
		\end{eqnarray*}
		For $\mathbf{J}_1$, by $\sup_{\tau\in[0,1]}\left|\frac{1}{T}\sum_{s=1}^{T}E(\mathbf{u}_s\mathbf{u}_s^\top)w_s(\tau)-\bm{\Omega}(\tau)\right| = o(1)$, we have
		\begin{eqnarray*}
			\|\mathbf{J}_1\|_{\delta/2} &\leq& \left\|\frac{1}{T}\sum_{t=1}^{T} \mathbf{y}_{t-1}^{(2)}\otimes \bm{\alpha}^\top(\tau_t)\bm{\Omega}^{-1}(\tau_t)\left(\frac{1}{T}\sum_{s=1}^{T}[E(\mathbf{u}_s\mathbf{u}_s^\top)-\mathbf{u}_s\mathbf{u}_s^\top] w_s(\tau_t)\right)\bm{\Omega}^{-1}(\tau_t)\mathbf{u}_t\right\|_{\delta/2}+o(1)\nonumber\\
			&\leq& \left\|\frac{1}{T}\sum_{t=1}^{T} \mathbf{y}_{t-1}^{(2)}\otimes \bm{\alpha}^\top(\tau_t)\bm{\Omega}^{-1}(\tau_t)\left(\frac{1}{T}\sum_{s=1}^{t-1}[E(\mathbf{u}_s\mathbf{u}_s^\top)-\mathbf{u}_s\mathbf{u}_s^\top] w_s(\tau_t)\right)\bm{\Omega}^{-1}(\tau_t)\mathbf{u}_t\right\|_{\delta/2}\nonumber\\
			&& + \left\|\frac{1}{T}\sum_{t=1}^{T} \mathbf{y}_{t-1}^{(2)}\otimes \bm{\alpha}^\top(\tau_t)\bm{\Omega}^{-1}(\tau_t)\left(\frac{1}{T}\sum_{s=t+1}^{T}[E(\mathbf{u}_s\mathbf{u}_s^\top)-\mathbf{u}_s\mathbf{u}_s^\top] w_s(\tau_t)\right)\bm{\Omega}^{-1}(\tau_t)\mathbf{u}_t\right\|_{\delta/2} + o(1)\nonumber\\
			&=&\mathbf{J}_{11} + \mathbf{J}_{12}.\nonumber
		\end{eqnarray*}
		For $\mathbf{J}_{11}$, by using Burkholder inequality for martingale differences, we have
		\begin{eqnarray*}
			\mathbf{J}_{11} &\leq& \frac{1}{T}\left(\sum_{t=1}^{T}\left\| \mathbf{y}_{t-1}^{(2)}\otimes \bm{\alpha}^\top(\tau_t)\bm{\Omega}^{-1}(\tau_t)\left(\frac{1}{T}\sum_{s=1}^{t-1}[E(\mathbf{u}_s\mathbf{u}_s^\top)-\mathbf{u}_s\mathbf{u}_s^\top] k_s(\tau_t)\right)\bm{\Omega}^{-1}(\tau_t)\mathbf{u}_t\right\|_{\delta/2}\right)^{1/2}\nonumber\\
			&\leq & O(1)\frac{1}{T}\left(\sum_{t=1}^{T}\left\| \mathbf{u}_t\mathbf{y}_{t-1}^{(2),\top}\right\|_\delta\left\|\frac{1}{T}\sum_{s=1}^{t-1}[E(\mathbf{u}_s\mathbf{u}_s^\top)-\mathbf{u}_s\mathbf{u}_s^\top] k_s(\tau_t)\right\|_{\delta}\right)^{1/2}\nonumber\\
			&=& O\left(T^{-1} (T\times \sqrt{T} \times 1/{\sqrt{Th}})^{1/2}\right) = o(1).\nonumber
		\end{eqnarray*}
		Similarly, we have $\mathbf{J}_{12} = o(1)$ and thus $\mathbf{J}_1 = o_P(1)$. In addition, by using uniform convergence result of $[\widehat{\bm{\Pi}}(\tau), \widehat{\bm{\Gamma}}(\tau)]$ and  $Th^8 \to 0$, we have $\mathbf{J}_4 = o_P(1)$.
		
		Next, consider $\mathbf{J}_2$ and $\mathbf{J}_3$. By the proof of Theorem \ref{Thm1} and Lemmas \ref{L5}--\ref{L7}, we have uniformly over $\tau \in [h,1-h]$,
		\begin{eqnarray*}
			&& [\widehat{\bm{\Pi}}(\tau)-\bm{\Pi}(\tau), \widehat{\bm{\Gamma}}(\tau)-\bm{\Gamma}(\tau)]\nonumber\\
			&=&\frac{1}{2}h^2\widetilde{c}_2[\bm{\alpha}^{(2)}(\tau)\bm{\beta}^\top,\bm{\Gamma}^{(2)}(\tau)]+ O(h^3) +O_P\left( h^2 \sqrt{\log T/(Th)}\right)\nonumber\\
			&&+ \frac{1}{Th}\sum_{t=1}^{T}\mathbf{u}_{t}\mathbf{w}_{t-1}^\top\bm{\Sigma}_{\mathbf{w}}^{-1}(\tau)  K\left(\frac{\tau_t-\tau}{h}\right) \left[ \begin{matrix}
				\bm{\beta}^\top & \mathbf{0}_{r_0\times d(p_0-1)}\nonumber\\
				\mathbf{0}_{d(p_0-1)\times r_0}& \mathbf{I}_{d(p_0-1)}\nonumber\\
			\end{matrix}\right] + O_P\left(\log T /(Th) \right).\nonumber
		\end{eqnarray*}
		In addition, for $1\leq t \leq \lfloor Th \rfloor +1$ and $\lfloor T(1-h) \rfloor +1\leq t \leq T$, since
		$$
		\sup_{\tau \in [0,1]}\left|[\widehat{\bm{\Pi}}(\tau), \widehat{\bm{\Gamma}}(\tau)] - [\bm{\Pi}(\tau), \bm{\Gamma}(\tau)]\right| = O_P\left(h^2 + \sqrt{\log T/(Th)} \right),
		$$
		we have
		\begin{eqnarray*}
			&&\frac{1}{T}\sum_{1\leq t \leq \lfloor Th \rfloor +1} \mathbf{y}_{t-1}^{(2)}\otimes \bm{\alpha}^\top(\tau_t)\bm{\Omega}^{-1}(\tau_t)\left(\frac{1}{T}\sum_{s=1}^{T}(\mathbf{u}_s-\widehat{\mathbf{u}}_s)\mathbf{u}_s^\top w_s(\tau_t)\right)\bm{\Omega}^{-1}(\tau_t)\mathbf{u}_t \nonumber\\
			&=& \frac{1}{\sqrt{T}} Th O_P\left(h^2 + \sqrt{\frac{\log T}{Th}}\right)=o_P(1)
		\end{eqnarray*}
		if $Th^6\to 0$, which implies that
		\begin{eqnarray*}
			\mathbf{J}_{2} &=& \frac{1}{T}\sum_{t=\lfloor Th\rfloor }^{\lfloor T(1-h)+1\rfloor} \mathbf{y}_{t-1}^{(2)}\otimes \bm{\alpha}^\top(\tau_t)\bm{\Omega}^{-1}(\tau_t)\left(\frac{1}{T}\sum_{s=1}^{T}(\mathbf{u}_s-\widehat{\mathbf{u}}_s)\mathbf{u}_s^\top w_s(\tau_t)\right)\bm{\Omega}^{-1}(\tau_t)\mathbf{u}_t + o_P(1)\nonumber\\
			&=& \frac{1}{T}\sum_{t=\lfloor Th\rfloor }^{\lfloor T(1-h)+1\rfloor} \mathbf{y}_{t-1}^{(2)}\otimes \bm{\alpha}^\top(\tau_t)\bm{\Omega}^{-1}(\tau_t)\left(\frac{1}{T}\sum_{s=1}^{T}\frac{1}{2}h^2\widetilde{c}_2[\bm{\alpha}^{(2)}(\tau_s),\bm{\Gamma}^{(2)}(\tau_s)]\mathbf{w}_{s-1}\mathbf{u}_s^\top w_s(\tau_t)\right)\bm{\Omega}^{-1}(\tau_t)\mathbf{u}_t\nonumber\\
			&& + \frac{1}{T}\sum_{t=\lfloor Th\rfloor }^{\lfloor T(1-h)+1\rfloor} \mathbf{y}_{t-1}^{(2)}\otimes \bm{\alpha}^\top(\tau_t)\bm{\Omega}^{-1}(\tau_t)\nonumber\\
			&&\times \left(\frac{1}{T}\sum_{s=1}^{T}\frac{1}{Th}\sum_{v=1}^{T}\mathbf{u}_{v}\mathbf{w}_{v-1}^\top\bm{\Sigma}_{\mathbf{w}}^{-1}(\tau_s)  K\left(\frac{\tau_v-\tau_s}{h}\right)\mathbf{w}_{s-1}\mathbf{u}_s^\top w_s(\tau_t)\right)\bm{\Omega}^{-1}(\tau_t)\mathbf{u}_t+o_P(1).\nonumber
		\end{eqnarray*}
		
		Then, by using Burkholder inequality for martingale differences again, we can prove that $\mathbf{J}_2 =o_P(1)$ and $\mathbf{J}_3 = o_P(1)$. Combining the above results, we have proved $\mathbf{K}_{T,17,2} = o_P(1)$.
		
		Using similar arguments, we can prove $\mathbf{K}_{T,17,3} = o_P(1)$, $\mathbf{K}_{T,17,4} = o_P(1)$ and $\mathbf{K}_{T,18}= o_P(1)$.
		
		The proof of part (3) is now completed.
		
		\medskip
		
		\noindent (4). By $\sup_{\tau\in[0,1]}\left|\frac{1}{T^{3/2}h}\Delta \mathbf{X}^\top\mathbf{K}(\tau)\mathbf{R}(\widehat{\bm{\alpha}})\right| = O_P(\sqrt{\log T/(Th)})$ and $\sup_{\tau\in[0,1]}\left|\bm{\alpha}(\tau_t) - \widehat{\bm{\alpha}}(\tau_t)\right| = O_P(h^2+\sqrt{\log T/(Th)})$, if $Th^8\to0$, we have
		\begin{eqnarray*}
			&&\frac{1}{T}\sum_{t=1}^{T}\widetilde{\mathbf{R}}_t(\widehat{\bm{\alpha}})\widehat{\bm{\Omega}}^{-1}(\tau_t)\left(\bm{\alpha}(\tau_t) - \widehat{\bm{\alpha}}(\tau_t)\right)\bm{\beta}^{\top}\mathbf{y}_{t-1}\nonumber \\
			&=& \frac{1}{T}\sum_{t=1}^{T}\mathbf{R}_t(\widehat{\bm{\alpha}})\widehat{\bm{\Omega}}^{-1}(\tau_t)\left(\bm{\alpha}(\tau_t)- \widehat{\bm{\alpha}}(\tau_t)\right)\bm{\beta}^{\top}\mathbf{z}_{t-1} \nonumber\\
			&& -\frac{1}{T}\sum_{t=1}^{T}\mathbf{R}^\top(\bm{\alpha})\mathbf{K}(\tau_t)\Delta \mathbf{X}^{*}\left(\Delta \mathbf{X}^{*,\top}\mathbf{K}(\tau_t) \Delta \mathbf{X}^*\right)^{-1}\left[\mathbf{I}_{d^2(p_0-1)},\mathbf{0}_{d^2(p_0-1)} \right]^\top \Delta\mathbf{X}_{t-1}\nonumber\\
			&&\times \widehat{\bm{\Omega}}^{-1}(\tau_t)\left(\bm{\alpha}(\tau_t) - \widehat{\bm{\alpha}}(\tau_t)\right)\bm{\beta}^{\top}\mathbf{z}_{t-1}\nonumber\\
			&=& \frac{1}{T}\sum_{t=1}^{T}\left(\mathbf{y}_{t-1}^{(2)}\otimes \bm{\alpha}^\top(\tau_t)\right)\bm{\Omega}^{-1}(\tau_t)\left(\bm{\alpha}(\tau_t)- \widehat{\bm{\alpha}}(\tau_t)\right)\bm{\beta}^{\top}\mathbf{z}_{t-1} + o_P(1).\nonumber
		\end{eqnarray*}
		
		Similar to the proof of part (3), we have 
		$$
		\frac{1}{T}\sum_{t=1}^{T}\left(\mathbf{y}_{t-1}^{(2)}\otimes \bm{\alpha}^\top(\tau_t)\right)\bm{\Omega}^{-1}(\tau_t)\left(\bm{\alpha}(\tau_t)- \widehat{\bm{\alpha}}(\tau_t)\right)\bm{\beta}^{\top}\mathbf{z}_{t-1} = o_P(1).
		$$
		The proof of part (4) is now completed.
		
		\medskip
		
		\noindent (5). For $\widehat{\bm{\Gamma}}(\tau,\widehat{\bm{\alpha}}\bm{\beta}^\top)-\bm{\Gamma}(\tau)$, we have
		\begin{eqnarray*}
			&&\widehat{\bm{\Gamma}}(\tau,\widehat{\bm{\alpha}}\bm{\beta}^\top)-\bm{\Gamma}(\tau) \nonumber\\
			&=& \sum_{t=1}^{T}\bm{\Gamma}(\tau_t)\Delta\mathbf{x}_{t-1}\Delta\mathbf{x}_{t-1}^{*,\top} K\left(\frac{\tau_t-\tau}{h}\right) \left(\sum_{t=1}^{T}\Delta\mathbf{x}_{t-1}^*\Delta\mathbf{x}_{t-1}^{*,\top} K\left(\frac{\tau_t-\tau}{h}\right)\right)^{-1}\left[\mathbf{I}_{d(p_0-1)},\mathbf{0}_{d(p_0-1)}\right]^\top-\bm{\Gamma}(\tau)\nonumber\\
			&& + \sum_{t=1}^{T}(\bm{\alpha}(\tau_t)-\widehat{\bm{\alpha}}(\tau_t))\bm{\beta}^\top\mathbf{y}_{t-1}\Delta\mathbf{x}_{t-1}^{*,\top} K\left(\frac{\tau_t-\tau}{h}\right) \left(\sum_{t=1}^{T}\Delta\mathbf{x}_{t-1}^*\Delta\mathbf{x}_{t-1}^{*,\top} K\left(\frac{\tau_t-\tau}{h}\right)\right)^{-1}\left[\mathbf{I}_{d(p_0-1)},\mathbf{0}_{d(p_0-1)}\right]^\top\nonumber\\
			&& + \sum_{t=1}^{T}\mathbf{u}_{t}\Delta\mathbf{x}_{t-1}^{*,\top} K\left(\frac{\tau_t-\tau}{h}\right) \left(\sum_{t=1}^{T}\Delta\mathbf{x}_{t-1}^*\Delta\mathbf{x}_{t-1}^{*,\top} K\left(\frac{\tau_t-\tau}{h}\right)\right)^{-1}\left[\mathbf{I}_{d(p_0-1)},\mathbf{0}_{d(p_0-1)}\right]^\top\nonumber\\
			&=&\mathbf{K}_{T,19} + \mathbf{K}_{T,20} + \mathbf{K}_{T,21}.\nonumber
		\end{eqnarray*}
		
		By standard arguments for local linear estimation method and the uniform convergence results in Lemmas \ref{L5} and \ref{L7}, we have
		$$
		\sup_{\tau\in[0,1]}\left|\widehat{\bm{\Gamma}}(\tau,\widehat{\bm{\alpha}}\bm{\beta}^\top)-\bm{\Gamma}(\tau)\right| = O_P\left(h^2 + \sqrt{\frac{\log T}{Th}}\right).
		$$ 
		
		In addition, as $\sup_{\tau\in[0,1]}\left|\frac{1}{T^{3/2}h}\mathbf{R}^\top(\widehat{\bm{\alpha}})\mathbf{K}(\tau)\Delta\mathbf{X}^*\right| = O_P\left(\sqrt{\log T /(Th)}\right) $, \\$\sup_{\tau\in[0,1]}\left|\widehat{\bm{\Omega}}(\tau) - \bm{\Omega}(\tau)\right| = O_P\left(h^2+\sqrt{\log T /(Th)}\right)$, $Th^6\to 0$ and $\frac{Th^2}{(\log T)^2}\to \infty$, we have
		\begin{eqnarray*}
			&&\frac{1}{T}\sum_{t=1}^{T}\widetilde{\mathbf{R}}_t(\widehat{\bm{\alpha}})\widehat{\bm{\Omega}}^{-1}(\tau_t)\left(\bm{\Gamma}(\tau_t) - \widehat{\bm{\Gamma}}(\tau_t,\widehat{\bm{\alpha}}\bm{\beta}^\top)\right)\Delta \mathbf{x}_{t-1}\nonumber\\
			&=&\frac{1}{T}\sum_{t=1}^{T}\mathbf{y}_{t-1}^{(2)}\otimes \bm{\alpha}^\top(\tau_t)\bm{\Omega}^{-1}(\tau_t)\left(\bm{\Gamma}(\tau_t) - \widehat{\bm{\Gamma}}(\tau_t,\widehat{\bm{\alpha}}\bm{\beta}^\top)\right)\Delta \mathbf{x}_{t-1} + o_P(1)\nonumber\\
			&=&\frac{1}{T}\sum_{t=\lfloor Th \rfloor}^{\lfloor T(1-h) \rfloor}\mathbf{y}_{t-1}^{(2)}\otimes \bm{\alpha}^\top(\tau_t)\bm{\Omega}^{-1}(\tau_t)\left(\bm{\Gamma}(\tau_t) - \widehat{\bm{\Gamma}}(\tau_t,\widehat{\bm{\alpha}}\bm{\beta}^\top)\right)\Delta \mathbf{x}_{t-1} + o_P(1).\nonumber
		\end{eqnarray*}

		Again, by standard arguments for local linear estimation method, uniformly over $\tau\in [h,1-h]$, we have
		\begin{eqnarray*}
			\mathbf{K}_{T,19} = \frac{1}{2}h^2\widetilde{c}_2\bm{\Gamma}^{(2)}(\tau) + o(h^2) + h^2O_P\left(h^2 + \sqrt{\frac{\log T}{Th}}\right).\nonumber
		\end{eqnarray*}
		
		Similarly to the proof of Lemma \ref{L4} (3), we have
		$$
		\frac{1}{T}\sum_{t=\lfloor Th \rfloor}^{\lfloor T(1-h) \rfloor}\mathbf{y}_{t-1}^{(2)}\otimes \bm{\alpha}^\top(\tau_t)\bm{\Omega}^{-1}(\tau_t)\bm{\Gamma}^{(2)}(\tau_t)\Delta \mathbf{x}_{t-1} = O_P(1).
		$$
		Hence, we have
		$$
		\frac{1}{T}\sum_{t=\lfloor Th \rfloor}^{\lfloor T(1-h) \rfloor}\mathbf{y}_{t-1}^{(2)}\otimes \bm{\alpha}^\top(\tau_t)\bm{\Omega}^{-1}(\tau_t)\times \mathbf{K}_{T,19}\times \Delta \mathbf{x}_{t-1} = o_P(1).
		$$
		
		In addition, similarly to the proof of part (4), by using Burkholder inequality for martingale differences, we can prove that 
		$$
		\frac{1}{T}\sum_{t=\lfloor Th \rfloor}^{\lfloor T(1-h) \rfloor}\mathbf{y}_{t-1}^{(2)}\otimes \bm{\alpha}^\top(\tau_t)\bm{\Omega}^{-1}(\tau_t)\times \mathbf{K}_{T,20}\times \Delta \mathbf{x}_{t-1} = o_P(1)
		$$
		and 
		$$
		\frac{1}{T}\sum_{t=\lfloor Th \rfloor}^{\lfloor T(1-h) \rfloor}\mathbf{y}_{t-1}^{(2)}\otimes \bm{\alpha}^\top(\tau_t)\bm{\Omega}^{-1}(\tau_t)\times \mathbf{K}_{T,21}\times \Delta \mathbf{x}_{t-1} = o_P(1).
		$$ 
		
		Note that in these two cases, we can use the uniform convergence results in Lemma \ref{L5} to replace the inverse of the sample covariance matrix with its population version. The proof is now completed.
	\end{proof}
	
	\medskip
	
	Define $\bm{\Gamma}_{p}(\tau)=[\bm{\Gamma}_{p,1}(\tau),\ldots,\bm{\Gamma}_{p,p-1}(\tau)]$, where $\bm{\Gamma}_{p,j}(\tau)=\bm{\Gamma}_{j}(\tau)$ for $1\leq j \leq p_0-1$ and $\bm{\Gamma}_{p,j}(\tau)=\mathbf{0}$ for $j > p_0-1$. Let $\mathbf{h}_{p,t-1}=[\mathbf{y}_{t-1}^\top,\Delta\mathbf{x}_{p,t-1}^\top]^\top$, $\Delta\mathbf{x}_{p,t-1}=[\Delta\mathbf{y}_{t-1}^\top,\ldots,\Delta\mathbf{y}_{t-p-1}^\top]^\top$, $\mathbf{h}_{p,t-1}^* = \left[\mathbf{h}_{p,t-1}^\top, \frac{\tau_t-\tau}{h}\mathbf{h}_{p,t-1}^\top \right]^\top$,
	$$
	\mathbf{M}_p(\tau_t) = [\bm{\Pi}(\tau_t),\bm{\Gamma}_p(\tau_t)]-[\bm{\Pi}(\tau),\bm{\Gamma}_p(\tau)]-[\bm{\Pi}^{(1)}(\tau),\bm{\Gamma}_p^{(1)}(\tau)](\tau_t-\tau)-\frac{1}{2}[\bm{\Pi}^{(2)}(\tau),\bm{\Gamma}_p^{(2)}(\tau)](\tau_t-\tau)^2,
	$$
	$\bm{\Gamma}_{\overline{p}}(\tau)=[\bm{\Gamma}_{p,p}(\tau),\ldots,\mathbf{\Gamma}_{p,P-1}(\tau)]$ and $\Delta\mathbf{x}_{\overline{p},t-1}=[\Delta\mathbf{x}_{t-p-2}^\top,\ldots,\Delta\mathbf{x}_{t-P-1}^\top]^\top$.
	
	\begin{proof}[Proof of Lemma \ref{L9}]
		\item
		\noindent (1). Since $p\geq p_0$, we have $\widehat{\mathbf{u}}_{p,t} = \mathbf{u}_t+\left[\bm{\Pi}(\tau_t) - \widehat{\bm{\Pi}}(\tau_t) ,\bm{\Gamma}_p(\tau_t)-\widehat{\bm{\Gamma}}_p(\tau_t)\right]\mathbf{h}_{p,t-1}$ and
		\begin{eqnarray*}
			&&\text{RSS}(p)\nonumber\\
			&=&\frac{1}{T}\sum_{t=1}^{T}\mathbf{u}_t^\top\mathbf{u}_t-2\frac{1}{T}\sum_{t=1}^{T}\left(\mathbf{u}_t-\widehat{\mathbf{u}}_{p,t}\right)^\top\mathbf{u}_t  \nonumber\\
			&&+\frac{1}{T}\sum_{t=1}^{T}\mathbf{h}_{p,t-1}^\top\left[\bm{\Pi}(\tau_t) - \widehat{\bm{\Pi}}(\tau_t) ,\bm{\Gamma}_p(\tau_t)-\widehat{\bm{\Gamma}}_p(\tau_t)\right]^\top\left[\bm{\Pi}(\tau_t) - \widehat{\bm{\Pi}}(\tau_t) ,\bm{\Gamma}_p(\tau_t)-\widehat{\bm{\Gamma}}_p(\tau_t)\right] \mathbf{h}_{p,t-1}\nonumber\\
			&=&  \frac{1}{T}\sum_{t=1}^{T}\mathbf{u}_t^\top\mathbf{u}_t+ K_{T,22}+K_{T,23}.\nonumber
		\end{eqnarray*}
		
		For $K_{T,23}$, by Theorem 1 (ii), we have $K_{T,23} = O_P\left( c_T^2 \right)$ with $c_T = h^2+(\log T/(Th))^{1/2}$.
		
		For $K_{T,22}$, write
		\begin{eqnarray*}
			&& \frac{1}{T}\sum_{t=1}^{T} \mathbf{u}_t\left(\mathbf{u}_t-\widehat{\mathbf{u}}_{p,t}\right)^\top = \frac{1}{T}\sum_{t=1}^{T} \mathbf{u}_t\mathbf{h}_{p,t-1}^\top\left[\widehat{\bm{\Pi}}(\tau_t) - \bm{\Pi}(\tau_t) ,\widehat{\bm{\Gamma}}_p(\tau_t)-\bm{\Gamma}_p(\tau_t)\right]^\top\nonumber\\
			&=&\frac{1}{2} h^2 \cdot \frac{1}{T}\sum_{t=1}^{T}\mathbf{u}_t[\mathbf{h}_{p,t-1}^\top,\mathbf{0}_{d^2p\times 1}^\top] \mathbf{S}_{T}^{+}(\tau_t) \left(\begin{matrix}
				\mathbf{S}_{T,2}(\tau_t) \\
				\mathbf{S}_{T,3}(\tau_t)
			\end{matrix}\right)\left[\bm{\alpha}^{(2)}(\tau_t)\bm{\beta}^\top ,\bm{\Gamma}_p^{(2)}(\tau_t)\right]^\top\nonumber\\
			&&+\frac{1}{T}\sum_{t=1}^{T}\mathbf{u}_t[\mathbf{h}_{p,t-1}^\top,\mathbf{0}_{d^2p\times 1}^\top]\mathbf{S}_{T}^{+}(\tau_t) \left(\sum_{s=1}^{T} \mathbf{h}_{p,s-1}^*\mathbf{h}_{p,s-1}^\top  \mathbf{M}_p^\top(\tau_s)K\left(\frac{\tau_s-\tau_t}{h}\right) \right)\nonumber\\
			&&+\frac{1}{T}\sum_{t=1}^{T}\mathbf{u}_t[\mathbf{h}_{p,t-1}^\top,\mathbf{0}_{d^2p\times 1}^\top] \mathbf{S}_{T}^{+}(\tau_t) \left(\sum_{s=1}^{T}\mathbf{h}_{p,s-1}^*\mathbf{u}_s^\top K\left(\frac{\tau_s-\tau_t}{h}\right)\right)\nonumber\\
			&=& K_{T,22,1}+K_{T,22,2}+K_{T,22,3},\nonumber
		\end{eqnarray*}
		where $\mathbf{S}_{T,l}(\tau) = \sum_{t=1}^{T}\mathbf{h}_{p,t-1}\mathbf{h}_{p,t-1}^\top \left(\frac{\tau_t-\tau}{h}\right)^lK\left(\frac{\tau_t-\tau}{h}\right)$ and
		$$
		\mathbf{S}_T(\tau) = \left[\begin{matrix}
			\mathbf{S}_{T,0}(\tau) & \mathbf{S}_{T,1}(\tau)\\ \mathbf{S}_{T,1}(\tau) & \mathbf{S}_{T,2}(\tau)
		\end{matrix} \right].
		$$
		
		By the uniform convergence results stated in Lemmas \ref{L5} and \ref{L7} and using similar arguments as those in the proof of Theorem \ref{Thm1}, we replace the weighed sample covariance with its limit plus the rate $O_P\left((\log T/(Th))^{1/2}\right)$, and hence
		\begin{eqnarray*}
			|K_{T,22,1}|+|K_{T,22,2}|=O_P\left(T^{-\frac{1}{2}}h^2+h^2(\log T/(Th))^{1/2}\right).\nonumber
		\end{eqnarray*}
		
		Meanwhile, by Lemmas \ref{L5} and \ref{L7} and the proof of Theorem \ref{Thm1}, we have
		\begin{eqnarray*}
			|K_{T,22,3}| &=& \frac{1}{T}\sum_{t=1}^{T} \mathbf{u}_t\mathbf{w}_{p,t-1}^\top \bm{\Sigma}_{\mathbf{w}}^{-1}(\tau_t) \left(\frac{1}{Th}\sum_{s=1}^{T}\mathbf{w}_{p,s-1}\mathbf{u}_s^\top K\left(\frac{\tau_s-\tau_t}{h}\right)\right)\nonumber\\
			&&+O_P\left((Th)^{-1/2}\cdot(h^2+\sqrt{\log T/(Th)})\right),\nonumber
		\end{eqnarray*}
		where $\mathbf{w}_{p,t-1} = [\mathbf{y}_{t-1}^\top\bm{\beta},\Delta\mathbf{x}_{p,t-1}^\top]^\top$. Then by using Burkholder inequality for martingale differences, we have
		\begin{eqnarray*}
			&&\frac{1}{T}\sum_{t=1}^{T} \mathbf{u}_t\mathbf{w}_{p,t-1}^\top \bm{\Sigma}_{\mathbf{w}}^{-1}(\tau_t) \left(\frac{1}{Th}\sum_{s=1}^{T}\mathbf{w}_{p,s-1}\mathbf{u}_s^\top K\left(\frac{\tau_s-\tau_t}{h}\right)\right)\nonumber\\
			&=&2\frac{1}{T}\sum_{t=1}^{T} \mathbf{u}_t\mathbf{w}_{p,t-1}^\top \bm{\Sigma}_{\mathbf{w}}^{-1}(\tau_t) \left(\frac{1}{Th}\sum_{s=1}^{t-1}\mathbf{w}_{p,s-1}\mathbf{u}_s^\top K\left(\frac{\tau_s-\tau_t}{h}\right)\right) + O_P(1/(Th))\nonumber\\
			&=&O_P(1/(T\sqrt{h})+1/(Th)).\nonumber
		\end{eqnarray*}
		
		Combining the above results, we have proved part (1).
		
		\medskip
		
		\noindent (2).  For $p < p_0$, we have $[\widehat{\bm{\Pi}}(\tau),\widehat{\bm{\Gamma}}_p(\tau)]-[\bm{\Pi}(\tau),\bm{\Gamma}_{p}(\tau)]\mathbf{Q}^{*,-1}(\tau) = \mathbf{B}_p(\tau)+o_P(1)$ uniformly over $\tau \in [0,1]$, where $\mathbf{B}_p(\tau)$ is a nonrandom bias term. Since 
		$$
		\widehat{\mathbf{u}}_{p,t}=\mathbf{u}_{t}+\left[\bm{\Pi}(\tau_t) - \widehat{\bm{\Pi}}(\tau_t) ,\bm{\Gamma}_p(\tau_t)-\widehat{\bm{\Gamma}}_p(\tau_t)\right]\mathbf{Q}^{*,-1}(\tau)\mathbf{w}_{p,t-1}+\bm{\Gamma}_{\overline{p}}(\tau_t)\Delta\mathbf{x}_{\overline{p},t-1},
		$$
		we have
		$$
		\text{RSS}(p)= \frac{1}{T}\sum_{t=1}^{T} \mathbf{u}_{t}^\top\mathbf{u}_{t}+\frac{1}{T}\sum_{t=1}^{T}\mathrm{tr}\left( [\mathbf{B}_p(\tau_t),\bm{\Gamma}_{\overline{p}}(\tau_t)]E\left(\mathbf{w}_{P,t-1}\mathbf{w}_{P,t-1}^\top\right)[\mathbf{B}_p(\tau_t),\bm{\Gamma}_{\overline{p}}(\tau_t)]^\top \right)+o_P(1).
		$$
		Since $[\mathbf{B}_p(\tau_t),\bm{\Gamma}_{\overline{p}}(\tau_t)]\neq 0$ and $E\left(\mathbf{w}_{P,t-1}\mathbf{w}_{P,t-1}^\top\right)$ is positively definite, the result follows.
	\end{proof}
	
	\begin{proof}[Proof of Lemma \ref{L10}]
		\item
		Noting that $\{\mathbf{u}_s^\top\mathbf{W}_{s-1}^\top\mathbf{H}_s\sum_{v=1}^{s-1}\mathbf{W}_{v-1}\mathbf{u}_v w_{s,v}\}_{s=2}^{T}$ is a sequence of martingale differences, we apply the martingale central limit theory to prove the asymptotic normality of $\widetilde{U}$. 
		
		We first verify the Lindeberg condition. Since $\{\bm{\varepsilon}_t\}$ is a sequence of independent random variables, we have for $\delta > 4$
		$$
		\|\mathbf{y}_t^*\|_{\delta} \leq \|\mathbf{W}_{t-1}\|_\delta \|\mathbf{u}_t\|_\delta =O(1)
		$$
		with $\mathbf{y}_t^* := \mathbf{W}_{t-1}\mathbf{u}_t$. Let $\mathbf{h}_{t-1}^* = \sum_{v=1}^{t-1}\mathbf{y}_v^*w_{t,v}$. By Cauchy-Schwarz inequality, we have
		\begin{eqnarray*}
			\sum_{t=2}^T \|\mathbf{y}_t^{*,\top}\mathbf{H}_s \mathbf{h}_{t-1}^*\|_{\delta/2}^{\delta/2} \leq O(1)\sum_{t=2}^{T} \|\mathbf{y}_t^{*}\|_{\delta}^{\delta/2} \|\mathbf{h}_{t-1}^*\|_{\delta}^{\delta/2}.\nonumber
		\end{eqnarray*}
		
		Then, in order to verify the Lindeberg condition, it suffices to show $\sum_{t=2}^{T}\|\mathbf{h}_{t-1}^*\|_{\delta}^{\delta/2} = o(1)$. Note that $\{\mathbf{y}_t\}$ is a sequence of martingale differences, then by Burkholder inequality, we have
		$$
		\|\mathbf{h}_{t-1}^*\|_{\delta}^{\delta} \leq O(1) \left\|\sum_{v=1}^{t-1}|\mathbf{y}_v^*w_{t,v}|^2 \right\|_{\delta/2}^{\delta/2} \leq \left(\sum_{v=1}^{t-1}\|\mathbf{y}_v^*w_{t,v}\|_{\delta}^2 \right)^{\delta/2} = O\left((\sum_{v=1}^{t-1}|w_{t,v}|^2)^{\delta/2}\right).
		$$
		
		Hence, we have $\sum_{t=2}^T \|\mathbf{y}_t^{*,\top}\mathbf{H}_t \mathbf{h}_{t-1}^*\|_{\delta/2}^{\delta/2} =  O\left(\sum_{t=2}^{T}(\sum_{v=1}^{t-1}|w_{t,v}|^2)^{\delta/4}\right)$. Let $a_t = \sum_{v=1}^{t-1}|w_{t,v}|^2$. Since 
		\begin{eqnarray*}
			\max_t|a_t| &=& \max_t\left|\sum_{i=1}^{t-1}\frac{1}{T^2h}\left[\int_{-1}^{1}K\left(u\right)K\left(u+\frac{i}{Th}\right)\mathrm{d}u\right]^2\right| \nonumber\\
			&\leq& \frac{1}{T^2h}\sum_{i=1}^{T}\left[\int_{-1}^{1}K\left(u\right)K\left(u+\frac{i}{Th}\right)\mathrm{d}u\right]^2\nonumber \\
			&=& \frac{1}{T}\int_{0}^{\infty}\left[\int_{-1}^{1}K\left(u\right)K\left(u+v\right)\mathrm{d}u\right]^2 \mathrm{d}v(1+o(1)) = O(1/T),\nonumber
		\end{eqnarray*}
		we have $\sum_{t=2}^T \|\mathbf{y}_t^{*,\top}\mathbf{H}_t \mathbf{h}_{t-1}^*\|_{\delta/2}^{\delta/2} =  O\left(\sum_{t=2}^{T}(\sum_{v=1}^{t-1}|w_{t,v}|^2)^{\delta/4}\right) = \max_t|a_t|^{\delta/4-1}=o(1)$ as $\delta > 4$.
		
		We next prove the convergence of conditional variance of $\widetilde{U}$, i.e., 
		$$
		\sum_{t=2}^{T}\mathrm{tr}\left(E(\mathbf{H}_t^\top\mathbf{y}_t^{*}\mathbf{y}_t^{*,\top}\mathbf{H}_t \mid  \mathcal{F}_{t-1})\mathbf{h}_{t-1}^*\mathbf{h}_{t-1}^{*,\top}\right) \to_P  \sum_{t=2}^{T}\mathrm{tr}\left(E(\mathbf{H}_t^\top\mathbf{y}_t^{*}\mathbf{y}_t^{*,\top}\mathbf{H}_t)E(\mathbf{h}_{t-1}^*\mathbf{h}_{t-1}^{*,\top})\right).
		$$
		We will prove this convergence result by showing
		$$
		\left\|\sum_{t=2}^{T}\mathrm{tr}\left(\left[E(\mathbf{H}_t^\top\mathbf{y}_t^{*}\mathbf{y}_t^{*,\top}\mathbf{H}_t \mid  \mathcal{F}_{t-1})-E(\mathbf{H}_t^\top\mathbf{y}_t^{*}\mathbf{y}_t^{*,\top}\mathbf{H}_t)\right]\mathbf{h}_{t-1}^*\mathbf{h}_{t-1}^{*,\top}\right)\right\|_1 = o(1)
		$$
		and
		$$
		\left\|\sum_{t=2}^{T}\mathrm{tr}\left(E(\mathbf{H}_t^\top\mathbf{y}_t^{*}\mathbf{y}_t^{*,\top}\mathbf{H}_t)\left[\mathbf{h}_{t-1}^*\mathbf{h}_{t-1}^{*,\top}-E(\mathbf{h}_{t-1}^*\mathbf{h}_{t-1}^{*,\top})\right]\right)\right\|_1 = o(1),
		$$
		respectively.
		
		Consider $\left\|\sum_{t=2}^{T}\mathrm{tr}\left(\left[E(\mathbf{H}_t^\top\mathbf{y}_t^{*}\mathbf{y}_t^{*,\top}\mathbf{H}_t \mid  \mathcal{F}_{t-1})-E(\mathbf{H}_t^\top\mathbf{y}_t^{*}\mathbf{y}_t^{*,\top}\mathbf{H}_t)\right]\mathbf{h}_{t-1}^*\mathbf{h}_{t-1}^{*,\top}\right)\right\|_1$. 
		
		Let $\mathbf{w}_t^* = E(\mathbf{H}_t^\top\mathbf{y}_t^{*}\mathbf{y}_t^{*,\top}\mathbf{H}_t \mid  \mathcal{F}_{t-1}) - E(\mathbf{H}_t^\top\mathbf{y}_t^{*}\mathbf{y}_t^{*,\top}\mathbf{H}_t)$ and thus $E(\mathbf{w}_t^*) = 0$. For any integer $I \geq 1$ introduce the truncated process $\mathbf{h}_{t-1,I}^* = E\left(\mathbf{h}_{t-1}^*|\mathcal{F}_{t-I}\right)$. Then $\mathbf{h}_{t-1,I}^* = 0$ if $t\leq I$ and $\mathbf{h}_{t-1,I}^* = \sum_{s=1}^{t-I} w_{s,t}\mathbf{y}_s$ for $1\leq I < t$. For $2 \leq t\leq T$, by using Burkholder inequality again,
		\begin{equation*}
			\|\mathbf{h}_{t-1,I}^*-\mathbf{h}_{t-1}^*\|_\delta^2 \leq M \max_t \|\mathbf{y}_t\|_\delta^2 \sum_{s=\max(1,t-I+1)}^{t-1}w_{s,t}^2 = O\left(\sum_{s=\max(1,t-I+1)}^{t-1}w_{s,t}^2\right).
		\end{equation*}
		
		Let $L(I) = \sum_{J=1}^{I}l(J)$ with $l(J) = \sum_{s=1}^{T-J}w_{s,s+J}^2$, $V(I)=\sum_{t=2}^{T}\mathrm{tr}\left[\mathbf{w}_t^* \mathbf{h}_{t-1,I}^*\mathbf{h}_{t-1,I}^{*,\top}\right]$ and 
		$$T(I) = \sum_{t=2}^{T}\mathrm{tr}\left[E(\mathbf{w}_t^*|\mathcal{F}_{t-I})\mathbf{h}_{t-1,I}^*\mathbf{h}_{t-1,I}^{*,\top}\right].
		$$
		Note that for any fixed integer $J$, we have
		\begin{eqnarray*}
			\sum_{s=1}^{T-J}w_{s,s+J}^2 &=& \sum_{s=1}^{T-J}\frac{1}{T^2h}\left[\int_{-1}^{1}K\left(u\right)K\left(u+\frac{J}{Th}\right)\mathrm{d}u\right]^2\nonumber \\
			&=& \frac{T-J}{T^2h}\left[\int_{-1}^{1}K\left(u\right)K\left(u+\frac{J}{Th}\right)\mathrm{d}u\right]^2 = O(1/(Th)).\nonumber
		\end{eqnarray*}
		In addition, by Cauchy--Schwarz inequality, if $I/(Th) \to 0$, we have
		\begin{eqnarray*}
			\|V(1)-V(I)\|_1 &\leq& \sum_{t=2}^{T} \left\|\mathrm{tr}\left[\mathbf{w}_t^*\left(\mathbf{h}_{t-1}^*\mathbf{h}_{t-1}^{*,\top}-\mathbf{h}_{t-1,I}^*\mathbf{h}_{t-1,I}^{*,\top}\right)\right]\right\|_1 \nonumber\\
			&\leq& \sum_{t=2}^{T}\|\mathbf{w}_t^*\|_2 \|\mathbf{h}_{t-1}^*-\mathbf{h}_{t-1,I}^*\|_4 \|\mathbf{h}_{t-1}^*+\mathbf{h}_{t-1,I}^*\|_4 \nonumber\\
			&\leq& M\sum_{t=2}^{T}\|\mathbf{h}_{t-1}^*-\mathbf{h}_{t-1,I}^*\|_4 a_t^{1/2}\nonumber\\
			&\leq& M\left\{\sum_{t=2}^{T}\|\mathbf{h}_{t-1}^*-\mathbf{h}_{t-1,I}^*\|_4^2\right\}^{1/2} \left\{\sum_{t=2}^{T}a_t\right\}^{1/2}\nonumber\\
			&=&O(1)[L(I)]^{1/2} \to 0,\nonumber
		\end{eqnarray*}
		since
		\begin{eqnarray*}
			\|\mathbf{h}_{t-1}^*+\mathbf{h}_{t-1,I}^*\|_4 &\leq& \left\{M\sum_{s=1}^{t-1}\|w_{s,t}\mathbf{y}_{s}\|_4^2\right\}^{1/2}+\left\{M\sum_{s=1}^{t-I}\|w_{s,t}\mathbf{y}_{s}\|_4^2\right\}^{1/2}\nonumber\\
			&=&O\left(\left\{\sum_{s=1}^{t-1}w_{s,t}^2\right\}^{1/2} \right)=O(a_t^{1/2}).\nonumber
		\end{eqnarray*}
		
		Define the projection operator $\mathcal{P}_t\bm{\xi} = E\left(\bm{\xi}\mid\mathcal{F}_{t}\right)-E\left(\bm{\xi}\mid\mathcal{F}_{t-1}\right)$. For $0\leq j\leq I-1$, let $U(j) = \sum_{t=2}^{T}\mathrm{tr}\left[\left(\mathcal{P}_{t-j} \mathbf{w}_t^* \right)\mathbf{h}_{t-1,I}^*\mathbf{h}_{t-1,I}^{*,\top}\right]$, then
		$$
		V(I)-T(I)= \sum_{t=2}^{T}\mathrm{tr}\left[\left(\sum_{j=0}^{I-1}\mathcal{P}_{t-j} \mathbf{w}_t^*\right) \mathbf{h}_{t-1,I}^*\mathbf{h}_{t-1,I}^{*,\top}\right] = \sum_{j=0}^{I-1}U(j).
		$$
		Note that $\left\{\left(\mathcal{P}_{t-j} \mathbf{w}_t^*\right) \mathbf{h}_{t-1,I}^*\mathbf{h}_{t-1,I}^{*,\top}\right\}_{t=2}^{T}$ forms a martingale difference sequence since
		$$
		E\left\{\left(\mathcal{P}_{t-j} \mathbf{w}_t^*\right) \mathbf{h}_{t-1,I}^*\mathbf{h}_{t-1,I}^{*,\top}|\mathcal{F}_{t-j-1}\right\} =\left[ E( \mathbf{w}_t^*|\mathcal{F}_{t-j-1})-E( \mathbf{w}_t^*|\mathcal{F}_{t-j-1})\right]\mathbf{h}_{t-1,I}^*\mathbf{h}_{t-1,I}^{*,\top}=0.
		$$
		
		By Cauchy--Schwarz inequality, since $\|\mathcal{P}_{t-j}\mathbf{w}_t^*\|_{\delta/2} \leq 2 \|\mathbf{w}_t^*\|_{\delta/2} < \infty$,
		\begin{eqnarray*}
			\|U(j)\|_{\delta/4}^{\delta/4}&\leq& M \sum_{t=2}^{T}\|\left(\mathcal{P}_{t-j} \mathbf{w}_t^*\right) \mathbf{h}_{t-1,I}^*\mathbf{h}_{t-1,I}^{*,\top}\|_{\delta/4}^{\delta/4} \leq M \sum_{t=2}^{T}\|\mathbf{h}_{t-1,I}^*\|_\delta^{\delta/2}\nonumber\\
			&\leq&M \sum_{t=2}^{T}a_t^{\delta/4} \leq M \max_t a_t^{\delta/4-1} \sum_{t=2}^{T}a_t = O\left(T^{1-\delta/4}\right).\nonumber
		\end{eqnarray*}
		In addition, by $\|V(1)-V(I)\|_1 \to 0$,
		\begin{eqnarray*}
			\|V(1)\|_1 &\leq& \|V(I)-T(I)\|_{\delta/4}+E|T(I)|+o(1)\nonumber \\
			&\leq& \sum_{j=0}^{I-1}\|U(j)\|_{\delta/4} + \max_t\|E(\mathbf{w}_t^*|\mathcal{F}_{t-I})\|_2 \sum_{t=2}^{T} \|\mathbf{h}_{t-1,I}^*\|_4^2 = o(1),\nonumber
		\end{eqnarray*}
		since $\max_t\|E(\mathbf{w}_t^*|\mathcal{F}_{t-I})\|_2 \to 0$  as $I \to \infty$. We then have proved 
		$$
		\left\|\sum_{t=2}^{T}\mathrm{tr}\left(\left[E(\mathbf{H}_t^\top\mathbf{y}_t^{*}\mathbf{y}_t^{*,\top}\mathbf{H}_t \mid  \mathcal{F}_{t-1})-E(\mathbf{H}_t^\top\mathbf{y}_t^{*}\mathbf{y}_t^{*,\top}\mathbf{H}_t)\right]\mathbf{h}_{t-1}^*\mathbf{h}_{t-1}^{*,\top}\right)\right\|_1 = o(1).
		$$
		
		Consider $\left\|\sum_{t=2}^{T}\mathrm{tr}\left(E(\mathbf{H}_t^\top\mathbf{y}_t^{*}\mathbf{y}_t^{*,\top}\mathbf{H}_t)\left[\mathbf{h}_{t-1}^*\mathbf{h}_{t-1}^{*,\top}-E(\mathbf{h}_{t-1}^*\mathbf{h}_{t-1}^{*,\top})\right]\right)\right\|_1$. For notational simplicity, we omit $E(\mathbf{H}_t^\top\mathbf{y}_t^{*}\mathbf{y}_t^{*,\top}\mathbf{H}_t)$ and write
		\begin{eqnarray*}
			&&\sum_{t=2}^{T}\mathrm{tr}\left[\mathbf{h}_{t-1}^*\mathbf{h}_{t-1}^{*,\top}-E\left(\mathbf{h}_{t-1}^*\mathbf{h}_{t-1}^{*,\top}\right)\right] \nonumber\\
			&=&\sum_{t=2}^{T}\sum_{s=1}^{t-1}\mathrm{tr}\left[\left(\mathbf{W}_{s-1}\mathbf{u}_s\mathbf{u}_{s}^\top\mathbf{W}_{s-1}^\top -E\left(\mathbf{W}_{s-1}\mathbf{u}_s\mathbf{u}_{s}^\top\mathbf{W}_{s-1}^\top\right)\right)w_{s,t}^2\right]\nonumber\\
			&&+2\sum_{t=3}^{T}\sum_{s_1=2}^{t-1}\sum_{s_2=1}^{s_1-1}\mathrm{tr}\left[\mathbf{W}_{s_1-1}\mathbf{u}_{s_1}\mathbf{u}_{s_2}\mathbf{W}_{s_2-1}^\top w_{s_1,t}w_{s_2,t}\right]\nonumber\\
			&=&|K_{T,24}| + 2|K_{T,25}|.\nonumber
		\end{eqnarray*}
		
		Consider $|K_{T,24}|$. Write
		\begin{eqnarray*}
			K_{T,24} &=& \sum_{t=2}^{T}\sum_{s=1}^{t-1}\mathrm{tr}\left[\left(\mathbf{u}_s\mathbf{u}_{s}^\top-\bm{\Omega}(\tau_s)\right)\mathbf{W}_{s-1}^\top\mathbf{W}_{s-1}\right]w_{s,t}^2\nonumber \\ 
			&&+\sum_{t=2}^{T}\sum_{s=1}^{t-1}\mathrm{tr}\left[\bm{\Omega}(\tau_s)\left(\mathbf{W}_{s-1}^\top\mathbf{W}_{s-1}-E\left(\mathbf{W}_{s-1}^\top\mathbf{W}_{s-1}\right)\right)\right]w_{s,t}^2 \nonumber\\
			&=&\frac{1}{T}\sum_{s=1}^{T-1}\mathrm{tr}\left[\left(\mathbf{u}_s\mathbf{u}_{s}^\top-\bm{\Omega}(\tau_s)\right)\mathbf{W}_{s-1}^\top\mathbf{W}_{s-1}\right]\left(T\sum_{t=s+1}^{T}w_{s,t}^2\right)\nonumber\\ 
			&&+\frac{1}{T}\sum_{s=1}^{T-1}\mathrm{tr}\left[\bm{\Omega}(\tau_s)\left(\mathbf{W}_{s-1}^\top\mathbf{W}_{s-1}-E\left(\mathbf{W}_{s-1}^\top\mathbf{W}_{s-1}\right)\right)\right]\left(T\sum_{t=s+1}^{T}w_{s,t}^2\right)  \nonumber\\
			&=&K_{T,24,1}+K_{T,24,2}.\nonumber
		\end{eqnarray*}
		
		Since $\mathrm{tr}\left[\left(\mathbf{u}_s\mathbf{u}_{s}^\top-\bm{\Omega}(\tau_s)\right)\mathbf{W}_{s-1}^\top\mathbf{W}_{s-1}\right]$, $s=1,2,\ldots$ are a martingale difference sequence and $T\sum_{t=1}^{T}w_{s,t}^2=O(1)$, we have $K_{T,24,1} = o_P(1)$. In addition, similarly to the proof of Lemma \ref{L4} (the convergence of the estimator of the sample covariance of time-varying VMA($\infty$) processes), we have $K_{T,24,2}= o_P(1)$.
		
		Next, consider $K_{T,25}$. By Cauchy--Schwarz inequality and Burkholder inequality,
		\begin{eqnarray*}
			\|K_{T,25}\|_2^2 &\leq&M \sum_{s_1=2}^{T-1}\|\mathrm{tr}\left[\mathbf{W}_{s_1-1}\bm{\eta}_{s_1}\sum_{s_2=1}^{s_1-1}\mathbf{u}_{s_2}\mathbf{W}_{s_2-1}^\top \sum_{t=s_1+1}^{T}w_{s_1,t}w_{s_2,t}\right]\|_2^2\nonumber \\
			&\leq& M \sum_{s_1=2}^{T-1}\|\mathbf{W}_{s_1-1}\mathbf{u}_{s_1}\|_4^2\|\sum_{s_2=1}^{s_1-1}\mathbf{W}_{s_2-1}\mathbf{u}_{s_2} \sum_{t=s_1+1}^{T}w_{s_1,t}w_{s_2,t}\|_4^2 \nonumber\\
			&\leq& M \sum_{s_1=2}^{T-1}\|\mathbf{W}_{s_1-1}\mathbf{u}_{s_1}\|_4^2\sum_{s_2=1}^{s_1-1}\|\mathbf{W}_{s_2-1}\mathbf{u}_{s_2} \sum_{t=s_1+1}^{T}w_{s_1,t}w_{s_2,t}\|_4^2 \nonumber\\
			&=&O\left(\sum_{s_1=2}^{T-1}\sum_{s_2=1}^{s_1-1}\left(\sum_{t=s_1+1}^{T}w_{s_1,t}w_{s_2,t}\right)^2 \right)=O(1/T)\nonumber
		\end{eqnarray*}
		provided that
		\begin{eqnarray*}
			&& \sum_{k=1}^{T-1}\sum_{t=1}^{k-1}\left[\sum_{j=k+1}^{T}w_{k,j}w_{t,j}\right]^2 \nonumber\\
			&\leq& \frac{1}{T^4h^2}\sum_{k=1}^{T-1}\sum_{t=1}^{k-1}\left(\sum_{j=1+k}^{T}\left[\int_{-1}^{1}K\left(u\right)K\left(u+\frac{j-k}{Th}\right)\mathrm{d}u\right]^2\right)\nonumber \\
			&&\times\left(\sum_{j=1+k}^{T}\left[\int_{-1}^{1}K\left(u\right)K\left(u+\frac{j-t}{Th}\right)\mathrm{d}u\right]^2\right)\nonumber\\
			&\leq&\frac{M}{T^3h}\sum_{k=1}^{T-1}\sum_{t=1}^{k-1}\sum_{j=1+k}^{T}\left[\int_{-1}^{1}K\left(u\right)K\left(u+\frac{j-t}{Th}\right)\mathrm{d}u\right]^2 \nonumber \\
			&\leq& \frac{M}{T^3h}\sum_{k=1}^{T-1}\sum_{j=2}^{T}\sum_{t=1}^{j-1}\left[\int_{-1}^{1}K\left(u\right)K\left(u+\frac{j-t}{Th}\right)\mathrm{d}u\right]^2 =O(1/T).\nonumber
		\end{eqnarray*}
		
		Combining the above results, the proof is now completed.
	\end{proof}

\end{document}